\documentclass[11pt]{article}
\pdfoutput=1
\usepackage[T1]{fontenc}
\usepackage[utf8]{inputenc}
\usepackage{lmodern}
\usepackage[a4paper,margin=28mm]{geometry}
\usepackage{microtype}
\usepackage{amsmath,amssymb,amsfonts,amsthm,mathtools}
\usepackage{graphicx}
\usepackage{booktabs,longtable,multirow,array}
\usepackage{braket}
\usepackage{dsfont}
\usepackage{xcolor}
\usepackage{tikz}
\usetikzlibrary{arrows.meta,decorations.pathmorphing}
\usetikzlibrary{calc,patterns}
\usepackage{authblk}
\usepackage{fancyhdr}
\usepackage[numbers,sort&compress]{natbib}
\usepackage{doi}
\usepackage{hyperref}
\hypersetup{
	colorlinks=true,
	linkcolor=blue, 
	citecolor=red, 
	urlcolor=blue   
}

\usetikzlibrary{decorations.markings}

\allowdisplaybreaks
\newcommand{\abs}[1]{\left\lvert#1\right\rvert}
\DeclareRobustCommand{\id}{\mathds{1}}
\newcommand{\eps}{\varepsilon}
\newcommand\Bell{{\boldsymbol\ell}}
\newcommand{\sixj}[6]{%
	\begin{Bmatrix}
		#1&#2&#3\\
		#4&#5&#6
\end{Bmatrix}}

\newtheorem{theorem}{Theorem}[section]
\newtheorem{lemma}[theorem]{Lemma}

\newtheorem{conjecture}[theorem]{Conjecture}
\theoremstyle{definition}

\theoremstyle{remark}
\newtheorem{remark}[theorem]{Remark}

\begin{document}

\fancypagestyle{titlepage}{
  \fancyhf{} 
  \rhead{\small RIKEN-iTHEMS-Report-26} 
  \renewcommand{\headrulewidth}{0pt} 
}
	
	\title{Bulk OPE Coefficients of the $E$-Series Virasoro Minimal Models}
	\author[1,2]{Amaury Lhoste}
	\author[1,5]{Jiaxin Qiao}
	\author[1,3,4,5,6]{Masahito Yamazaki}
	\affil[1]{Kavli Institute for the Physics and Mathematics of the Universe,
		UTIAS, University of Tokyo, Kashiwa, Chiba 277-8583, Japan}
	\affil[2]{École normale supérieure, Université PSL,
		45 rue d'Ulm, 75005 Paris, France}
	\affil[3]{Department of Physics, Graduate School of Science, University of Tokyo, Hongo 7-3-1, Bunkyo-ku, Tokyo 113-0033, Japan}
	\affil[4]{Trans-Scale Quantum Science Institute, University of Tokyo, Hongo 7-3-1, Bunkyo-ku, Tokyo 113-0033, Japan}
    \affil[5]{Center for Data-Driven Discovery (CD3), Kavli IPMU, University of Tokyo, Chiba 277-8583, Japan}
	\affil[6]{Center for Interdisciplinary Theoretical and Mathematical Sciences,
		RIKEN, Saitama 351-0198, Japan}
	\date{}
    
	\maketitle
    \thispagestyle{titlepage}
	
	\begin{abstract}
		We determine candidate exact bulk OPE coefficients for all the \(E_6\), \(E_7\), and \(E_8\) Virasoro minimal models from high-precision numerical bootstrap calculations. The results support existence and uniqueness up to field redefinitions. An analytic extension theorem reduces these questions to a finite bootstrap problem, for which we provide explicit algebraic data and exhaustive numerical checks. The remaining finite consistency and uniqueness claims are conjectural. Our calculation assumes only Virasoro symmetry, bulk locality (permutation symmetry) and OPE associativity (crossing equations). 
	\end{abstract}
		
	\tableofcontents
	\vspace{1em}
	
	\newpage
	\section{Introduction}
	
	Virasoro minimal models are among the few interacting conformal field theories
	for which one may determine the complete local operator algebra
	exactly~\cite{Belavin:1984vu}.  Modular invariance, nonnegative integral
	multiplicities, and a unique vacuum classify their torus spectra by pairs of
	ADE Dynkin diagrams~\cite{Cappelli:1986hf,Cappelli:1987xt,Kato:1987td}.  The modular
	invariant specifies which left- and right-moving Virasoro representations
	occur, but it does not by itself determine the local theory: one needs to find the operator product expansion (OPE) coefficients that satisfy a collection of self-consistency conditions.
	
	The purpose of this paper is to give explicit bulk OPE coefficients for every Virasoro minimal model with an exceptional modular invariant ($E$ series).  For all coprime pairs \((p,q)\), with \(p,q>1\) and
	\begin{equation}
		q=12,\ 18,\ 30
		\qquad\text{for}\qquad
		E_6,\ E_7,\ E_8,
		\label{eq:q-values}
	\end{equation}
	respectively, we give expressions for all bulk-primary OPE coefficients in one fixed normalization.  The result includes both the unitary and nonunitary theories in the three families. Many partial results on Virasoro minimal models already exist in the literature. We briefly review what was known and what is new here.
	
	\subsection{What it means to solve an exceptional minimal model}
	
	Several logically different achievements are often described by saying that a
	rational CFT is ``solved''.  It is useful to distinguish them here.
	\begin{center}
		\begin{tabular}{p{0.25\textwidth}p{0.67\textwidth}}
			\toprule
			Level of description & Data determined \\
			\midrule
			Modular invariant & Spectrum, including multiplicities. \\[1em]
			\shortstack[l]{Bulk data} &
			Three-point functions on the sphere (bulk OPE coefficients). \\[1em]
			\shortstack[l]{Boundary data} & Conformal boundary conditions and related OPE coefficients.\\[1em]
			\shortstack[l]{Interface} & Conformal/topological line defects and their fusions, junctions, etc.\\[1em]
			\shortstack[l]{$\cdots$}  & $\cdots$ \\
			\bottomrule
		\end{tabular}
	\end{center}
	The present work concerns the bulk data.  It is compatible with, but is not a
	claim to replace, the more abstract categorical description of the rational CFTs \cite{Moore:1988uz,Moore:1988qv}.
	Conversely, an abstract description does not by itself produce the
	coefficient tables used in a direct correlation-function computation.
	
	By solving the bulk data, we mean giving explicit expressions for all the bulk OPE coefficients, including their magnitudes and phases. This distinction is essential for the comparison with earlier results below. Once the two-point
	functions have been normalized, the OPE coefficients admit the field-basis transformation
	\begin{equation}\label{eq:gauge}
		C_{IJK}\longrightarrow \eps_I\eps_J\eps_K C_{IJK},
		\qquad \eps_I=\pm1,
	\end{equation}
	which corresponds to independent field redefinitions, $V_I\rightarrow\eps_I V_I$, with the vacuum sign fixed to $+1$.  Thus a complete answer
	means a complete list of representative $C_{IJK}$'s in one specified operator-sign gauge.
	
	\subsection{Earlier explicit calculations}
	
	We summarize earlier calculations of related OPE coefficients in Table~\ref{tab:early_calculation}. We also include the $SU(2)$ WZW results, which are relevant because their reduced crossing equations involve the same quantum-group fusion data as the first-Kac-label sector studied here, after the root-of-unity and normalization conventions are matched. Section~\ref{sec:s=1} explains the label and parameter map.
	
	Accordingly, our \(E_6\) results overlap both the early unitary analysis~\cite{Furlan:1989ra} and the recent all-central-charge formulas of Ref.~\cite{Nivesvivat:2025odb}. In the unitary cases, our \(E_7\) and \(E_8\) results also overlap partially with the earlier calculations~\cite{Fuchs:1989prl,Fuchs:1989kz,Fuchs:1989Z2,Fuchs:1989zp,Douglas:1988rv,Kato:1988ct,Petkova:1994zs}. The main new explicit data are the complete all-central-charge \(E_7\) and \(E_8\) OPE coefficients, including nonunitary theories and with all relative signs fixed. We summarize the comparison with the literature in Section~\ref{sec:s=1} and give a detailed discussion in Appendix~\ref{app:literature-comparison}.
	
	\begingroup
	\small\renewcommand{\arraystretch}{1.12}\setlength{\tabcolsep}{4pt}
	\begin{longtable}{@{}>{\raggedright\arraybackslash}p{.20\textwidth}>{\raggedright\arraybackslash}p{.15\textwidth}>{\raggedright\arraybackslash}p{.16\textwidth}>{\raggedright\arraybackslash}p{.42\textwidth}@{}}
		\caption{Earlier explicit calculations and their scope. ``All $p$'' means
			all admissible coprime minimal-model parameters, including nonunitary
			ones. The $SU(2)$ WZW examples have positive integer level and are unitary.}
		\label{tab:early_calculation}\\
		\toprule Work & Theory/cases & Parameters & Data supplied \\\midrule
		\endfirsthead
		\toprule Work & Theory/cases & Parameters & Data supplied \\\midrule
		\endhead
		Dotsenko--Fateev \cite{Dotsenko:1984nm,Dotsenko:1984ad,Dotsenko:1985hi}
		& Virasoro $A$ & Unitary and nonunitary
		& Explicit diagonal bulk structure constants from Coulomb-gas integrals.\\[3pt]
		Petkova; Runkel \cite{Petkova:1988cy,Petkova:1988yf,Runkel:1998a,Runkel:1999dz}
		& Virasoro $A,D$ & Unitary and nonunitary
		& Bulk coefficients. Runkel also gives boundary data and uniqueness up
		to field normalization within the specified sewing problem.\\[3pt]
		Furlan--Ganchev--Petkova \cite{Furlan:1989ra}
		& Virasoro $E_6$ & $(p,q)=(11,12)$
		& Local unitary coefficients and sign relations; also larger
		fermionic and quasilocal algebras outside the local bulk spectrum.\\[3pt]
		Fateev--Zamolodchikov \cite{Zamolodchikov:1986bd}
		& $SU(2)$ WZW $A$ & Integer levels
		& Explicit diagonal coefficients; reference factors for non-diagonal
		WZW comparisons.\\[3pt]
		Di Francesco \cite{DiFrancesco:1988sc}
		& $SU(2)$ WZW $D,E$ & Integer levels
		& Partial non-diagonal results.\\[3pt]
		Fuchs--Klemm; Fuchs--Klemm--Scheich \cite{Fuchs:1989prl,Fuchs:1989kz,Fuchs:1989Z2,Fuchs:1989zp}
		& $SU(2)$ WZW $E_6,E_7,E_8$ & Levels $10,16,28$
		& Explicit exceptional formulas, relative squared magnitudes and
		additional zeros. The general method also treats minimal-model
		fusion matrices.\\[3pt]
		Douglas--Trivedi; Kato--Kitazawa \cite{Douglas:1988rv,Kato:1988ct}
		& $SU(2)$ WZW $E_7$ & Level $16$
		& Algebraic squares and positive reference-normalized ratios, in
		different conventions. Douglas--Trivedi also treats $D$ theories.\\[3pt]
		Petkova--Zuber \cite{Petkova:1994zs}
		& Virasoro $A,D,E$ & Unitary $p=q-1$
		& Detailed $s=1$ relative data: all
		magnitudes in terms of Pasquier matrices, and explicit $D,E_6,E_7$ sign relations. They report computing
		the $E_8$ spinning signs, without printing the full list. Their discussion
		also addresses second-label signs.\\[3pt]
		Nivesvivat--Ribault \cite{Nivesvivat:2025odb}
		& Virasoro $E_6,E_7,E_8$ & All $p$, all Kac labels
		& Nonchiral fusion rules for all three families; full signed $E_6$
		formulas, partial $E_7,E_8$ coefficients and numerical tests.\\
		\bottomrule
	\end{longtable}
	\endgroup

	\subsection{Other aspects of rational CFTs}
	
	Complementary constructions use graphs, cells, module categories, and algebra objects. The relation between
	ADE data and Ocneanu cells organized boundary conditions, twisted amplitudes,
	and generalized chiral vertex operators
	\cite{Petkova:2000ip,PetkovaZuber:2001Ocneanu}.  Kirillov--Ostrik and Ostrik
	placed the ADE theories in the language of algebra objects, module categories,
	and weak Hopf algebras~\cite{KirillovOstrik:2002,Ostrik:2003}. The Fuchs--Runkel--Schweigert TFT approach constructs full rational-CFT correlators from a symmetric special Frobenius algebra in the modular tensor category of chiral data, expressing bulk, boundary, and bulk--boundary structure constants as invariants of ribbon graphs in three-manifolds \cite{Fuchs:2001triangulations,Fuchs:2002cm,Fuchs:2004xi}. The resulting correlators on oriented surfaces of arbitrary genus, with or without boundary, were proved to satisfy mapping-class-group invariance and bulk and boundary factorization \cite{Fjelstad:2005}.

	These constructions provide a consistent framework for rational CFTs. There is
	also a converse relevant to a direct bulk approach: a local rational bulk
	algebra with a unit, a nondegenerate invariant two-point pairing, one vacuum,
	and the prescribed ADE modular-invariant spectrum admits compatible boundary
	data. In the categorical formulation, the spectrum fixes the required quantum
	dimension, and the bulk algebra is the full center of a boundary algebra \cite[Theorems~3.4 and~3.22]{KongRunkel:2008Cardy}. Under these assumptions, every admissible ADE bulk solution has a compatible boundary realization. This converse does not by itself prove uniqueness for a fixed spectrum. Our approach computes the bulk OPE coefficients directly from bulk locality and crossing equations, without introducing boundary data.
	
	Within the conformal-net framework, Kawahigashi and Longo established the existence and uniqueness, up to isomorphism, of the full unitary $c<1$ conformal theories associated with the ADE modular invariants \cite{Kawahigashi:2002px,Kawahigashi:2003gi}. This gives an abstract uniqueness result for the corresponding bulk OPE data, up to field redefinitions, and includes all six unitary
	$E$-series models: $(A_{10},E_6),(E_6,A_{12}),(A_{16},E_7),(E_7,A_{18}),
	(A_{28},E_8),(E_8,A_{30})$. Their analysis is restricted to unitary theories and does not provide explicit formulas for the bulk OPE coefficients. Our aim is an explicit presentation of the bulk data in a fixed Virasoro normalization and operator-sign convention, suitable for direct correlation-function calculations. Our calculations supply explicit coefficients for both unitary and nonunitary $E$-series models. The assumptions and level of rigor of our arguments are stated below.

	\subsection{Summary of the result}
	For coprime integers $p,q>1$, the central charge is
	\begin{equation}
		\begin{split}
			c=1-\frac{6(p-q)^2}{pq}\,.
		\end{split}
	\end{equation}
	For the $E$-series minimal model, $q=12,18,30$ for $E_6,E_7,E_8$, respectively. We normalize the Virasoro primary operators $V_I$ by
	\begin{equation}\label{eq:normalization}
		\begin{split}
			\braket{V_I(z_1,\bar{z}_1)V_J(z_2,\bar{z}_2)}=\frac{\delta_{IJ}}{(z_2-z_1)^{2h_I}(\bar{z}_2-\bar{z}_1)^{2\bar{h}_I}}\,,
		\end{split}
	\end{equation}
	and fix the convention of OPE coefficients $C_{IJK}$ to be
	\begin{equation}\label{eq:ope_convention}
		V_I(0)V_J(z,\bar z)
		=\sum_K C_{IJK}
		z^{h_K-h_I-h_J}\bar z^{\bar h_K-\bar h_I-\bar h_J}
		\bigl(V_K(0)+\text{Virasoro descendants}\bigr).
	\end{equation}
	
	The finite seed computations support the conjecture that,\footnote{By ``finite seed computation'', we mean solving the conformal bootstrap equations in the $s=1$ sector of the theory. Here ``$s$'' refers to the second Kac index.} for each $E$-series minimal model, bulk locality and OPE associativity admit a unique solution $\{C_{IJK}\}$, modulo the operator-sign freedom \eqref{eq:gauge}. Section~\ref{sec:global_solution} proves that existence and uniqueness for the seed problem imply existence and uniqueness for the full conformal bootstrap problem.
	
	We use the odd-$s$ Kac representatives of
	Section~\ref{sec:bootstrap}. A triplet is admissible if its fields
	belong to the spectrum and satisfy both truncated chiral fusion rules;
	channels forbidden by chiral fusion have zero coefficient.
	For an admissible triplet $[IJK]$, the candidate OPE coefficients
	factorize as
	\begin{equation}\label{eq:factorization}
		\begin{aligned}
			C_{IJK}
			&=\alpha_{IJK}C^{\mathrm{(ref)}}_{IJK},\\
			C^{\mathrm{(ref)}}_{IJK}
			&:=D_{(r_I,s_I)(r_J,s_J)(r_K,s_K)}
			D_{(\bar r_I,s_I)(\bar r_J,s_J)(\bar r_K,s_K)}.
		\end{aligned}
	\end{equation}
	Here $(r_I,s_I)$ and $(\bar r_I,s_I)$ are the left- and right-moving
	Kac indices of the primary field $I$, and $D$ is a universal chiral
	coefficient of the generalized minimal model
	\cite{Gabai:2024qum}. The OPE data specific to the exceptional models
	are encoded in the reduced coefficients $\alpha_{IJK}$.
	
	For fixed $q$, the reduced coefficients are independent of the odd
	second Kac labels $s_I,s_J,s_K$. In the operator-sign convention
	of this paper, their
	$p$-dependence is given by
	\begin{equation}\label{eq:periodic_phase}
		\begin{split}
			&\left.\alpha_{IJK}\right|_{p=kq+\epsilon w}
			=\varphi_{IJK}^{\,k}\,
			\mathcal K_\epsilon\!\left(\alpha_{IJK}^{(w)}\right),\\
			&k\in\mathbb Z_{\ge0},\qquad
			\epsilon\in\{+1,-1\},\qquad
			1\le w<q/2,\qquad \gcd(w,q)=1.
		\end{split}
	\end{equation}
	Writing $\ell_I=(r_I-1)/2$ and
	$\bar\ell_I=(\bar r_I-1)/2$, the phase is
	\begin{equation}\label{def:phi}
		\varphi_{IJK}
		=i^{
			(\lfloor\ell_I\rfloor+\lfloor\ell_J\rfloor+\lfloor\ell_K\rfloor)_2
			+
			(\lfloor\bar\ell_I\rfloor+\lfloor\bar\ell_J\rfloor
			+\lfloor\bar\ell_K\rfloor)_2
		}.
	\end{equation}
	Here $\lfloor x\rfloor$ is the floor function, and
	\begin{equation}\label{def:mod2}
		(n)_2:=n-2\lfloor n/2\rfloor\in\{0,1\}
	\end{equation}
	denotes the residue of an integer $n$ modulo two.
	The map $\mathcal K_\epsilon$ either leaves its argument unchanged
	or takes its complex conjugate:
	\begin{equation}
		\mathcal K_\epsilon(z)=
		\begin{cases}
			z,           & \epsilon=+1,\\
			\overline z, & \epsilon=-1.
		\end{cases}
	\end{equation}
	Thus the full $p$-dependence follows from finitely many reference
	residues. The required reference coefficients are tabulated in
	Section~\ref{sec:summary_alpha_s1}.
	
	Section~\ref{sec:reconstruction} collects direct formulas for every
	reduced coefficient $\alpha$ and the universal chiral coefficients
	$D$, including their absolute values and phases. Together with
	\eqref{eq:factorization}, these formulas allow any admissible OPE
	coefficient to be evaluated. The operator convention used there
	coincides with the table convention for $E_7$ and $E_8$; for $E_6$,
	the explicit field-sign conversion is stated in that section.
	
	\par\medskip
	\noindent\textbf{Notation and reconstruction.}
	We list the main notation used in this paper: 
	\begin{center}
	\begin{tabular}{p{0.45\textwidth}p{0.48\textwidth}}
	\toprule
	Symbol & Meaning \\
	\midrule
	$(p,q)$ & Coprime parameters, $p,q>1$ \\[0.5em]
	$(r,s),(\bar r,s)$ & Left and right Kac labels, with odd $s$ \\[0.5em]
	$\ell=\frac{r-1}{2}$, $\bar\ell=\frac{\bar{r}-1}{2}$ & First-label quantum-group spins \\[0.5em]
	$I=(\ell,\bar\ell,s)$ & Full Virasoro-primary label \\[0.5em]
	$S_I=h_I-\bar h_I$ & Integer physical spin \\[0.5em]
	$w$ & Reference residue; $w=1$ is formal, not a physical $p=1$ model \\[0.5em]
	$q_+,q_-$ & quantum-group parameters \\
	\bottomrule
	\end{tabular}
	\end{center}

	\subsection{Outline of the paper}
	The paper is organized as follows.
	Section~\ref{sec:reconstruction} collects explicit formulas for reconstructing
	the OPE coefficients, allowing readers to use the results without following
	the technical derivations.
	Section~\ref{sec:bootstrap} introduces the $E$-series spectra and
	fusion-matrix conventions and derives the reduced bulk crossing equations.
	Section~\ref{sec:s=1} presents the finite $s=1$ bootstrap calculation,
	the resulting coefficient tables, and their numerical verification.
	Section~\ref{sec:global_solution} proves that existence and uniqueness
	for all Kac labels reduce to the corresponding questions for the exact
	finite seed problem.
	Section~\ref{sec:conclusion} discusses unitarity, parity and other
	symmetries, Galois conjugation, and directions for future work.
	
	\section{Explicit formulas for reconstructing OPE coefficients}
	\label{sec:reconstruction}
	
	This section provides a self-contained summary of the formulas needed
	to reconstruct the OPE coefficients. It is intended for readers who wish
	to use the results without following the technical derivations, and may
	be skipped without loss of continuity. 
	
	For an $E$-series minimal model labelled by $(p,q)$, the full OPE coefficients are
	reconstructed as
	\begin{equation*}
		\begin{aligned}
			&C_{(r_1,\bar r_1,s_1)(r_2,\bar r_2,s_2)(r_3,\bar r_3,s_3)} =
			\alpha^{(p)}_{(r_1,\bar r_1)(r_2,\bar r_2)(r_3,\bar r_3)}
			\,D_{(r_1,s_1)(r_2,s_2)(r_3,s_3)}
			\,D_{(\bar r_1,s_1)(\bar r_2,s_2)(\bar r_3,s_3)}.
		\end{aligned}
	\end{equation*}
	We present empirically determined formulas for every reduced OPE
	coefficient $\alpha_{IJK}$ of our proposed solution. These formulas
	have been checked exhaustively against the finite data set (which covers all cases). We also
	give explicit expressions for the universal chiral minimal-model
	coefficients $D$. Each coefficient is expressed as an absolute value
	multiplied by an explicit phase.
	
	The notation in this section is local: several symbols have meanings
	different from those used elsewhere in the paper. We write $\alpha$
	for the reduced coefficients in the parity-invariant operator convention
	used here, and $\alpha^{\mathrm{tab}}$ for those in the convention of tables \ref{tab:E6-I1-results}--\ref{tab:E8-I13-results} and their $p$-dependent reconstruction. These conventions
	coincide for $E_7$ and $E_8$; the explicit operator-sign conversion
	for $E_6$ is given below. The spectrum and Virasoro fusion rules are
	understood (otherwise see Section \ref{sec:bootstrap}). In the reduced-coefficient formulas, we choose odd second
	Kac labels and suppress them.
	
	\subsection{Common notation for the reduced OPE coefficients}
	We label primary fields by their left- and right-moving Kac indices,
	$\Phi_{(r,s)(\bar r,s)}$, with a common second Kac label.
	We choose representatives with odd $s$, so that
	$s=1,3,\ldots,p-2$.
	The reduced OPE coefficients $\alpha_{IJK}$ are independent of the
	second Kac labels, which we therefore suppress.
	Below, we reuse the symbol $s$ for the triple of sector labels,
	ordered according to the convention specified for each family.
	
	For an ordered triple, abbreviate
	$\alpha^{(p)}=\alpha^{(p)}_{I_1I_2I_3}$.
	Write $I_i=(a_i,x_i,y_i)$, with $r_i=R_{a_i}(x_i)$,
	$\bar r_i=R_{a_i}(y_i)$, and put
	$s=(a_1,a_2,a_3)$, $x=(x_1,x_2,x_3)$, $y=(y_1,y_2,y_3)$.
	The block labels and their ordered chiral labels are
	\begin{center}
		\begin{tabular}{ccll}
			\toprule
			Family &$q$& $a$ & $(R_a(0),R_a(1),\ldots)$\\\midrule
			$E_6$&12&$\id,\ \sigma,\ \varepsilon$
			&$(1,7),\ (4,8),\ (5,11)$\\
			$E_7$&18&$\id,\ B,\ C,\ D$
			&$(1,17),\ (5,13),\ (7,11),\ (3,15,9)$\\
			$E_8$&30&$\id,\ \tau$
			&$(1,11,29,19),\ (7,13,23,17)$\\\bottomrule
		\end{tabular}
	\end{center}
	Both $x_i,y_i\in\{0,1,\ldots\}$ range over the positions in their block, except that
	$E_7$ block $D$ requires $x_i=2$ or $y_i=2$.
	
	\vspace{2mm}
	\noindent\textbf{Sorting convention}
	
	For $E_6,E_7$, stably sort the fields in the displayed block order:
	retain their relative order within each block and let $n_{\rm inv}$
	count the inversions removed. More precisely,
	\begin{equation}
		\begin{split}
			n_{\rm inv}=\#\{(i,j):i<j,\ a_i>a_j\},\quad \id<\sigma<\varepsilon,\quad \id<B<C<D. 
		\end{split}
	\end{equation}
	Use these sorted slots on the right-hand
	sides below; $\alpha^{(p)}$ always refers to the original ordered triple.
	The $E_8$ ordering is specified below.
	
	\vspace{2mm}
	\noindent\textbf{Decomposition convention for $(p,q)$}
	
	For $p>1$ coprime to $q$, define $k,w,\epsilon$ by
	\begin{equation}
		\begin{gathered}
			p=kq+\epsilon w,\qquad k\in\mathbb Z,\quad
			\epsilon\in\{\pm1\},\quad 0<w<q/2,\\
			h=\left(\sum_i\left\lfloor\tfrac{r_i-1}{2}\right\rfloor\right)_2
			+\left(\sum_i\left\lfloor\tfrac{\bar r_i-1}{2}\right\rfloor\right)_2.
		\end{gathered}
		\label{eq:alpha-intro-period}
	\end{equation}
	Here $(n)_2$ is the residue in $\{0,1\}$; all other arithmetic is ordinary
	integer arithmetic. 
	
	\vspace{2mm}
	\noindent\textbf{The common Pasquier-matrix convention}
	
	It is convenient to introduce the Pasquier matrix \cite{Pasquier:1987xj,Petkova:1994zs}. The starting point is the $E$-type Dynkin diagrams $G=E_6,E_7,E_8$.
	\begin{figure*}[ht]
		\centering
		\begin{tikzpicture}[
			x=0.6cm, y=0.6cm,
			line width=0.5pt,
			every node/.style={font=\small},
			dynkin vertex/.style={
				circle, draw, fill=white,
				inner sep=0pt, minimum size=4pt
			}
			]
			\foreach \rank/\offset in {6/0,7/6,8/13} {
				\begin{scope}[shift={(\offset,0)}]
					\draw (0,0) -- ({\rank-2},0)
					(2,0) -- (2,1);
					
					\node[dynkin vertex, label=above:{$\rank$}]
					at (2,1) {};
					
					\pgfmathtruncatemacro{\last}{\rank-1}
					\foreach \i in {1,...,\last} {
						\node[dynkin vertex, label=below:{$\i$}]
						at ({\i-1},0) {};
					}
					
					\node at ({(\rank-2)/2},-2) {$E_{\rank}$};
				\end{scope}
			}
		\end{tikzpicture}
	\end{figure*}
	
	Let $A_G$ be the adjacency matrix, defined by
	\begin{equation}
		\begin{split}
			(A_G)_{v_1\,v_2}=\begin{cases}
				1, & \text{if $v_1$ is next to $v_2$}, \\
				0, & \text{otherwise}. \\
			\end{cases}
		\end{split}
	\end{equation}
	Choose a real orthonormal eigenbasis
	\begin{equation}
		A_G\psi^{(r)}=2\cos(\pi r/q)\psi^{(r)},\qquad
		\sum_v\psi_v^{(r)}\psi_v^{(s)}=\delta_{rs},\qquad
		\psi_v^{(1)}>0.
	\end{equation}
	where $q=12,18,30$ for $E_6,E_7,E_8$, respectively. The exponent sets for $r$ are respectively
	\[
	\begin{aligned}
		\mathcal E_6&=\{1,4,5,7,8,11\},\\
		\mathcal E_7&=\{1,5,7,9,11,13,17\},\\
		\mathcal E_8&=\{1,7,11,13,17,19,23,29\}.
	\end{aligned}
	\]
	Let
	\begin{equation}
		M^G_{rst}=\sum_{v\in G}
		\frac{\psi_v^{(r)}\psi_v^{(s)}\psi_v^{(t)}}{\psi_v^{(1)}},
		\qquad
		\mathcal A_G=\sqrt{\left|M^G_{r_1r_2r_3}
			M^G_{\bar r_1\bar r_2\bar r_3}\right|},
		\label{eq:alpha-intro-M}
	\end{equation}
	where the symmetric tensor $M^G_{rst}$ is the entry of the Pasquier matrix $(M_r)_{st}=M_{rst}$. 
	
	The reduced OPE coefficients in the $E_6$ and $E_8$ minimal models satisfy
	\begin{equation}
		\begin{split}
			|\alpha^{(p)}|=\mathcal{A}_G\,,
		\end{split}
	\end{equation}
	while the $E_7$ minimal models contain absolute values in different patterns and are treated differently below. 
	
	Absolute values are entrywise; all square roots are
	nonnegative. Phase functions may be set to zero wherever the coefficient
	vanishes. In this notation the parity transformation is $x\leftrightarrow y$.

	\subsection{Reduced OPE coefficient for $E_6$}
	\label{subsec:alpha-intro-e6}
	With $\kappa=(w-1)/4\in\{0,1\}$, the reduced OPE coefficient for $E_6$ is
	\begin{equation}
		\boxed{\begin{aligned}
				\alpha^{(p)}
				={}&i^{kh+\epsilon[L_s(x)+L_s(y)]}(-1)^{Q_{s,w}(x)+Q_{s,w}(y)
					+n_{\rm inv}\sum_{i:a_i\ne\id}(x_i+y_i)}\mathcal A_{E_6}.
		\end{aligned}}
		\label{eq:alpha-intro-e6}
	\end{equation}
	For a chiral triple $b=(b_1,b_2,b_3)$, the required functions are
	\[
	\begin{array}{c|c|l}
		s&L_s(b)&Q_{s,w}(b)\\\hline
		\id\id\id&(b_1+b_2+b_3)_2&0\\
		\id\sigma\sigma&(b_1+b_2+b_3)_2
		&b_1b_2+\kappa b_3(1+b_1+b_2)\\
		\id\varepsilon\varepsilon&b_1&b_1(1-\kappa+b_2)\\
		\sigma\sigma\varepsilon&(b_1+b_2)_2&b_1(1+b_3+\kappa b_2)
	\end{array}
	\]
	The $M$ entries vanish on the omitted sector patterns. The operator
	convention in \eqref{eq:alpha-intro-e6} differs from the tables by
	$\Phi_{a,x,y}=(-1)^{\gamma_{a,w}(x,y)}\Phi^{\mathrm{tab}}_{a,x,y}$, where
	\begin{equation}
		\begin{aligned}
			\alpha^{(p)}
			&=(-1)^{\sum_i\gamma_{a_i,w}(x_i,y_i)}
			\alpha^{(p),\mathrm{tab}},\\
			\gamma_{\id,w}(x,y)&=x+\kappa x(1-y),\qquad
			\gamma_{\sigma,w}(x,y)=x,\qquad
			\gamma_{\varepsilon,w}(x,y)=(1-\kappa)x(1-y).
		\end{aligned}
		\label{eq:alpha-intro-e6-gauge}
	\end{equation}
	The same sign product converts back to the table convention, which
	is not itself parity invariant. No such redefinition is made for $E_7,E_8$.

	\subsection{Reduced OPE coefficient for $E_7$}
	\label{subsec:alpha-intro-e7}
	
	Use the block order $\id<B<C<D$ and the stable sorting convention above.
	The reduced OPE coefficient for $E_7$ is
	\begin{equation}
		\boxed{\begin{aligned}
				\alpha^{(p)}
				={}&i^{kh+\epsilon\sum_i\chi_w(I_i)}
				(-1)^{\,n_{\rm inv}\sum\limits_{i:a_i\in\{C,D\}}S_i+Q_{s,w}}\times\begin{cases}
					\lambda_s(\eta)\,\mathcal A_{E_7},&n=0,\\[2pt]
					0,        & n=1,\ s=DDD,\\[2pt]
					2^{-1/4}, & n=1,\ s\ne DDD,\\[2pt]
					2^{-1/2}, & n=2,\ s=CDD,\\[2pt]
					1,        & n=2,\ s\ne CDD,\\[2pt]
					2^{1/4},  & n=3.
				\end{cases}
		\end{aligned}}
		\label{eq:alpha-intro-e7-alpha}
	\end{equation}
	For each sorted slot define
	\begin{gather}
		z_i=\mathbf1_{x_i=1},\quad t_i=\mathbf1_{x_i=2},\qquad
		\bar z_i=\mathbf1_{y_i=1},\quad\bar t_i=\mathbf1_{y_i=2},\nonumber\\
		S_i=z_i+\bar z_i,\quad H_i=t_i+\bar t_i,\quad
		C_i=z_i\bar z_i,\quad K_i=t_i\bar t_i,\nonumber\\
		A_{ij}=z_iz_j+\bar z_i\bar z_j,\quad
		B_{ij}=z_i\bar z_j+\bar z_i z_j,\quad
		V_{ij}=z_i\bar t_j+\bar z_i t_j.
		\label{eq:alpha-intro-e7-variables}
	\end{gather}
	Here $n=\sum_{i:a_i=D}(1-K_i)$ counts the fields
	$(3,9),(15,9),(9,3),(9,15)$. Their labels $3,15$ are not $E_7$ graph
	exponents, so $\mathcal A_{E_7}$ is used only when $n=0$.
	All cases require
	\[
	|r_1-r_2|+1\le r_3\le\min(r_1+r_2-1,35-r_1-r_2)
	\]
	and the same condition for $\bar r_i$; otherwise $\alpha^{(p)}=0$.
	With $\eta=(\sum_i\mathbf1_{x_i\ne y_i})_2$,
	\begin{equation}
		\lambda_s(\eta)=
		\begin{cases}
			\sqrt{1+2\eta},&s=BCC,\\
			1-\eta,&s=CCC,\\
			1,&\text{otherwise},
		\end{cases}
		\qquad
		\begin{array}{c|ccc}
			a&\chi_1&\chi_5&\chi_7\\\hline
			\id&0&0&(S)_2\\
			B&0&(S)_2&0\\
			C&(S)_2&0&0\\
			D&H-2K-S&S&S
		\end{array}
		\label{eq:alpha-intro-e7-local}
	\end{equation}
	The entries $\chi_w$ are bits.  The sign table uses only the following
	additional combinations:
	\begin{align}
		Z&=(z_1+z_2+z_3)(\bar z_1+\bar z_2+\bar z_3),
		&\mathcal B&=1+S_1+S_2,\nonumber\\
		\mathcal F&=S_1H_3+V_{23},
		&\mathcal G&=A_{13}+V_{13}+V_{23},\nonumber\\
		\mathcal H&=S_2+V_{23}+V_{32}+H_3V_{12},
		&\mathcal K&=S_1S_2(1+K_3)+S_1K_3+K_1(S_2+S_3).
		\label{eq:alpha-intro-e7-table-variables}
	\end{align}
	\begin{center}
		\footnotesize
		\setlength{\tabcolsep}{3pt}
		\renewcommand{\arraystretch}{1.25}
		\begin{tabular}{c|ccc}
			$s$&$Q_{s,1}$&$Q_{s,5}$&$Q_{s,7}$\\\hline
			$\id\id\id$&$0$&$0$&$S_1+B_{12}+A_{23}$\\
			$\id BB$&$0$&$S_3+A_{12}$&$B_{12}$\\
			$\id CC$&$S_2$&$A_{12}$&$S_1S_2$\\
			$\id DD$&$S_2+H_2+V_{12}$&$S_1+S_2+V_{12}$&$S_1+S_2+V_{12}+C_1$\\
			$BBB$&$0$&$0$&$B_{12}$\\
			$BBC$&$S_1$&$S_2$&$A_{13}+B_{12}$\\
			$BBD$&$S_1+\mathcal F$&$S_1+S_2+H_3+\mathcal G$&$S_1+S_2+C_1+C_2+\mathcal G$\\
			$BCC$&$\mathcal B+S_2+S_3+Z$&$\mathcal B+Z+A_{23}$&$\mathcal B+S_1(S_2+S_3)+B_{23}$\\
			$BCD$&$S_2+H_3+\mathcal F$&$S_1+\mathcal G$&$A_{12}+C_1+C_2+\mathcal G$\\
			$BDD$&$S_2+V_{12}$&$S_1+S_3+H_2+V_{12}$&$S_3+C_1+V_{12}$\\
			$CCC$&$Z$&$Z$&$Z+\sum_iC_i+z_1z_2z_3 +\bar z_1\bar z_2\bar z_3$\\
			$CCD$&$H_3+\mathcal F$&$S_2+A_{12}+\mathcal G$&$S_1+A_{12}+C_1+C_2+\mathcal G$\\
			$CDD$&$\mathcal H+z_1t_2t_3 +\bar z_1\bar t_2\bar t_3$
			&$\mathcal H+A_{12}+A_{13}$
			&$\mathcal H+A_{12}+A_{13}+C_1 +t_2t_3+\bar t_2\bar t_3$\\
			$DDD$&$1+\sum_iK_i +\sum_{\rm cyc}S_iK_{i+1}$&$S_1+S_2+\mathcal K$&$S_3+\mathcal K$\\
		\end{tabular}
	\end{center}
	Here $K_4=K_1$; omitted block patterns have zero magnitude, so set
	$Q_{s,w}=0$ there.

	\subsection{Reduced OPE coefficient for $E_8$}
	\label{subsec:alpha-intro-e8}
	
	The final formula for $E_8$ is
	\begin{equation}
		\boxed{\begin{aligned}
				\alpha^{(p)}
				={}&i^{kh+\epsilon\sum_{i=1}^3[d_{a_i,w}(x_i)+d_{a_i,w}(y_i)]} (-1)^{G_{s,w}(x)+G_{s,w}(y)+\sum_{i=1}^3K_{a_i,w}(x_i,y_i)}
				\mathcal A_{E_8}.
		\end{aligned}}
		\label{eq:alpha-intro-e8}
	\end{equation}
	For a chiral label \(b_i\in\{0,1,2,3\}\), put
	\(z_i=(b_i)_2\), \(t_i=\lfloor b_i/2\rfloor\), and define
	\[
	X=\Bigl(\frac{w-1}{2}\Bigr)_2,\qquad
	Y=(\lfloor w/4\rfloor)_2,\qquad
	T=X+Y-2XY,\qquad U=X+Y-XY.
	\]
	The phases, with \(b=z+2t\), are
	\[
	\begin{aligned}
		d_{\id,w}(b)&=(1+U)z+(1+T)t+zt,\\
		d_{\tau,w}(b)&=(1+U+T)z+(T+Y)t-zt,\\
		K_{\id,w}(x,y)&=(1-U)xy+\delta_{x,3}\delta_{y,3},\\
		K_{\tau,w}(x,y)&=(1-XY)xy
		+(1-Y)\delta_{x,2}\delta_{y,2}+Y\delta_{x,3}\delta_{y,3}.
	\end{aligned}
	\]
	For \(s=\id\id\id,\tau\tau\tau\), define
	\[
	n=\sum_i z_i,\qquad E=\Bigl(\sum_i t_i\Bigr)_2,\qquad
	R=\sum_{i<j}t_it_j,\qquad
	\Omega_j=\sum_{i=1}^3z_it_{i+1+j},
	\]
	where indices in \(\Omega_j\) are cyclic modulo three. Then
	\[
	\begin{aligned}
		G_{\id\id\id,w}(b)
		&=(1-U)\mathbf1_{n>0}+U\mathbf1_{n=3}
		+\Omega_U+(1+T+U\mathbf1_{n=3})R,\\
		G_{\tau\tau\tau,w}(b)
		&=\mathbf1_{n=3}+T\mathbf1_{n>0}+E
		+\Omega_{(Y+UE)_2}+(1+T+Un)R.
	\end{aligned}
	\]
	For the mixed pattern $\id\tau\tau$, cyclically rotate the fields to put the unique
	\(\id\) sector in slot \(1\), retaining the resulting order of slots
	\(2,3\). Set
	\[
	m=z_2+z_3,\qquad e=t_2+t_3,\qquad R=t_2t_3,\qquad
	D=z_2t_2+z_3t_3.
	\]
	Then
	\begin{equation}
		\begin{aligned}
			G_{\id\tau\tau,w}(b)
			={}&(1-z_1)(z_3+t_3)+z_3t_3\\
			&+\begin{cases}
				(1-X)R+[(1-U)(1-z_2)+Y]e+XYz_2,
				&z_1=0,\\[2pt]
				(1-XY)(1-R)+Y(e+R),
				&z_1=1,\ m=0,\\[2pt]
				\begin{aligned}
					&(1-Y)[(e+X)(t_1+D)+Xt_1D]\\[-2pt]
					&\quad+Y[X+(1-X)(1-e)(1-D)],
				\end{aligned}
				&z_1=1,\ m=1,\\[2pt]
				\begin{aligned}
					&t_1[X+(1-X)e]+Y(X+e)\\[-2pt]
					&\quad+[1-Y(1-X)]R,
				\end{aligned}
				&z_1=1,\ m=2.
			\end{cases}
		\end{aligned}
		\label{eq:alpha-intro-e8-mixed}
	\end{equation}
	Set \(G=0\) for \(\id\id\tau\) and its cyclic images, whose
	\(M\)-entries vanish. All counts above are ordinary integers unless
	explicitly reduced modulo two. Formula~\eqref{eq:alpha-intro-e8} is
	manifestly invariant under \(x\leftrightarrow y\), since each \(K\) is
	symmetric and the same functions occur on both chiralities.
	
	\subsection{Universal chiral minimal-model coefficients}
	\label{subsec:alpha-intro-chiral}
	Restore the second Kac labels $s_i$ and the original field order.
	For coprime $p,q\ge2$, let $r=(r_1,r_2,r_3)$ and $s=(s_1,s_2,s_3)$
	satisfy the truncated Virasoro fusion rules separately, with
	$1\le r_i<q$, $1\le s_i<p$. Even second labels are allowed here.
	Put
	\begin{equation}
		\begin{split}
			\beta=\sqrt{\frac pq},\qquad
			P_i=\frac{r_i\beta-s_i\beta^{-1}}2,\qquad
			P_0=\frac{\beta-\beta^{-1}}2.
		\end{split}
	\end{equation}
	Let $\Gamma_\beta$ be the normalized Barnes double Gamma function
	with periods $\beta,\beta^{-1}$, and put $Q_\beta=\beta+\beta^{-1}$.
	For $\beta>0$ and ${\rm Re} (x)>0$, it is defined by the convergent integral
	\begin{equation}
		\log\Gamma_\beta(x)
		=\int_0^\infty\frac{dv}{v}\Biggl[ \frac{e^{-xv}-e^{-Q_\beta v/2}}
		{(1-e^{-\beta v})(1-e^{-v/\beta})} -\frac{(Q_\beta/2-x)^2}{2}e^{-v}
		-\frac{Q_\beta/2-x}{v}
		\Biggr].
		\label{eq:alpha-intro-barnes-definition}
	\end{equation}
	This normalization gives $\Gamma_\beta(Q_\beta/2)=1$.
	The function extends meromorphically to the complex plane through
	the shift relations
	\begin{equation}
		\begin{split}
			\frac{\Gamma_\beta(x+\beta)}{\Gamma_\beta(x)}
			=\frac{\sqrt{2\pi}\,\beta^{\beta x-1/2}}{\Gamma(\beta x)},
			\qquad
			\frac{\Gamma_\beta(x+\beta^{-1})}{\Gamma_\beta(x)}
			=\frac{\sqrt{2\pi}\,\beta^{-x/\beta+1/2}}{\Gamma(x/\beta)}.
		\end{split}
	\end{equation}
	Away from its poles, values outside ${\rm Re} (x)>0$ can therefore be
	evaluated by shifting the argument into this half-plane.
	
	Using the momenta defined above, introduce the symmetric product
	\begin{equation}
		\mathcal R_{ijk}
		=\prod_{\epsilon_1,\epsilon_2,\epsilon_3=\pm1}
		\Gamma_\beta\!\left(
		\frac{Q_\beta}{2}
		+\epsilon_1P_i+\epsilon_2P_j+\epsilon_3P_k
		\right)^{-1}.
		\label{eq:alpha-intro-chiral-R}
	\end{equation}
	The chiral coefficient in the convention used throughout the paper is
	\begin{equation}
		\boxed{\displaystyle
			D_{(r_1,s_1)(r_2,s_2)(r_3,s_3)}
			=e^{i\pi\Phi(r,s;\beta)}
			\left|
			\frac{\mathcal R_{123}^{\,2}\mathcal R_{000}}
			{\mathcal R_{011}\mathcal R_{022}\mathcal R_{033}}
			\right|^{1/4}.}
		\label{eq:alpha-intro-chiral-D}
	\end{equation}
	Here the fourth root is the positive real root, and the explicit
	phase $\Phi(r,s;\beta)$ is given below. The overall normalization
	of $\Gamma_\beta$ cancels from this expression.
	For $\beta^2=p/q$ with coprime $p,q\ge2$, and with both Kac-label
	triples satisfying their respective truncated fusion rules, every
	Barnes factor appearing in this expression is finite and nonzero.
	The formula can therefore be evaluated directly at the rational point,
	without a limiting prescription.
	
	For the explicit phase, put $u=r$ or $s$ below; empty sums vanish:
	\begin{equation}\label{def:D_phase_notation}
		\begin{split}
			m_u&=\frac{u_1+u_2+u_3-1}{2}, \\
			F_\beta(a)&=\sum_{k=1}^{a}\lfloor k\beta^2\rfloor,\qquad
			N_\beta^{(n)}(a,b)=\sum_{k=n+1}^{n+a}
			\min\left[b,\max(0,\lfloor k\beta^2\rfloor-n)\right],\quad n=0,1, \\
				Z_\beta(u)&=
				\sum_i\bigl[F_\beta(u_i-1)-F_\beta(m_u-u_i)\bigr]-F_\beta(m_u)
				+\frac12\left(\sum_i\lfloor u_i\beta^2\rfloor-\lfloor\beta^2\rfloor\right)\\
				&\quad+\sum_i\min(u_i,\lfloor\beta^{-2}\rfloor)
				-2\min(m_u,\lfloor\beta^{-2}\rfloor)
				-\min(1,\lfloor\beta^{-2}\rfloor).
		\end{split}
	\end{equation}
	Then, with cyclic indices $s_0=s_3$, $s_4=s_1$,
	\begin{equation}
		\begin{aligned}
			\Phi(r,s;\beta)={}&
			\frac14\sum_{i=1}^3r_i(s_{i-1}-s_{i+1})
			+\frac12\bigl[Z_\beta(r)-Z_{\beta^{-1}}(s)\bigr]\\
			&-\frac12\sum_{i=1}^3\bigl[
			N_\beta^{(0)}(r_i-1,s_i-1)+N_\beta^{(1)}(r_i-1,s_i-1)\bigr]\\
			&+\sum_{i=1}^3N_\beta^{(0)}(m_r-r_i,m_s-s_i)
			+N_\beta^{(1)}(m_r-1,m_s-1).
		\end{aligned}
		\label{eq:alpha-intro-chiral-phase}
	\end{equation}
	This includes the full continuation phase: $4\Phi\in\mathbb Z$.
	
	A derivation from the finite-product expression, including the
	continuation phase, is given in Appendix~\ref{app:barnes-representation}.

	\newpage
	\section{\texorpdfstring{$E$-series minimal models and reduced crossing}{E-series minimal models and reduced crossing}}\label{sec:bootstrap}
	
	\subsection{Spectrum and Kac-table conventions}
	
	We use the standard minimal-model conventions
	\begin{equation}
		c=1-\frac{6(p-q)^2}{pq},
		\qquad
		h_{r,s}=\frac{(pr-qs)^2-(p-q)^2}{4pq},
	\end{equation}
	with coprime integers $p,q>1$, \(1\leq r\leq q-1\), \(1\leq s\leq p-1\), and the Kac-table
	identification
	\begin{equation}
		(r,s)\sim(q-r,p-s).
	\end{equation}
	Because \(q\in\{12,\,18,\,30\}\) for the $E$ series and \(\gcd(p,q)=1\), every $E$-series minimal model has odd \(p\).  We choose the representative with odd
	second Kac index \(s\).  The torus partition functions of the $E$-series minimal models are
	\cite{Cappelli:1986hf,Cappelli:1987xt}
	\begin{equation}\label{def:model}
		\begin{split}
			Z_{(A_{p-1}, E_6)} &= \sum_{\substack{s=1\\s\text{ odd}}}^{p-1}\Big{(}|\chi_{1,s} + \chi_{7,s}|^2 + |\chi_{4,s} + \chi_{8,s}|^2 + |\chi_{5,s} + \chi_{11,s}|^2\Big{)} \\
			Z_{(A_{p-1}, E_7)} &= \sum_{\substack{s=1\\s\text{ odd}}}^{p-1}\Big{(}|\chi_{1,s} + \chi_{17,s}|^2 + |\chi_{5,s} + \chi_{13,s}|^2 + |\chi_{7,s} + \chi_{11,s}|^2+|\chi_{9,s}|^2 \\
			&\qquad\qquad+\chi_{9,s}(\bar{\chi}_{3,s}+\bar{\chi}_{15,s}) +(\chi_{3,s}+\chi_{15,s})\bar{\chi}_{9,s}\Big{)} \\
			Z_{(A_{p-1}, E_8)} &= \sum_{\substack{s=1\\s\text{ odd}}}^{p-1}\Big{(}|\chi_{1,s} + \chi_{11,s} + \chi_{19,s} + \chi_{29,s}|^2 + |\chi_{7,s} + \chi_{13,s} + \chi_{17,s} + \chi_{23,s}|^2\Big{)} \\
		\end{split}
	\end{equation}
	Each product of holomorphic and antiholomorphic characters in
	Eq.~\eqref{def:model} corresponds to a Virasoro primary.  We label it by
	\begin{equation}
		I=(\ell,\bar\ell,s),
		\qquad r=2\ell+1,
		\qquad \bar r=2\bar\ell+1,
	\end{equation}
	where the equality of the left and right second Kac indices is an important property of the \(E\)-series spectra.
	
	For later use, define
	\begin{equation}\label{def:Rm}
		\mathcal R_m(a,b)=
		\bigl\{|a-b|+1,\ |a-b|+3,\ldots,
		\min(a+b-1,2m-a-b-1)\bigr\}.
	\end{equation}
	The chiral Virasoro fusion rules allow \((r_3,s_3)\) in the product of
	\((r_1,s_1)\) and \((r_2,s_2)\) precisely when \cite{DiFrancesco:1997nk}
	\begin{equation}\label{chiral_fusion_rule}
		r_3\in\mathcal R_q(r_1,r_2),
		\qquad
		s_3\in\mathcal R_p(s_1,s_2).
	\end{equation}
	In our $E$-series notation, these conditions read
	\begin{equation}
		\begin{split}
			&\ell_3\in\{|\ell_1-\ell_2|,\,|\ell_1-\ell_2|+1,\,\ldots\,,\min(\ell_1+\ell_2,q-2-\ell_1-\ell_2)\}\,, \\
			&\bar\ell_3\in\{|\bar\ell_1-\bar\ell_2|,\,|\bar\ell_1-\bar\ell_2|+1,\,\ldots\,,\min(\bar\ell_1+\bar\ell_2,q-2-\bar\ell_1-\bar\ell_2)\}\,, \\
			&s_3\in\mathcal{R}_p(s_1,s_2)\,. \\
		\end{split}
	\end{equation}
	We call a full triplet \((I,J,K)\) \emph{admissible} when the left and right
	chiral fusion rules are both satisfied and all three full fields occur in the
	spectrum.  OPE coefficients for triplets forbidden by chiral fusion are zero.
	
	\subsection{Locality and OPE associativity}
	
	With the OPE convention \eqref{eq:ope_convention}, the normalization condition \eqref{eq:normalization} gives
	\begin{equation}
		\begin{split}
			C_{II\id}=1\,,
		\end{split}
	\end{equation}
	where $\id$ denotes the identity operator.
	
	We work in Euclidean signature and only consider bosonic theories. We impose locality condition, which, for bosonic fields, means that correlation functions are single-valued away from coincident insertions
	and symmetric under simultaneous exchange of a field and its position:
	\begin{equation}\label{eq:locality_condition}
		\big\langle\cdots V_I(x)V_J(y)\cdots\big\rangle
		=\big\langle\cdots V_J(y)V_I(x)\cdots\big\rangle,
		\qquad x\neq y.
	\end{equation}
	This is a statement about Euclidean correlation functions. Applying it to
	the OPE convention \eqref{eq:ope_convention} gives
	\begin{equation}\label{eq:permutation_odd}
		C_{IJK}=(-1)^{S_I+S_J+S_K}C_{JIK},
	\end{equation}
	where \(S_I=h_I-\bar h_I\in\mathbb Z\) is the spin.
	
	We impose OPE associativity,
	\begin{equation}\label{eq:ope_associativity}
		\begin{split}
			\left(V_I(x)V_J(y)\right)V_K(z)=V_I(x)\left(V_J(y)V_K(z)\right)\,.
		\end{split}
	\end{equation}
	On each side, we first apply the OPE \eqref{eq:ope_convention} to the pair of operators in parentheses, and compare the resulting correlation functions.
	
	Taking expectation values on both sides and using \eqref{eq:ope_convention} gives
	\begin{equation}\label{eq:permutation_even}
		C_{IJK}=C_{JKI}\,.
	\end{equation}
	
	Thus $C_{IJK}$ is invariant under even permutations and transforms according to Eq.~\eqref{eq:permutation_odd} under odd permutations:
	\begin{equation}\label{eq:permutation}
		\begin{split}
			C_{IJK}=C_{JKI}=(-1)^{S_I+S_J+S_K}C_{JIK}\,.
		\end{split}
	\end{equation}

	\subsection{Crossing equations and the fusion matrix}
	
	OPE associativity, \((V_IV_J)V_K=V_I(V_JV_K)\), is equivalent to the crossing
	equations for four-point functions, schematically,
	\begin{equation}\label{eq:crossing-diagram}
		\resizebox{0.72\columnwidth}{!}{%
			\begin{tikzpicture}[
				baseline=-0.5ex,
				line/.style={line width=0.6pt,line cap=round},
				vertex/.style={circle,fill=black,inner sep=1.1pt},
				lab/.style={font=\scriptsize}
				]
				
				\node at (-3,0) {$\displaystyle\sum_M$};
				
				\coordinate (sl) at (-2,0);
				\coordinate (sr) at (-1,0);
				
				\draw[line] (sl) -- ++(-0.55, 0.55)
				node[lab,anchor=south east] {$I$};
				\draw[line] (sl) -- ++(-0.55,-0.55)
				node[lab,anchor=north east] {$J$};
				\draw[line] (sr) -- ++( 0.55,-0.55)
				node[lab,anchor=north west] {$K$};
				\draw[line] (sr) -- ++( 0.55, 0.55)
				node[lab,anchor=south west] {$L$};
				
				\draw[line] (sl) --
				node[lab,midway,above] {$M$} (sr);
				
				\node[vertex] at (sl) {};
				\node[vertex] at (sr) {};
				
				\node at (0,0) {$=$};
				
				\node at (0.48,0) {$\displaystyle\sum_N$};
				
				\coordinate (tt) at (1.55, 0.38);
				\coordinate (tb) at (1.55,-0.38);
				
				\draw[line] (tt) -- ++(-0.60, 0.42)
				node[lab,anchor=south east] {$I$};
				\draw[line] (tt) -- ++( 0.60, 0.42)
				node[lab,anchor=south west] {$L$};
				\draw[line] (tb) -- ++(-0.60,-0.42)
				node[lab,anchor=north east] {$J$};
				\draw[line] (tb) -- ++( 0.60,-0.42)
				node[lab,anchor=north west] {$K$};
				
				\draw[line] (tt) --
				node[lab,midway,right] {$N$} (tb);
				
				\node[vertex] at (tt) {};
				\node[vertex] at (tb) {};
				
			\end{tikzpicture}%
		}
	\end{equation}
	The precise bootstrap equation reads
	\begin{equation}\label{eq:crossing}
		\begin{split}
			\sum_{M}C_{IJM}C_{KLM}F_{IJKL}^{(s),M}=\sum_{N}C_{JKN}C_{LIN}F_{IJKL}^{(t),N}\,.
		\end{split}
	\end{equation}
	Here \(F_{IJKL}^{(s),M}\) and \(F_{IJKL}^{(t),N}\) are the nonchiral
	\(s\)- and \(t\)-channel Virasoro blocks.  They are products of holomorphic
	and antiholomorphic blocks,
	\begin{equation}
		\begin{split}
			&F_{IJKL}^{(s/t),M}(z_i,\bar{z}_i)=\mathcal{F}_{(r_i,s_i)}^{(s/t),(r_M,s_M)}(z_i)\mathcal{F}_{(\bar{r}_i,\bar{s}_i)}^{(s/t),(\bar{r}_M,\bar{s}_M)}(\bar{z}_i)\,, \\
		\end{split}
	\end{equation}
	where the subscript \(i\) collectively denotes the labels and coordinates of
	the four external operators. 
	
	We solve Eq.~\eqref{eq:crossing} using the fusion transformation of chiral
	Virasoro blocks~\cite{Moore:1988uz,Moore:1988qv},
	\begin{equation}\label{eq:fusion_transf}
		\begin{split}
			\mathcal{F}_{(r_i,s_i)}^{(s),(r_5,s_5)}(z_i)=\sum_{6}\left[\mathcal{K}_{(r_i,s_i)}\right]_{56}\mathcal{F}_{(r_i,s_i)}^{(t),(r_6,s_6)}(z_i)\,,
		\end{split}
	\end{equation}
    The fusion matrix $\left[\mathcal{K}_{(r_i,s_i)}\right]_{56}$ depends on six pairs of Kac labels, $(r_i,s_i)$ with $i=1,\ldots,6$. The chiral Virasoro blocks solve the BPZ equations \cite{Belavin:1984vu}; Eq.~\eqref{eq:fusion_transf} expresses their $s$-channel basis in the $t$-channel basis. The nontrivial statement is that $\mathcal{K}$ is given by a normalized product of the \(6j\) symbols of
	\(U_{q_+}(\mathfrak{sl}_2)\otimes U_{q_-}(\mathfrak{sl}_2)\):
	\begin{equation}\label{fusion_kernel}
		\begin{split}
			\left[\mathcal{K}_{(r_i,s_i)}\right]_{56}=\frac{D_{236}D_{416}}{D_{125}D_{345}}\,\sixj{\ell^+_1}{\ell^+_2}{\ell^+_5}{\ell^+_3}{\ell^+_4}{\ell^+_6}_{q_+}
			\sixj{\ell^-_1}{\ell^-_2}{\ell^-_5}
			{\ell^-_3}{\ell^-_4}{\ell^-_6}_{q_-}\,. \\
		\end{split}
	\end{equation}
	Here
	\(D_{ijk}\equiv D_{(r_i,s_i)(r_j,s_j)(r_k,s_k)}\) are chiral
	generalized-minimal-model coefficients~\cite{Gabai:2024qum}, and
	\begin{equation}
		\begin{split}
			r_i&=2\ell^+_i+1,\qquad
			\bar r_i=2\bar\ell^+_i+1,\\
			s_i&=2\ell^-_i+1,\qquad
			\bar s_i=2\bar\ell^-_i+1,\\
			q_+&=e^{i\pi p/q}\,,\qquad q_-=e^{i\pi q/p}\,.
		\end{split}
	\end{equation}
	The brackets $\{\ldots\}_{q_{\pm}}$ are the quantum-group $6j$ symbols, whose expressions are as follows. Define
	the two finite lists
	\begin{align*}
		\mathcal A={}&(\ell_1+\ell_2+\ell_5,\,
		\ell_3+\ell_4+\ell_5,\,
		\ell_1+\ell_4+\ell_6,\,
		\ell_2+\ell_3+\ell_6),\\
		\mathcal B={}&(\ell_1+\ell_2+\ell_3+\ell_4,\,
		\ell_1+\ell_3+\ell_5+\ell_6,\,
		\ell_2+\ell_4+\ell_5+\ell_6).
	\end{align*}
	The $6j$ symbol of $U_{q}(\mathfrak{sl}_2)$ is then given by \cite{Kirillov:1989repr}
	\begin{equation}\label{def:sixj}
		\begin{aligned}
			\sixj{\ell_1}{\ell_2}{ \ell_5}{ \ell_3}{ \ell_4}{ \ell_6}_{q}&=
			(-1)^{\ell_1+\ell_2+\ell_3+\ell_4}
			\sqrt{[2\ell_5+1]_q[2\ell_6+1]_q} \\
			&\quad\times\Delta(\ell_1,\ell_2,\ell_5)
			\Delta(\ell_3,\ell_4,\ell_5) \Delta(\ell_1,\ell_4,\ell_6)
			\Delta(\ell_2,\ell_3,\ell_6)\\
			&\quad \times
			\sum_{z=z_{\min}}^{z_{\max}}
			\frac{(-1)^z[z+1]_q!}
			{\displaystyle
				\prod_{a\in\mathcal A}[z-a]_q!\,
				\prod_{b\in\mathcal B}[b-z]_q!}\, .
		\end{aligned}
	\end{equation}
	where
	\begin{equation}
		z_{\min}=\max\mathcal A,\qquad z_{\max}=\min\mathcal B.
	\end{equation}
	For an admissible sextuple $(\ell_1,\ell_2,\ell_3,\ell_4,\ell_5,\ell_6)$, these endpoints are integers and every factorial
	argument in the sum is a nonnegative integer.  Furthermore,
	\begin{equation}
		\Delta(a,b,c)
		=
		\sqrt{
			\frac{
				[-a+b+c]_q!\,
				[a-b+c]_q!\,
				[a+b-c]_q!
			}{
				[a+b+c+1]_q!
			}
		}\, .
	\end{equation}
	The $q$-numbers and $q$-factorials are defined by\footnote{The ``$q$'' used for $6j$ symbols should not be confused with the ``$q$'' that labels the minimal models $\mathcal{M}(p,q)$.}
	\begin{equation}\label{def:qnum}
		\begin{split}
			[n]_q:=\frac{q^n-q^{-n}}{q-q^{-1}}\,,\quad [n]_q!:=\prod_{k=1}^{n}[k]_q\,.
		\end{split}
	\end{equation}
	
	The relationship between fusion matrices and quantum-group $6j$ symbols
	was developed through two complementary approaches: Coulomb-gas constructions \cite{Petkova:1988yf,Felder:1989wv,Furlan:1989ra,Hou:1990gq,Gomez:1990er,Gomez:1989sw} and the analysis of polynomial consistency equations for fusion and braiding
	\cite{Moore:1988uz,Moore:1988qv,Alvarez-Gaume:1988bek,Alvarez-Gaume:1988izd,Alvarez-Gaume:1989blj}.
	Our notation, normalization, and square-root conventions for
	Eq.~\eqref{fusion_kernel} are specified in
	Appendix~\ref{app:fusion_kernel}.

	Now we return to the bootstrap problem. All divisions by chiral reference factors below are restricted to
	triplets satisfying the truncated fusion rules. For coprime physical
	parameters $p,q>1$, the reference factors for such triplets are finite
	and nonzero: the truncated factorial ranges avoid vanishing $q$-numbers,
	and the corresponding Gamma and $Y$ products have no poles or zeros in
	these ranges; see Appendix~\ref{app:fusion_kernel}. Thus the normalization
	\eqref{eq:factorization} and the divisions in \eqref{fusion_kernel} are
	well-defined on admissible channels. This does not assert nonvanishing
	of the physical coefficient $C_{IJK}$, whose additional zeros are carried
	by $\alpha_{IJK}$.
	
	Using Eq.~\eqref{fusion_kernel}, the bootstrap equation \eqref{eq:crossing} is
	equivalent to
	\begin{equation}\label{eq:reduced-crossing}
		\begin{split}
			&\sum_{I_5}
			\alpha_{I_1I_2I_5}\,\alpha_{I_3I_4I_5}
			\sixj{\ell_1}{\ell_2}{\ell_5}{\ell_3}{\ell_4}{\ell_6}_{q_+}\sixj{\ell^-_1}{\ell^-_2}{\ell^-_5}
			{\ell^-_3}{\ell^-_4}{\ell^-_6}_{q_-}
			\sixj{\bar\ell_1}{\bar\ell_2}{\bar\ell_5}
			{\bar\ell_3}{\bar\ell_4}{\bar\ell_6}_{q_+}
			\sixj{\bar\ell^-_1}{\bar\ell^-_2}{\bar\ell^-_5}
			{\bar\ell^-_3}{\bar\ell^-_4}{\bar\ell^-_6}_{q_-} \\
			&=
			\begin{cases}
				\alpha_{I_2I_3I_6}\,\alpha_{I_4I_1I_6},& s_6=\bar{s}_6\ \text{and}\ I_6\in\text{the spectrum},\\
				0,&\text{otherwise},
			\end{cases}\\
			&s_i=2\ell_i^-+1,\qquad
			\bar s_i=2\bar\ell_i^-+1,\qquad i=1,\ldots,6\,.
		\end{split}
	\end{equation}
	Here $\alpha_{IJK}$ is defined by Eq.~\eqref{eq:factorization}.
	For the physical external and $s$-channel fields, we use the shorthand
	$I_i=(\ell_i,\bar\ell_i,s_i)$, $i=1,\ldots,5$, since
	$s_i=\bar s_i$. The trial $t$-channel field instead has the full label
	$I_6=(\ell_6,s_6;\bar\ell_6,\bar s_6)$, with independent
	$s_6$ and $\bar s_6$. When they agree, this again reduces to our
	triple notation. The vanishing right-hand side for
	$s_6\neq\bar s_6$ also constrains the OPE coefficients.
	
	In the $E$-series minimal models, the reference OPE coefficients
	\begin{equation}
		C^\text{(ref)}_{IJK}\equiv D_{(r_I,s_I)(r_J,s_J)(r_K,s_K)}
		D_{(\bar r_I,s_I)(\bar r_J,s_J)(\bar r_K,s_K)}
	\end{equation}
	are symmetric under permutations of the indices. Indeed, the chiral reference coefficient $D_{(r_I,s_I)(r_J,s_J)(r_K,s_K)}$ satisfies the permutation properties
	\begin{equation}
		\begin{split}
			&D_{(r_I,s_I)(r_J,s_J)(r_K,s_K)}=D_{(r_J,s_J)(r_K,s_K)(r_I,s_I)} \\
			&=(-1)^{(r_I-1)(s_J-1)}i^{(r_I-1)(s_I-1)+(r_J-1)(s_J-1)+(r_K-1)(s_K-1)}\,D_{(r_J,s_J)(r_I,s_I)(r_K,s_K)}\,.
		\end{split}
	\end{equation}
	Then the invariance of $C^\text{(ref)}_{IJK}$ under even permutations is immediate, while the invariance under the odd permutations follows by direct calculation:
	\begin{equation}
		\begin{split}
			\frac{C^\text{(ref)}_{IJK}}{C^\text{(ref)}_{JIK}}=(-1)^{(r_I+\bar{r}_I-2)(s_J-1)}i^{(r_I+\bar{r}_I-2)(s_I-1)+(r_J+\bar{r}_J-2)(s_J-1)+(r_K+\bar{r}_K-2)(s_K-1)}\,.
		\end{split}
	\end{equation}
	On the right-hand side, the $(-1)$ factor is equal to one, because $s_J$ is always odd. For the $i$ factor, $r_I+\bar{r}_I$ is even and $s_I$ is odd, so the exponent is a multiple of $4$. Thus the right-hand side is equal to one, showing the invariance under odd permutations.
	
	Consequently, $\alpha_{IJK}$
	inherits the permutation properties of $C_{IJK}$:
	\begin{equation}\label{eq:permutation_alpha}
		\boxed{
			\alpha_{IJK}=\alpha_{KIJ}=(-1)^{S_I+S_J+S_K}\alpha_{JIK}\,.
		}
	\end{equation}
	Moreover, $C^\text{(ref)}_{II\id}=1$, which under our normalization
	implies $\alpha_{II\id}=1$.
	
	The $q_+$ $6j$ symbols are independent of the second Kac
	indices, whereas all $s_i$ dependence is confined to the
	$q_-$ factors. Moreover, the parity of the spin is independent
	of $s$:
	\begin{equation}
		S_I
		=
		\frac{(\ell-\bar\ell)
			(\ell+\bar\ell+1)p}{q}-(\ell-\bar\ell)s\,.
		\label{eq:spin}
	\end{equation}
	These observations motivate an $s_i$-independent extension of the reduced OPE
	coefficients. Section~\ref{sec:s=1} constructs candidate seed data ($s=1$); the subsequent lifting argument applies to any exact seed solution. 
	
	The concrete problem is therefore to solve the bootstrap equations \eqref{eq:reduced-crossing}, subject to the
	permutation constraints \eqref{eq:permutation_alpha} and the normalization
	conditions $\alpha_{II\id}=1$.
	
	\newpage
	\section{\texorpdfstring{Solving the $s=1$ sector}{Solving the s=1 sector}}\label{sec:s=1}
	
	The operators with $s=1$ form a closed OPE subalgebra, so their bootstrap
	equations can be analyzed independently of the remaining operators. In this case, \eqref{eq:reduced-crossing} simplifies to
		\begin{equation}\label{eq:reduced-crossing_s=1}
		\boxed{\begin{aligned}
			&\sum_{I_5}
			\alpha_{I_1I_2I_5}\,\alpha_{I_3I_4I_5}
			\sixj{\ell_1}{\ell_2}{\ell_5}{\ell_3}{\ell_4}{\ell_6}_{q_+}
			\sixj{\bar\ell_1}{\bar\ell_2}{\bar\ell_5}
			{\bar\ell_3}{\bar\ell_4}{\bar\ell_6}_{q_+} \\
			&=
			\begin{cases}
				\alpha_{I_2I_3I_6}\,\alpha_{I_4I_1I_6},&  I_6\in\text{the spectrum},\\
				0,&\text{otherwise},
			\end{cases}
		\end{aligned}}
	\end{equation}
	Here $I_i=(\ell_i,\bar\ell_i,1)$ for all the labels.
	
	The simplified bootstrap equations \eqref{eq:reduced-crossing_s=1} also appear in the
	$E$-type $SU(2)$ WZW CFTs \cite{DiFrancesco:1988sc,Douglas:1988rv}. There, the Virasoro algebra is extended to the Kac-Moody algebra $\widehat{\mathfrak{su}}(2)_{10}$, $\widehat{\mathfrak{su}}(2)_{16}$, $\widehat{\mathfrak{su}}(2)_{28}$ for $E_6$, $E_7$, $E_8$, respectively; the corresponding quantum-group parameters are \cite{Gaberdiel:1994vv,Rehren:1994jc}
	\begin{equation}
		\begin{split}
			q_\text{WZW}=\begin{cases}
				e^{i\pi/{12}} & \text{for }E_6,\ SU(2)_{10}\,, \\
				e^{i\pi/{18}} & \text{for }E_7,\ SU(2)_{16}\,, \\
				e^{i\pi/{30}} & \text{for }E_8,\ SU(2)_{28}\,. \\
			\end{cases}
		\end{split}
	\end{equation}
	The comparison concerns reduced crossing equations, not equality of the
	minimal-model and WZW CFTs or their central charges. The field-label map
	is $r=2\ell+1$, $\bar r=2\bar\ell+1$, with the WZW affine spins
	$(\ell,\bar\ell)$ and the minimal-model second labels fixed to $s=1$.
	Writing $q_{\rm WZW}=e^{i\pi/q}$, our parameter is $q_+=(q_\text{WZW})^p$.
	At the formal residue $p=w=1$, the normalized $6j$ symbols coincide with
	those entering the WZW equations. To compare constants, one must also
	divide out the respective chiral coupling factors and match the
	conformal-block and operator-sign conventions, as in
	Eq.~\eqref{eq:factorization}. Raw OPE coefficients are not identified
	by this comparison. The phase and conjugation transformations below
	relate other representatives of a residue class; the remaining
	coprime residues require the corresponding root-of-unity evaluations.
	Our calculation includes all these residues, encompassing both unitary
	and nonunitary minimal models.
	
	The main part of this section is a numerical calculation. It produces candidate algebraic solutions of the seed
	bootstrap equations and supports their uniqueness after fixing the
	operator-sign freedom. We do not supply an exact algebraic certificate
	or certified numerical error bounds for this finite step.
	
	\subsection{\texorpdfstring{Finite reduction in \(p\)}{Finite reduction in p}}\label{subsec:finite_p_reduction}
	
	Although each $E$-series family contains infinitely many minimal models, solving the $s=1$ sector requires only a finite computation, for the following reasons. First, by \eqref{def:model}, the number of $s=1$
	operators depends only on $q$ and is therefore fixed within each family: the
	$E_6$, $E_7$, and $E_8$ series contain $12$, $17$, and $32$ such operators,
	respectively. Second, the $q_-$ $6j$ symbols equal unity in this sector, so only
	the $q_+$ symbols contribute (as we have seen in \eqref{eq:reduced-crossing_s=1}). Third, the latter are periodic in the phase of
	the quantum-group parameter. Writing
	\begin{equation}\label{eq:thetap}
		q_+=e^{i\theta_+},
		\qquad
		\theta_+=\frac{\pi p}{q},
	\end{equation}
	the complete normalized $6j$ symbol has the periodicity
	\begin{equation}
		\theta_+\longmapsto\theta_++2\pi,
		\qquad\text{equivalently}\qquad
		p\longmapsto p+2q.
	\end{equation}
	Although separate square roots can change sign under this continuation,
	their monodromies cancel in Eq.~\eqref{def:sixj}. To see this, we count the number of $\Delta$ factors in \eqref{def:sixj} containing
	two half-integer spins. Each such $\Delta$ contributes a minus sign, whereas a $\Delta$ factor with three integer spins does not. Since each half-integer spin appears in two $\Delta$ factors, the product of the signs from all $\Delta$ factors is 
	$$(-1)^{2L}\,,\quad L=\sum_{i=1}^{6}\ell_i$$
	The Virasoro fusion rule makes
	$\ell_1+\ell_2+\ell_3+\ell_4\in\mathbb{Z}$, and therefore $2L\equiv 2\ell_5+2\ell_6\pmod 2$.
	The square roots in the first line of \eqref{def:sixj} contribute $(-1)^{2\ell_5+2\ell_6}$, cancelling
	the product of the $\Delta$ signs. On the other hand, the rational Racah sum is unchanged.
	
	Therefore, it suffices to consider the finitely many potentially inequivalent values
	\begin{equation}\label{eq:finite_cases}
		\begin{split}
			p&=kq+w, \\
			k&=0,1,
			\quad
			1\leqslant w<q,
			\quad
			\gcd(w,q)=1.
		\end{split}
	\end{equation}
	
	We next reduce the number of reference cases based on the properties of the $6j$ symbols and the $E$-series spectra.
	
	\subsubsection{\texorpdfstring{From $(p,q)$ to $(p+q,q)$}{From (p,q) to (p+q,q)}}
	For $p=kq+w$, the $k$ dependence of the $6j$ symbols can be absorbed into a phase described in \eqref{def:phi}. To see this, consider the phase change of the $6j$ symbol under
	$$(p,q)\rightarrow (p+q,q)$$
	which is equivalent to
	$$\theta_+\rightarrow\theta_++\pi$$
	in \eqref{eq:thetap}. For each $q$-number, the phase change is
	\begin{equation}
		\begin{split}
			[n]_q\rightarrow e^{-(n-1)\pi i}\,[n]_q\,.
		\end{split}
	\end{equation}
	So the phase changes of the prefactors in \eqref{def:sixj} are given by
	\begin{equation}
		\begin{split}
			\sqrt{[2\ell+1]_q}&\rightarrow i^{-2\ell}\sqrt{[2\ell+1]_q}\,, \\
			\Delta(a,b,c)&\rightarrow i^{(a+b+c)(1-a-b-c)+4(ab+ac+bc)}\Delta(a,b,c)
		\end{split}
	\end{equation}
	For the phase change of the factor
	$$\frac{(-1)^z[z+1]_q!}
	{\displaystyle
		\prod_{a\in\mathcal A}[z-a]_q!\,
		\prod_{b\in\mathcal B}[b-z]_q!}$$
	the $z$-dependent part is
	$$(-1)^{3z^2-z-4(\ell_1+\ell_2+\ell_3+\ell_4+\ell_5+\ell_6)z}=(-1)^{3z^2-z}=1$$
	so it is actually $z$-independent:
	\begin{equation}\label{eq:phase_z_term}
		\begin{split}
			\frac{(-1)^z[z+1]_q!}
			{\displaystyle
				\prod_{a\in\mathcal A}[z-a]_q!\,
				\prod_{b\in\mathcal B}[b-z]_q!}\rightarrow (-1)^{\left(\sum \ell_i\right)^2+\sum\ell_i^2}\frac{(-1)^z[z+1]_q!}
				{\displaystyle
				\prod_{a\in\mathcal A}[z-a]_q!\,
				\prod_{b\in\mathcal B}[b-z]_q!}
		\end{split}
	\end{equation}
	
	When all $\ell_i$ are integers, the phase factor in \eqref{eq:phase_z_term} is equal to $1$. Then the total phase change of the $6j$ symbol is 
	$$\frac{\psi(\ell_1,\ell_2,\ell_5)\psi(\ell_3,\ell_4,\ell_5)}{\psi(\ell_1,\ell_4,\ell_6)\psi(\ell_2,\ell_3,\ell_6)}\,,\quad \text{with }\psi(a,b,c)=i^{-(a+b+c)^2}\,.$$
	This can already be absorbed into redefinitions of $\alpha_{IJK}$. Therefore, one can simply assign
	\begin{equation}
		\begin{split}
			\alpha_{I_1I_2I_3}\big{|}_{(p+q,q)}=\left[\psi(\ell_1,\ell_2,\ell_3)\psi(\bar\ell_1,\bar\ell_2,\bar\ell_3)\right]^{-1}\alpha_{I_1I_2I_3}\big{|}_{(p,q)}
		\end{split}
	\end{equation}
	for $E_7$ and $E_8$ series, where all the quantum-group spins are integers, as well as the integer-spin (odd $r$) sector of the $E_6$ series. If we solve bootstrap equations for parameter $(p,q)$, this assignment manifestly solves the integer-spin bootstrap equations for parameter $(p+q,q)$.
	
	When we include half-integer $\ell_i$'s in the $6j$ symbol, the above analysis fails. If we consider all the representations in the fusion ring $SU(2)_N$, which contains the representations with total spins
	$$\ell=0,\frac{1}{2},1,\ldots, \frac{N}{2}\,,$$
	then for $N>0$ we cannot find a function $\psi$ which is symmetric in its arguments, and such that the phase change of the $6j$ symbol is globally of the form\footnote{One obstruction is from the case $\ell_1=\ell_2=\ell_3=\ell_4=\ell$ and $\ell_5=\ell_6=0$, for which the $6j$ symbol is equal to $(-1)^{2\ell}/[2\ell+1]_q$. When $q\rightarrow-q$, the $6j$ symbol changes sign for half-integer $\ell$, but the ratio of identical $\psi(\ell,\ell,0)$ must equal $1$.}
	\begin{equation}
		\begin{split}
			\{6j\}\rightarrow \frac{\psi(\ell_1,\ell_2,\ell_5)\psi(\ell_3,\ell_4,\ell_5)}{\psi(\ell_1,\ell_4,\ell_6)\psi(\ell_2,\ell_3,\ell_6)}\times\{6j\}
		\end{split}
	\end{equation}
	We require $\psi(a,b,c)$ to be permutation invariant so that the redefinition preserves the permutation properties of $\alpha_{IJK}$. 
	
	However, the above requirement is stronger than we need. It suffices to impose the following weaker conditions on the full phase factors $\varphi_{IJK}$ which combine chiral and antichiral parts:
	\begin{itemize}
        \item $\varphi_{IJK}$ is permutation invariant.
		\item When all the $(\ell_i,\bar\ell_i)$ are in the spectrum, the phase change on the product of the $6j$ symbols in \eqref{eq:reduced-crossing_s=1} can be absorbed into $\varphi$:
		\begin{equation}
			\begin{split}
				\{6j\ \text{of}\ \ell\}_{q_+}\{6j\ \text{of}\ \bar\ell\}_{q_+}\rightarrow\frac{\varphi_{236}\varphi_{146}}{\varphi_{125}\varphi_{345}}\{6j\ \text{of}\ \ell\}_{q_+}\{6j\ \text{of}\ \bar\ell\}_{q_+}\,.
			\end{split}
		\end{equation}
		\item When all the $(\ell_i,\bar\ell_i)$ except for $(\ell_6,\bar\ell_6)$ are in the spectrum, the phase change of the product of $6j$ symbols in \eqref{eq:reduced-crossing_s=1} multiplying $\varphi_{125}\varphi_{345}$ does not depend on $(\ell_5,\bar\ell_5)$ in the sum.
	\end{itemize}
	We only need to search for such $\varphi_{IJK}$ for $E_6$. It turns out that the construction \eqref{def:phi} satisfies these requirements.\footnote{By enumerating all possible $(\ell_i,\bar\ell_i)$ ($i=1,2,3,4,5$) in the spectrum and all possible $(\ell_6,\bar\ell_6)$ allowed by the fusion rules, we have verified that $\varphi$ in \eqref{def:phi} compensates the phase change of the $6j$ symbols.}
	
	Therefore, the solution at $p=kq+w$ is related to the solution at $p=w$ by
	\begin{equation}
		\alpha_{IJK}\big|_{p=kq+w}
		=\left(\varphi_{IJK}\right)^k\alpha_{IJK}\big|_{p=w}\,.
	\end{equation}
	So it remains only to solve the $k=0$
	cases in \eqref{eq:finite_cases}.
	
	\subsubsection{\texorpdfstring{From $(p,q)$ to $(-p,q)$}{From (p,q) to (-p,q)}}
	Changing $p$ to $-p$ complex conjugates the $6j$ symbols. The
	solution at $(-p,q)$ is therefore obtained by complex conjugating the solution
	at $(p,q)$:
	\begin{equation}
		\alpha_{IJK}\big{|}_{(p,q)}\longrightarrow
		\alpha_{IJK}\big{|}_{(-p,q)}
		=\left(\alpha_{IJK}\big{|}_{(p,q)}\right)^*\,.
	\end{equation}
	This, together with the previous observations, reduces the reference residues to
	\begin{equation}
		\begin{split}
			1\leqslant w\leqslant q/2\,,\quad \gcd(w,q)=1\,.
		\end{split}
	\end{equation}
	There are thus nine reference cases: \(w=1,5\) for \(q=12\)
	(\(E_6\)); \(w=1,5,7\) for \(q=18\) (\(E_7\)); and
	\(w=1,7,11,13\) for \(q=30\) (\(E_8\)).
	\footnote{The symbol \(w\) labels a residue class used to evaluate the
		root-of-unity data.  A displayed reference value such as \(w=1\) is not
		itself asserted to define a physical \(p=1\) minimal model.} Section~\ref{subsec:galois} further relates the reference cases within each family by Galois conjugation. We nevertheless compute them separately to obtain independent checks.
	
	In principle, one could solve Eq.~\eqref{eq:reduced-crossing} symbolically in
	each case.\footnote{Additional extended-algebra or defect constraints may simplify an analytic solution, but we do not pursue that approach here. Our aim is to determine the bulk data from Virasoro crossing, locality, and the specified spectrum.} In practice, this is prohibitively expensive. We therefore make
	several further reductions that render the computation feasible on a laptop.
	
	\subsection{\texorpdfstring{$\mathbb{Z}_2$ simple current and triplet orbits}{Z2 simple current and triplet orbits}}\label{sec:triplet_orbit}
	
	The $E$-series minimal models contain the following nontrivial invertible
	$\mathbb{Z}_2$ primaries, which act as order-two simple currents
	\cite{Schellekens:1990xy,Nivesvivat:2025odb}:
	\begin{equation}
		\begin{split}
			E_6:\quad&
			V_{(5,5,1)},
			\\
			E_7:\quad&
			V_{(8,0,1)},\ 
			V_{(0,8,1)},\ 
			V_{(8,8,1)},
			\\
			E_8:\quad&
			V_{(14,0,1)},\ 
			V_{(0,14,1)},\ 
			V_{(14,14,1)}.
		\end{split}
	\end{equation}
	
	Let $\ell_{\mathbb{Z}_2}=5,8,14$ for $E_6,E_7,E_8$, respectively. Fusion with an order-two simple current acts by reflection on the first Kac label:
	\begin{equation}
		\begin{split}
			[\ell_{\mathbb{Z}_2}]\times[\ell]=[\ell_{\mathbb{Z}_2}-\ell]\,.
		\end{split}
	\end{equation}
	
	Together with the identity operator, these $\mathbb{Z}_2$ primaries form the fusion subalgebra $\mathbb{Z}_2$ for $E_6$ series, and $\mathbb{Z}_2\times\mathbb{Z}_2$ for $E_7$ and $E_8$ series.
	
	We now discuss OPE coefficients involving these $\mathbb{Z}_2$ operators. For convenience, we denote them by $(\ell_{\mathbb{Z}_2},\ell_{\mathbb{Z}_2})$ (diagonal), $(\ell_{\mathbb{Z}_2},0)$ (chiral) and $(0,\ell_{\mathbb{Z}_2})$ (antichiral). The second Kac label is fixed to $s=1$ and is suppressed below.
	
	\subsubsection{\texorpdfstring{Diagonal $\mathbb{Z}_2$ operator}{Diagonal Z2 operator}}
	
	Consider the four-point function
	\begin{equation}
		\begin{split}
			\braket{
				(\ell_{\mathbb{Z}_2},\ell_{\mathbb{Z}_2})
				(\ell_1,\bar\ell_1)
				(\ell_2,\bar\ell_2)
				(\ell_3,\bar\ell_3)
			}\,.
		\end{split}
	\end{equation}
	The corresponding crossing equation is
	\begin{equation}\label{crossing:Z2}
		\begin{split}
			&
			\alpha_{(\ell_{\mathbb{Z}_2},\ell_{\mathbb{Z}_2})(\ell_1,\bar\ell_1)(\ell_{\mathbb{Z}_2}-\ell_1,\ell_{\mathbb{Z}_2}-\bar\ell_1)}
			\alpha_{(\ell_2,\bar\ell_2)(\ell_3,\bar\ell_3)(\ell_{\mathbb{Z}_2}-\ell_1,\ell_{\mathbb{Z}_2}-\bar\ell_1)}
			\\
			&\hspace{1cm}\times
			\sixj{\ell_{\mathbb{Z}_2}}{\ell_1}{\ell_{\mathbb{Z}_2}-\ell_1}{\ell_2}{\ell_3}{\ell_{\mathbb{Z}_2}-\ell_3}_{q_+}
			\sixj{\ell_{\mathbb{Z}_2}}{\bar\ell_1}{\ell_{\mathbb{Z}_2}-\bar\ell_1}{\bar\ell_2}{\bar\ell_3}{\ell_{\mathbb{Z}_2}-\bar\ell_3}_{q_+}
			\\
			&=
			\alpha_{(\ell_3,\bar\ell_3)(\ell_{\mathbb{Z}_2},\ell_{\mathbb{Z}_2})(\ell_{\mathbb{Z}_2}-\ell_3,\ell_{\mathbb{Z}_2}-\bar\ell_3)}
			\alpha_{(\ell_1,\bar\ell_1)(\ell_2,\bar\ell_2)(\ell_{\mathbb{Z}_2}-\ell_3,\ell_{\mathbb{Z}_2}-\bar\ell_3)}\,.
		\end{split}
	\end{equation}
	
	Now perform the change of variables
	\begin{equation}
		\begin{split}
			(\ell_3,\bar\ell_3)
			\longrightarrow
			(\ell_{\mathbb{Z}_2}-\ell_3,\ell_{\mathbb{Z}_2}-\bar\ell_3)\,.
		\end{split}
	\end{equation}
	Then \eqref{crossing:Z2} becomes
	\begin{equation}
		\begin{split}
			&
			\alpha_{(\ell_{\mathbb{Z}_2},\ell_{\mathbb{Z}_2})(\ell_1,\bar\ell_1)(\ell_{\mathbb{Z}_2}-\ell_1,\ell_{\mathbb{Z}_2}-\bar\ell_1)}
			\alpha_{(\ell_2,\bar\ell_2)(\ell_{\mathbb{Z}_2}-\ell_3,\ell_{\mathbb{Z}_2}-\bar\ell_3)(\ell_{\mathbb{Z}_2}-\ell_1,\ell_{\mathbb{Z}_2}-\bar\ell_1)}
			\\
			&\hspace{1cm}\times
			\sixj{\ell_{\mathbb{Z}_2}}{\ell_1}{\ell_{\mathbb{Z}_2}-\ell_1}{\ell_2}{\ell_{\mathbb{Z}_2}-\ell_3}{\ell_3}_{q_+}
			\sixj{\ell_{\mathbb{Z}_2}}{\bar\ell_1}{\ell_{\mathbb{Z}_2}-\bar\ell_1}{\bar\ell_2}{\ell_{\mathbb{Z}_2}-\bar\ell_3}{\bar\ell_3}_{q_+}
			\\
			&=
			\alpha_{(\ell_{\mathbb{Z}_2}-\ell_3,\ell_{\mathbb{Z}_2}-\bar\ell_3)(\ell_{\mathbb{Z}_2},\ell_{\mathbb{Z}_2})(\ell_3,\bar\ell_3)}
			\alpha_{(\ell_1,\bar\ell_1)(\ell_2,\bar\ell_2)(\ell_3,\bar\ell_3)}\,.
		\end{split}
	\end{equation}
	If the displayed denominators are nonzero, this is equivalently
	\begin{equation}
		\begin{split}
			&\frac{
				\alpha_{(\ell_1,\bar\ell_1)(\ell_2,\bar\ell_2)(\ell_3,\bar\ell_3)}
			}{
				\alpha_{(\ell_{\mathbb{Z}_2}-\ell_1,\ell_{\mathbb{Z}_2}-\bar\ell_1)(\ell_2,\bar\ell_2)(\ell_{\mathbb{Z}_2}-\ell_3,\ell_{\mathbb{Z}_2}-\bar\ell_3)}
			} \\
			&=
			\frac{
				\alpha_{(\ell_{\mathbb{Z}_2},\ell_{\mathbb{Z}_2})(\ell_1,\bar\ell_1)(\ell_{\mathbb{Z}_2}-\ell_1,\ell_{\mathbb{Z}_2}-\bar\ell_1)}
			}{
				\alpha_{(\ell_{\mathbb{Z}_2}-\ell_3,\ell_{\mathbb{Z}_2}-\bar\ell_3)(\ell_{\mathbb{Z}_2},\ell_{\mathbb{Z}_2})(\ell_3,\bar\ell_3)}
			} \\
			&\quad\times\sixj{\ell_{\mathbb{Z}_2}}{\ell_1}{\ell_{\mathbb{Z}_2}-\ell_1}{\ell_2}{\ell_{\mathbb{Z}_2}-\ell_3}{\ell_3}_{q_+}
			\sixj{\ell_{\mathbb{Z}_2}}{\bar\ell_1}{\ell_{\mathbb{Z}_2}-\bar\ell_1}{\bar\ell_2}{\ell_{\mathbb{Z}_2}-\bar\ell_3}{\bar\ell_3}_{q_+}\,.
		\end{split}
	\end{equation}
	The $6j$ symbol in this equation takes the value $\pm1$. Its explicit form is given by
	\begin{equation}\label{sixj_reflection_notation}
		\begin{split}
			\sixj{\ell_{\mathbb{Z}_2}}{\ell_1}{\ell_{\mathbb{Z}_2}-\ell_1}{\ell_2}{\ell_{\mathbb{Z}_2}-\ell_3}{\ell_3}_{q_+}
			&=
			(-1)^{\ell_1+\ell_2-\ell_3}(-1)^{f_{q,k,w}(\ell_1,\ell_2,\ell_3)}\,, \\
			f_{q,k,w}(\ell_1,\ell_2,\ell_3)=&k(\ell_1+\ell_2-\ell_3)(\ell_2+\ell_3-\ell_1+1) \\
			&\quad+\Phi_{q,w}(2\ell_1)+\Phi_{q,w}(2\ell_3+1) \\
			&\quad+\Phi_{q,w}(\ell_1-\ell_2+\ell_3)+\Phi_{q,w}(\ell_1+\ell_2+\ell_3+1)\,, \\
			q_+&=e^{i\pi\left(k+\tfrac{w}{q}\right)}
		\end{split}
	\end{equation}
	where the function $\Phi$ is defined by
	\begin{equation}
		\begin{split}
			\Phi_{N,r}(n):=\sum_{j=1}^{n}\left\lfloor \frac{r\,j}{N}\right\rfloor\,.
		\end{split}
	\end{equation}
	The above notation is motivated by the fact that $f_{q,0,1}(\ell_1,\ell_2,\ell_3)=0$.
	
	Therefore, we obtain
	\begin{equation}
		\begin{split}
			\frac{
				\alpha_{(\ell_1,\bar\ell_1)(\ell_2,\bar\ell_2)(\ell_3,\bar\ell_3)}
			}{
				\alpha_{(\ell_{\mathbb{Z}_2}-\ell_1,\ell_{\mathbb{Z}_2}-\bar\ell_1)(\ell_2,\bar\ell_2)(\ell_{\mathbb{Z}_2}-\ell_3,\ell_{\mathbb{Z}_2}-\bar\ell_3)}
			}
			&=
			\frac{
				\alpha_{(\ell_{\mathbb{Z}_2},\ell_{\mathbb{Z}_2})(\ell_1,\bar\ell_1)(\ell_{\mathbb{Z}_2}-\ell_1,\ell_{\mathbb{Z}_2}-\bar\ell_1)}
			}{
				\alpha_{(\ell_{\mathbb{Z}_2},\ell_{\mathbb{Z}_2})(\ell_3,\bar\ell_3)(\ell_{\mathbb{Z}_2}-\ell_3,\ell_{\mathbb{Z}_2}-\bar\ell_3)}
			}
			(-1)^{\ell_{123}+\bar\ell_{123}} \\
			&\quad\times(-1)^{f_{q,k,w}(\ell_1,\ell_2,\ell_3)+f_{q,k,w}(\bar{\ell}_1,\bar{\ell}_2,\bar{\ell}_3)}
		\end{split}
	\end{equation}
	where
	\begin{equation}
		\begin{split}
			\ell_{ijk}:=\ell_i+\ell_j-\ell_k\,,
			\qquad
			\bar\ell_{ijk}:=\bar\ell_i+\bar\ell_j-\bar\ell_k\,.
		\end{split}
	\end{equation}
	After relabeling the indices, the above relation can be rewritten as
	\begin{equation}\label{alpha:constr_Z2diagonal1}
		\boxed{\begin{aligned}
			\frac{
				\alpha_{(\ell_1,\bar\ell_1)(\ell_2,\bar\ell_2)(\ell_3,\bar\ell_3)}
			}{
				\alpha_{(\ell_{\mathbb{Z}_2}-\ell_1,\ell_{\mathbb{Z}_2}-\bar\ell_1)(\ell_{\mathbb{Z}_2}-\ell_2,\ell_{\mathbb{Z}_2}-\bar\ell_2)(\ell_3,\bar\ell_3)}
			}
			&=
			\frac{
				\alpha_{(\ell_{\mathbb{Z}_2},\ell_{\mathbb{Z}_2})(\ell_2,\bar\ell_2)(\ell_{\mathbb{Z}_2}-\ell_2,\ell_{\mathbb{Z}_2}-\bar\ell_2)}
			}{
				\alpha_{(\ell_{\mathbb{Z}_2},\ell_{\mathbb{Z}_2})(\ell_1,\bar\ell_1)(\ell_{\mathbb{Z}_2}-\ell_1,\ell_{\mathbb{Z}_2}-\bar\ell_1)}
			}
			(-1)^{\ell_{231}+\bar\ell_{231}} \\
			&\quad\times(-1)^{f_{q,k,w}(\ell_2,\ell_3,\ell_1)+f_{q,k,w}(\bar{\ell}_2,\bar{\ell}_3,\bar{\ell}_1)}
		\end{aligned}
				}
	\end{equation}
	
	We first determine the quantities on the right-hand side of \eqref{alpha:constr_Z2diagonal1}. This gives explicit relations of the form
	$$\frac{\alpha_{IJK}}{\alpha_{\sigma(I)\sigma(J)K}}=\text{a known phase}$$
	These relations organize the triplets $[IJK]$ into orbits and reduce the number of representative OPE coefficients. If an orbit coefficient vanishes, the ratio notation is replaced by the cross-multiplied relation, which also propagates zeros. OPE coefficients of the form $\alpha_{\sigma(\id)A\sigma(A)}$ are phases, as follows from their squared values below.\footnote{$\sigma(\id)$ is the label for the $\mathbb{Z}_2$ operator itself, since the pre-image is the identity operator.}
	
	Let us now specialize to the coefficient
	\begin{equation}
		\begin{split}
			\alpha_{(\ell_{\mathbb{Z}_2},\ell_{\mathbb{Z}_2})(\ell,\bar\ell)(\ell_{\mathbb{Z}_2}-\ell,\ell_{\mathbb{Z}_2}-\bar\ell)}\,.
		\end{split}
	\end{equation}
	This is obtained from the previous relation by setting
	\begin{equation}
		\begin{split}
			(\ell_1,\bar\ell_1)=(\ell,\bar\ell)\,,
			\qquad
			(\ell_2,\bar\ell_2)=(\ell_{\mathbb{Z}_2},\ell_{\mathbb{Z}_2})\,,
			\qquad
			(\ell_3,\bar\ell_3)=(\ell_{\mathbb{Z}_2}-\ell,\ell_{\mathbb{Z}_2}-\bar\ell)\,.
		\end{split}
	\end{equation}
	We then find
	\begin{equation}
		\begin{split}
			\frac{
				\alpha_{(\ell,\bar\ell)(\ell_{\mathbb{Z}_2},\ell_{\mathbb{Z}_2})(\ell_{\mathbb{Z}_2}-\ell,\ell_{\mathbb{Z}_2}-\bar\ell)}
			}
			{
				\alpha_{(\ell_{\mathbb{Z}_2}-\ell,\ell_{\mathbb{Z}_2}-\bar\ell)(0,0)(\ell_{\mathbb{Z}_2}-\ell,\ell_{\mathbb{Z}_2}-\bar\ell)}
			}
			&=
			\frac{
				\alpha_{(\ell_{\mathbb{Z}_2},\ell_{\mathbb{Z}_2})(\ell_{\mathbb{Z}_2},\ell_{\mathbb{Z}_2})(0,0)}
			}{
				\alpha_{(\ell_{\mathbb{Z}_2},\ell_{\mathbb{Z}_2})(\ell,\bar\ell)(\ell_{\mathbb{Z}_2}-\ell,\ell_{\mathbb{Z}_2}-\bar\ell)}
			}(-1)^{2(\ell+\bar{\ell})} \\
			&\quad\times(-1)^{f_{q,k,w}(\ell_{\mathbb{Z}_2},\ell_{\mathbb{Z}_2}-\ell,\ell)+f_{q,k,w}(\ell_{\mathbb{Z}_2},\ell_{\mathbb{Z}_2}-\bar{\ell},\bar{\ell})}\,.
		\end{split}
	\end{equation}
	Using the normalization condition for $\alpha$ involving the identity operator, permutation symmetries and the fact that $\ell+\bar{\ell}$ is always an integer in the $E$-series minimal models, the above expression is equivalent to
	\begin{equation}
		\label{alpha:constr_Z2diagonal2}
		\begin{split}
			&
			\left[\alpha_{(\ell_{\mathbb{Z}_2},\ell_{\mathbb{Z}_2})(\ell,\bar\ell)(\ell_{\mathbb{Z}_2}-\ell,\ell_{\mathbb{Z}_2}-\bar\ell)}\right]^2
			\\
			&=(-1)^{S_{\ell,\bar{\ell}}+S_{\ell_{\mathbb{Z}_2}-\ell,\ell_{\mathbb{Z}_2}-\bar{\ell}}}(-1)^{f_{q,k,w}(\ell_{\mathbb{Z}_2},\ell_{\mathbb{Z}_2}-\ell,\ell)+f_{q,k,w}(\ell_{\mathbb{Z}_2},\ell_{\mathbb{Z}_2}-\bar{\ell},\bar{\ell})}\,.
		\end{split}
	\end{equation}
	Here we have used the fact that the parity of the spin $S_{\ell,\bar\ell,s}\bmod2$ does not depend on $s$. This constraint implies that $\alpha_{(\ell_{\mathbb{Z}_2},\ell_{\mathbb{Z}_2})(\ell,\bar\ell)(\ell_{\mathbb{Z}_2}-\ell,\ell_{\mathbb{Z}_2}-\bar\ell)}$ can only take values in $\{\pm 1,\pm i\}$.

	\subsubsection{\texorpdfstring{Chiral/antichiral $\mathbb{Z}_2$ operator}{Chiral/antichiral Z2 operator}}
	For $E_7$- and $E_8$-series minimal models, there exist chiral and antichiral $\mathbb{Z}_2$ operators. The analyses of the chiral and antichiral $\mathbb{Z}_2$ operators parallel the analysis of the diagonal $\mathbb{Z}_2$ operator. Here we only state the results.
	
	The analysis of the chiral $\mathbb{Z}_2$ operator results in 
	\begin{equation}\label{alpha:constr_Z2chiral1}
		\boxed{\begin{aligned}
			\frac{
				\alpha_{(\ell_1,\bar\ell_1)(\ell_2,\bar\ell_2)(\ell_3,\bar\ell_3)}
			}{
				\alpha_{(\ell_{\mathbb{Z}_2}-\ell_1,\bar\ell_1)(\ell_{\mathbb{Z}_2}-\ell_2,\bar\ell_2)(\ell_3,\bar\ell_3)}
			}
			&=
			\frac{
				\alpha_{(\ell_{\mathbb{Z}_2},0)(\ell_2,\bar\ell_2)(\ell_{\mathbb{Z}_2}-\ell_2,\bar\ell_2)}
			}{
				\alpha_{(\ell_{\mathbb{Z}_2},0)(\ell_1,\bar\ell_1)(\ell_{\mathbb{Z}_2}-\ell_1,\bar\ell_1)}
			}
			(-1)^{\ell_{231}+f_{q,k,w}(\ell_2,\ell_3,\ell_1)} \\
		\end{aligned}}
	\end{equation}
	In particular
	\begin{equation}\label{alpha:constr_Z2chiral2}
		\begin{split}
			&
			\left[\alpha_{(\ell_{\mathbb{Z}_2},0)(\ell,\bar\ell)(\ell_{\mathbb{Z}_2}-\ell,\bar\ell)}\right]^2=(-1)^{S_{\ell_{\mathbb{Z}_2},0}+S_{\ell,\bar{\ell}}+S_{\ell_{\mathbb{Z}_2}-\ell,\bar{\ell}}}(-1)^{f_{q,k,w}(\ell_{\mathbb{Z}_2},\ell_{\mathbb{Z}_2}-\ell,\ell)}\,.
		\end{split}
	\end{equation}
	
	The analysis of the antichiral $\mathbb{Z}_2$ operator results in 
	\begin{equation}\label{alpha:constr_Z2antichiral1}
		\boxed{\begin{aligned}
			\frac{
				\alpha_{(\ell_1,\bar\ell_1)(\ell_2,\bar\ell_2)(\ell_3,\bar\ell_3)}
			}{
				\alpha_{(\ell_1,\ell_{\mathbb{Z}_2}-\bar\ell_1)(\ell_2,\ell_{\mathbb{Z}_2}-\bar\ell_2)(\ell_3,\bar\ell_3)}
			}
			&=
			\frac{
				\alpha_{(0,\ell_{\mathbb{Z}_2})(\ell_2,\bar\ell_2)(\ell_2,\ell_{\mathbb{Z}_2}-\bar\ell_2)}
			}{
				\alpha_{(0,\ell_{\mathbb{Z}_2})(\ell_1,\bar\ell_1)(\ell_1,\ell_{\mathbb{Z}_2}-\bar\ell_1)}
			}
			(-1)^{\bar\ell_{231}+f_{q,k,w}(\bar\ell_2,\bar\ell_3,\bar\ell_1)} \\
		\end{aligned}}
	\end{equation}
	In particular
	\begin{equation}\label{alpha:constr_Z2antichiral2}
		\begin{split}
			&
			\left[\alpha_{(0,\ell_{\mathbb{Z}_2})(\ell,\bar\ell)(\ell,\ell_{\mathbb{Z}_2}-\bar\ell)}\right]^2=(-1)^{S_{0,\ell_{\mathbb{Z}_2}}+S_{\ell,\bar{\ell}}+S_{\ell,\ell_{\mathbb{Z}_2}-\bar{\ell}}}(-1)^{f_{q,k,w}(\ell_{\mathbb{Z}_2},\ell_{\mathbb{Z}_2}-\bar\ell,\bar\ell)}\,.
		\end{split}
	\end{equation}
	Here we have used the fact that $\ell$ and $\bar{\ell}$ are integers for all the operators in $E_7$ and $E_8$.

	\subsection{Algorithm}
	
	We solve the seed bootstrap equations \eqref{eq:reduced-crossing_s=1} for each reference residue $(p,q)$. The computation proceeds in three stages.
	\newline
	\newline\noindent\textbf{Pre-bootstrap stage.}
	
	This stage identifies vanishing OPE coefficients. For each representative
	$\alpha_{IJK}$, we analyze the four-point function
	$$\braket{V_IV_JV_JV_I}\,.$$
	The crossing equations are linear in the squares $(\alpha_{IJK})^2$
	in the $s$-channel and the products $\alpha_{IIK}\alpha_{JJK}$
	in the $t$-channel. The latter include the normalized identity
	contribution $\alpha_{II\id}\alpha_{JJ\id}=1$. When the relevant
	linear system has sufficient rank, this normalization allows us
	to determine the quantities themselves, rather than only their
	ratios. In the present calculation, the ranks and the resulting
	vanishing coefficients are determined numerically.
	
	A principal motivation for this stage is to investigate the
	selection rules associated with the extended symmetries.
	The extended-character decompositions of the torus partition
	functions \eqref{def:model} suggest Ising-type sector selection
	rules for $E_6$ and Fibonacci-type sector selection rules for $E_8$:
	\begin{equation}\label{extended_selcetion_rule}
		\begin{split}
			\alpha_{IJK}=0\quad\text{if }[IJK]\ \text{is forbidden by the extended symmetries}\,.
		\end{split}
	\end{equation}
	
	To test these selection rules without assuming them, we perform
	an independent, unrestricted $\braket{ABBA}$ calculation for the
	formal seeds $(p,q)=(1,12)$ and $(1,30)$. For each ordered pair
	$A,B$, we retain every physical $s$-channel field allowed by
	the separate left- and right-moving Virasoro fusion rules.
	We determine the squared coefficients using the crossing
	equations for the identity target ($\ell_6=\bar\ell_6=0$) and the chirally admissible
	target channels absent from the physical spectrum. Thus, the
	calculation uses only the prescribed spectrum, Virasoro symmetry,
	permutation symmetry, and OPE associativity. The extended-sector
	exclusions are imposed only in subsequent bootstrap stages.
	
    The unrestricted computations provide numerical evidence for these exclusions. Once certified exactly, the relevant rank and vanishing relations would extend to the other admissible parameters $p$ by Galois conjugation; see
	Section~\ref{subsec:galois}.
	
	The numerical rank tests find full column rank for the systems
	associated with all $144$ ordered pairs in $E_6$ and all $1024$
	ordered pairs in $E_8$. Among the reconstructed candidates,
	all $2$ forbidden unordered triples in $E_6$ and all $90$
	forbidden unordered triples in $E_8$ have vanishing coefficients.
	The $E_8$ solution also contains $85$ additional zero unordered
	triples that are allowed by the extended fusion rule.
	\newline
	\newline\noindent\textbf{Correlator-by-correlator bootstrap.}
	
	In this stage, we organize $\alpha_{IJK}$ into triplet orbits under the $\mathbb{Z}_2$ relations \eqref{alpha:constr_Z2diagonal1}, \eqref{alpha:constr_Z2chiral1}, \eqref{alpha:constr_Z2antichiral1} and the permutation relations
	\eqref{eq:permutation_alpha}. 
	The number of representative coefficients before imposing the remaining crossing equations equals the number of triplet orbits, as summarized in Table~\ref{tab:counts}.
	
	\begin{table}[ht]
		\centering
		\begin{tabular}{cccc}
			\toprule
			Series & $q$ & $s=1$ primaries & Triplet orbits\\
			\midrule
			$E_6$ & $12$ & $12$ & $18$\\
			$E_7$ & $18$ & $17$ & $32$\\
			$E_8$ & $30$ & $32$ & $107$\\
			\bottomrule
		\end{tabular}
		\caption{Numbers of $s=1$ primaries and nontrivial representative triplets in
			the three exceptional series. The last column counts variables after
			the orbit identifications and the preliminary sector exclusions,
			before the remaining crossing equations are imposed; it does not
			count free parameters of a solution. Identity and simple-current
			couplings are fixed separately.}
		\label{tab:counts}
	\end{table}
	
	For each orbit $\mathcal{C}_k$, all OPE coefficients are
	expressed in terms of a single representative variable $x_k$. We then solve
	the crossing equations for correlators of the form $\braket{ABBA}$. These
	equations determine all squares $x_k^2$, together with a subset of the
	relations among different $x_k$. Finally, after fixing the operator-sign freedom, we use crossing equations of other correlators to determine the remaining relative signs. We choose to proceed with the correlators in the following order:
	$$\braket{ABBA}\longrightarrow \braket{AAAB}\longrightarrow\braket{AABC}\longrightarrow\braket{ABCD}.$$
	
	In the reported $E_6$ and $E_8$ runs, the $\braket{ABBA}$ and $\braket{AAAB}$ stages determine every representative after gauge fixing. The $E_7$ implementation also uses $\braket{AABC}$: it determines $24$ of the $32$ representatives after $\braket{ABBA}$, $31$ after $\braket{AAAB}$, and all $32$ after $\braket{AABC}$. These counts include the zero classes. This describes the implemented determination stages; it does not rule out a different strategy using fewer correlator types. These results support uniqueness of the seed solution. Proving uniqueness would additionally require certifying the relevant ranks and showing that the crossing equations fix all
	relative signs up to field redefinitions.
	
	\begin{remark}
		For a fixed ordered quartet, we introduce auxiliary variables for the products
		that occur in the two channels,
		\begin{equation}
			Y^{(s)}_{5}=\alpha_{125}\alpha_{345},
			\qquad
			Y^{(t)}_{6}=\alpha_{236}\alpha_{416}.
		\end{equation}
		Equation~\eqref{eq:reduced-crossing} is linear in these variables.  Every solution of the original quadratic system induces a solution of the linearized system, so solving the latter cannot discard a nonlinear branch. When the relevant ratios are uniquely fixed by the linear system, the correlator gives relations of the form
		$$Y^{(s)}_{5}/Y^{(t)}_{6}=\ldots,\quad Y^{(s)}_{5}/Y^{(s)}_{5'}=\ldots$$ 
		with explicit right-hand sides and nonzero denominators. Ratios are formed only when their denominators are known to be nonzero. Otherwise, the corresponding linear relations are retained without division, so that solutions with vanishing OPE products are not excluded. These relations constrain the original factorized products; the auxiliary variables are not assumed to factorize automatically.
	\end{remark}
	\vspace{2mm}
	\noindent\textbf{Verification of all correlators.}
	
	After determining candidate OPE coefficients from a subset of
	correlators, we check them against the crossing equations for all
	ordered external quartets, including $\braket{AABC}$ and general
	$\braket{ABCD}$ correlators. These finite-precision checks provide numerical evidence for existence. An exact proof would require verifying that the candidates satisfy the complete system exactly.

	\subsection{\texorpdfstring{Summary of results: $s=1$}{Summary of results: s=1}}\label{sec:summary_alpha_s1}
	
	In principle, the entire calculation can be carried out symbolically
	using exact algebraic arithmetic. Indeed, $q_+$ is a root of unity,
	so the quantum-group $6j$ symbols are algebraic. After fixing the
	field normalizations, the bootstrap equations for the reduced OPE
	coefficients therefore form a finite polynomial system with algebraic
	coefficients, and every isolated solution has algebraic coordinates.
	A fully symbolic implementation is nevertheless computationally
	expensive. We instead evaluate the $6j$ symbols at high precision,
	solve the linear systems numerically, and use \texttt{RootApproximant}
	to reconstruct candidate algebraic expressions for the reduced OPE
	coefficients. The reconstructed values are then tested numerically
	against all the bootstrap equations. Algebraic reconstruction
	and small numerical residuals do not by themselves establish exact
	crossing identities or provide certified error bounds.
	
	As described above, it suffices to consider the following nine
	reference cases:
	\begin{equation}
		\begin{split}
			E_6:&\quad (p,q)=(1,12),(5,12)\,, \\
			E_7:&\quad (p,q)=(1,18),(5,18),(7,18)\,, \\
			E_8:&\quad (p,q)=(1,30),(7,30),(11,30),(13,30)\,. \\
		\end{split}
	\end{equation}
	We solve the seed bootstrap equations numerically for each case.
	We emphasize that these equations
	\eqref{eq:reduced-crossing_s=1} remain well defined at $p=1$,
	although this parameter choice does not define a Virasoro minimal model.
	
	\paragraph{Numerical precision.}
	We solved the bootstrap equations at $50$-digit working precision and
	reconstructed algebraic expressions for the OPE coefficients using
	\texttt{RootApproximant}.
	For the subsequent verification, these expressions were evaluated afresh
	at $120$-digit working precision, corresponding to $100$ requested digits
	plus $20$ guard digits. The cached $6j$ symbols were also retained at
	$120$ digits during the verification stage, and the crossing sums were computed at the same precision. 
	We imposed the strict absolute acceptance criterion
	\[
	\Delta \equiv
	\bigl|\mathrm{LHS}-\mathrm{RHS}\bigr| < 10^{-100}
	\]
	on every crossing equation.
	
	\paragraph{Exhaustive crossing checks.}
	A separately implemented Python checker tested the reconstructed OPE coefficients using the cached $6j$ symbols. We enumerated every ordered quartet of fields in the field set used in our bootstrap, including
	repeated fields and the identity. These sets contain $12$, $17$, and $32$
	fields for $E_6$, $E_7$, and $E_8$, respectively.
	For each quartet, we tested every $t$-channel component of the $s$--$t$
	crossing relation allowed by the left- and right-moving chiral fusion
	rules. This includes intermediate-channel pairs outside the physical
	spectrum, for which the corresponding right-hand side must vanish.
	Identity normalization and the vanishing of fusion-forbidden OPE
	coefficients were checked separately.
	
	Table~\ref{tab:exhaustive-crossing-checks} lists the coverage and the
	maximum absolute residual, $\Delta_{\max}$, over all tested crossing
	equations in each case.
	The quartet counts include all orderings and quartets with no admissible
	$t$-channel pair; such quartets contribute no equations.
	Each complex equation is counted once, without separating its real and
	imaginary parts. The counts include vanishing equations and are not
	counts of algebraically independent constraints. All tested equations
	passed.
	
	\paragraph{Code, data and software.}The ancillary package contains exact coefficient data and programs in
	Python and the Wolfram Language, runnable locally on a personal computer.
	Reading the JSON tables, rendering saved reports and checking file hashes
	require only the Python standard library. Repeating the numerical
	crossing checks uses \texttt{mpmath}; regenerating their algebraic inputs,
	solving the bootstrap systems and repeating the exact Galois comparisons
	also require a licensed Wolfram kernel. The package records the tested
	software versions and provides separate commands for inspecting saved
	results and rerunning the calculations.\footnote{The authors used
		Python 3.13.2, \texttt{mpmath} 1.4.1 and Wolfram Mathematica 14.3
		on macOS 26.6.2.}

    \paragraph{OPE coefficient tables and compact formulas.}
    The tables in Sections~\ref{subsec:result_E6}--\ref{subsec:result_E8}
    give the representative reduced coefficients $\alpha_{IJK}$ for these
    nine reference residues. Identity and simple-current couplings are fixed
    in Appendix~\ref{app:simple_current_conventions}. Compact formulas for
    all OPE coefficients are provided in Section~\ref{sec:reconstruction}.

    \paragraph{Comparison with the earlier calculations.}
    We have compared our results with the earlier WZW and unitary minimal-model calculations \cite{Fuchs:1989kz,Fuchs:1989zp,Douglas:1988rv,Kato:1988ct,Furlan:1989ra,Petkova:1994zs} and with the recent $E_6$ formulas of Nivesvivat and Ribault~\cite{Nivesvivat:2025odb}. After matching conventions, we find agreement for the magnitudes and zeros in the overlapping cases. For $E_6$ and $E_7$, this comparison also includes the complete signed relative coefficients in the common reference sector \cite{Petkova:1994zs,Nivesvivat:2025odb}. Appendix~\ref{app:literature-comparison} presents the comparison paper by paper and specifies the scope of the signed checks.
	
	\begin{table}[htbp]
		\centering
		\caption{Exhaustive numerical crossing checks for the nine reference cases. The maximal absolute residual $\Delta_{\max}$ is rounded to six significant figures. All equations satisfy $\Delta<10^{-100}$.}
		\label{tab:exhaustive-crossing-checks}
		\setlength{\tabcolsep}{10pt}
		\renewcommand{\arraystretch}{1.2}
		\resizebox{\textwidth}{!}{%
			\begin{tabular}{crrc}
				\toprule
				$(p,q)$
				& \shortstack{Number of\\ordered quartets}
				& \shortstack{Number of\\equations checked}
				& \shortstack{Maximal absolute\\residual $\Delta_{\max}$} \\
				\midrule
				$(1,12)$  & $20\,736$     & $18\,688$
				& $1.10369\times10^{-119}$ \\
				$(5,12)$  & $20\,736$     & $18\,688$
				& $9.44912\times10^{-119}$ \\
				\midrule
				$(1,18)$  & $83\,521$     & $359\,073$
				& $1.20050\times10^{-119}$ \\
				$(5,18)$  & $83\,521$     & $359\,073$
				& $6.56792\times10^{-118}$ \\
				$(7,18)$  & $83\,521$     & $359\,073$
				& $1.85884\times10^{-118}$ \\
				\midrule
				$(1,30)$  & $1\,048\,576$ & $7\,747\,200$
				& $1.00687\times10^{-119}$ \\
				$(7,30)$  & $1\,048\,576$ & $7\,747\,200$
				& $5.70263\times10^{-117}$ \\
				$(11,30)$ & $1\,048\,576$ & $7\,747\,200$
				& $2.54900\times10^{-118}$ \\
				$(13,30)$ & $1\,048\,576$ & $7\,747\,200$
				& $4.35470\times10^{-118}$ \\
				\bottomrule
			\end{tabular}%
		}
	\end{table}

	\clearpage
	\subsubsection{\texorpdfstring{$E_6$ series}{E6 series}}\label{subsec:result_E6}
	For the $E_6$ series, we fix the remaining operator-sign freedom by setting the OPE coefficients for the following triplets to the values listed in the corresponding tables:
	\begin{equation}
		\begin{split}
			&(0,3)(0,3)(0,3) \\
			&(3,0)(3,0)(3,0) \\
			&(3,3)(3,3)(3,3) \\
			&(\tfrac{3}{2},\tfrac{3}{2})(\tfrac{3}{2},\tfrac{7}{2})(2,5) \\
			&(\tfrac{3}{2},\tfrac{3}{2})(\tfrac{7}{2},\tfrac{3}{2})(5,2) \\
		\end{split}
	\end{equation}
	We verified at the pre-bootstrap stage that all these coefficients are nonzero. For each such triplet, we identify the representative OPE coefficient $x_k$ of its triplet orbit. We then fix the sign of $x_k$ by choosing the principal value of the square root $\sqrt{(x_k)^2}$, where $(x_k)^2$ is determined by the $\braket{ABBA}$ bootstrap equations.
	
	The representative reduced coefficients for $(p,q)=(1,12)$ and $(5,12)$ are listed below.
	
	\begin{table}[ht]
		\centering
		\caption{Representative reduced coefficients for $(p,q)=(1,12)$.}
		\label{tab:E6-I1-results}
		\setlength{\tabcolsep}{4pt}
		\renewcommand{\arraystretch}{1.45}
		\resizebox{\textwidth}{!}{%
			\begin{tabular}{c@{\hspace{6pt}}cccccc}
				\toprule
				No. & 1 & 2 & 3 & 4 & 5 & 6 \\
				\midrule
				Representative
				& $(0,3)(0,3)(0,3)$
				& $(0,3)(\tfrac32,\tfrac32)(\tfrac32,\tfrac32)$
				& $(0,3)(\tfrac32,\tfrac32)(\tfrac32,\tfrac72)$
				& $(0,3)(\tfrac32,\tfrac72)(\tfrac32,\tfrac72)$
				& $(0,3)(2,2)(2,2)$
				& $(0,3)(2,2)(2,5)$ \\
				$\alpha_{IJK}$
				& $\mathrm{i}\sqrt[4]{2}$
				& $\dfrac{\mathrm{i}}{\sqrt[4]{2}}$
				& $\sqrt[4]{\dfrac{3}{2}}$
				& $-\dfrac{\mathrm{i}}{\sqrt[4]{2}}$
				& $-\mathrm{i}\sqrt[4]{2}$
				& $-\mathrm{i}$ \\
				\midrule
				No. & 7 & 8 & 9 & 10 & 11 & 12 \\
				\midrule
				Representative
				& $(\tfrac32,\tfrac32)(\tfrac32,\tfrac32)(2,2)$
				& $(\tfrac32,\tfrac32)(\tfrac32,\tfrac32)(3,0)$
				& $(\tfrac32,\tfrac32)(\tfrac32,\tfrac32)(3,3)$
				& $(\tfrac32,\tfrac32)(\tfrac32,\tfrac72)(2,2)$
				& $(\tfrac32,\tfrac32)(\tfrac32,\tfrac72)(2,5)$
				& $(\tfrac32,\tfrac32)(\tfrac32,\tfrac72)(3,3)$ \\
				$\alpha_{IJK}$
				& $\sqrt{\dfrac{3}{2}}$
				& $-\dfrac{\mathrm{i}}{\sqrt[4]{2}}$
				& $\dfrac{1}{\sqrt{2}}$
				& $\mathrm{i}\sqrt[4]{\dfrac{3}{4}}$
				& $\mathrm{i}\sqrt[4]{\dfrac{3}{2}}$
				& $-\mathrm{i}\sqrt[4]{\dfrac{3}{4}}$ \\
				\midrule
				No. & 13 & 14 & 15 & 16 & 17 & 18 \\
				\midrule
				Representative
				& $(\tfrac32,\tfrac72)(\tfrac32,\tfrac72)(2,2)$
				& $(\tfrac32,\tfrac72)(\tfrac32,\tfrac72)(3,0)$
				& $(\tfrac32,\tfrac72)(\tfrac32,\tfrac72)(3,3)$
				& $(2,2)(2,2)(3,0)$
				& $(2,2)(2,2)(3,3)$
				& $(2,5)(2,5)(3,0)$ \\
				$\alpha_{IJK}$
				& $-\sqrt{\dfrac{3}{2}}$
				& $-\dfrac{\mathrm{i}}{\sqrt[4]{2}}$
				& $-\dfrac{1}{\sqrt{2}}$
				& $\mathrm{i}\sqrt[4]{2}$
				& $\sqrt{2}$
				& $\mathrm{i}\sqrt[4]{2}$ \\
				\bottomrule
			\end{tabular}%
		}
	\end{table}
	
	\begin{table*}[ht]
		\centering
		\caption{Representative reduced coefficients for $(p,q)=(5,12)$.}
		\label{tab:E6-I5-results}
		\setlength{\tabcolsep}{4pt}
		\renewcommand{\arraystretch}{1.45}
		\resizebox{\textwidth}{!}{%
			\begin{tabular}{c@{\hspace{6pt}}cccccc}
				\toprule
				No. & 1 & 2 & 3 & 4 & 5 & 6 \\
				\midrule
				Representative
				& $(0,3)(0,3)(0,3)$
				& $(0,3)(\tfrac32,\tfrac32)(\tfrac32,\tfrac32)$
				& $(0,3)(\tfrac32,\tfrac32)(\tfrac32,\tfrac72)$
				& $(0,3)(\tfrac32,\tfrac72)(\tfrac32,\tfrac72)$
				& $(0,3)(2,2)(2,2)$
				& $(0,3)(2,2)(2,5)$ \\
				$\alpha_{IJK}$
				& $\mathrm{i}\sqrt[4]{2}$
				& $\dfrac{\mathrm{i}}{\sqrt[4]{2}}$
				& $\sqrt[4]{\dfrac{3}{2}}$
				& $\dfrac{\mathrm{i}}{\sqrt[4]{2}}$
				& $\mathrm{i}\sqrt[4]{2}$
				& $\mathrm{i}$ \\
				\midrule
				No. & 7 & 8 & 9 & 10 & 11 & 12 \\
				\midrule
				Representative
				& $(\tfrac32,\tfrac32)(\tfrac32,\tfrac32)(2,2)$
				& $(\tfrac32,\tfrac32)(\tfrac32,\tfrac32)(3,0)$
				& $(\tfrac32,\tfrac32)(\tfrac32,\tfrac32)(3,3)$
				& $(\tfrac32,\tfrac32)(\tfrac32,\tfrac72)(2,2)$
				& $(\tfrac32,\tfrac32)(\tfrac32,\tfrac72)(2,5)$
				& $(\tfrac32,\tfrac32)(\tfrac32,\tfrac72)(3,3)$ \\
				$\alpha_{IJK}$
				& $\sqrt{\dfrac{3}{2}}$
				& $\dfrac{\mathrm{i}}{\sqrt[4]{2}}$
				& $\dfrac{1}{\sqrt{2}}$
				& $\mathrm{i}\sqrt[4]{\dfrac{3}{4}}$
				& $\mathrm{i}\sqrt[4]{\dfrac{3}{2}}$
				& $-\mathrm{i}\sqrt[4]{\dfrac{3}{4}}$ \\
				\midrule
				No. & 13 & 14 & 15 & 16 & 17 & 18 \\
				\midrule
				Representative
				& $(\tfrac32,\tfrac72)(\tfrac32,\tfrac72)(2,2)$
				& $(\tfrac32,\tfrac72)(\tfrac32,\tfrac72)(3,0)$
				& $(\tfrac32,\tfrac72)(\tfrac32,\tfrac72)(3,3)$
				& $(2,2)(2,2)(3,0)$
				& $(2,2)(2,2)(3,3)$
				& $(2,5)(2,5)(3,0)$ \\
				$\alpha_{IJK}$
				& $\sqrt{\dfrac{3}{2}}$
				& $\dfrac{\mathrm{i}}{\sqrt[4]{2}}$
				& $\dfrac{1}{\sqrt{2}}$
				& $\mathrm{i}\sqrt[4]{2}$
				& $\sqrt{2}$
				& $\mathrm{i}\sqrt[4]{2}$ \\
				\bottomrule
			\end{tabular}%
		}
	\end{table*}
	\clearpage
	\newpage
	\subsubsection{\texorpdfstring{$E_7$ series}{E7 series}}\label{subsec:result_E7}
	For the $E_7$ series, we fix the remaining operator-sign freedom by assigning the coefficients of the following triplets the values listed in the corresponding tables: 
	\begin{equation}
		\begin{split}
			&(2,2)(2,2)(2,2) \\
			&(3,3)(3,3)(3,3) \\
			&(4,4)(4,4)(4,4) \\
			&(4,1)(4,1)(4,1) \\
			&(1,4)(1,4)(1,4) \\
		\end{split}
	\end{equation}
	The representative reduced coefficients for $(p,q)=(1,18)$, $(5,18)$ and $(7,18)$ are listed below.
	
	\begin{table}[ht]
		\centering
		\caption{Representative reduced coefficients for $(p,q)=(1,18)$.}
		\label{tab:E7-I1-results}
		\setlength{\tabcolsep}{4pt}
		\renewcommand{\arraystretch}{1.45}
		\resizebox{\textwidth}{!}{%
			\begin{tabular}{c@{\hspace{6pt}}cccccccc}
				\toprule
				No. & 1 & 2 & 3 & 4 & 5 & 6 & 7 & 8 \\
				\midrule
				Representative
				& $(1,4)(1,4)(1,4)$
				& $(1,4)(1,4)(2,2)$
				& $(1,4)(2,2)(2,2)$
				& $(1,4)(2,2)(3,3)$
				& $(1,4)(3,3)(3,3)$
				& $(1,4)(3,3)(4,1)$
				& $(1,4)(3,3)(4,4)$
				& $(1,4)(4,1)(4,4)$
				\\
				$\alpha_{IJK}$
				& $\mathrm{i}\sqrt[4]{2}$
				& $-1$
				& $\dfrac{\mathrm{i}}{\sqrt[4]{2}}$
				& $-\dfrac{\mathrm{i}}{\sqrt[4]{2}}$
				& $-\dfrac{\mathrm{i}}{\sqrt[4]{2}}$
				& $-\dfrac{1}{\sqrt{2}}$
				& $\dfrac{\mathrm{i}}{\sqrt[4]{2}}$
				& $-1$
				\\
				\midrule
				
				No. & 9 & 10 & 11 & 12 & 13 & 14 & 15 & 16 \\
				\midrule
				Representative
				& $(1,4)(4,4)(4,4)$
				& $(2,2)(2,2)(2,2)$
				& $(2,2)(2,2)(3,3)$
				& $(2,2)(2,2)(4,1)$
				& $(2,2)(2,2)(4,4)$
				& $(2,2)(3,3)(3,3)$
				& $(2,2)(3,3)(3,5)$
				& $(2,2)(3,3)(4,1)$
				\\
				$\alpha_{IJK}$
				& $0$
				& $1$
				& $1$
				& $\dfrac{\mathrm{i}}{\sqrt[4]{2}}$
				& $\dfrac{1}{\sqrt{2}}$
				& $-\dfrac{1}{2}$
				& $\dfrac{\mathrm{i}\sqrt{3}}{2}$
				& $-\dfrac{\mathrm{i}}{\sqrt[4]{2}}$
				\\
				\midrule
				
				No. & 17 & 18 & 19 & 20 & 21 & 22 & 23 & 24 \\
				\midrule
				Representative
				& $(2,2)(3,3)(4,4)$
				& $(2,2)(3,3)(5,3)$
				& $(2,2)(3,3)(5,5)$
				& $(2,2)(4,1)(4,1)$
				& $(2,2)(4,4)(4,4)$
				& $(3,3)(3,3)(3,3)$
				& $(3,3)(3,3)(3,5)$
				& $(3,3)(3,3)(4,1)$
				\\
				$\alpha_{IJK}$
				& $\dfrac{1}{\sqrt{2}}$
				& $\dfrac{\mathrm{i}\sqrt{3}}{2}$
				& $\dfrac{1}{2}$
				& $-1$
				& $1$
				& $1$
				& $0$
				& $-\dfrac{\mathrm{i}}{\sqrt[4]{2}}$
				\\
				\midrule
				
				No. & 25 & 26 & 27 & 28 & 29 & 30 & 31 & 32 \\
				\midrule
				Representative
				& $(3,3)(3,3)(4,4)$
				& $(3,3)(3,3)(5,3)$
				& $(3,3)(3,3)(5,5)$
				& $(3,3)(4,1)(4,4)$
				& $(3,3)(4,4)(4,4)$
				& $(4,1)(4,1)(4,1)$
				& $(4,1)(4,4)(4,4)$
				& $(4,4)(4,4)(4,4)$
				\\
				$\alpha_{IJK}$
				& $\dfrac{1}{\sqrt{2}}$
				& $0$
				& $-1$
				& $\dfrac{\mathrm{i}}{\sqrt[4]{2}}$
				& $0$
				& $\mathrm{i}\sqrt[4]{2}$
				& $0$
				& $\sqrt{2}$
				\\
				\bottomrule
			\end{tabular}%
		}
	\end{table}
	
	\begin{table}[ht]
		\centering
		\caption{Representative reduced coefficients for $(p,q)=(5,18)$.}
		\label{tab:E7-I5-results}
		\setlength{\tabcolsep}{4pt}
		\renewcommand{\arraystretch}{1.45}
		\resizebox{\textwidth}{!}{%
			\begin{tabular}{c@{\hspace{6pt}}cccccccc}
				\toprule
				No. & 1 & 2 & 3 & 4 & 5 & 6 & 7 & 8 \\
				\midrule
				Representative
				& $(1,4)(1,4)(1,4)$
				& $(1,4)(1,4)(2,2)$
				& $(1,4)(2,2)(2,2)$
				& $(1,4)(2,2)(3,3)$
				& $(1,4)(3,3)(3,3)$
				& $(1,4)(3,3)(4,1)$
				& $(1,4)(3,3)(4,4)$
				& $(1,4)(4,1)(4,4)$
				\\
				$\alpha_{IJK}$
				& $\sqrt[4]{2}$
				& $-1$
				& $-\dfrac{1}{\sqrt[4]{2}}$
				& $\dfrac{1}{\sqrt[4]{2}}$
				& $\dfrac{1}{\sqrt[4]{2}}$
				& $\dfrac{1}{\sqrt{2}}$
				& $\dfrac{1}{\sqrt[4]{2}}$
				& $1$
				\\
				\midrule
				
				No. & 9 & 10 & 11 & 12 & 13 & 14 & 15 & 16 \\
				\midrule
				Representative
				& $(1,4)(4,4)(4,4)$
				& $(2,2)(2,2)(2,2)$
				& $(2,2)(2,2)(3,3)$
				& $(2,2)(2,2)(4,1)$
				& $(2,2)(2,2)(4,4)$
				& $(2,2)(3,3)(3,3)$
				& $(2,2)(3,3)(3,5)$
				& $(2,2)(3,3)(4,1)$
				\\
				$\alpha_{IJK}$
				& $0$
				& $1$
				& $1$
				& $-\dfrac{1}{\sqrt[4]{2}}$
				& $\dfrac{1}{\sqrt{2}}$
				& $-\dfrac{1}{2}$
				& $-\dfrac{\sqrt{3}}{2}$
				& $\dfrac{1}{\sqrt[4]{2}}$
				\\
				\midrule
				
				No. & 17 & 18 & 19 & 20 & 21 & 22 & 23 & 24 \\
				\midrule
				Representative
				& $(2,2)(3,3)(4,4)$
				& $(2,2)(3,3)(5,3)$
				& $(2,2)(3,3)(5,5)$
				& $(2,2)(4,1)(4,1)$
				& $(2,2)(4,4)(4,4)$
				& $(3,3)(3,3)(3,3)$
				& $(3,3)(3,3)(3,5)$
				& $(3,3)(3,3)(4,1)$
				\\
				$\alpha_{IJK}$
				& $\dfrac{1}{\sqrt{2}}$
				& $-\dfrac{\sqrt{3}}{2}$
				& $\dfrac{1}{2}$
				& $-1$
				& $1$
				& $1$
				& $0$
				& $\dfrac{1}{\sqrt[4]{2}}$
				\\
				\midrule
				
				No. & 25 & 26 & 27 & 28 & 29 & 30 & 31 & 32 \\
				\midrule
				Representative
				& $(3,3)(3,3)(4,4)$
				& $(3,3)(3,3)(5,3)$
				& $(3,3)(3,3)(5,5)$
				& $(3,3)(4,1)(4,4)$
				& $(3,3)(4,4)(4,4)$
				& $(4,1)(4,1)(4,1)$
				& $(4,1)(4,4)(4,4)$
				& $(4,4)(4,4)(4,4)$
				\\
				$\alpha_{IJK}$
				& $\dfrac{1}{\sqrt{2}}$
				& $0$
				& $-1$
				& $\dfrac{1}{\sqrt[4]{2}}$
				& $0$
				& $\sqrt[4]{2}$
				& $0$
				& $\sqrt{2}$
				\\
				\bottomrule
			\end{tabular}%
		}
	\end{table}
	
	\begin{table}[ht]
		\centering
		\caption{Representative reduced coefficients for $(p,q)=(7,18)$.}
		\label{tab:E7-I7-results}
		\setlength{\tabcolsep}{4pt}
		\renewcommand{\arraystretch}{1.45}
		\resizebox{\textwidth}{!}{%
			\begin{tabular}{c@{\hspace{6pt}}cccccccc}
				\toprule
				No. & 1 & 2 & 3 & 4 & 5 & 6 & 7 & 8 \\
				\midrule
				Representative
				& $(1,4)(1,4)(1,4)$
				& $(1,4)(1,4)(2,2)$
				& $(1,4)(2,2)(2,2)$
				& $(1,4)(2,2)(3,3)$
				& $(1,4)(3,3)(3,3)$
				& $(1,4)(3,3)(4,1)$
				& $(1,4)(3,3)(4,4)$
				& $(1,4)(4,1)(4,4)$
				\\
				$\alpha_{IJK}$
				& $\sqrt[4]{2}$
				& $1$
				& $\dfrac{1}{\sqrt[4]{2}}$
				& $\dfrac{1}{\sqrt[4]{2}}$
				& $\dfrac{1}{\sqrt[4]{2}}$
				& $\dfrac{1}{\sqrt{2}}$
				& $-\dfrac{1}{\sqrt[4]{2}}$
				& $1$
				\\
				\midrule
				
				No. & 9 & 10 & 11 & 12 & 13 & 14 & 15 & 16 \\
				\midrule
				Representative
				& $(1,4)(4,4)(4,4)$
				& $(2,2)(2,2)(2,2)$
				& $(2,2)(2,2)(3,3)$
				& $(2,2)(2,2)(4,1)$
				& $(2,2)(2,2)(4,4)$
				& $(2,2)(3,3)(3,3)$
				& $(2,2)(3,3)(3,5)$
				& $(2,2)(3,3)(4,1)$
				\\
				$\alpha_{IJK}$
				& $0$
				& $1$
				& $1$
				& $\dfrac{1}{\sqrt[4]{2}}$
				& $\dfrac{1}{\sqrt{2}}$
				& $-\dfrac{1}{2}$
				& $-\dfrac{\sqrt{3}}{2}$
				& $\dfrac{1}{\sqrt[4]{2}}$
				\\
				\midrule
				
				No. & 17 & 18 & 19 & 20 & 21 & 22 & 23 & 24 \\
				\midrule
				Representative
				& $(2,2)(3,3)(4,4)$
				& $(2,2)(3,3)(5,3)$
				& $(2,2)(3,3)(5,5)$
				& $(2,2)(4,1)(4,1)$
				& $(2,2)(4,4)(4,4)$
				& $(3,3)(3,3)(3,3)$
				& $(3,3)(3,3)(3,5)$
				& $(3,3)(3,3)(4,1)$
				\\
				$\alpha_{IJK}$
				& $\dfrac{1}{\sqrt{2}}$
				& $-\dfrac{\sqrt{3}}{2}$
				& $-\dfrac{1}{2}$
				& $1$
				& $1$
				& $1$
				& $0$
				& $\dfrac{1}{\sqrt[4]{2}}$
				\\
				\midrule
				
				No. & 25 & 26 & 27 & 28 & 29 & 30 & 31 & 32 \\
				\midrule
				Representative
				& $(3,3)(3,3)(4,4)$
				& $(3,3)(3,3)(5,3)$
				& $(3,3)(3,3)(5,5)$
				& $(3,3)(4,1)(4,4)$
				& $(3,3)(4,4)(4,4)$
				& $(4,1)(4,1)(4,1)$
				& $(4,1)(4,4)(4,4)$
				& $(4,4)(4,4)(4,4)$
				\\
				$\alpha_{IJK}$
				& $\dfrac{1}{\sqrt{2}}$
				& $0$
				& $1$
				& $-\dfrac{1}{\sqrt[4]{2}}$
				& $0$
				& $\sqrt[4]{2}$
				& $0$
				& $\sqrt{2}$
				\\
				\bottomrule
			\end{tabular}%
		}
	\end{table}

	\subsubsection{\texorpdfstring{$E_8$ series}{E8 series}}\label{subsec:result_E8}
	For the $E_8$ series, we fix the remaining operator-sign freedom by assigning the coefficients of the following triplets the values listed in the corresponding tables: 
	\begin{equation}
		\begin{split}
			&(3,3)(3,3)(3,3) \\
			&(5,5)(5,5)(5,5) \\
			&(6,6)(6,6)(6,6) \\
			&(0,5)(0,5)(0,5) \\
			&(5,0)(5,0)(5,0) \\
			&(3,6)(3,6)(3,6) \\
			&(6,3)(6,3)(6,3) \\
			&(5,9)(5,9)(5,9) \\
			&(9,5)(9,5)(9,5) \\
		\end{split}
	\end{equation}
	The representative reduced coefficients for $(p,q)=(1,30)$, $(7,30)$, $(11,30)$ and $(13,30)$ are listed below.
	
	\clearpage
	\begingroup
    \noindent\begin{minipage}{\textwidth}
	\fontsize{6}{6.5}\selectfont
	\setlength{\tabcolsep}{1.6pt}
	\renewcommand{\arraystretch}{1.15}
	\setlength{\aboverulesep}{0.4ex}
	\setlength{\belowrulesep}{0.65ex}
	\refstepcounter{table}\label{tab:E8-I1-results}
	\begin{center}
		{\normalsize\textbf{Table~\thetable: Representative reduced coefficients for $(p,q)=(1,30)$.}}
		\par\vspace{3pt}
		\makebox[\textwidth][c]{\resizebox{1.15\textwidth}{!}{%
				\begin{tabular}{c@{\hspace{3pt}}cccccccc}
					\toprule
					No. & 1 & 2 & 3 & 4 & 5 & 6 & 7 & 8 \\*
					\midrule
					Triplet
					& $(0,5)(0,5)(0,5)$
					& $(0,5)(0,5)(0,9)$
					& $(0,5)(3,3)(3,3)$
					& $(0,5)(3,3)(3,6)$
					& $(0,5)(3,3)(3,8)$
					& $(0,5)(3,6)(3,6)$
					& $(0,5)(3,6)(3,8)$
					& $(0,5)(5,0)(5,5)$ \\*[1.5pt]
					$\alpha_{IJK}$
					& $\mathrm{i}\sqrt{\dfrac{3}{2}}$
					& $\dfrac{\mathrm{i}\sqrt[4]{5}}{\sqrt{2}}$
					& $-\mathrm{i}$
					& $-\dfrac{\sqrt[4]{5}}{\sqrt{2}}$
					& $\dfrac{\mathrm{i}}{\sqrt{2}}$
					& $-\dfrac{\mathrm{i}}{\sqrt{2}}$
					& $\dfrac{\sqrt[4]{5}}{\sqrt{2}}$
					& $-1$ \\[2pt]
					\midrule
					
					No. & 9 & 10 & 11 & 12 & 13 & 14 & 15 & 16 \\*
					\midrule
					Triplet
					& $(0,5)(5,5)(5,5)$
					& $(0,5)(5,5)(5,9)$
					& $(0,5)(6,3)(6,3)$
					& $(0,5)(6,3)(6,6)$
					& $(0,5)(6,3)(6,8)$
					& $(0,5)(6,6)(6,6)$
					& $(0,5)(6,6)(6,8)$
					& $(3,3)(3,3)(3,3)$ \\*[1.5pt]
					$\alpha_{IJK}$
					& $\mathrm{i}\sqrt{\dfrac{3}{2}}$
					& $\dfrac{\mathrm{i}\sqrt[4]{5}}{\sqrt{2}}$
					& $-\mathrm{i}$
					& $\dfrac{\sqrt[4]{5}}{\sqrt{2}}$
					& $-\dfrac{\mathrm{i}}{\sqrt{2}}$
					& $-\dfrac{\mathrm{i}}{\sqrt{2}}$
					& $\dfrac{\sqrt[4]{5}}{\sqrt{2}}$
					& $\sqrt{\dfrac{5}{3}}$ \\[2pt]
					\midrule
					
					No. & 17 & 18 & 19 & 20 & 21 & 22 & 23 & 24 \\*
					\midrule
					Triplet
					& $(3,3)(3,3)(3,6)$
					& $(3,3)(3,3)(5,0)$
					& $(3,3)(3,3)(5,5)$
					& $(3,3)(3,3)(6,3)$
					& $(3,3)(3,3)(6,6)$
					& $(3,3)(3,6)(3,6)$
					& $(3,3)(3,6)(3,8)$
					& $(3,3)(3,6)(5,5)$ \\*[1.5pt]
					$\alpha_{IJK}$
					& $\dfrac{\mathrm{i}\sqrt[4]{5}}{\sqrt{3}}$
					& $-\mathrm{i}$
					& $1$
					& $\dfrac{\mathrm{i}\sqrt[4]{5}}{\sqrt{3}}$
					& $\dfrac{1}{\sqrt{3}}$
					& $-\sqrt{\dfrac{5}{6}}$
					& $\dfrac{\mathrm{i}\sqrt[4]{5}}{\sqrt{2}}$
					& $-\dfrac{\mathrm{i}\sqrt[4]{5}}{\sqrt{2}}$ \\[2pt]
					\midrule
					
					No. & 25 & 26 & 27 & 28 & 29 & 30 & 31 & 32 \\*
					\midrule
					Triplet
					& $(3,3)(3,6)(5,9)$
					& $(3,3)(3,6)(6,3)$
					& $(3,3)(3,6)(6,6)$
					& $(3,3)(3,6)(6,8)$
					& $(3,3)(5,0)(6,3)$
					& $(3,3)(5,0)(8,3)$
					& $(3,3)(5,5)(6,3)$
					& $(3,3)(5,5)(6,6)$ \\*[1.5pt]
					$\alpha_{IJK}$
					& $\dfrac{\mathrm{i}}{\sqrt{2}}$
					& $-\dfrac{1}{\sqrt{3}}$
					& $\dfrac{\mathrm{i}\sqrt[4]{5}}{\sqrt{6}}$
					& $\dfrac{1}{\sqrt{2}}$
					& $-\dfrac{\sqrt[4]{5}}{\sqrt{2}}$
					& $-\dfrac{\mathrm{i}}{\sqrt{2}}$
					& $-\dfrac{\mathrm{i}\sqrt[4]{5}}{\sqrt{2}}$
					& $\dfrac{\sqrt{5}}{2}$ \\[2pt]
					\midrule
					
					No. & 33 & 34 & 35 & 36 & 37 & 38 & 39 & 40 \\*
					\midrule
					Triplet
					& $(3,3)(5,5)(6,8)$
					& $(3,3)(5,5)(8,3)$
					& $(3,3)(5,5)(8,6)$
					& $(3,3)(5,5)(8,8)$
					& $(3,3)(6,3)(6,3)$
					& $(3,3)(6,3)(6,6)$
					& $(3,3)(6,3)(8,3)$
					& $(3,3)(6,3)(8,6)$ \\*[1.5pt]
					$\alpha_{IJK}$
					& $\dfrac{\mathrm{i}\sqrt[4]{5}}{2}$
					& $\dfrac{1}{\sqrt{2}}$
					& $\dfrac{\mathrm{i}\sqrt[4]{5}}{2}$
					& $-\dfrac{1}{2}$
					& $-\sqrt{\dfrac{5}{6}}$
					& $\dfrac{\mathrm{i}\sqrt[4]{5}}{\sqrt{6}}$
					& $\dfrac{\mathrm{i}\sqrt[4]{5}}{\sqrt{2}}$
					& $\dfrac{1}{\sqrt{2}}$ \\[2pt]
					\midrule
					
					No. & 41 & 42 & 43 & 44 & 45 & 46 & 47 & 48 \\*
					\midrule
					Triplet
					& $(3,3)(6,6)(6,6)$
					& $(3,3)(6,6)(6,8)$
					& $(3,3)(6,6)(8,6)$
					& $(3,3)(6,6)(8,8)$
					& $(3,6)(3,6)(3,6)$
					& $(3,6)(3,6)(3,8)$
					& $(3,6)(3,6)(5,0)$
					& $(3,6)(3,6)(5,5)$ \\*[1.5pt]
					$\alpha_{IJK}$
					& $\dfrac{\sqrt{5/3}}{2}$
					& $-\dfrac{\mathrm{i}\sqrt[4]{5}}{2}$
					& $-\dfrac{\mathrm{i}\sqrt[4]{5}}{2}$
					& $-\dfrac{\sqrt{3}}{2}$
					& $\mathrm{i}\sqrt{\dfrac{2}{3}}\sqrt[4]{5}$
					& $0$
					& $-\mathrm{i}$
					& $\dfrac{1}{\sqrt{2}}$ \\[2pt]
					\midrule
					
					No. & 49 & 50 & 51 & 52 & 53 & 54 & 55 & 56 \\*
					\midrule
					Triplet
					& $(3,6)(3,6)(5,9)$
					& $(3,6)(3,6)(6,3)$
					& $(3,6)(3,6)(6,6)$
					& $(3,6)(3,6)(6,8)$
					& $(3,6)(3,6)(6,11)$
					& $(3,6)(5,0)(6,6)$
					& $(3,6)(5,0)(8,6)$
					& $(3,6)(5,5)(6,3)$ \\*[1.5pt]
					$\alpha_{IJK}$
					& $-\dfrac{\sqrt[4]{5}}{\sqrt{2}}$
					& $-\dfrac{\mathrm{i}\sqrt[4]{5}}{\sqrt{6}}$
					& $\sqrt{\dfrac{2}{3}}$
					& $0$
					& $-\dfrac{\mathrm{i}}{\sqrt{2}}$
					& $\dfrac{\sqrt[4]{5}}{\sqrt{2}}$
					& $\dfrac{\mathrm{i}}{\sqrt{2}}$
					& $-\dfrac{\sqrt{5}}{2}$ \\[2pt]
					\midrule
					No. & 57 & 58 & 59 & 60 & 61 & 62 & 63 & 64 \\*
					\midrule
					Triplet
					& $(3,6)(5,5)(6,6)$
					& $(3,6)(5,5)(6,8)$
					& $(3,6)(5,5)(6,11)$
					& $(3,6)(5,5)(8,3)$
					& $(3,6)(5,5)(8,6)$
					& $(3,6)(5,5)(8,8)$
					& $(3,6)(5,5)(8,11)$
					& $(3,6)(6,3)(6,3)$ \\*[1.5pt]
					$\alpha_{IJK}$
					& $\dfrac{\mathrm{i}\sqrt[4]{5}}{2}$
					& $\dfrac{\sqrt{5}}{2}$
					& $-\dfrac{\sqrt[4]{5}}{2}$
					& $-\dfrac{\mathrm{i}\sqrt[4]{5}}{2}$
					& $-\dfrac{1}{2}$
					& $\dfrac{\mathrm{i}\sqrt[4]{5}}{2}$
					& $-\dfrac{\mathrm{i}}{2}$
					& $-\dfrac{\mathrm{i}\sqrt[4]{5}}{\sqrt{6}}$ \\[2pt]
					\midrule
					
					No. & 65 & 66 & 67 & 68 & 69 & 70 & 71 & 72 \\*
					\midrule
					Triplet
					& $(3,6)(6,3)(6,6)$
					& $(3,6)(6,3)(6,8)$
					& $(3,6)(6,3)(8,3)$
					& $(3,6)(6,3)(8,6)$
					& $(3,6)(6,3)(8,8)$
					& $(3,6)(6,6)(6,6)$
					& $(3,6)(6,6)(6,8)$
					& $(3,6)(6,6)(8,6)$ \\*[1.5pt]
					$\alpha_{IJK}$
					& $-\dfrac{\sqrt{5/3}}{2}$
					& $-\dfrac{\mathrm{i}\sqrt[4]{5}}{2}$
					& $-\dfrac{1}{\sqrt{2}}$
					& $\dfrac{\mathrm{i}\sqrt[4]{5}}{2}$
					& $-\dfrac{\sqrt{3}}{2}$
					& $-\dfrac{\mathrm{i}\sqrt[4]{5}}{\sqrt{3}}$
					& $0$
					& $-1$ \\[2pt]
					\midrule
					
					No. & 73 & 74 & 75 & 76 & 77 & 78 & 79 & 80 \\*
					\midrule
					Triplet
					& $(3,6)(6,6)(8,8)$
					& $(5,0)(5,0)(5,0)$
					& $(5,0)(5,0)(9,0)$
					& $(5,0)(5,5)(5,5)$
					& $(5,0)(5,5)(9,5)$
					& $(5,0)(6,3)(6,3)$
					& $(5,0)(6,3)(8,3)$
					& $(5,0)(6,6)(6,6)$ \\*[1.5pt]
					$\alpha_{IJK}$
					& $0$
					& $\mathrm{i}\sqrt{\dfrac{3}{2}}$
					& $\dfrac{\mathrm{i}\sqrt[4]{5}}{\sqrt{2}}$
					& $\mathrm{i}\sqrt{\dfrac{3}{2}}$
					& $\dfrac{\mathrm{i}\sqrt[4]{5}}{\sqrt{2}}$
					& $-\dfrac{\mathrm{i}}{\sqrt{2}}$
					& $\dfrac{\sqrt[4]{5}}{\sqrt{2}}$
					& $-\dfrac{\mathrm{i}}{\sqrt{2}}$ \\[2pt]
					\midrule
					
					No. & 81 & 82 & 83 & 84 & 85 & 86 & 87 & 88 \\*
					\midrule
					Triplet
					& $(5,0)(6,6)(8,6)$
					& $(5,5)(5,5)(5,5)$
					& $(5,5)(5,5)(5,9)$
					& $(5,5)(5,5)(9,5)$
					& $(5,5)(5,5)(9,9)$
					& $(5,5)(6,3)(6,3)$
					& $(5,5)(6,3)(6,6)$
					& $(5,5)(6,3)(6,8)$ \\*[1.5pt]
					$\alpha_{IJK}$
					& $\dfrac{\sqrt[4]{5}}{\sqrt{2}}$
					& $\dfrac{3}{2}$
					& $\dfrac{\sqrt{3}\sqrt[4]{5}}{2}$
					& $\dfrac{\sqrt{3}\sqrt[4]{5}}{2}$
					& $-\dfrac{\sqrt{5}}{2}$
					& $\dfrac{1}{\sqrt{2}}$
					& $\dfrac{\mathrm{i}\sqrt[4]{5}}{2}$
					& $\dfrac{1}{2}$ \\[2pt]
					\midrule
					
					No. & 89 & 90 & 91 & 92 & 93 & 94 & 95 & 96 \\*
					\midrule
					Triplet
					& $(5,5)(6,3)(8,3)$
					& $(5,5)(6,3)(8,6)$
					& $(5,5)(6,3)(8,8)$
					& $(5,5)(6,6)(6,6)$
					& $(5,5)(6,6)(6,8)$
					& $(5,5)(6,6)(8,6)$
					& $(5,5)(6,6)(8,8)$
					& $(6,3)(6,3)(6,3)$ \\*[1.5pt]
					$\alpha_{IJK}$
					& $\dfrac{\mathrm{i}\sqrt[4]{5}}{\sqrt{2}}$
					& $-\dfrac{\sqrt{5}}{2}$
					& $\dfrac{\mathrm{i}\sqrt[4]{5}}{2}$
					& $\dfrac{1}{2}$
					& $\dfrac{\mathrm{i}\sqrt[4]{5}}{2}$
					& $\dfrac{\mathrm{i}\sqrt[4]{5}}{2}$
					& $-\dfrac{\sqrt{5}}{2}$
					& $\mathrm{i}\sqrt{\dfrac{2}{3}}\sqrt[4]{5}$ \\[2pt]
					\midrule
					
					No. & 97 & 98 & 99 & 100 & 101 & 102 & 103 & 104 \\*
					\midrule
					Triplet
					& $(6,3)(6,3)(6,6)$
					& $(6,3)(6,3)(8,3)$
					& $(6,3)(6,3)(8,6)$
					& $(6,3)(6,6)(6,6)$
					& $(6,3)(6,6)(6,8)$
					& $(6,3)(6,6)(8,6)$
					& $(6,3)(6,6)(8,8)$
					& $(6,6)(6,6)(6,6)$ \\*[1.5pt]
					$\alpha_{IJK}$
					& $\sqrt{\dfrac{2}{3}}$
					& $0$
					& $0$
					& $-\dfrac{\mathrm{i}\sqrt[4]{5}}{\sqrt{3}}$
					& $-1$
					& $0$
					& $0$
					& $\dfrac{2}{\sqrt{3}}$ \\[2pt]
					\midrule
					
					No. & 105 & 106 & 107 & & & & & \\*
					\midrule
					Triplet
					& $(6,6)(6,6)(6,8)$
					& $(6,6)(6,6)(8,6)$
					& $(6,6)(6,6)(8,8)$
					& & & & & \\*[1.5pt]
					$\alpha_{IJK}$
					& $0$
					& $0$
					& $0$
					& & & & & \\[2pt]
					\bottomrule
				\end{tabular}%
		}}
	\end{center}
	    \end{minipage}
    \endgroup
	\clearpage
	\begingroup
    \noindent\begin{minipage}{\textwidth}
	\fontsize{6}{6.5}\selectfont
	\setlength{\tabcolsep}{1.6pt}
	\renewcommand{\arraystretch}{1.15}
	\setlength{\aboverulesep}{0.4ex}
	\setlength{\belowrulesep}{0.65ex}
	\refstepcounter{table}\label{tab:E8-I7-results}
	\begin{center}
		{\normalsize\textbf{Table~\thetable: Representative reduced coefficients for $(p,q)=(7,30)$.}}
		\par\vspace{3pt}
		\makebox[\textwidth][c]{\resizebox{1.15\textwidth}{!}{%
				\begin{tabular}{c@{\hspace{3pt}}cccccccc}
					\toprule
					No. & 1 & 2 & 3 & 4 & 5 & 6 & 7 & 8 \\*
					\midrule
					Triplet
					& $(0,5)(0,5)(0,5)$
					& $(0,5)(0,5)(0,9)$
					& $(0,5)(3,3)(3,3)$
					& $(0,5)(3,3)(3,6)$
					& $(0,5)(3,3)(3,8)$
					& $(0,5)(3,6)(3,6)$
					& $(0,5)(3,6)(3,8)$
					& $(0,5)(5,0)(5,5)$ \\*[1.5pt]
					$\alpha_{IJK}$
					& $\sqrt{\dfrac{3}{2}}$
					& $\dfrac{\sqrt[4]{5}}{\sqrt{2}}$
					& $-1$
					& $-\dfrac{\sqrt[4]{5}}{\sqrt{2}}$
					& $\dfrac{1}{\sqrt{2}}$
					& $\dfrac{1}{\sqrt{2}}$
					& $\dfrac{\sqrt[4]{5}}{\sqrt{2}}$
					& $1$ \\[2pt]
					\midrule
					
					No. & 9 & 10 & 11 & 12 & 13 & 14 & 15 & 16 \\*
					\midrule
					Triplet
					& $(0,5)(5,5)(5,5)$
					& $(0,5)(5,5)(5,9)$
					& $(0,5)(6,3)(6,3)$
					& $(0,5)(6,3)(6,6)$
					& $(0,5)(6,3)(6,8)$
					& $(0,5)(6,6)(6,6)$
					& $(0,5)(6,6)(6,8)$
					& $(3,3)(3,3)(3,3)$ \\*[1.5pt]
					$\alpha_{IJK}$
					& $\sqrt{\dfrac{3}{2}}$
					& $\dfrac{\sqrt[4]{5}}{\sqrt{2}}$
					& $-1$
					& $-\dfrac{\sqrt[4]{5}}{\sqrt{2}}$
					& $\dfrac{1}{\sqrt{2}}$
					& $\dfrac{1}{\sqrt{2}}$
					& $\dfrac{\sqrt[4]{5}}{\sqrt{2}}$
					& $\sqrt{\dfrac{5}{3}}$ \\[2pt]
					\midrule
					
					No. & 17 & 18 & 19 & 20 & 21 & 22 & 23 & 24 \\*
					\midrule
					Triplet
					& $(3,3)(3,3)(3,6)$
					& $(3,3)(3,3)(5,0)$
					& $(3,3)(3,3)(5,5)$
					& $(3,3)(3,3)(6,3)$
					& $(3,3)(3,3)(6,6)$
					& $(3,3)(3,6)(3,6)$
					& $(3,3)(3,6)(3,8)$
					& $(3,3)(3,6)(5,5)$ \\*[1.5pt]
					$\alpha_{IJK}$
					& $-\dfrac{\sqrt[4]{5}}{\sqrt{3}}$
					& $-1$
					& $1$
					& $-\dfrac{\sqrt[4]{5}}{\sqrt{3}}$
					& $\dfrac{1}{\sqrt{3}}$
					& $\sqrt{\dfrac{5}{6}}$
					& $\dfrac{\sqrt[4]{5}}{\sqrt{2}}$
					& $\dfrac{\sqrt[4]{5}}{\sqrt{2}}$ \\[2pt]
					\midrule
					
					No. & 25 & 26 & 27 & 28 & 29 & 30 & 31 & 32 \\*
					\midrule
					Triplet
					& $(3,3)(3,6)(5,9)$
					& $(3,3)(3,6)(6,3)$
					& $(3,3)(3,6)(6,6)$
					& $(3,3)(3,6)(6,8)$
					& $(3,3)(5,0)(6,3)$
					& $(3,3)(5,0)(8,3)$
					& $(3,3)(5,5)(6,3)$
					& $(3,3)(5,5)(6,6)$ \\*[1.5pt]
					$\alpha_{IJK}$
					& $-\dfrac{1}{\sqrt{2}}$
					& $\dfrac{1}{\sqrt{3}}$
					& $-\dfrac{\sqrt[4]{5}}{\sqrt{6}}$
					& $-\dfrac{1}{\sqrt{2}}$
					& $-\dfrac{\sqrt[4]{5}}{\sqrt{2}}$
					& $-\dfrac{1}{\sqrt{2}}$
					& $\dfrac{\sqrt[4]{5}}{\sqrt{2}}$
					& $\dfrac{\sqrt{5}}{2}$ \\[2pt]
					\midrule
					
					No. & 33 & 34 & 35 & 36 & 37 & 38 & 39 & 40 \\*
					\midrule
					Triplet
					& $(3,3)(5,5)(6,8)$
					& $(3,3)(5,5)(8,3)$
					& $(3,3)(5,5)(8,6)$
					& $(3,3)(5,5)(8,8)$
					& $(3,3)(6,3)(6,3)$
					& $(3,3)(6,3)(6,6)$
					& $(3,3)(6,3)(8,3)$
					& $(3,3)(6,3)(8,6)$ \\*[1.5pt]
					$\alpha_{IJK}$
					& $\dfrac{\sqrt[4]{5}}{2}$
					& $\dfrac{1}{\sqrt{2}}$
					& $\dfrac{\sqrt[4]{5}}{2}$
					& $-\dfrac{1}{2}$
					& $\sqrt{\dfrac{5}{6}}$
					& $-\dfrac{\sqrt[4]{5}}{\sqrt{6}}$
					& $\dfrac{\sqrt[4]{5}}{\sqrt{2}}$
					& $-\dfrac{1}{\sqrt{2}}$ \\[2pt]
					\midrule
					
					No. & 41 & 42 & 43 & 44 & 45 & 46 & 47 & 48 \\*
					\midrule
					Triplet
					& $(3,3)(6,6)(6,6)$
					& $(3,3)(6,6)(6,8)$
					& $(3,3)(6,6)(8,6)$
					& $(3,3)(6,6)(8,8)$
					& $(3,6)(3,6)(3,6)$
					& $(3,6)(3,6)(3,8)$
					& $(3,6)(3,6)(5,0)$
					& $(3,6)(3,6)(5,5)$ \\*[1.5pt]
					$\alpha_{IJK}$
					& $\dfrac{\sqrt{5/3}}{2}$
					& $\dfrac{\sqrt[4]{5}}{2}$
					& $\dfrac{\sqrt[4]{5}}{2}$
					& $-\dfrac{\sqrt{3}}{2}$
					& $\sqrt{\dfrac{2}{3}}\sqrt[4]{5}$
					& $0$
					& $-1$
					& $-\dfrac{1}{\sqrt{2}}$ \\[2pt]
					\midrule
					
					No. & 49 & 50 & 51 & 52 & 53 & 54 & 55 & 56 \\*
					\midrule
					Triplet
					& $(3,6)(3,6)(5,9)$
					& $(3,6)(3,6)(6,3)$
					& $(3,6)(3,6)(6,6)$
					& $(3,6)(3,6)(6,8)$
					& $(3,6)(3,6)(6,11)$
					& $(3,6)(5,0)(6,6)$
					& $(3,6)(5,0)(8,6)$
					& $(3,6)(5,5)(6,3)$ \\*[1.5pt]
					$\alpha_{IJK}$
					& $-\dfrac{\sqrt[4]{5}}{\sqrt{2}}$
					& $-\dfrac{\sqrt[4]{5}}{\sqrt{6}}$
					& $-\sqrt{\dfrac{2}{3}}$
					& $0$
					& $-\dfrac{\mathrm{i}}{\sqrt{2}}$
					& $-\dfrac{\sqrt[4]{5}}{\sqrt{2}}$
					& $-\dfrac{1}{\sqrt{2}}$
					& $\dfrac{\sqrt{5}}{2}$ \\[2pt]
					\midrule
					No. & 57 & 58 & 59 & 60 & 61 & 62 & 63 & 64 \\*
					\midrule
					Triplet
					& $(3,6)(5,5)(6,6)$
					& $(3,6)(5,5)(6,8)$
					& $(3,6)(5,5)(6,11)$
					& $(3,6)(5,5)(8,3)$
					& $(3,6)(5,5)(8,6)$
					& $(3,6)(5,5)(8,8)$
					& $(3,6)(5,5)(8,11)$
					& $(3,6)(6,3)(6,3)$ \\*[1.5pt]
					$\alpha_{IJK}$
					& $-\dfrac{\sqrt[4]{5}}{2}$
					& $\dfrac{\sqrt{5}}{2}$
					& $\dfrac{\mathrm{i}\sqrt[4]{5}}{2}$
					& $\dfrac{\sqrt[4]{5}}{2}$
					& $-\dfrac{1}{2}$
					& $-\dfrac{\sqrt[4]{5}}{2}$
					& $\dfrac{\mathrm{i}}{2}$
					& $-\dfrac{\sqrt[4]{5}}{\sqrt{6}}$ \\[2pt]
					\midrule
					
					No. & 65 & 66 & 67 & 68 & 69 & 70 & 71 & 72 \\*
					\midrule
					Triplet
					& $(3,6)(6,3)(6,6)$
					& $(3,6)(6,3)(6,8)$
					& $(3,6)(6,3)(8,3)$
					& $(3,6)(6,3)(8,6)$
					& $(3,6)(6,3)(8,8)$
					& $(3,6)(6,6)(6,6)$
					& $(3,6)(6,6)(6,8)$
					& $(3,6)(6,6)(8,6)$ \\*[1.5pt]
					$\alpha_{IJK}$
					& $\dfrac{\sqrt{5/3}}{2}$
					& $-\dfrac{\sqrt[4]{5}}{2}$
					& $-\dfrac{1}{\sqrt{2}}$
					& $\dfrac{\sqrt[4]{5}}{2}$
					& $\dfrac{\sqrt{3}}{2}$
					& $\dfrac{\sqrt[4]{5}}{\sqrt{3}}$
					& $0$
					& $1$ \\[2pt]
					\midrule
					
					No. & 73 & 74 & 75 & 76 & 77 & 78 & 79 & 80 \\*
					\midrule
					Triplet
					& $(3,6)(6,6)(8,8)$
					& $(5,0)(5,0)(5,0)$
					& $(5,0)(5,0)(9,0)$
					& $(5,0)(5,5)(5,5)$
					& $(5,0)(5,5)(9,5)$
					& $(5,0)(6,3)(6,3)$
					& $(5,0)(6,3)(8,3)$
					& $(5,0)(6,6)(6,6)$ \\*[1.5pt]
					$\alpha_{IJK}$
					& $0$
					& $\sqrt{\dfrac{3}{2}}$
					& $\dfrac{\sqrt[4]{5}}{\sqrt{2}}$
					& $\sqrt{\dfrac{3}{2}}$
					& $\dfrac{\sqrt[4]{5}}{\sqrt{2}}$
					& $\dfrac{1}{\sqrt{2}}$
					& $\dfrac{\sqrt[4]{5}}{\sqrt{2}}$
					& $\dfrac{1}{\sqrt{2}}$ \\[2pt]
					\midrule
					
					No. & 81 & 82 & 83 & 84 & 85 & 86 & 87 & 88 \\*
					\midrule
					Triplet
					& $(5,0)(6,6)(8,6)$
					& $(5,5)(5,5)(5,5)$
					& $(5,5)(5,5)(5,9)$
					& $(5,5)(5,5)(9,5)$
					& $(5,5)(5,5)(9,9)$
					& $(5,5)(6,3)(6,3)$
					& $(5,5)(6,3)(6,6)$
					& $(5,5)(6,3)(6,8)$ \\*[1.5pt]
					$\alpha_{IJK}$
					& $\dfrac{\sqrt[4]{5}}{\sqrt{2}}$
					& $\dfrac{3}{2}$
					& $\dfrac{\sqrt{3}\sqrt[4]{5}}{2}$
					& $\dfrac{\sqrt{3}\sqrt[4]{5}}{2}$
					& $-\dfrac{\sqrt{5}}{2}$
					& $-\dfrac{1}{\sqrt{2}}$
					& $-\dfrac{\sqrt[4]{5}}{2}$
					& $\dfrac{1}{2}$ \\[2pt]
					\midrule
					
					No. & 89 & 90 & 91 & 92 & 93 & 94 & 95 & 96 \\*
					\midrule
					Triplet
					& $(5,5)(6,3)(8,3)$
					& $(5,5)(6,3)(8,6)$
					& $(5,5)(6,3)(8,8)$
					& $(5,5)(6,6)(6,6)$
					& $(5,5)(6,6)(6,8)$
					& $(5,5)(6,6)(8,6)$
					& $(5,5)(6,6)(8,8)$
					& $(6,3)(6,3)(6,3)$ \\*[1.5pt]
					$\alpha_{IJK}$
					& $-\dfrac{\sqrt[4]{5}}{\sqrt{2}}$
					& $-\dfrac{\sqrt{5}}{2}$
					& $-\dfrac{\sqrt[4]{5}}{2}$
					& $\dfrac{1}{2}$
					& $\dfrac{\sqrt[4]{5}}{2}$
					& $\dfrac{\sqrt[4]{5}}{2}$
					& $-\dfrac{\sqrt{5}}{2}$
					& $\sqrt{\dfrac{2}{3}}\sqrt[4]{5}$ \\[2pt]
					\midrule
					
					No. & 97 & 98 & 99 & 100 & 101 & 102 & 103 & 104 \\*
					\midrule
					Triplet
					& $(6,3)(6,3)(6,6)$
					& $(6,3)(6,3)(8,3)$
					& $(6,3)(6,3)(8,6)$
					& $(6,3)(6,6)(6,6)$
					& $(6,3)(6,6)(6,8)$
					& $(6,3)(6,6)(8,6)$
					& $(6,3)(6,6)(8,8)$
					& $(6,6)(6,6)(6,6)$ \\*[1.5pt]
					$\alpha_{IJK}$
					& $-\sqrt{\dfrac{2}{3}}$
					& $0$
					& $0$
					& $\dfrac{\sqrt[4]{5}}{\sqrt{3}}$
					& $1$
					& $0$
					& $0$
					& $\dfrac{2}{\sqrt{3}}$ \\[2pt]
					\midrule
					
					No. & 105 & 106 & 107 & & & & & \\*
					\midrule
					Triplet
					& $(6,6)(6,6)(6,8)$
					& $(6,6)(6,6)(8,6)$
					& $(6,6)(6,6)(8,8)$
					& & & & & \\*[1.5pt]
					$\alpha_{IJK}$
					& $0$
					& $0$
					& $0$
					& & & & & \\[2pt]
					\bottomrule
				\end{tabular}%
		}}
	\end{center}
	    \end{minipage}
    \endgroup
	\clearpage
	\begingroup
    \noindent\begin{minipage}{\textwidth}
	\fontsize{6}{6.5}\selectfont
	\setlength{\tabcolsep}{1.6pt}
	\renewcommand{\arraystretch}{1.15}
	\setlength{\aboverulesep}{0.4ex}
	\setlength{\belowrulesep}{0.65ex}
	\refstepcounter{table}\label{tab:E8-I11-results}
	\begin{center}
		{\normalsize\textbf{Table~\thetable: Representative reduced coefficients for $(p,q)=(11,30)$.}}
		\par\vspace{3pt}
		\makebox[\textwidth][c]{\resizebox{1.15\textwidth}{!}{%
				\begin{tabular}{c@{\hspace{3pt}}cccccccc}
					\toprule
					No. & 1 & 2 & 3 & 4 & 5 & 6 & 7 & 8 \\*
					\midrule
					Triplet
					& $(0,5)(0,5)(0,5)$
					& $(0,5)(0,5)(0,9)$
					& $(0,5)(3,3)(3,3)$
					& $(0,5)(3,3)(3,6)$
					& $(0,5)(3,3)(3,8)$
					& $(0,5)(3,6)(3,6)$
					& $(0,5)(3,6)(3,8)$
					& $(0,5)(5,0)(5,5)$ \\*[1.5pt]
					$\alpha_{IJK}$
					& $\sqrt{\dfrac{3}{2}}$
					& $\dfrac{\mathrm{i}\sqrt[4]{5}}{\sqrt{2}}$
					& $1$
					& $\dfrac{\mathrm{i}\sqrt[4]{5}}{\sqrt{2}}$
					& $-\dfrac{\mathrm{i}}{\sqrt{2}}$
					& $\dfrac{1}{\sqrt{2}}$
					& $-\dfrac{\sqrt[4]{5}}{\sqrt{2}}$
					& $1$ \\[2pt]
					\midrule
					
					No. & 9 & 10 & 11 & 12 & 13 & 14 & 15 & 16 \\*
					\midrule
					Triplet
					& $(0,5)(5,5)(5,5)$
					& $(0,5)(5,5)(5,9)$
					& $(0,5)(6,3)(6,3)$
					& $(0,5)(6,3)(6,6)$
					& $(0,5)(6,3)(6,8)$
					& $(0,5)(6,6)(6,6)$
					& $(0,5)(6,6)(6,8)$
					& $(3,3)(3,3)(3,3)$ \\*[1.5pt]
					$\alpha_{IJK}$
					& $\sqrt{\dfrac{3}{2}}$
					& $\dfrac{\mathrm{i}\sqrt[4]{5}}{\sqrt{2}}$
					& $1$
					& $-\dfrac{\mathrm{i}\sqrt[4]{5}}{\sqrt{2}}$
					& $\dfrac{\mathrm{i}}{\sqrt{2}}$
					& $\dfrac{1}{\sqrt{2}}$
					& $-\dfrac{\sqrt[4]{5}}{\sqrt{2}}$
					& $\sqrt{\dfrac{5}{3}}$ \\[2pt]
					\midrule
					
					No. & 17 & 18 & 19 & 20 & 21 & 22 & 23 & 24 \\*
					\midrule
					Triplet
					& $(3,3)(3,3)(3,6)$
					& $(3,3)(3,3)(5,0)$
					& $(3,3)(3,3)(5,5)$
					& $(3,3)(3,3)(6,3)$
					& $(3,3)(3,3)(6,6)$
					& $(3,3)(3,6)(3,6)$
					& $(3,3)(3,6)(3,8)$
					& $(3,3)(3,6)(5,5)$ \\*[1.5pt]
					$\alpha_{IJK}$
					& $\dfrac{\mathrm{i}\sqrt[4]{5}}{\sqrt{3}}$
					& $1$
					& $1$
					& $\dfrac{\mathrm{i}\sqrt[4]{5}}{\sqrt{3}}$
					& $\dfrac{1}{\sqrt{3}}$
					& $\sqrt{\dfrac{5}{6}}$
					& $-\dfrac{\sqrt[4]{5}}{\sqrt{2}}$
					& $\dfrac{\mathrm{i}\sqrt[4]{5}}{\sqrt{2}}$ \\[2pt]
					\midrule
					
					No. & 25 & 26 & 27 & 28 & 29 & 30 & 31 & 32 \\*
					\midrule
					Triplet
					& $(3,3)(3,6)(5,9)$
					& $(3,3)(3,6)(6,3)$
					& $(3,3)(3,6)(6,6)$
					& $(3,3)(3,6)(6,8)$
					& $(3,3)(5,0)(6,3)$
					& $(3,3)(5,0)(8,3)$
					& $(3,3)(5,5)(6,3)$
					& $(3,3)(5,5)(6,6)$ \\*[1.5pt]
					$\alpha_{IJK}$
					& $-\dfrac{1}{\sqrt{2}}$
					& $-\dfrac{1}{\sqrt{3}}$
					& $-\dfrac{\mathrm{i}\sqrt[4]{5}}{\sqrt{6}}$
					& $\dfrac{\mathrm{i}}{\sqrt{2}}$
					& $\dfrac{\mathrm{i}\sqrt[4]{5}}{\sqrt{2}}$
					& $\dfrac{\mathrm{i}}{\sqrt{2}}$
					& $\dfrac{\mathrm{i}\sqrt[4]{5}}{\sqrt{2}}$
					& $\dfrac{\sqrt{5}}{2}$ \\[2pt]
					\midrule
					
					No. & 33 & 34 & 35 & 36 & 37 & 38 & 39 & 40 \\*
					\midrule
					Triplet
					& $(3,3)(5,5)(6,8)$
					& $(3,3)(5,5)(8,3)$
					& $(3,3)(5,5)(8,6)$
					& $(3,3)(5,5)(8,8)$
					& $(3,3)(6,3)(6,3)$
					& $(3,3)(6,3)(6,6)$
					& $(3,3)(6,3)(8,3)$
					& $(3,3)(6,3)(8,6)$ \\*[1.5pt]
					$\alpha_{IJK}$
					& $\dfrac{\sqrt[4]{5}}{2}$
					& $\dfrac{\mathrm{i}}{\sqrt{2}}$
					& $\dfrac{\sqrt[4]{5}}{2}$
					& $\dfrac{1}{2}$
					& $\sqrt{\dfrac{5}{6}}$
					& $-\dfrac{\mathrm{i}\sqrt[4]{5}}{\sqrt{6}}$
					& $-\dfrac{\sqrt[4]{5}}{\sqrt{2}}$
					& $\dfrac{\mathrm{i}}{\sqrt{2}}$ \\[2pt]
					\midrule
					
					No. & 41 & 42 & 43 & 44 & 45 & 46 & 47 & 48 \\*
					\midrule
					Triplet
					& $(3,3)(6,6)(6,6)$
					& $(3,3)(6,6)(6,8)$
					& $(3,3)(6,6)(8,6)$
					& $(3,3)(6,6)(8,8)$
					& $(3,6)(3,6)(3,6)$
					& $(3,6)(3,6)(3,8)$
					& $(3,6)(3,6)(5,0)$
					& $(3,6)(3,6)(5,5)$ \\*[1.5pt]
					$\alpha_{IJK}$
					& $\dfrac{\sqrt{5/3}}{2}$
					& $-\dfrac{\sqrt[4]{5}}{2}$
					& $-\dfrac{\sqrt[4]{5}}{2}$
					& $\dfrac{\sqrt{3}}{2}$
					& $\mathrm{i}\sqrt{\dfrac{2}{3}}\sqrt[4]{5}$
					& $0$
					& $1$
					& $\dfrac{1}{\sqrt{2}}$ \\[2pt]
					\midrule
					
					No. & 49 & 50 & 51 & 52 & 53 & 54 & 55 & 56 \\*
					\midrule
					Triplet
					& $(3,6)(3,6)(5,9)$
					& $(3,6)(3,6)(6,3)$
					& $(3,6)(3,6)(6,6)$
					& $(3,6)(3,6)(6,8)$
					& $(3,6)(3,6)(6,11)$
					& $(3,6)(5,0)(6,6)$
					& $(3,6)(5,0)(8,6)$
					& $(3,6)(5,5)(6,3)$ \\*[1.5pt]
					$\alpha_{IJK}$
					& $\dfrac{\mathrm{i}\sqrt[4]{5}}{\sqrt{2}}$
					& $\dfrac{\mathrm{i}\sqrt[4]{5}}{\sqrt{6}}$
					& $\sqrt{\dfrac{2}{3}}$
					& $0$
					& $-\dfrac{1}{\sqrt{2}}$
					& $-\dfrac{\mathrm{i}\sqrt[4]{5}}{\sqrt{2}}$
					& $-\dfrac{\mathrm{i}}{\sqrt{2}}$
					& $-\dfrac{\sqrt{5}}{2}$ \\[2pt]
					\midrule
					No. & 57 & 58 & 59 & 60 & 61 & 62 & 63 & 64 \\*
					\midrule
					Triplet
					& $(3,6)(5,5)(6,6)$
					& $(3,6)(5,5)(6,8)$
					& $(3,6)(5,5)(6,11)$
					& $(3,6)(5,5)(8,3)$
					& $(3,6)(5,5)(8,6)$
					& $(3,6)(5,5)(8,8)$
					& $(3,6)(5,5)(8,11)$
					& $(3,6)(6,3)(6,3)$ \\*[1.5pt]
					$\alpha_{IJK}$
					& $-\dfrac{\mathrm{i}\sqrt[4]{5}}{2}$
					& $-\dfrac{\mathrm{i}\sqrt{5}}{2}$
					& $-\dfrac{\mathrm{i}\sqrt[4]{5}}{2}$
					& $-\dfrac{\sqrt[4]{5}}{2}$
					& $-\dfrac{\mathrm{i}}{2}$
					& $-\dfrac{\mathrm{i}\sqrt[4]{5}}{2}$
					& $-\dfrac{\mathrm{i}}{2}$
					& $\dfrac{\mathrm{i}\sqrt[4]{5}}{\sqrt{6}}$ \\[2pt]
					\midrule
					
					No. & 65 & 66 & 67 & 68 & 69 & 70 & 71 & 72 \\*
					\midrule
					Triplet
					& $(3,6)(6,3)(6,6)$
					& $(3,6)(6,3)(6,8)$
					& $(3,6)(6,3)(8,3)$
					& $(3,6)(6,3)(8,6)$
					& $(3,6)(6,3)(8,8)$
					& $(3,6)(6,6)(6,6)$
					& $(3,6)(6,6)(6,8)$
					& $(3,6)(6,6)(8,6)$ \\*[1.5pt]
					$\alpha_{IJK}$
					& $-\dfrac{\sqrt{5/3}}{2}$
					& $-\dfrac{\sqrt[4]{5}}{2}$
					& $-\dfrac{\mathrm{i}}{\sqrt{2}}$
					& $\dfrac{\sqrt[4]{5}}{2}$
					& $\dfrac{\sqrt{3}}{2}$
					& $\dfrac{\mathrm{i}\sqrt[4]{5}}{\sqrt{3}}$
					& $0$
					& $-\mathrm{i}$ \\[2pt]
					\midrule
					
					No. & 73 & 74 & 75 & 76 & 77 & 78 & 79 & 80 \\*
					\midrule
					Triplet
					& $(3,6)(6,6)(8,8)$
					& $(5,0)(5,0)(5,0)$
					& $(5,0)(5,0)(9,0)$
					& $(5,0)(5,5)(5,5)$
					& $(5,0)(5,5)(9,5)$
					& $(5,0)(6,3)(6,3)$
					& $(5,0)(6,3)(8,3)$
					& $(5,0)(6,6)(6,6)$ \\*[1.5pt]
					$\alpha_{IJK}$
					& $0$
					& $\sqrt{\dfrac{3}{2}}$
					& $\dfrac{\mathrm{i}\sqrt[4]{5}}{\sqrt{2}}$
					& $\sqrt{\dfrac{3}{2}}$
					& $\dfrac{\mathrm{i}\sqrt[4]{5}}{\sqrt{2}}$
					& $\dfrac{1}{\sqrt{2}}$
					& $-\dfrac{\sqrt[4]{5}}{\sqrt{2}}$
					& $\dfrac{1}{\sqrt{2}}$ \\[2pt]
					\midrule
					
					No. & 81 & 82 & 83 & 84 & 85 & 86 & 87 & 88 \\*
					\midrule
					Triplet
					& $(5,0)(6,6)(8,6)$
					& $(5,5)(5,5)(5,5)$
					& $(5,5)(5,5)(5,9)$
					& $(5,5)(5,5)(9,5)$
					& $(5,5)(5,5)(9,9)$
					& $(5,5)(6,3)(6,3)$
					& $(5,5)(6,3)(6,6)$
					& $(5,5)(6,3)(6,8)$ \\*[1.5pt]
					$\alpha_{IJK}$
					& $-\dfrac{\sqrt[4]{5}}{\sqrt{2}}$
					& $\dfrac{3}{2}$
					& $\dfrac{\mathrm{i}\sqrt{3}\sqrt[4]{5}}{2}$
					& $\dfrac{\mathrm{i}\sqrt{3}\sqrt[4]{5}}{2}$
					& $\dfrac{\sqrt{5}}{2}$
					& $\dfrac{1}{\sqrt{2}}$
					& $-\dfrac{\mathrm{i}\sqrt[4]{5}}{2}$
					& $\dfrac{\mathrm{i}}{2}$ \\[2pt]
					\midrule
					
					No. & 89 & 90 & 91 & 92 & 93 & 94 & 95 & 96 \\*
					\midrule
					Triplet
					& $(5,5)(6,3)(8,3)$
					& $(5,5)(6,3)(8,6)$
					& $(5,5)(6,3)(8,8)$
					& $(5,5)(6,6)(6,6)$
					& $(5,5)(6,6)(6,8)$
					& $(5,5)(6,6)(8,6)$
					& $(5,5)(6,6)(8,8)$
					& $(6,3)(6,3)(6,3)$ \\*[1.5pt]
					$\alpha_{IJK}$
					& $-\dfrac{\sqrt[4]{5}}{\sqrt{2}}$
					& $\dfrac{\mathrm{i}\sqrt{5}}{2}$
					& $-\dfrac{\mathrm{i}\sqrt[4]{5}}{2}$
					& $\dfrac{1}{2}$
					& $-\dfrac{\sqrt[4]{5}}{2}$
					& $-\dfrac{\sqrt[4]{5}}{2}$
					& $\dfrac{\sqrt{5}}{2}$
					& $\mathrm{i}\sqrt{\dfrac{2}{3}}\sqrt[4]{5}$ \\[2pt]
					\midrule
					
					No. & 97 & 98 & 99 & 100 & 101 & 102 & 103 & 104 \\*
					\midrule
					Triplet
					& $(6,3)(6,3)(6,6)$
					& $(6,3)(6,3)(8,3)$
					& $(6,3)(6,3)(8,6)$
					& $(6,3)(6,6)(6,6)$
					& $(6,3)(6,6)(6,8)$
					& $(6,3)(6,6)(8,6)$
					& $(6,3)(6,6)(8,8)$
					& $(6,6)(6,6)(6,6)$ \\*[1.5pt]
					$\alpha_{IJK}$
					& $\sqrt{\dfrac{2}{3}}$
					& $0$
					& $0$
					& $\dfrac{\mathrm{i}\sqrt[4]{5}}{\sqrt{3}}$
					& $-\mathrm{i}$
					& $0$
					& $0$
					& $\dfrac{2}{\sqrt{3}}$ \\[2pt]
					\midrule
					
					No. & 105 & 106 & 107 & & & & & \\*
					\midrule
					Triplet
					& $(6,6)(6,6)(6,8)$
					& $(6,6)(6,6)(8,6)$
					& $(6,6)(6,6)(8,8)$
					& & & & & \\*[1.5pt]
					$\alpha_{IJK}$
					& $0$
					& $0$
					& $0$
					& & & & & \\[2pt]
					\bottomrule
				\end{tabular}%
		}}
	\end{center}
	    \end{minipage}
    \endgroup
	\clearpage
	\begingroup
    \noindent\begin{minipage}{\textwidth}
	\fontsize{6}{6.5}\selectfont
	\setlength{\tabcolsep}{1.6pt}
	\renewcommand{\arraystretch}{1.15}
	\setlength{\aboverulesep}{0.4ex}
	\setlength{\belowrulesep}{0.65ex}
	\refstepcounter{table}\label{tab:E8-I13-results}
	\begin{center}
		{\normalsize\textbf{Table~\thetable: Representative reduced coefficients for $(p,q)=(13,30)$.}}
		\par\vspace{3pt}
		\makebox[\textwidth][c]{\resizebox{1.15\textwidth}{!}{%
				\begin{tabular}{c@{\hspace{3pt}}cccccccc}
					\toprule
					No. & 1 & 2 & 3 & 4 & 5 & 6 & 7 & 8 \\*
					\midrule
					Triplet
					& $(0,5)(0,5)(0,5)$
					& $(0,5)(0,5)(0,9)$
					& $(0,5)(3,3)(3,3)$
					& $(0,5)(3,3)(3,6)$
					& $(0,5)(3,3)(3,8)$
					& $(0,5)(3,6)(3,6)$
					& $(0,5)(3,6)(3,8)$
					& $(0,5)(5,0)(5,5)$ \\*[1.5pt]
					$\alpha_{IJK}$
					& $\sqrt{\dfrac{3}{2}}$
					& $\dfrac{\mathrm{i}\sqrt[4]{5}}{\sqrt{2}}$
					& $1$
					& $-\dfrac{\mathrm{i}\sqrt[4]{5}}{\sqrt{2}}$
					& $\dfrac{1}{\sqrt{2}}$
					& $\dfrac{1}{\sqrt{2}}$
					& $\dfrac{\mathrm{i}\sqrt[4]{5}}{\sqrt{2}}$
					& $1$ \\[2pt]
					\midrule
					
					No. & 9 & 10 & 11 & 12 & 13 & 14 & 15 & 16 \\*
					\midrule
					Triplet
					& $(0,5)(5,5)(5,5)$
					& $(0,5)(5,5)(5,9)$
					& $(0,5)(6,3)(6,3)$
					& $(0,5)(6,3)(6,6)$
					& $(0,5)(6,3)(6,8)$
					& $(0,5)(6,6)(6,6)$
					& $(0,5)(6,6)(6,8)$
					& $(3,3)(3,3)(3,3)$ \\*[1.5pt]
					$\alpha_{IJK}$
					& $\sqrt{\dfrac{3}{2}}$
					& $\dfrac{\mathrm{i}\sqrt[4]{5}}{\sqrt{2}}$
					& $1$
					& $\dfrac{\mathrm{i}\sqrt[4]{5}}{\sqrt{2}}$
					& $-\dfrac{1}{\sqrt{2}}$
					& $\dfrac{1}{\sqrt{2}}$
					& $\dfrac{\mathrm{i}\sqrt[4]{5}}{\sqrt{2}}$
					& $\sqrt{\dfrac{5}{3}}$ \\[2pt]
					\midrule
					
					No. & 17 & 18 & 19 & 20 & 21 & 22 & 23 & 24 \\*
					\midrule
					Triplet
					& $(3,3)(3,3)(3,6)$
					& $(3,3)(3,3)(5,0)$
					& $(3,3)(3,3)(5,5)$
					& $(3,3)(3,3)(6,3)$
					& $(3,3)(3,3)(6,6)$
					& $(3,3)(3,6)(3,6)$
					& $(3,3)(3,6)(3,8)$
					& $(3,3)(3,6)(5,5)$ \\*[1.5pt]
					$\alpha_{IJK}$
					& $\dfrac{\mathrm{i}\sqrt[4]{5}}{\sqrt{3}}$
					& $1$
					& $1$
					& $\dfrac{\mathrm{i}\sqrt[4]{5}}{\sqrt{3}}$
					& $\dfrac{1}{\sqrt{3}}$
					& $\sqrt{\dfrac{5}{6}}$
					& $\dfrac{\mathrm{i}\sqrt[4]{5}}{\sqrt{2}}$
					& $-\dfrac{\mathrm{i}\sqrt[4]{5}}{\sqrt{2}}$ \\[2pt]
					\midrule
					
					No. & 25 & 26 & 27 & 28 & 29 & 30 & 31 & 32 \\*
					\midrule
					Triplet
					& $(3,3)(3,6)(5,9)$
					& $(3,3)(3,6)(6,3)$
					& $(3,3)(3,6)(6,6)$
					& $(3,3)(3,6)(6,8)$
					& $(3,3)(5,0)(6,3)$
					& $(3,3)(5,0)(8,3)$
					& $(3,3)(5,5)(6,3)$
					& $(3,3)(5,5)(6,6)$ \\*[1.5pt]
					$\alpha_{IJK}$
					& $-\dfrac{1}{\sqrt{2}}$
					& $-\dfrac{1}{\sqrt{3}}$
					& $-\dfrac{\mathrm{i}\sqrt[4]{5}}{\sqrt{6}}$
					& $\dfrac{1}{\sqrt{2}}$
					& $-\dfrac{\mathrm{i}\sqrt[4]{5}}{\sqrt{2}}$
					& $-\dfrac{1}{\sqrt{2}}$
					& $-\dfrac{\mathrm{i}\sqrt[4]{5}}{\sqrt{2}}$
					& $\dfrac{\sqrt{5}}{2}$ \\[2pt]
					\midrule
					
					No. & 33 & 34 & 35 & 36 & 37 & 38 & 39 & 40 \\*
					\midrule
					Triplet
					& $(3,3)(5,5)(6,8)$
					& $(3,3)(5,5)(8,3)$
					& $(3,3)(5,5)(8,6)$
					& $(3,3)(5,5)(8,8)$
					& $(3,3)(6,3)(6,3)$
					& $(3,3)(6,3)(6,6)$
					& $(3,3)(6,3)(8,3)$
					& $(3,3)(6,3)(8,6)$ \\*[1.5pt]
					$\alpha_{IJK}$
					& $-\dfrac{\mathrm{i}\sqrt[4]{5}}{2}$
					& $-\dfrac{1}{\sqrt{2}}$
					& $-\dfrac{\mathrm{i}\sqrt[4]{5}}{2}$
					& $\dfrac{1}{2}$
					& $\sqrt{\dfrac{5}{6}}$
					& $-\dfrac{\mathrm{i}\sqrt[4]{5}}{\sqrt{6}}$
					& $\dfrac{\mathrm{i}\sqrt[4]{5}}{\sqrt{2}}$
					& $\dfrac{1}{\sqrt{2}}$ \\[2pt]
					\midrule
					
					No. & 41 & 42 & 43 & 44 & 45 & 46 & 47 & 48 \\*
					\midrule
					Triplet
					& $(3,3)(6,6)(6,6)$
					& $(3,3)(6,6)(6,8)$
					& $(3,3)(6,6)(8,6)$
					& $(3,3)(6,6)(8,8)$
					& $(3,6)(3,6)(3,6)$
					& $(3,6)(3,6)(3,8)$
					& $(3,6)(3,6)(5,0)$
					& $(3,6)(3,6)(5,5)$ \\*[1.5pt]
					$\alpha_{IJK}$
					& $\dfrac{\sqrt{5/3}}{2}$
					& $\dfrac{\mathrm{i}\sqrt[4]{5}}{2}$
					& $\dfrac{\mathrm{i}\sqrt[4]{5}}{2}$
					& $\dfrac{\sqrt{3}}{2}$
					& $\mathrm{i}\sqrt{\dfrac{2}{3}}\sqrt[4]{5}$
					& $0$
					& $1$
					& $\dfrac{1}{\sqrt{2}}$ \\[2pt]
					\midrule
					
					No. & 49 & 50 & 51 & 52 & 53 & 54 & 55 & 56 \\*
					\midrule
					Triplet
					& $(3,6)(3,6)(5,9)$
					& $(3,6)(3,6)(6,3)$
					& $(3,6)(3,6)(6,6)$
					& $(3,6)(3,6)(6,8)$
					& $(3,6)(3,6)(6,11)$
					& $(3,6)(5,0)(6,6)$
					& $(3,6)(5,0)(8,6)$
					& $(3,6)(5,5)(6,3)$ \\*[1.5pt]
					$\alpha_{IJK}$
					& $-\dfrac{\mathrm{i}\sqrt[4]{5}}{\sqrt{2}}$
					& $\dfrac{\mathrm{i}\sqrt[4]{5}}{\sqrt{6}}$
					& $\sqrt{\dfrac{2}{3}}$
					& $0$
					& $-\dfrac{\mathrm{i}}{\sqrt{2}}$
					& $\dfrac{\mathrm{i}\sqrt[4]{5}}{\sqrt{2}}$
					& $\dfrac{1}{\sqrt{2}}$
					& $-\dfrac{\sqrt{5}}{2}$ \\[2pt]
					\midrule
					
					No. & 57 & 58 & 59 & 60 & 61 & 62 & 63 & 64 \\*
					\midrule
					Triplet
					& $(3,6)(5,5)(6,6)$
					& $(3,6)(5,5)(6,8)$
					& $(3,6)(5,5)(6,11)$
					& $(3,6)(5,5)(8,3)$
					& $(3,6)(5,5)(8,6)$
					& $(3,6)(5,5)(8,8)$
					& $(3,6)(5,5)(8,11)$
					& $(3,6)(6,3)(6,3)$ \\*[1.5pt]
					$\alpha_{IJK}$
					& $\dfrac{\mathrm{i}\sqrt[4]{5}}{2}$
					& $\dfrac{\sqrt{5}}{2}$
					& $-\dfrac{\sqrt[4]{5}}{2}$
					& $\dfrac{\mathrm{i}\sqrt[4]{5}}{2}$
					& $\dfrac{1}{2}$
					& $\dfrac{\mathrm{i}\sqrt[4]{5}}{2}$
					& $\dfrac{\mathrm{i}}{2}$
					& $\dfrac{\mathrm{i}\sqrt[4]{5}}{\sqrt{6}}$ \\[2pt]
					\midrule
					
					No. & 65 & 66 & 67 & 68 & 69 & 70 & 71 & 72 \\*
					\midrule
					Triplet
					& $(3,6)(6,3)(6,6)$
					& $(3,6)(6,3)(6,8)$
					& $(3,6)(6,3)(8,3)$
					& $(3,6)(6,3)(8,6)$
					& $(3,6)(6,3)(8,8)$
					& $(3,6)(6,6)(6,6)$
					& $(3,6)(6,6)(6,8)$
					& $(3,6)(6,6)(8,6)$ \\*[1.5pt]
					$\alpha_{IJK}$
					& $-\dfrac{\sqrt{5/3}}{2}$
					& $\dfrac{\mathrm{i}\sqrt[4]{5}}{2}$
					& $-\dfrac{1}{\sqrt{2}}$
					& $-\dfrac{\mathrm{i}\sqrt[4]{5}}{2}$
					& $\dfrac{\sqrt{3}}{2}$
					& $\dfrac{\mathrm{i}\sqrt[4]{5}}{\sqrt{3}}$
					& $0$
					& $-1$ \\[2pt]
					\midrule
					
					No. & 73 & 74 & 75 & 76 & 77 & 78 & 79 & 80 \\*
					\midrule
					Triplet
					& $(3,6)(6,6)(8,8)$
					& $(5,0)(5,0)(5,0)$
					& $(5,0)(5,0)(9,0)$
					& $(5,0)(5,5)(5,5)$
					& $(5,0)(5,5)(9,5)$
					& $(5,0)(6,3)(6,3)$
					& $(5,0)(6,3)(8,3)$
					& $(5,0)(6,6)(6,6)$ \\*[1.5pt]
					$\alpha_{IJK}$
					& $0$
					& $\sqrt{\dfrac{3}{2}}$
					& $\dfrac{\mathrm{i}\sqrt[4]{5}}{\sqrt{2}}$
					& $\sqrt{\dfrac{3}{2}}$
					& $\dfrac{\mathrm{i}\sqrt[4]{5}}{\sqrt{2}}$
					& $\dfrac{1}{\sqrt{2}}$
					& $\dfrac{\mathrm{i}\sqrt[4]{5}}{\sqrt{2}}$
					& $\dfrac{1}{\sqrt{2}}$ \\[2pt]
					\midrule
					
					No. & 81 & 82 & 83 & 84 & 85 & 86 & 87 & 88 \\*
					\midrule
					Triplet
					& $(5,0)(6,6)(8,6)$
					& $(5,5)(5,5)(5,5)$
					& $(5,5)(5,5)(5,9)$
					& $(5,5)(5,5)(9,5)$
					& $(5,5)(5,5)(9,9)$
					& $(5,5)(6,3)(6,3)$
					& $(5,5)(6,3)(6,6)$
					& $(5,5)(6,3)(6,8)$ \\*[1.5pt]
					$\alpha_{IJK}$
					& $\dfrac{\mathrm{i}\sqrt[4]{5}}{\sqrt{2}}$
					& $\dfrac{3}{2}$
					& $\dfrac{\mathrm{i}\sqrt{3}\sqrt[4]{5}}{2}$
					& $\dfrac{\mathrm{i}\sqrt{3}\sqrt[4]{5}}{2}$
					& $\dfrac{\sqrt{5}}{2}$
					& $\dfrac{1}{\sqrt{2}}$
					& $\dfrac{\mathrm{i}\sqrt[4]{5}}{2}$
					& $-\dfrac{1}{2}$ \\[2pt]
					\midrule
					
					No. & 89 & 90 & 91 & 92 & 93 & 94 & 95 & 96 \\*
					\midrule
					Triplet
					& $(5,5)(6,3)(8,3)$
					& $(5,5)(6,3)(8,6)$
					& $(5,5)(6,3)(8,8)$
					& $(5,5)(6,6)(6,6)$
					& $(5,5)(6,6)(6,8)$
					& $(5,5)(6,6)(8,6)$
					& $(5,5)(6,6)(8,8)$
					& $(6,3)(6,3)(6,3)$ \\*[1.5pt]
					$\alpha_{IJK}$
					& $\dfrac{\mathrm{i}\sqrt[4]{5}}{\sqrt{2}}$
					& $-\dfrac{\sqrt{5}}{2}$
					& $\dfrac{\mathrm{i}\sqrt[4]{5}}{2}$
					& $\dfrac{1}{2}$
					& $\dfrac{\mathrm{i}\sqrt[4]{5}}{2}$
					& $\dfrac{\mathrm{i}\sqrt[4]{5}}{2}$
					& $\dfrac{\sqrt{5}}{2}$
					& $\mathrm{i}\sqrt{\dfrac{2}{3}}\sqrt[4]{5}$ \\[2pt]
					\midrule
					
					No. & 97 & 98 & 99 & 100 & 101 & 102 & 103 & 104 \\*
					\midrule
					Triplet
					& $(6,3)(6,3)(6,6)$
					& $(6,3)(6,3)(8,3)$
					& $(6,3)(6,3)(8,6)$
					& $(6,3)(6,6)(6,6)$
					& $(6,3)(6,6)(6,8)$
					& $(6,3)(6,6)(8,6)$
					& $(6,3)(6,6)(8,8)$
					& $(6,6)(6,6)(6,6)$ \\*[1.5pt]
					$\alpha_{IJK}$
					& $\sqrt{\dfrac{2}{3}}$
					& $0$
					& $0$
					& $\dfrac{\mathrm{i}\sqrt[4]{5}}{\sqrt{3}}$
					& $-1$
					& $0$
					& $0$
					& $\dfrac{2}{\sqrt{3}}$ \\[2pt]
					\midrule
					
					No. & 105 & 106 & 107 & & & & & \\*
					\midrule
					Triplet
					& $(6,6)(6,6)(6,8)$
					& $(6,6)(6,6)(8,6)$
					& $(6,6)(6,6)(8,8)$
					& & & & & \\*[1.5pt]
					$\alpha_{IJK}$
					& $0$
					& $0$
					& $0$
					& & & & & \\[2pt]
					\bottomrule
				\end{tabular}%
		}}
	\end{center}
	    \end{minipage}
    \endgroup
	\clearpage

	\newpage
	\section{Existence and uniqueness for all Kac labels}\label{sec:global_solution}
    
    Section~\ref{sec:summary_alpha_s1} presents candidate seed solutions and the numerical evidence supporting them. In this section, we establish the correspondence between exact seed solutions and solutions for all Kac labels.
	
	We now consider the general reduced OPE coefficients
	$$\alpha_{(\ell_1,\bar\ell_1,s_1)(\ell_2,\bar\ell_2,s_2)(\ell_3,\bar\ell_3,s_3)}\,.$$
	
	\begin{theorem}[Seed-to-full correspondence]\label{thm:seed-correspondence}
		Fix a physical coprime pair $(p,q)$ and an exceptional spectrum in Eq.~\eqref{def:model}, with the odd-$s$ representatives, two-point normalization~\eqref{eq:normalization}, and finite nonzero reference factors~\eqref{eq:factorization}. Let $\mathcal S_{\mathrm{full}}$ and $\mathcal S_{\mathrm{seed}}$ denote the exact normalized solution spaces of bulk locality and crossing, and let $\mathcal G_{\mathrm{full}}$ and $\mathcal G_{\mathrm{seed}}$ be their respective primary-field sign groups, with vacuum sign fixed. Restriction to $s=1$ induces a bijection
		\[
		\mathcal S_{\mathrm{full}}/\mathcal G_{\mathrm{full}}
		\ \longleftrightarrow\
		\mathcal S_{\mathrm{seed}}/\mathcal G_{\mathrm{seed}}.
		\]
		Its inverse is the $s$-independent lift~\eqref{eq:lift}. Hence exact existence and uniqueness modulo field signs for the seed problem imply the corresponding all-label result.
	\end{theorem}
	The next two subsections prove the extension and its uniqueness. The theorem is a statement about exact solution spaces, irrespective of whether a numerical candidate has been certified as a member. The proof uses the complete truncated $SU(2)$ fusion spaces, complex orthogonality, and the specific nonvanishing recoupling entries described below; a multiplicity-free Cartesian spectrum alone would not suffice.

	\subsection{Existence}
	
	Write \(\id_s=V_{(0,0,s)}\), and denote a generic operator $V_{(\ell,\bar\ell,s)}$ by a capital letter, for example $A_s$. Given the
	$s=1$ solution, define
	\begin{equation}\label{eq:lift}
		\alpha_{A_{s_1}B_{s_2}C_{s_3}}
		:=
		\alpha_{A_1B_1C_1}
	\end{equation}
	for every admissible triplet \((A_{s_1},B_{s_2},C_{s_3})\). The lift also preserves the normalization, since $\alpha_{A_sA_s\id_1}=\alpha_{A_1A_1\id_1}=1$.
	
	First, \eqref{eq:lift} satisfies the permutation relations \eqref{eq:permutation_alpha} because spin parity does not depend on $s$ (see Eq.~\eqref{eq:spin}):
	\begin{equation}
		\begin{split}
			(-1)^{S_{(\ell,\bar\ell,s)}}=(-1)^{S_{(\ell,\bar\ell,1)}}
		\end{split}
	\end{equation}
	for any $(\ell,\bar{\ell})$ in the $E$-series spectrum.
	
	Second, the lift \eqref{eq:lift} satisfies the bootstrap equations \eqref{eq:reduced-crossing}. To verify this, we first sum over the second Kac index on the left-hand side of \eqref{eq:reduced-crossing}. As shown in Eq.~\eqref{def:model}, the $E$-series spectra are diagonal in the second Kac label, and each allowed odd second label occurs with multiplicity one for every allowed pair of first Kac labels. The sum therefore runs over the complete intermediate fusion space for the second Kac label. Complex orthogonality of the normalized $6j$ symbols then gives
	\begin{equation}
		\begin{split}
			\sum_{\ell^-_5}\sixj{\ell^-_1}{\ell^-_2}{\ell^-_5}
			{\ell^-_3}{\ell^-_4}{\ell^-_6}_{q_-}
			\sixj{\ell^-_1}{\ell^-_2}{\ell^-_5}
			{\ell^-_3}{\ell^-_4}{\bar\ell^-_6}_{q_-}
			=\delta_{\ell^-_6,\bar\ell^-_6}\,.
		\end{split}
	\end{equation}
	For equal left and right second labels in the target channel, this reduces the remaining equation to the $s=1$ bootstrap equation, Eq.~\eqref{eq:reduced-crossing_s=1}, including its zero right-hand side when the target first-label pair is absent from the spectrum. The Kronecker delta also sets all forbidden mixed second-label target channels to zero. This orthogonality relation contains no complex conjugation and requires no positivity assumption, so the argument applies equally to nonunitary models.
	
	Thus an exact seed solution yields a solution of the normalized bulk locality and OPE-associativity equations for every admissible second Kac label.
	
	\subsection{Uniqueness}
	
	Uniqueness of the extension follows from the multiplicity-free Cartesian structure of the $E$-series spectra together with the properties of the truncated $SU(2)$ fusion category. We state the result for these spectra, with the normalization and representative conventions used throughout this paper.
	\begin{lemma}\label{lemma:uniqueness_from_s=1}
		Fix an $E$-series spectrum in Eq.~\eqref{def:model} at physical coprime parameters $p,q>1$, using the odd-$s$ representatives. Its primary fields are $A_s=V_{(r,s;\bar r,s)}$, with one copy of each allowed first-label pair $A=(r,\bar r)$ for every $s=1,3,\ldots,p-2$. Impose the two-point normalization \eqref{eq:normalization} and use the reference factors of Eq.~\eqref{eq:factorization}, finite and nonzero on admissible triplets and normalized by $C^{\mathrm{(ref)}}_{II\id}=1$.
		Then every solution of the reduced locality and crossing equations is related by the operator-sign redefinitions \eqref{eq:gauge} to the lift \eqref{eq:lift} of its $s=1$ restriction. In particular, fixing the seed solution fixes its extension uniquely up to field signs.
	\end{lemma}
	
	The multiplicity-free hypothesis is important: when several primaries carry the same left and right Virasoro representations, transformations preserving the two-point form can mix them, and the field-basis freedom need not reduce to independent signs. Such spectra are outside the scope of this lemma.
	
	The numerical evidence for seed uniqueness and its limitations are discussed in Section~\ref{sec:s=1}. The present lemma concerns exact solutions.

	Runkel previously obtained bulk OPE coefficients and uniqueness up to field redefinitions for the $A$ and $D$ series within the bulk--boundary sewing problem \cite{Runkel:1998a,Runkel:1999dz}. Our proof uses bulk equations directly and exploits the truncated $SU(2)$ fusion structure of the diagonal second-label sector. This is a distinction between methods. It does not posit additional boundaryless ADE bulk solutions: under the local Frobenius-algebra, unique-vacuum and ADE-spectrum hypotheses, the converse full-center construction supplies a consistent boundary completion \cite{KongRunkel:2008Cardy}.
	
	The proof has three steps:
	\begin{enumerate}
		\item $\alpha_{\id_{s_1}\id_{s_2}\id_{s_3}}=1$ after fixing the signs of the $\id_s$ operators;
		\item $\alpha_{A_{s_1}A_{s_2}\id_{s_3}}=1$ after fixing the signs of the $A_s$ operators;
		\item $\alpha_{A_{s_1}B_{s_2}C_{s_3}}=\alpha_{A_{1}B_{1}C_{1}}$.
	\end{enumerate}

	\subsubsection{\texorpdfstring{Step 1: the diagonal \((1,s)\) sector}{Step 1: the diagonal (1,s) sector}}
	The $(1,s)$ sector contains only scalar operators, so the corresponding
	$\alpha_{IJK}$ are fully symmetric under permutations. Here, $s$ runs over all allowed odd integers
	\begin{equation}
		\begin{split}
			s=1,3,\ldots,p-2\,.
		\end{split}
	\end{equation}
	For the physical $E$-series minimal models, $p\geqslant5$. 
	
	To simplify the notation, we switch to the quantum-group spins using
	\begin{equation}
		\begin{split}
			s=2\ell+1\,,\quad\ell=0,1,\ldots,\frac{p-3}{2}\,.
		\end{split}
	\end{equation}
	In this step only, $\ell$ denotes the quantum-group spin associated with the second Kac label; elsewhere, it denotes the spin associated with the first Kac label.
	
	We define
	\begin{equation}
		\mathsf{F}^{(\ell)}_{\ell_5,\ell_6}
		:=
		\sixj{\ell}{1}{\ell_5}{1}{\ell}{\ell_6}_{q_-},
		\qquad
		\ell_5\in\{\ell-1,\ell,\ell+1\},
		\quad
		\ell_6\in\{0,1,2\}.
	\end{equation}
	As a matrix with indices $\ell_5$ and $\ell_6$, its entries are
	\begin{equation}
		\begin{split}
			\mathsf{F}^{(\ell)}_{\ell-1,0}
			&=
			\mathcal{N}_{\ell}
			\sqrt{
				[2]_{q_-}[4]_{q_-}
				[2\ell-1]_{q_-}[2\ell]_{q_-}[2\ell+2]_{q_-}
			},
			\\
			\mathsf{F}^{(\ell)}_{\ell-1,1}
			&=
			-\mathcal{N}_{\ell}
			[2]_{q_-}[2\ell+2]_{q_-}
			\sqrt{[3]_{q_-}[2\ell-1]_{q_-}},
			\\
			\mathsf{F}^{(\ell)}_{\ell-1,2}
			&=
			\mathcal{N}_{\ell}
			[2]_{q_-}[2\ell+2]_{q_-}
			\sqrt{[2\ell+3]_{q_-}},
			\\[3pt]
			\mathsf{F}^{(\ell)}_{\ell,0}
			&=
			-\mathcal{N}_{\ell}
			\sqrt{
				[2]_{q_-}[4]_{q_-}
				[2\ell]_{q_-}[2\ell+1]_{q_-}[2\ell+2]_{q_-}
			},
			\\
			\mathsf{F}^{(\ell)}_{\ell,1}
			&=
			\mathcal{N}_{\ell}
			\left(
			[2\ell+3]_{q_-}-[2\ell-1]_{q_-}
			\right)
			\sqrt{[3]_{q_-}[2\ell+1]_{q_-}},
			\\
			\mathsf{F}^{(\ell)}_{\ell,2}
			&=
			\mathcal{N}_{\ell}
			[2]_{q_-}^{\,2}
			\sqrt{
				[2\ell-1]_{q_-}
				[2\ell+1]_{q_-}
				[2\ell+3]_{q_-}
			},
			\\[3pt]
			\mathsf{F}^{(\ell)}_{\ell+1,0}
			&=
			\mathcal{N}_{\ell}
			\sqrt{
				[2]_{q_-}[4]_{q_-}
				[2\ell]_{q_-}[2\ell+2]_{q_-}[2\ell+3]_{q_-}
			},
			\\
			\mathsf{F}^{(\ell)}_{\ell+1,1}
			&=
			\mathcal{N}_{\ell}
			[2]_{q_-}[2\ell]_{q_-}
			\sqrt{[3]_{q_-}[2\ell+3]_{q_-}},
			\\
			\mathsf{F}^{(\ell)}_{\ell+1,2}
			&=
			\mathcal{N}_{\ell}
			[2]_{q_-}[2\ell]_{q_-}
			\sqrt{[2\ell-1]_{q_-}},
		\end{split}
	\end{equation}
	where
	\begin{equation}
		\begin{split}
			\mathcal{N}_{\ell}
			=
			\frac{1}{
				\sqrt{
					[2]_{q_-}[3]_{q_-}[4]_{q_-}
					[2\ell]_{q_-}[2\ell+1]_{q_-}[2\ell+2]_{q_-}
				}
			}\,.
		\end{split}
	\end{equation}
	For $\ell=1,2,\ldots,\tfrac{p-3}{2}$, $\mathcal{N}_\ell$ is finite and nonzero because the arguments of all $q_-$-numbers $[n]_{q_-}$ in its denominator lie between $1$ and $p-1$.\footnote{Given $q_-=e^{i\pi q/p}$, we have $[n]_{q_-}=\tfrac{\sin\left(nq\pi/p\right)}{\sin\left(q\pi/p\right)}$. Since $q$ and $p$ are coprime, $[n]_{q_-}=0$ if and only if $n$ is a multiple of $p$.} We now examine the entries of $\mathsf{F}^{(\ell)}$ case by case:
	\begin{itemize}
		\item When $\ell=0$, $\mathsf{F}^{(\ell)}$ is the $1\times1$ matrix with entry $1$.
		\item When $\ell=1,2,\ldots,\tfrac{p-5}{2}$, all entries of $\mathsf{F}^{(\ell)}$ are nonzero. This range is empty for $p=5$. For $p\geqslant7$, the only entry requiring a separate check
		is the central one, $\mathsf{F}^{(\ell)}_{\ell,1}$. Using $q_-=e^{i\pi q/p}$,
		\begin{equation}
			[2\ell+3]_{q_-}-[2\ell-1]_{q_-}
			=4\cos\left(\tfrac{(2\ell+1)\pi q}{p}\right)
			\cos\left(\tfrac{\pi q}{p}\right)\,.
		\end{equation}
		Since $p$ is odd and $q$ is even, neither cosine can vanish. 
		\item When $\ell=\tfrac{p-3}{2}$, the spin-$(\ell+1)$ representation is absent from the fusion ring $SU(2)_{p-2}$, and $[2\ell+3]_{q_-}=[p]_{q_-}=0$. In this case, $\mathsf{F}^{(\ell)}$ restricts to a $2\times2$ matrix whose allowed entries are all nonzero.
	\end{itemize}
	
	Thus every entry of $\mathsf{F}^{(\ell)}$ on its allowed fusion channels is nonzero.
	
	Consider the bootstrap equation for
	$\braket{\id_{s}\id_{3}\id_{3}\id_{s}}$, where
	$s=2\ell+1$. The corresponding $6j$ symbols in the bootstrap equations are precisely $\mathsf{F}^{(\ell)}$ discussed above. We use the following
	elementary linear-algebra lemma.
	\begin{lemma}\label{lemma:linear_algebra}
		Let $U_{ab}$ ($a,b=1,2,\ldots,N$) be a complex orthogonal matrix acting on
		$\mathbb{C}^N$, i.e.,
		\begin{equation}\label{eq:U_ortho}
			\begin{split}
				\sum_{a=1}^NU_{ab}U_{ac}=\delta_{bc}\,,\quad\sum_{a=1}^NU_{ba}U_{ca}=\delta_{bc}\,.
			\end{split}
		\end{equation}
		Suppose that two vectors $x,y\in\mathbb{C}^N$ satisfy
		\begin{equation}\label{eq:lemma_assumption}
			\begin{split}
				\sum_{a}x_a U_{ab}U_{ac}=y_b\delta_{bc}\,,
			\end{split}
		\end{equation}
		Then $x_a=y_b$ for every pair $(a,b)$ such that $U_{ab}\neq0$.
	\end{lemma}
	\begin{proof}
		Multiplying both sides of \eqref{eq:lemma_assumption} by $U_{dc}$ and summing
		over $c$, Eq.~\eqref{eq:U_ortho} gives
		\begin{equation}
			\begin{split}
				x_d U_{db}=y_{b}U_{db}\,.
			\end{split}
		\end{equation}
		It follows that $x_d=y_b$ whenever $U_{db}\neq0$.
	\end{proof}
	
	For $\braket{\id_{s}\id_{3}\id_{3}\id_{s}}$,
	identify
	$U_{\ell_5,\ell_6}=\mathsf{F}^{(\ell)}_{\ell_5,\ell_6}$,
	$x_{\ell_5}=(\alpha_{\id_{3}\id_{s}\id_{s_5}})^2$,
	and
	$y_{\ell_6}=\alpha_{\id_{3}\id_{3}\id_{s_6}}
	\alpha_{\id_{s}\id_{s}\id_{s_6}}$.
	For $s=1,3,5,\ldots,p-2$, Lemma~\ref{lemma:linear_algebra} and the nonvanishing
	of all matrix entries give
	\begin{equation}
		\begin{split}
			&(\alpha_{\id_{3}\id_{s}\id_{s_5}})^2=\alpha_{\id_{3}\id_{3}\id_{s_6}}\alpha_{\id_{s}\id_{s}\id_{s_6}}=1\,, \\
			&\forall s_5\in\mathcal{R}_p(s,3)\,,\quad s_6\in\mathcal{R}_p(3,3)\cap\mathcal{R}_p(s,s)\,,
		\end{split}
	\end{equation}
	where $\mathcal{R}_p(s_1,s_2)$ was defined in \eqref{def:Rm}. The generic cases give $3\times3$ matrices, with $s_5=s-2,s,s+2$ and $s_6=1,3,5$, while the endpoint cases $s=1$ and $s=p-2$ give, respectively, $1\times1$ and
	$2\times2$ restrictions of the same fusion matrix. Their allowed entries are
	also nonzero, so the same conclusion holds. We may therefore fix the
	operator-sign gauge by imposing
	\begin{equation}\label{eq:id_gauge}
		\begin{split}
			\alpha_{\id_{3}\id_{3}\id_{3}}&=1\,, \\
			\alpha_{\id_{3}\id_{s-2}\id_{s}}&=1\,,\quad\text{for}\ s=5,7,\ldots,p-2\,.
		\end{split}
	\end{equation}
	
	This exhausts the sign freedom of the $\id_s$ operators. We now show
	that, in this gauge,
	$\alpha_{\id_{s_1}\id_{s_2}\id_{s_3}}=1$ for every
	admissible triplet $(s_1,s_2,s_3)$.
	
	If one of the $s_i$ equals $1$, the result follows from the normalization
	condition. Otherwise, order the indices as
	$s_1\geqslant s_2\geqslant s_3\geqslant3$.
	
	When $s_2\geqslant5$, the triplet $(s_1-2,s_2-2,s_3)$ is also admissible. The
	original admissibility conditions are
	\begin{equation}
		\begin{split}
			\abs{\ell_1-\ell_2}\leqslant\ell_3\leqslant\ell_1+\ell_2\,,\quad \ell_1+\ell_2+\ell_3\leqslant p-2\,,
		\end{split}
	\end{equation}
	while those for $(s_1-2,s_2-2,s_3)$ follow from
	\begin{equation}
		\begin{split}
			&\abs{(\ell_1-1)-(\ell_2-1)}\leqslant\ell_3\,, \\
			&\ell_3\leqslant\ell_2\leqslant \ell_1+\ell_2-2=(\ell_1-1)+(\ell_2-1)\,, \\
			&(\ell_1-1)+(\ell_2-1)+\ell_3\leqslant p-2-2\leqslant p-2\,.
		\end{split}
	\end{equation}
	
	Now consider the bootstrap equation for
	\begin{equation}
		\begin{split}
			\braket{\id_{s_1-2}\id_{3}\id_{s_2}\id_{s_3}}\,.
		\end{split}
	\end{equation}
	Set
	\begin{equation}
		a:=\ell_1+\ell_2-\ell_3,\qquad
		b:=\ell_1+\ell_2+\ell_3.
	\end{equation}
	\begin{equation}
		\begin{aligned}
			\sixj{\ell_1-1}{1}{\ell_1}{\ell_2}{\ell_3}{\ell_2-1}_{q_-}
			&=\left[
			\frac{
				[a-1]_{q_-}[a]_{q_-}[b]_{q_-}[b+1]_{q_-}
			}{
				[2\ell_1-1]_{q_-}
				[2\ell_1]_{q_-}
				[2\ell_2]_{q_-}
				[2\ell_2+1]_{q_-}
			}
			\right]^{1/2}.
		\end{aligned}
	\end{equation}
	Recall that $[n]_{q_-}=0$ if and only if $n$ is a multiple of $p$. The
	displayed $6j$ symbol is nonzero because
	\begin{itemize}
		\item we have $\ell_1+\ell_2-\ell_3-1\geqslant2+(\ell_2-\ell_3)-1\geqslant1$ and $\ell_1+\ell_2-\ell_3-1\leqslant\ell_1+\ell_2-1\leqslant p-3$, which imply that $[a-1]_{q_-}$ and $[a]_{q_-}$ are nonzero;
		\item we have $1\leqslant \ell_1+\ell_2+\ell_3\leqslant p-2$, which implies that $[b]_{q_-}$ and $[b+1]_{q_-}$ are nonzero;
		\item $2\leqslant\ell_1,\ell_2\leqslant\tfrac{p-3}{2}$, which implies that all the factors in the denominator are finite and nonzero.
	\end{itemize} 
	Then applying Lemma~\ref{lemma:linear_algebra} to this $6j$-symbol matrix gives
	\begin{equation}\label{eq:id_sss}
		\begin{split}
			\alpha_{\id_{s_1-2}\id_{3}\id_{s_1}}\alpha_{\id_{s_2}\id_{s_3}\id_{s_1}}=\alpha_{\id_{3}\id_{s_2}\id_{s_2-2}}\alpha_{\id_{s_3}\id_{s_1-2}\id_{s_2-2}}\,.
		\end{split}
	\end{equation}
	Together with \eqref{eq:id_gauge}, this gives
	\begin{equation}\label{eq:elementary_move_id}
		\begin{split}
			\alpha_{\id_{s_1}\id_{s_2}\id_{s_3}}=\alpha_{\id_{s_1-2}\id_{s_2-2}\id_{s_3}}
		\end{split}
	\end{equation}
	for all triples $(s_1,s_2,s_3)$ satisfying $s_1\geqslant s_2\geqslant s_3\geqslant3$ and $s_2\geqslant5$.
	
	Equation~\eqref{eq:elementary_move_id}, together with permutation symmetry, defines an elementary move. Whenever two
	indices are at least $5$, the two largest may be decreased by two without
	changing $\alpha$. Repeating the move finitely many times reaches either
	$(3,3,3)$ or $(5,3,3)$, where $\alpha=1$ by \eqref{eq:id_gauge}. Therefore,
	\begin{equation}\label{eq:ids}
		\begin{split}
			\alpha_{\id_{s_1}\id_{s_2}\id_{s_3}}=1
		\end{split}
	\end{equation}
	for every admissible triplet.

	\subsubsection{Step 2: two equal non-identity labels}
	
	For the second step, let $A\neq\id$ and consider the bootstrap equation for
	$\braket{A_1\id_s\id_sA_1}$. Each channel contains a single operator, and all $6j$ symbols equal one. The equation therefore gives
	\begin{equation}
		\begin{split}
			\alpha_{A_1\id_s A_s}\,\alpha_{\id_s A_1 A_s}=\alpha_{\id_s\id_s\id_1}\,\alpha_{A_1 A_1\id_1}\,.
		\end{split}
	\end{equation}
	Using $(-1)^{S_{\id_s}}=1$ and $(-1)^{S_{A_s}}=(-1)^{S_{A_1}}$ for the $E$ series, the permutation relations give
	\begin{equation}
		\begin{split}
			\left(\alpha_{A_sA_1\id_s}\right)^2=1
		\end{split}
	\end{equation}
	We fix the operator-sign ambiguity of the $A_s$
	operators by imposing
	\begin{equation}\label{eq:gauge_As}
		\begin{split}
			\alpha_{A_sA_1\id_s}=1\,,\forall s\,,
		\end{split}
	\end{equation}
	where the $s=1$ case follows from the normalization. These choices fix the signs of $A_s$ relative to the chosen sign of $A_1$.
	
	Next consider the bootstrap equation for
	$\langle A_{s_1}A_1\id_{s_2}\id_{s_3}\rangle$. Only $\id_{s_1}$ appears in the $s$-channel, and only $A_{s_2}$ appears in the $t$-channel. All $6j$ symbols equal one, so the bootstrap equation becomes
	\begin{equation}
		\alpha_{A_{s_1}A_1\id_{s_1}}\,\alpha_{\id_{s_2}\id_{s_3}\id_{s_1}}=\alpha_{A_{1}\id_{s_2}A_{s_2}}\,\alpha_{\id_{s_3}A_{s_1}A_{s_2}}\,.
	\end{equation}
	By \eqref{eq:ids} and \eqref{eq:gauge_As}, we conclude that
	\begin{equation}\label{eq:two-A}
		\begin{split}
			\alpha_{A_{s_1}A_{s_2}\id_{s_3}}=1
		\end{split}
	\end{equation}
	for every admissible triplet $(s_1,s_2,s_3)$.

	\subsubsection{Step 3: the general OPE coefficient}
	
	All operator-sign freedom outside the chosen $s=1$ gauge has now been fixed, consistently with
	\eqref{eq:lift}. It remains to prove that every other coefficient also equals
	its lifted value. For a generic admissible triplet
	$(A_{s_1},B_{s_2},C_{s_3})$, we consider the bootstrap equations for the following two four-point functions $$\braket{A_1\id_{s_1}B_{s_2}C_{s_3}}\,,\quad\braket{B_1\id_{s_3}C_{s_3}A_1}.$$
	In both cases, each channel contains a single operator and all $6j$ symbols equal one, giving
	\begin{equation}
		\begin{split}
			\alpha_{A_1\id_{s_1}A_{s_1}}\,\alpha_{A_{s_1}B_{s_2}C_{s_3}}&=\alpha_{\id_{s_1}B_{s_2}B_{s_3}}\,\alpha_{A_1 B_{s_3}C_{s_3}}\,, \\
			\alpha_{B_1\id_{s_3}B_{s_3}}\,\alpha_{C_{s_3}A_1B_{s_3}}&=\alpha_{\id_{s_3}C_{s_3}C_1}\,\alpha_{A_1B_1C_1}\,.
		\end{split}
	\end{equation}
	By \eqref{eq:two-A}, the above equations imply
	\begin{equation}
		\begin{split}
			\alpha_{A_{s_1}B_{s_2}C_{s_3}}=\alpha_{A_1 B_{s_3}C_{s_3}}=\alpha_{A_1B_1C_1}\,.
		\end{split}
	\end{equation}
	This proves that the extension of fixed seed data is unique up to the operator-sign redefinitions \eqref{eq:gauge}. Combined with the existence argument above, it proves the claimed conditional result for all Kac labels.

	\section{Discussion}\label{sec:conclusion}
	We have obtained explicit candidate bulk OPE coefficients for all
	$E$-series minimal models. In the normalization convention \eqref{eq:normalization}, the general expression is given in Eqs.~\eqref{eq:factorization}--\eqref{def:phi}. A uniform treatment of the three families is possible because the spectrum is diagonal in the second Kac index $s$. After the universal $s$ dependence is factored out, the exceptional bootstrap problem depends only on the first Kac indices; the resulting coefficients are listed in Tables~\ref{tab:E6-I1-results}--\ref{tab:E8-I13-results}.

	The finite seed computations support the following conjecture, whose all-label conclusion follows from exact seed existence and uniqueness by Section~\ref{sec:global_solution}.
	\begin{conjecture}\label{thm:main}
		For every admissible coprime pair $(p,q)$ and each $E_6$-, $E_7$-, or $E_8$-series spectrum, the coefficients reconstructed from Eqs.~\eqref{eq:factorization}--\eqref{def:phi}, Tables~\ref{tab:E6-I1-results}--\ref{tab:E8-I13-results}, and the operator-sign conventions of Appendix~\ref{app:simple_current_conventions} satisfy bulk locality and OPE associativity with the two-point normalization \eqref{eq:normalization}. Every solution with this spectrum and normalization is related to these coefficients by the independent primary-field sign redefinitions \eqref{eq:gauge}.
	\end{conjecture}
	
	For the six unitary models, existence and uniqueness within the conformal-net framework were established in Ref.~\cite{Kawahigashi:2003gi}. Conjecture \ref{thm:main} is formulated directly for bulk locality and crossing, without assuming reflection positivity.

	\subsection{Unitary case}\label{subsec:unitary}
	In a unitary CFT, the OPE coefficients obey an additional conjugation condition. With our normalization \eqref{eq:normalization} and the reflection convention specified below, unitarity implies
	\begin{equation}\label{eq:unitarity}
		\begin{split}
			C_{IJK}^*=(-1)^{S_I+S_J+S_K}C_{IJK}\,.
		\end{split}
	\end{equation}
	
	Let us briefly review how \eqref{eq:unitarity} is obtained. Choose a Euclidean time coordinate $\tau$ and set
	\begin{equation}
		\begin{split}
			z=x+i\tau\,,\quad\bar{z}=x- i\tau\,.
		\end{split}
	\end{equation}
	In Euclidean QFT, unitarity is expressed by \emph{reflection positivity}. In the normalization \eqref{eq:normalization}, the reflection conjugates the primary fields according to
	\begin{equation}
		\begin{split}
			\left[V_I(z,\bar{z})\right]^\dagger=(-1)^{S_I}\,V_I(\bar{z},z)\,.
		\end{split}
	\end{equation}
	The factor $(-1)^{S_I}$ compensates the coordinate phase and makes the reflected two-point function $\braket{\left[V_I(z,\bar{z})\right]^\dagger V_I(z,\bar{z})}$ positive. Then the three-point functions must satisfy
	\begin{equation}
		\begin{split}
			\braket{V_I(z_1,\bar{z}_1)V_J(z_2,\bar{z}_2)V_K(z_3,\bar{z}_3)}^*=\braket{\left[V_I(z_1,\bar{z}_1)\right]^\dagger \left[V_J(z_2,\bar{z}_2)\right]^\dagger \left[V_K(z_3,\bar{z}_3)\right]^\dagger }\,.
		\end{split}
	\end{equation}
	Substituting \eqref{eq:ope_convention} gives \eqref{eq:unitarity}.\footnote{The alternative convention $z=\tau+ix$, $\bar z=\tau-ix$ gives the same condition \eqref{eq:unitarity}. With field components normalized as in \eqref{eq:normalization}, the reflection relation becomes $[V_I(z,\bar z)]^\dagger=V_I(-\bar z,-z)$. The factor $(-1)^{S_I+S_J+S_K}$ then arises from the reflected OPE coordinate factors.}
	 
	For Virasoro minimal models $\mathcal{M}(p,q)$, unitarity only occurs when $\abs{p-q}=1$ \cite{Friedan:1983xq}. For $E$-series minimal models, it means
	\begin{equation}\label{E_unitary}
		\begin{split}
			p=\begin{cases}
				11,13 & \text{for }E_6\,, \\
				17,19 & \text{for }E_7\,, \\
				29,31 & \text{for }E_8\,. \\
			\end{cases}
		\end{split}
	\end{equation}
	The $p=q-1$ and $p=q+1$ members have realizations in the critical dense and dilute ADE lattice families, respectively \cite{Pasquier:1986jc,Pasquier:1987xj,Warnaar:1992gj,Roche:1992ww,Warnaar:1993zn,OBrien:1995jsl}. The unitary interpretation of these models should be distinguished from checking the conjugation convention of the reconstructed coefficients.

	At the unitary central charges, the reconstructed relative coefficients obey
	\begin{equation}
		\alpha_{IJK}^*=(-1)^{S_I+S_J+S_K}\alpha_{IJK}\,.
	\end{equation}
	To infer \eqref{eq:unitarity} from this observation and \eqref{eq:factorization}, the additional statement needed is
	\begin{equation}\label{eq:reference-reality-needed}
		\bigl(C^{\mathrm{(ref)}}_{IJK}\bigr)^*=C^{\mathrm{(ref)}}_{IJK}
	\end{equation}
	for physical triplets, with the correlated branch prescription of Appendix~\ref{app:fusion_kernel}. 
	
	Lemma~\ref{lem:unitary-reference-reality}
	proves the required reality \eqref{eq:reference-reality-needed} for every
	admissible physical triplet at $p=q\pm1$. Together with the conjugation
	property of the reduced coefficients, this establishes
	Eq.~\eqref{eq:unitarity} for the displayed candidate OPE data.
	
	\begin{remark}
		Individual chiral factors need not be real. For example, at $(p,q)=(11,12)$ the formula in Appendix~\ref{app:fusion_kernel} gives
		\begin{equation}\label{eq:chiral-reality-counterexample}
			D_{(1,3),(4,1),(4,3)}^{\,2}=-\frac{13}{32}\,.
		\end{equation}
		This coefficient is purely imaginary, whereas the corresponding diagonal
		left--right product is real.
	\end{remark}

	\subsection{Parity symmetry}\label{sec:parity}
	
	According to the torus partition functions \eqref{def:model}, the spectra of all the $E$-series minimal models are invariant
	under the parity map\footnote{More generally, the parity transformation maps the Kac labels
		$(r,s;\bar r,\bar s)$ to $(\bar r,\bar s;r,s)$.}
	\begin{equation}
		\begin{split}
			P(I)=(\bar\ell,\ell,s)\,,
			\qquad
			I=(\ell,\bar\ell,s)\,.
		\end{split}
	\end{equation}
	This map sends any solution of the bootstrap equations
	\eqref{eq:reduced-crossing} to another solution,
	\begin{equation}
		\begin{split}
			C_{IJK}\longrightarrow
			\widetilde C_{IJK}
			:=C_{P(I)P(J)P(K)}\,.
		\end{split}
	\end{equation}
	The permutation properties \eqref{eq:permutation} are manifestly preserved, since parity
	maps the spin $S_I$ to $-S_I$, while $(-1)^{S_I}$ remains unchanged.
	
	Assuming the uniqueness statement of Conjecture~\ref{thm:main}, the parity-transformed solution is related to the original one by globally defined field signs:
	\begin{equation}
		\widetilde C_{IJK}=\epsilon_I\epsilon_J\epsilon_K C_{IJK},
		\qquad \epsilon_I\in\{\pm1\}\,.
	\end{equation}
	Consequently,
	\begin{equation}\label{def:parity_on_op}
		(P\cdot V_I)(z,\bar z):=\epsilon_I V_{P(I)}(\bar z,z)
	\end{equation}
	defines the parity transformation of the operators. Any correlator is invariant when all the operators transform under \eqref{def:parity_on_op}:
	\begin{equation}
		\begin{split}
			\braket{V_{I_1}(x_1)\,V_{I_2}(x_2)\,\ldots\,V_{I_n}(x_n)}=\braket{(P\cdot V_{I_1})(x_1)\,(P\cdot V_{I_2})(x_2)\,\ldots\,(P\cdot V_{I_n})(x_n)}\,.
		\end{split}
	\end{equation}
	In this sense, we say that the $E$-series minimal models are parity symmetric.

	A sufficient, stronger condition is invariance under the bare exchange of left- and right-moving labels,
	\begin{equation}\label{parity_assumption}
		C_{IJK}=C_{P(I)P(J)P(K)}\,.
	\end{equation}
	This cannot always be achieved in a parity-symmetric CFT, because the chosen primary operator basis may contain a parity-odd scalar, which transforms under parity as $(P\cdot V_I)(z,\bar{z})=-V_I(\bar{z},z)$. Then consider a pair of distinct labels related by parity: $J$ and $P(J)$. The parity symmetry implies that
	\begin{equation}
		\begin{split}
			C_{IJJ}=-C_{IP(J)P(J)}\,.
		\end{split}
	\end{equation}
	If $C_{IJJ}\neq0$, it is an obstruction to \eqref{parity_assumption} because it cannot be fixed by operator-sign redefinition. Appendix \ref{app:parity_counter_example} gives a concrete example using three decoupled free bosons.
	
	If a CFT has an operator-sign gauge satisfying \eqref{parity_assumption}, then the choice $\epsilon_I=1$ for \eqref{def:parity_on_op} makes the parity symmetry manifest. Every scalar primary is then even under this chosen parity. To establish this statement for $E$-series minimal models, one must verify the identity for the complete seed data in a specified gauge, including coefficients reconstructed by the simple-current relations. The left--right-symmetric reference factor, the symmetric phase \eqref{def:phi}, complex conjugation, and the all-label lifting argument then propagate the identity to the other cases.
	
	Now the concrete problem is to find a suitable collection of $\eta_I\in\{\pm1\}$, with $\eta_{(0,0)}=1$, such that
	\begin{equation}\label{parity:to_check}
		\begin{split}
			\alpha'_{IJK}
			=\eta_I\eta_J\eta_K\,\alpha_{IJK}\, ,
			\qquad
			\alpha'_{IJK}
			=\alpha'_{P(I)P(J)P(K)}\,.
		\end{split}
	\end{equation}
	This equality was verified exactly for every ordered triple of fields
	using the analytic OPE data. In the conventions of Tables \ref{tab:E6-I1-results}--\ref{tab:E8-I13-results}, we summarize the parity checks as follows.
	\begin{itemize}
		\item For $E_6$ at $(p,q)=(1,12)$, one can assign
		$\eta_I=-1$ to $I=(3,0),(\tfrac72,\tfrac32),(5,2)$; Eq.~\eqref{parity:to_check} then holds;
		at $(p,q)=(5,12)$, it suffices to assign
		$\eta_I=-1$ to
		$I=(\frac72,\frac32)$.
		All other operator signs may be kept positive.
		\item All
		$E_7$ and $E_8$ reference cases already satisfy \eqref{parity:to_check} with all $\eta_I=1$.
	\end{itemize}
	 
    The exact check covers every ordered triple in each of the nine
    reference seeds. Together with the transport and lifting arguments,
    this establishes the stronger parity condition
    \eqref{parity_assumption} for our candidate OPE data throughout the
    $E$-series families. Conditional on Conjecture~\ref{thm:main},
    it justifies the assumption used in Ref.~\cite{Nivesvivat:2025odb}
    for all $E$-series minimal models.

	\subsection{Additional global symmetries}\label{sec:symmetries}
	
	The analytic bootstrap equations use Virasoro symmetry, permutation symmetry (bulk locality), and crossing symmetry (OPE associativity). The candidate coefficients also exhibit additional sector selection rules. Their interpretation should be distinguished from the numerical evidence for uniqueness of the unrestricted bootstrap solution.

	Virasoro symmetry and null-state decoupling imply the chiral fusion rules \eqref{chiral_fusion_rule}, while permutation symmetry constrains the OPE coefficients. In particular,
	\begin{equation}
		\alpha_{AAB}=(-1)^{S_B}\alpha_{AAB},
	\end{equation}
	so that $\alpha_{AAB}=0$ whenever $B$ has odd integer spin. Further selection rules are discussed in \cite{Nivesvivat:2025odb}.

	\paragraph{Parity symmetry.} As discussed in Section~\ref{sec:parity}, the $E$-series minimal models have parity symmetry. In a suitable operator-sign gauge, their OPE coefficients satisfy the stronger parity condition \eqref{parity_assumption}.

	\paragraph{Extended chiral algebra.}
	
	The $E_6$ and $E_8$ partition functions in \eqref{def:model} are sums of absolute squares of extended characters, reflecting extensions of the Virasoro chiral algebra. For $E_7$, the chiral extension is generated by the $\mathbb{Z}_2$ simple current, whose effects were already discussed in Section~\ref{sec:triplet_orbit}. We therefore focus below on the $E_6$ and $E_8$ cases.

	For $E_6$, the three exceptional sectors may be labelled $\id$, $\sigma$, and $\epsilon$. Their sector selection rules have the Ising fusion-ring form
	\begin{equation}\label{fusion:Ising}
		\begin{split}
			&\id\times\id=\id\,,\quad \id\times\sigma=\sigma\,,\quad \id\times\epsilon=\epsilon\,, \\
			&\sigma\times\sigma=\id+\epsilon\,,\quad\sigma\times\epsilon=\sigma\,,\quad\epsilon\times\epsilon=\id\,.
		\end{split}
	\end{equation} 
	These are rules for the exceptional sector labels; the second Kac labels must independently satisfy their fusion conditions. They explain some zeros allowed by the separate Virasoro fusion rules. Their interpretation as extended-algebra fusion rules uses the chiral extension, rather than the character decomposition alone.
	
	For $E_8$, the two exceptional sectors have Fibonacci-type fusion rules,
	\begin{equation}\label{fusion:LY}
		\id\times\id=\id\,,\quad \id\times\tau=\tau\,,\quad \tau\times\tau=\id+\tau .
	\end{equation}
	The last rule permits both sectors in the $\tau\times\tau$ channel, while the first two exclude one sector. Sharing this fusion ring does not identify the extended theory with the Lee--Yang CFT: conformal weights, central charge, and associativity and braiding data contain further information.
	
	The corresponding WZW extensions are associated with $\widehat{\mathfrak{su}}(2)_{10}\subset\widehat{\mathfrak{so}}(5)_1$ and $\widehat{\mathfrak{su}}(2)_{28}\subset(\widehat{\mathfrak g}_2)_1$ \cite{Christe:1987hj}. Their extended sector fusion rings have the forms \eqref{fusion:Ising} and \eqref{fusion:LY}; see Appendix~C of \cite{Petkova:1994zs}. This is a comparison of exceptional sector data, not an identification of the WZW and minimal-model CFTs.

	One $\mathbb Z_2\times\mathbb Z_2$ orbit of $E_8$ zeros is not explained by the preceding sector and permutation rules \cite{Nivesvivat:2025odb}. A representative is
	\begin{equation}\label{E8_extra_zero}
		\alpha_{(8,8)(8,8)(8,8)}=0\,.
	\end{equation}
	Below we would like to review two different arguments which lead to this zero. The first argument uses the relation between scalar relative OPE constants and the Pasquier algebra \cite{Pasquier:1987xj,Petkova:1994zs}, coming from the lattice realization of the minimal model \cite{Pasquier:1986jc}. The second argument directly uses the underlying discrete symmetries of the minimal model, known as the Ocneanu fusion algebra \cite{Ocneanu2000paths,Petkova:2000ip,PetkovaZuber:2001Ocneanu,Coquereaux:2001di}.
	
	\paragraph{Pasquier algebra.}
	
	Pasquier showed that the unitary Virasoro minimal model
	$\mathcal{M}(q-1,q)$ associated with the ADE pairing
	$(A_{q-2},G_q)$ can be realized as the continuum limit of a
	critical RSOS model \cite{Pasquier:1986jc}. The lattice degrees
	of freedom are heights taking values in the nodes of the Dynkin
	diagram $G$. Adjacent heights must correspond to adjacent graph
	nodes, and each allowed configuration is assigned a Boltzmann
	weight. For example, for $G=E_6$, the heights take values in
	$\{1,2,3,4,5,6\}$.
	
	Let $A_G$ be the adjacency matrix of $G$, and choose a real
	orthonormal eigenbasis:
	\begin{equation}\label{def:graph_psi}
		A_G\psi^a=2\cos\!\left(\frac{\pi a}{q}\right)\psi^a,
		\qquad
		\sum_{v\in G}\psi_v^a\psi_v^b=\delta_{ab},
		\qquad
		\psi_v^1>0\,.
	\end{equation}
	Here $a$ runs over the graph exponents, with an additional
	multiplicity label understood when necessary. Pasquier
	introduced the local observables
	\begin{equation}
		\Phi^a(x)
		=\frac{\psi_{v_x}^a}{\psi_{v_x}^1},
	\end{equation}
	where $v_x$ is the height at the lattice site $x$. Their
	pointwise products form the Pasquier algebra:
	\begin{equation}\label{pasquier_algebra}
		\Phi^a(x)\Phi^b(x)
		=\sum_c M^G_{abc}\Phi^c(x),
		\qquad
		M^G_{abc}
		=\sum_{v\in G}
		\frac{\psi_v^a\psi_v^b\psi_v^c}{\psi_v^1}\,.
	\end{equation}
	
	In the continuum limit, these observables have scalar CFT
	fields as their leading scaling contributions. Their
	two-point functions behave as
	\begin{equation}\label{graph_two_pt}
		\braket{\Phi^a(x)\Phi^b(y)}
		=\delta_{ab}
		\left(
		\frac{\mathcal{A}_a}{\abs{x-y}^{2\Delta_a}}
		+\text{subleading terms}
		\right),
	\end{equation}
	where $\Delta_a$ is the scaling dimension of the leading
	field. Using the transfer-matrix formalism, Petkova and Zuber
	showed that the lattice three-point functions factorize as
	\cite{Petkova:1994zs}
	\begin{equation}\label{eq:lattice_three_point}
		\braket{\Phi^a(x)\Phi^b(y)\Phi^c(z)}_G
		=M^G_{abc}\,F_{abc}(x,y,z).
	\end{equation}
	Here the lattice correlators use the common modified-trace and boundary
	prescription of Petkova and Zuber~\cite{Petkova:1994zs}. The graph dependence is contained in $M^G_{abc}$. The function
	$F_{abc}$ depends only on the Coxeter number $q$, the graph
	eigenvalues $\gamma^a=2\cos(\pi a/q)$ and their counterparts
	for $b,c$, the insertion sites $x,y,z$, the lattice dimensions
	$L_1,L_2$, and the spectral parameter $u$ that controls the
	Boltzmann weights and anisotropy.
	
	Thus, for matching labels and lattice parameters, $F_{abc}$
	is the same for $G$ and for the $A$ diagram with the same
	Coxeter number. Whenever the denominator is nonzero, this
	gives the exact lattice ratio
	\begin{equation}\label{graph_three_pt_ratio}
		\frac{
			\braket{\Phi^a(x)\Phi^b(y)\Phi^c(z)}_G
		}{
			\braket{\Phi^a(x)\Phi^b(y)\Phi^c(z)}_A
		}
		=\frac{M^G_{abc}}{M^A_{abc}}\,.
	\end{equation}
	Orthonormality of the graph eigenvectors implies
	$M^G_{ab1}=\delta_{ab}$. Applying the same factorization to two-point functions therefore shows that the operator normalizations are independent of the graph. Consequently,
	the right-hand side of \eqref{graph_three_pt_ratio} also gives
	the ratio of the corresponding continuum OPE coefficients
	when the fields have matching two-point normalizations.
	
	For the $A$ diagram,
	$M^A_{abc}=N_{abc}\in\{0,1\}$, where $N_{abc}$ are the
	$A$-graph fusion coefficients entering the diagonal
	minimal-model fusion rules. On fusion-allowed triplets,
	$N_{abc}=1$, and the scalar OPE coefficients satisfy
	\begin{equation}\label{eq:pasquier_scalar_ope}
		C^G_{abc}=C^A_{abc}\,M^G_{abc}\,.
	\end{equation}
	Here $C^A_{abc}$ is the corresponding diagonal minimal-model
	OPE coefficient. Fusion-forbidden triplets have vanishing
	OPE coefficients.
	
	Now consider the $E_8$ minimal model and take $a=b=c=17$.
	The leading continuum contribution to $\Phi^{17}$ is the
	scalar field with $r=s=17$.\footnote{In the lattice scaling expansion, $\Phi^a$ can contain
		scalar fields $V_{(a,s)(a,s)}$ with $s\equiv a\pmod2$.
		For $a<q-1$, the lightest contribution has $r=s=a$.
		These Kac labels are given before conversion to our odd-$s$ representatives. For even $a$, the corresponding representative is $(r,s)=(q-a,p-a)$.}
	An explicit calculation gives
	\begin{equation}
		M^{E_8}_{17,17,17}=0\,,
	\end{equation}
	and hence
	\begin{equation}
		C^{E_8}_{(17,17)(17,17)(17,17)}=0\,.
	\end{equation}
	In our conventions, $r=\bar r=17$ corresponds to
	$\ell=\bar\ell=8$. Since the reduced OPE coefficients are
	independent of $s$, this establishes
	\eqref{E8_extra_zero} in the unitary case. The extended
	Galois-conjugation argument in
	Section~\ref{subsec:galois} transports this zero to the
	corresponding nonunitary solutions.
	
	This also explains the zero
	$C_{13,13,13}^{\mathrm{even}}=0$ of
	Nivesvivat and Ribault~\cite{Nivesvivat:2025odb}.
	Their scalar $(13,s)$ is represented in our odd-$s$ convention by
	$(13,s)$ for odd $s$ and $(17,p-s)$ for even $s$.
	An even sum of the original second labels therefore gives one or
	three first labels equal to $17$; the corresponding Pasquier
	coefficients satisfy $M_{13,13,17}=M_{17,17,17}=0$.
	
	More generally, \eqref{eq:pasquier_scalar_ope} determines the scalar
	OPE coefficients with $r_i=s_i<q-1$. The $s$-independence of the
	reduced coefficients $\alpha_{IJK}$ then determines the scalar OPE
	coefficients with the corresponding first Kac labels after conversion
	to odd-$s$ representatives and arbitrary allowed
	second Kac labels. The remaining endpoint label $r=q-1$ is covered
	by the diagonal simple-current relations derived above.

	\paragraph{Ocneanu fusion algebra and the generalized Pasquier algebra.}
	
	The topological defects of rational CFT can be studied through twisted torus
	partition functions \cite{Petkova:2000ip,PetkovaZuber:2001Ocneanu,Coquereaux:2001di}.
	Their fusion defines the Ocneanu fusion algebra,
	\begin{equation}
		\mathcal L_\alpha\otimes\mathcal L_\beta
		\simeq\bigoplus_\gamma\widetilde N_{\alpha\beta}{}^\gamma
		\mathcal L_\gamma.
	\end{equation}
	Together with their topological junctions and associators, these lines form
	a fusion category describing invertible and noninvertible symmetries
	\cite{Chang:2018iay}. Below we explain how these symmetries constrain bulk OPE
	coefficients, following the argument of Petkova-Zuber \cite{PetkovaZuber:2001Ocneanu}.
	
	We specialize to the $E$ series, whose full Virasoro spectra have
	multiplicities $Z_{ij}\in\{0,1\}$. A Virasoro-preserving topological defect therefore
	acts by a scalar on each bulk representation:
	\begin{equation}
		\mathcal L_\alpha[V(x)]=\lambda_\alpha(V)V(x),
		\qquad d_\alpha:=\lambda_\alpha(\mathbf1).
	\end{equation}
	Here brackets denote a counterclockwise defect loop surrounding the
	insertion, and $\alpha$ labels a simple defect. We use the two-point
	normalization and OPE convention fixed above.
	
	Consider\footnote{The construction of Petkova and Zuber
		\cite{PetkovaZuber:2001Ocneanu} uses two topological lines. A single line suffices for the identity derived here.}
	\begin{equation}\label{eq:ocn-dressed-four-point}
		G_{ABBA;\alpha}
		:=\bigl\langle\mathcal L_\alpha[A(x_1)B(x_2)]
		B(x_3)A(x_4)\bigr\rangle.
	\end{equation}
	The geometry of operator and line insertions is shown in Fig.~\ref{fig:ocn-loop}. 
	\begin{figure}[ht]
		\centering
		\resizebox{.25\linewidth}{!}{%
			\begin{tikzpicture}[x=1cm,y=1cm,font=\small,
				line cap=round,line join=round,>=Stealth,
				defect/.style={line width=1.05pt},
				oriented/.style={postaction={decorate},decoration={markings,
						mark=at position .12 with {\arrow{Stealth}},
						mark=at position .62 with {\arrow{Stealth}}}}]
				\draw[defect,oriented] (2.2,.65) arc[start angle=0,end angle=360,x radius=2.2,y radius=1.2];
				\node at (0,2.1) {$\mathcal L_\alpha$};
				\foreach \x/\lab in {-1.25/A(x_1),1.25/B(x_2)} {
					\fill (\x,.65) circle[radius=1.7pt];
					\node[below=5pt] at (\x,.65) {$\lab$};
				}
				\foreach \x/\lab in {1.25/B(x_3),-1.25/A(x_4)} {
					\fill (\x,-1.05) circle[radius=1.7pt];
					\node[below=5pt] at (\x,-1.05) {$\lab$};
				}
			\end{tikzpicture}%
		}
		\caption{A topological defect loop surrounding $A(x_1)$ and $B(x_2)$.
			Arrows specify the orientation of $\mathcal L_\alpha$.}
		\label{fig:ocn-loop}
	\end{figure}
	
	The $s$-channel expansion of $G_{ABBA;\alpha}$ is
	\begin{equation}\label{eq:ocn-s-channel}
		G_{ABBA;\alpha}
		=\sum_C C_{ABC}C_{BAC}\lambda_\alpha(C)
		\mathcal F_{ABBA}^{(s);C},
	\end{equation}
	where $\mathcal F^{(s);C}$ denotes the product of the left- and
	right-moving Virasoro blocks, including the common coordinate factors. The $s$-channel OPE is available because $A(x_1)$ and $B(x_2)$ can approach each other without crossing the defect. In the $t$ channel, however, $B(x_2)$ and $B(x_3)$ lie on opposite sides of the defect.
	
	To access the $t$ channel, insert an identity defect chord separating
	$A(x_1)$ and $B(x_2)$ and perform an $F$-move. This expresses the correlator as
	\begin{equation}\label{eq:ocn-network-expansion}
		G_{ABBA;\alpha}
		=\sum_{\beta,m,n}\mathsf f_{\alpha;\beta mn}
		G^{\mathrm{lasso}}_{ABBA;\alpha,\beta,m,n}.
	\end{equation}
	The resulting network consists of two $\alpha$ loops connected by a
	$\beta$ line, with topological junctions $m,n$, as in
	Fig.~\ref{fig:ocn-fmove}. Junction spaces are understood with the displayed
	orientations; reversing a line replaces it by its dual. In particular,
	one must distinguish $\alpha$ and $\alpha^\vee$ unless self-duality has
	been established. Shrinking a lasso produces a defect-ending field.
	The subsequent OPE of this field with an ordinary bulk field takes values
	in the corresponding $\beta$-twisted sector.
	
		\begin{figure}[ht]
		\centering
		\resizebox{.7\linewidth}{!}{%
			\begin{tikzpicture}[x=1cm,y=1cm,font=\small,
				line cap=round,line join=round,>=Stealth,
				defect/.style={line width=1.0pt},
				invisible/.style={gray!65,dashed,line width=.6pt},
				oriented/.style={postaction={decorate},decoration={markings,
						mark=at position .17 with {\arrow{Stealth}},
						mark=at position .67 with {\arrow{Stealth}}}}]
				\begin{scope}[shift={(2.3,3.05)}]
					\draw[defect,oriented] (2.05,0) arc[start angle=0,end angle=360,x radius=2.05,y radius=.88];
					\draw[invisible] (0,-.88) -- (0,.88);
					\fill (0,-.88) circle[radius=1.2pt];
					\fill (0,.88) circle[radius=1.2pt];
					\node[fill=white,inner sep=2pt] at (0,0) {$\mathbf1$};
					\node at (-1.25,1) {$\alpha$};
					\foreach \x/\lab in {-1.15/A_1,1.15/B_2} {
						\fill (\x,0) circle[radius=1.5pt];
						\node[below=5pt] at (\x,0) {$\lab$};
					}
					\foreach \x/\lab in {-1.15/A_4,1.15/B_3} {
						\fill (\x,-1.46) circle[radius=1.5pt];
						\node[below=5pt] at (\x,-1.46) {$\lab$};
					}
				\end{scope}
				\node at (6,3.05) {$\displaystyle=\sum_{\beta,m,n}\mathsf f_{\alpha;\beta mn}$};
				\begin{scope}[shift={(9.6,3.05)}]
					\foreach \x/\lab in {-1.18/A_1,1.18/B_2} {
						\draw[defect,oriented] (\x+.82,0) arc[start angle=0,end angle=360,x radius=.82,y radius=.88];
						\fill (\x,0) circle[radius=1.5pt];
						\node[below=5pt] at (\x,0) {$\lab$};
						\node at (\x,1.10) {$\alpha$};
					}
					\draw[defect,->] (-.36,0) -- (.36,0);
					\node[above=4pt] at (0,0) {$\beta$};
					\fill (-.36,0) circle[radius=2pt];
					\fill (.36,0) circle[radius=2pt];
					\node at (-0.58,-0.13) {$m$};
					\node at (0.58,-.13) {$n$};
					\foreach \x/\lab in {-1.18/A_4,1.18/B_3} {
						\fill (\x,-1.46) circle[radius=1.5pt];
						\node[below=5pt] at (\x,-1.46) {$\lab$};
					}
				\end{scope}
			\end{tikzpicture}%
		}
		\caption{An identity chord and an $F$-move resolve the loop into two lassos joined by a simple defect $\beta$. The coefficients $\mathsf f_{\alpha;\beta mn}$ include the stated junction normalization.}
		\label{fig:ocn-fmove}
	\end{figure}
	
	In the $E$-series minimal model, the only simple defect that admits the Virasoro vacuum representation in its twisted sector is the identity defect $\mathbf1$. It means that only the identity defect contributes to the $t$-channel vacuum block. Furthermore, the topological junction is one-dimensional when $\beta=\mathbf1$, corresponding to the unit junction, so we can neglect the indices $m,n$ in the identity channel.
	
	\begin{figure}[ht]
		\centering
		\resizebox{.75\linewidth}{!}{%
			\begin{tikzpicture}[x=1cm,y=1cm,font=\small,
				line cap=round,line join=round,>=Stealth,
				defect/.style={line width=1.0pt},
				invisible/.style={gray!65,dashed,line width=.6pt},
				oriented/.style={postaction={decorate},decoration={markings,
						mark=at position .17 with {\arrow{Stealth}},
						mark=at position .67 with {\arrow{Stealth}}}}]
				\node at (.25,-1.10) {$\displaystyle\mathsf f_{\alpha;\id}$};
				\begin{scope}[shift={(2.75,-1.10)}]
					\foreach \x/\lab in {-1.12/A_1,1.12/B_2} {
						\draw[defect,oriented] (\x+.73,0) arc[start angle=0,end angle=360,x radius=.73,y radius=.70];
						\fill (\x,0) circle[radius=1.5pt];
						\node[below=5pt] at (\x,0) {$\lab$};
						\node at (\x,.94) {$\alpha$};
					}
					\foreach \x/\lab in {-1.12/A_4,1.12/B_3} {
						\fill (\x,-1.34) circle[radius=1.5pt];
						\node[below=5pt] at (\x,-1.34) {$\lab$};
					}
				\end{scope}
				\node[align=center] at (7.5,-1.2)
				{$\displaystyle
					=\quad\frac{\lambda_\alpha(A)\lambda_\alpha(B)}{d_\alpha}
					\,\braket{A_1B_2B_3A_4}$};
			\end{tikzpicture}%
		}
		\caption{The contribution from the $\beta=\mathbf1$ term.}
		\label{fig:lasso_vacuum}
	\end{figure}
	
	Fig.~\ref{fig:lasso_vacuum} demonstrates the contribution from $\beta=\mathbf1$. Taking $A=B=\id$ gives $d_\alpha=f_{\alpha;\mathbf1}d_\alpha^2$, so we get $f_{\alpha;\mathbf1}=d_\alpha^{-1}$. Hence the $t$-channel expansion of $G_{ABBA;\alpha}$ is
	\begin{equation}\label{eq:ocn-t-vacuum}
		G_{ABBA;\alpha}
		=\frac{\lambda_\alpha(A)\lambda_\alpha(B)}{d_\alpha}
		\mathcal F_{ABBA}^{(t);\id}
		+\text{non-vacuum blocks}.
	\end{equation}

	Applying the bulk fusion kernel to \eqref{eq:ocn-s-channel} and matching
	the vacuum block gives
	\begin{equation}\label{eq:ocn-ward-physical}
		\sum_C C_{ABC}C_{BAC}\lambda_\alpha(C)
		\mathcal K^{\mathrm{bulk}}_{C,\id}
		=\frac{\lambda_\alpha(A)\lambda_\alpha(B)}{d_\alpha}.
	\end{equation}
	Write $\ell_X=(r_X-1)/2$, $\bar\ell_X=(\bar r_X-1)/2$ and
	$j_X=(s_X-1)/2$. In our $6j$ convention, the full kernel is
	\begin{equation}\label{eq:ocn-kernel-full}
		\begin{split}
			\mathcal K^{\mathrm{bulk}}_{C,\mathbf1}
			={}&\frac{1}{C^{\mathrm{(ref)}}_{ABC}C^{\mathrm{(ref)}}_{BAC}}
			\sixj{\ell_A}{\ell_B}{\ell_C}{\ell_B}{\ell_A}{0}_{q_+}
			\sixj{\bar\ell_A}{\bar\ell_B}{\bar\ell_C}
			{\bar\ell_B}{\bar\ell_A}{0}_{q_+}\ \left(\sixj{j_A}{j_B}{j_C}{j_B}{j_A}{0}_{q_-}\right)^2.
		\end{split}
	\end{equation}
	The square is legitimate here because both $t$-channel representations are the vacuum and the $E$-series spectra obey $s_X=\bar s_X$. Using
	\begin{equation}\label{eq:ocn-vacuum-sixj}
		\sixj{a}{b}{c}{b}{a}{0}_{Q}
		=(-1)^{a+b-c}
		\sqrt{\frac{[2c+1]_Q}{[2a+1]_Q[2b+1]_Q}},
	\end{equation}
	with the correlated branches specified in the paper, define
	\begin{equation}\label{eq:ocn-rho}
		\rho_X:=(-1)^{\ell_X+\bar\ell_X}
		[r_X]_{q_+}^{1/2}
		[\bar{r}_X]_{q_+}^{1/2}[s_X]_{q_-}.
	\end{equation}
	The second Kac labels are odd in our representatives, so $j_X$ is integral.
	Consequently,
	\begin{equation}
		C^{\mathrm{(ref)}}_{ABC}C^{\mathrm{(ref)}}_{BAC}
		\mathcal K^{\mathrm{bulk}}_{C,\mathbf1}
		=\frac{\rho_C}{\rho_A\rho_B}.
	\end{equation}
	Eq.~\eqref{eq:ocn-ward-physical} therefore becomes
	\begin{equation}\label{eq:ocn-ward-reduced}
		\Psi_A(\alpha)\Psi_B(\alpha)
		=\sum_C\alpha_{ABC}\alpha_{BAC}\Psi_C(\alpha),
		\qquad
		\Psi_A(\alpha):=\rho_A\frac{\lambda_\alpha(A)}{d_\alpha}.
	\end{equation}
	
	For the complete set of Virasoro-preserving simple defects,
	$N_{\mathrm{defect}}=\sum_{i,j}Z_{ij}^2$. Since the $E$-series spectra are
	multiplicity-free, this equals the number of bulk primary fields,
	$\sum_{i,j}Z_{ij}$. Completeness of the defect-operator basis then makes the full matrix $\Psi_A(\alpha)$ invertible. The nonzero row and column rescalings in its definition preserve this property. It follows that
	\begin{equation}\label{eq:ocn-generalized-pasquier}
		\boxed{\alpha_{ABC}\alpha_{BAC}=\widetilde M_{AB}{}^C,
		\qquad
		\widetilde M_{AB}{}^C
		:=\sum_\alpha\Psi_A(\alpha)\Psi_B(\alpha)
		(\Psi^{-1})_\alpha{}^C.}
	\end{equation}
	The constants $\widetilde M$ define the generalized Pasquier algebra, dual to the Ocneanu defect fusion algebra. They are distinct from the defect fusion multiplicities $\widetilde N_{\alpha\beta}{}^\gamma$.
	
	We have thus recovered Petkova and Zuber's four-point-function
	analysis in our normalization. The argument does not require
	unitarity. For the $E$ series, computing $\widetilde M$ explicitly
	requires the defect eigenvalues $\lambda_\alpha(A)$ in
	\eqref{eq:ocn-ward-reduced}. We summarize the resulting factorization
	below and give the eigenvalues and derivation in
	Appendix~\ref{app:ocneanu-E}.
	
	$\widetilde M$ is factorized into the first- and second-Kac-label factors,
	\begin{equation}\label{eq:general_Pasquier_factorization}
		\begin{split}
			\widetilde{M}^{(A_{p-1},G)}_{(r_1,\bar{r}_1,s_1)(r_2,\bar{r}_2,s_2)}{}^{(r_3,\bar{r}_3,s_3)}=M^{G,p}_{(r_1,\bar{r}_1)(r_2,\bar{r}_2)}{}^{(r_3,\bar{r}_3)} N_{s_1,s_2}{}^{s_3}\,.
		\end{split}
	\end{equation}
	Here $N_{s_1,s_2}{}^{s_3}\in\{0,1\}$ is the standard $A$-type fusion coefficient. Note that this factorization is compatible with our lifting result \eqref{eq:lift}: $\alpha_{ABC}$ is independent of the second Kac labels in a proper field-sign gauge.
	
	For $G=E_6,E_8$, the first Kac factor $M^{G,p}$ is given by factorized Pasquier tensors multiplied by a $p$-dependent phase
	\begin{equation}\label{gen_pasquier_fact}
		\begin{split}
			M^{G,p}_{(r_1,\bar{r}_1)(r_2,\bar{r}_2)}{}^{(r_3,\bar{r}_3)}=\abs{M^{G}_{r_1,r_2,r_3}\,M^{G}_{\bar{r}_1,\bar{r}_2,\bar{r}_3}}\times\text{(phase)}.
		\end{split}
	\end{equation}
	Here $M^{G}_{r_1,r_2,r_3}$ is given in \eqref{pasquier_algebra}. In particular, since $M^{E_8}_{17,17,17}=0$, eq.~\ref{gen_pasquier_fact} implies \eqref{E8_extra_zero} for all the $E_8$ minimal models, including the nonunitary ones.
	
	For $G=E_7$, the magnitude of $M^{E_7,p}$ is not the same as factorized Pasquier tensors. See Appendix~\ref{app:ocneanu-E}, e.g. \eqref{eq:ocn-app-T-E7} for its explicit magnitude.
	
	\subsection{Galois conjugation}
	\label{subsec:galois}
	
	Galois symmetry of rational-CFT modular and fusion data has
	long been studied \cite{Coste:1993af}.
	Its action on chiral fusion and braiding data, including
	normalization issues and the use of normal coefficient fields,
	was studied in \cite{Degiovanni1994equations}.
	Further treatments and applications of Galois conjugation
	include
	\cite{Ardonne:2011wxx,Freedman:2011dp,Davidovich:2013arith,Harvey:2019qzs}.
	Here we apply this mechanism to the reduced bulk bootstrap
	equations and give explicit extensions to the fields containing
	our $E$-series candidates.
	
	For the $E$ series, the reduced bootstrap equations
	\eqref{eq:reduced-crossing_s=1} depend on
	\begin{equation}
		q_+=\zeta_N^p,\qquad \zeta_N=e^{2\pi i/N},\qquad N=2q,
		\qquad
		q=\begin{cases}
			12,&E_6,\\
			18,&E_7,\\
			30,&E_8,
		\end{cases}
	\end{equation}
	where $\gcd(p,N)=1$. 
	
	For each fixed $q$, an algebraic solution of the $s=1$ reduced
	bootstrap equations at one admissible $p$ can be transported to every other admissible $p$ by an extended Galois automorphism. The corresponding transported solutions have the same zero/nonzero pattern.
	
	Here, we use the formal reference point $p=1$
	as the source of the Galois action; it is not a physical minimal
	model.
	
	\paragraph{Rationalized bootstrap equations.}
	To expose the cyclotomic dependence of the equations, define,
	for triplets allowed by both chiral fusion rules,
	\begin{equation}
		\begin{aligned}
			\widetilde\alpha_{IJK}&:=d_{IJK}\alpha_{IJK},\\
			d_{IJK}&:=\Delta(\ell_I,\ell_J,\ell_K)
			\Delta(\bar\ell_I,\bar\ell_J,\bar\ell_K)\\
			&\quad\times
			\left([2\ell_I+1]_{q_+}[2\bar\ell_I+1]_{q_+}
			[2\ell_J+1]_{q_+}[2\bar\ell_J+1]_{q_+}
			[2\ell_K+1]_{q_+}[2\bar\ell_K+1]_{q_+}\right)^{1/4}.
		\end{aligned}
		\label{eq:galois-rescaled-alpha}
	\end{equation}
	Fusion-forbidden coefficients are set to zero throughout.
	All fractional powers follow the logarithmic prescription
	\eqref{eq:qnumber_log_interior}. The crossing equations then take
	the schematic form
	\begin{equation}
		\widetilde\alpha\,\widetilde\alpha\,
		\widetilde F\,\widetilde F
		=\widetilde\alpha\,\widetilde\alpha,
		\label{eq:reduced_crossing_rewriting}
	\end{equation}
	with the rescaled $6j$ symbols
	\begin{equation}
		\widetilde F
		=F\,
		\frac{\Delta(\ell_2,\ell_3,\ell_6)
			\Delta(\ell_1,\ell_4,\ell_6)}
		{\Delta(\ell_1,\ell_2,\ell_5)
			\Delta(\ell_3,\ell_4,\ell_5)}
		\left(\frac{[2\ell_6+1]_{q_+}}
		{[2\ell_5+1]_{q_+}}\right)^{1/2}.
		\label{eq:galois-rescaled-F}
	\end{equation}
	Inserting the Racah formula \eqref{def:sixj}, the square roots
	cancel or combine into integer powers. Thus $\widetilde F$ is a
	rational function of quantum numbers with rational coefficients.
	
	Let $K_N=\mathbb Q(\zeta_N)$. For each admissible $p$, the
	cyclotomic automorphism
	\begin{equation}
		\sigma_p:K_N\longrightarrow K_N,\qquad
		\sigma_p(\zeta_N)=\zeta_N^p
		\label{eq:galois-base-map-p}
	\end{equation}
	maps the rationalized equations at $p=1$ to those at $p$.
	The coefficients $\widetilde\alpha^{(1)}_{IJK}$ need not belong to
	$K_N$, so the action must also be extended to them. For an algebraic
	solution, form the coefficient field
	\begin{equation}
		E=K_N\bigl(\{\widetilde\alpha^{(1)}_{IJK}\}\bigr)
	\end{equation}
	and let $L$ be its normal closure over $\mathbb Q$.
	This passage to a normal closure is also used in the construction
	of Galois twists of algebraic fusion and modular data
	\cite{Davidovich:2013arith}; see also \cite{Buican:2021axn}
	for a physics-oriented discussion. Since $L/\mathbb Q$ is normal and, in characteristic zero, separable, it is Galois. Applying \cite[Proposition~7.4]{milne2022} to $\mathbb Q\subset K_N\subset L$, we obtain an automorphism $\widehat{\sigma}_p$ of $L$ whose restriction to $K_N$ is $\sigma_p$.\footnote{In Milne's notation, $(\Omega,F,E)=(L,\mathbb Q,K_N)$, and the $F$-homomorphism $E\to\Omega$ is $\sigma_p:K_N\to K_N\subset L$. Milne's $E$ denotes our cyclotomic field $K_N$, not the coefficient field $E=K_N(\{\widetilde{\alpha}_{IJK}\})$ introduced above.} The extension need not preserve the smaller field $E$ itself. The transported coefficients
	\begin{equation}
		\widetilde\alpha^{\sigma_p}_{IJK}
		:=\widehat\sigma_p\bigl(\widetilde\alpha^{(1)}_{IJK}\bigr)
		\label{eq:galois-full-transport}
	\end{equation}
	satisfy the rationalized equations at $q_+=\zeta_N^p$.
	Once the extension is fixed, it acts simultaneously on all
	coefficients and preserves their algebraic relations and
	zero/nonzero pattern.
	
	\paragraph{Explicit extensions.}
	For the algebraic candidates obtained in this work, normal fields
	containing the rescaled source coefficients are
	\begin{equation}
		\begin{aligned}
			L_6&=K_{24}\bigl(3^{1/4}\bigr),\\
			L_7&=K_{36}\bigl(\sqrt2,\sqrt{[3]},\sqrt{[5]}\bigr),\\
			L_8&=K_{60}\bigl(\sqrt{[7]},\sqrt{[7]+1},
			\sqrt{[11]+1},\sqrt{[3]-2}\bigr).
		\end{aligned}
		\label{eq:galois-normal-fields}
	\end{equation}
	Here $[n]=[n]_{\zeta_N}$, and the defining radicals are positive
	at the source $p=1$. The defining radicands have independent
	square classes over the respective cyclotomic fields, giving extension degrees
	$2$, $8$, and $16$.
	The following tables specify one lift of each $\sigma_p$.
	Every expression on the right-hand side is a source element;
	the action on the cyclotomic field is always
	\eqref{eq:galois-base-map-p}, including
	$\widehat\sigma_p(i)=i^p$.
	
	For $E_6$, define $u_6=3^{1/4}$. Since
	$u_6^2=\zeta_{24}^2+\zeta_{24}^{-2}=\sqrt3$ and $\sigma_{\pm5}(\sqrt{3})=-\sqrt{3}$, we may choose
	\begin{equation}
		\begin{array}{c|c}
			p\ \text{mod 12}&\widehat\sigma_p(u_6)\\\hline
			\pm1&u_6\\
			\pm5&i u_6 \\
		\end{array}
		\label{eq:galois-E6-explicit-p}
	\end{equation}
	
	For $E_7$, set
	\begin{equation}
		u_7=\sqrt2,\qquad v_7=\sqrt{[3]},\qquad w_7=\sqrt{[5]}.
	\end{equation}
	A choice of images is
	\begin{equation}
		\renewcommand{\arraystretch}{1.5}
		\begin{array}{c|ccc}
			p\ \text{mod 18}&\widehat\sigma_p(u_7)&\widehat\sigma_p(v_7)
			&\widehat\sigma_p(w_7)\\\hline
			\pm1&u_7&v_7&w_7\\
			\pm5&u_7&v_7/w_7&-i v_7^3/w_7^2\\
			\pm7&u_7&-i w_7/v_7^2&-i w_7/v_7^3
		\end{array}
		\label{eq:galois-E7-explicit-p}
	\end{equation}
	
	For $E_8$, define
	\begin{equation}
		u=\sqrt{[7]},\qquad v=\sqrt{[7]+1},\qquad
		w=\sqrt{[11]+1},\qquad t=\sqrt{[3]-2},
		\label{eq:galois-E8-generators}
	\end{equation}
	and introduce the cyclotomic elements
	\begin{equation}
		\begin{aligned}
			x&=\zeta_{60}^2+\zeta_{60}^{-2},\qquad
			\varphi=x^3-3x=\frac{1+\sqrt5}{2},\\
			\gamma&=\zeta_{60}^5+\zeta_{60}^{-5}=\sqrt3.
		\end{aligned}
	\end{equation}
	For compactness, we omit the subscript $8$ in the following table.
	
	Then an explicit choice of the extension is
	\begin{equation}
		\renewcommand{\arraystretch}{1.8}
		\begin{array}{c|cccc}
			p\ \text{mod 30}&\widehat\sigma_p(u)&\widehat\sigma_p(v)
			&\widehat\sigma_p(w)&\widehat\sigma_p(t)\\\hline
			\pm1&u&v&w&t\\
			\pm7&\dfrac{i(x+1)}{u}
			&i(\varphi-1)t
			&\dfrac{x^2-x-1}{\gamma}w
			&i(2-x)(x+1)uv\\[3pt]
			\pm11&\dfrac{\varphi^2}{u}
			&(2-x)(x^2+x-1)uv
			&x(4-x^2)w
			&i(\varphi-1)ut\\[3pt]
			\pm13&\dfrac{i(x^2-3)}u
			&(2-x)(x+1)^2ut
			&\dfrac{(2-x)(2x^2+4x+1)}{\gamma}w
			&i(\varphi-1)v
		\end{array}
		\label{eq:galois-E8-explicit-p}
	\end{equation}
	In this case, the source coefficient field is $E=K_{60}(u,v,w)$, of degree
	$8$ over $K_{60}$. Its failure to be normal is explicit in
	\begin{equation}
		\widehat\sigma_7(v)=i(\varphi-1)t\notin E.
	\end{equation}
	Adjoining $t$ supplies the missing conjugates and gives the normal
	closure $L_8$.
	
	The radical images in the tables are grouped by $p$ modulo $q$,
	whereas the cyclotomic action is determined by $p$ modulo $2q$.
	The assignments preserve the algebraic relations among the generators
	and therefore define extensions of $\sigma_p$. They give concrete
	maps in \eqref{eq:galois-full-transport}. These are choices of lifts
	for individual $p$; no composition law
	$\widehat\sigma_p\widehat\sigma_{p'}=\widehat\sigma_{pp'}$
	is assumed.
	
	\paragraph{Transport in the original normalization.}
	Restoring the normalization of the reduced OPE coefficients gives
	\begin{equation}
		\alpha^{\mathrm{tr}}_{IJK}(\zeta_N^p)
		:=\frac{\widehat\sigma_p
			\bigl(\widetilde\alpha^{(1)}_{IJK}\bigr)}
		{d_{IJK}(\zeta_N^p)}.
		\label{eq:galois-original-normalization}
	\end{equation}
	The denominator is evaluated at the target parameter with the
	prescribed lifted logarithms. This step is essential: the fields
	\eqref{eq:galois-normal-fields} contain the rescaled coefficients
	but need not contain $\alpha_{IJK}$ and $d_{IJK}$ separately.
	Thus \eqref{eq:galois-original-normalization} does not require a
	direct Galois action on the original coefficients.
	
	For each family, applying these extensions to the $p=1$ candidate
	and restoring the target normalization reproduces the coefficients
	of the other reference cases, up to primary-field signs:
	\begin{equation}
		\alpha^{(p)}_{IJK}
		=\epsilon_I\epsilon_J\epsilon_K\,
		\alpha^{\mathrm{tr}}_{IJK}(\zeta_N^p),\qquad
		\epsilon_I\in\{\pm1\}.
	\end{equation}
	We checked this agreement independently against the seed
	data.\footnote{The explicit Galois extensions and the code
		implementing their action were generated by OpenAI's GPT-6 Astra.
		The authors independently verified the transported coefficients
		against their bootstrap results. The field calculations and
		primary-field sign comparisons are included in the ancillary files.}
	
	Galois conjugation does not preserve absolute values in general.
	For the present candidates, however, each fusion-allowed triplet
	satisfies
	\begin{equation}
		c_{IJK}:=\bigl(\alpha^{(1)}_{IJK}\bigr)^4
		\in\mathbb Q_{\geq0}.
	\end{equation}
	Moreover, the fourth power of the normalization is cyclotomic:
	$R_{IJK}:=d_{IJK}(\zeta_N)^4\in K_N$, with
	$d_{IJK}(\zeta_N^p)^4=\sigma_p(R_{IJK})$.
	It follows that
	\begin{equation}
		\begin{aligned}
			\bigl(\alpha^{\mathrm{tr}}_{IJK}(\zeta_N^p)\bigr)^4
			&=\frac{\widehat\sigma_p
				\bigl((\widetilde\alpha^{(1)}_{IJK})^4\bigr)}
			{d_{IJK}(\zeta_N^p)^4}\\
			&=\frac{\sigma_p(R_{IJK})c_{IJK}}
			{\sigma_p(R_{IJK})}
			=c_{IJK}.
		\end{aligned}
		\label{eq:galois-fourth-powers}
	\end{equation}
	Hence these transported candidates have $p$-independent absolute
	values, irrespective of the primary-field sign convention.
	
	More generally, an extension of $\sigma_p$ to
	$\overline{\mathbb Q}$ gives a bijection between the normalized
	algebraic solution spaces at $q_+=\zeta_N$ and
	$q_+=\zeta_N^p$, compatible with primary-field signs.
	An exact existence and uniqueness result for one seed in each
	family would therefore extend to every admissible $p$.
	The explicit field calculations above concern our algebraic
	candidates and do not establish exact seed crossing or uniqueness.
	The extension to all Kac labels is a separate step, proved in
	Section~\ref{sec:global_solution}.

	\subsection{Outlook}\label{subsec:outlook}
	
	The bulk OPE data obtained here provide a starting point for exploring
	both the structure of the exceptional theories and their dynamics away
	from criticality. We conclude with several directions for further study.

    \paragraph{Exact verification.} An important next step is the exact verification of the seed solutions. Our seed coefficients are algebraic expressions reconstructed from numerical solutions and tested through exhaustive high-precision crossing checks. Evaluating the reduced fusion matrices in suitable algebraic number fields, with consistent choices of algebraic branches, would allow every crossing equation to be verified as an exact identity. This would establish the seed solutions algebraically and provide exact input for our extension arguments to arbitrary Kac labels and allowed central charges.
    
    \paragraph{Arithmetic structure.}
    In our normalization, the tabulated reduced coefficients satisfy
    $\alpha_{IJK}^{4}\in\mathbb Q$. Together with the Galois transport
    of Section~\ref{subsec:galois}, this implies that their magnitudes
    are independent of $p$. It would be interesting to derive this
    restricted arithmetic form directly from the bootstrap equations
    and understand its origin.

	\paragraph{Sewing consistency.}
	Our analysis has concentrated on sphere correlators. Matching our coefficients to the TFT construction of Fuchs, Runkel, and Schweigert \cite{Fuchs:2002cm,Fuchs:2004xi,Fjelstad:2005} would relate these explicit data to the sewing constraints for full correlators, with chiral consistency governed by the Moore-Seiberg relations \cite{Moore:1988qv}. Such a comparison would provide an independent check, while the tables offer directly usable input for physics applications.
	
	\paragraph{Boundary CFTs.}
	The bulk data can also serve as a starting point for exploring the
	boundary physics of the exceptional models. A natural goal is to
	determine boundary OPE coefficients and bulk--boundary couplings for
	all elementary Virasoro-preserving boundary conditions. Using
	the known ADE boundary states and the cell/TFT constructions
	\cite{Behrend:1999bn,PetkovaZuber:2001Ocneanu,Fuchs:2004xi},
	one could derive uniform formulas compatible with our bulk conventions. These would allow
	systematic studies of boundary correlators and boundary RG flows
	throughout the unitary and nonunitary $E$ series.

    \paragraph{Free-field realizations.} Free-field methods have been used to determine OPE coefficients of the unitary $(A_{10},E_6)$ minimal model and exceptional $SU(2)$ WZW theories \cite{Furlan:1989ra,Fuchs:1989kz,Fuchs:1989zp}; the $E_6$ WZW theory also admits a free-fermion description \cite{Goddard:1985jp,Bouwknegt:1986gc}. A complementary goal is a uniform realization of the full $E$-series minimal-model operator algebras in screened free fields, including the nonunitary cases. The diagonal pairing construction of Ref.~\cite{Gabai:2024qum} and generalized orbifolds \cite{Frohlich:2009gb} suggest possible starting points. Such a construction could derive the OPE coefficients directly and explain the exceptional left--right pairings and selection rules at the operator level.

    \paragraph{Minimal strings.} An  application of the $E$-series minimal models is to minimal string theory, whose continuum description combines a minimal model with Liouville theory and conformal ghosts \cite{Seiberg:2004at}. ADE models coupled to gravity have also been studied through random-surface and matrix constructions \cite{Kostov:1989eg,Kostov:1995xw}. More recently, Ref.~\cite{Rodriguez:2025rte} computed continuum string amplitudes, with its exceptional calculations mostly restricted to $E_6$ by the available explicit OPE data. Our results supply the matter-sector input needed to extend these calculations to $E_7$ and $E_8$, including nonunitary models. Computing sphere four-point and torus one-point amplitudes would provide concrete observables for comparison with proposed matrix descriptions of exceptional minimal strings.
	
	\paragraph{RG flows and Hamiltonian truncation.}
	An important motivation for computing OPE coefficients is to use the
	conformal fixed point as a starting point for studying interacting
	theories away from criticality. Our results provide input for conformal
	perturbation theory and Hamiltonian truncation
	\cite{Yurov:1989yu,Hogervorst:2014rta,Rychkov:2014eea,
		James:2017cpc,Anand:2020gnn}.
	They could enable quantitative tests of proposed exceptional RG flows
	\cite{Ravanini:1991pu,Benedetti:2024utz} and the truncation applications
	suggested in Refs.~\cite{Nakayama:2022svf,Katsevich:2024sov},
	as well as explorations of more general perturbations and their
	infrared phases.

	\section*{Acknowledgments}
	
	We thank Mikhail Kapranov, Yu Nakayama, and Mengyang Zhang for useful
	discussions, and Robert Coquereaux, Yu Nakayama, and Jean-Bernard Zuber
	for helpful feedback on the draft. We are especially grateful to Sylvain Ribault for his many suggestions for improving the manuscript. JQ also thanks Barak Gabai, Victor Gorbenko, Bernardo Zan, and Aleksandr Zhabin for earlier discussions during his time in Lausanne. The work of JQ and MY was supported by World Premier
	International Research Center Initiative (WPI), MEXT, Japan, and by the Center
	for Data-Driven Discovery, Kavli IPMU (WPI). MY was also supported in part by the JSPS KAKENHI Grant No.~23K25865; by JST, Japan (CREST Grant No.~JPMJCR26XA, Moonshot R\&D Grant No.~JPMJMS256E); and by the IBM-UTokyo-sponsored research.
	The mathematical derivations and the writing of the manuscript were carried out primarily by the authors. OpenAI’s ChatGPT, including GPT-6 Astra, GPT-5.6 Sol and other models, was used under their direction to assist with Mathematica and Python code development, simplification of mathematical arguments, summarizing empirical formulas, language editing, and literature searches.

	\appendix
	
	\section{\texorpdfstring{Fusion matrix, \(6j\) symbol, and chiral minimal-model OPE coefficients}{Fusion matrix, 6j symbol, and chiral minimal-model OPE coefficients}}\label{app:fusion_kernel}
	
	The main purpose of this appendix is to fix our notation and
	conventions for Virasoro fusion matrices, quantum-group $6j$ symbols,
	and chiral OPE coefficients. For detailed treatments of minimal-model
	fusion matrices and their quantum-group description, see
	Refs.~\cite{Felder:1989wv,Furlan:1989ra,Hou:1990gq,Teschner:1995}.
	Subtleties of rational specialization and fusion-rule truncation
	are discussed in
	Refs.~\cite{dotsenko1988lectures,GanchevPetkova:1989}.
	
	Related formulas for degenerate Virasoro fields at generic central
	charge, including normalization-dependent chiral coupling factors,
	were derived and further developed in the Liouville approach
	\cite{Cremmer:1994kc,Cremmer:1994}.
	See also Ref.~\cite{Gabai:2024qum} for a recent Coulomb-gas treatment
	and Ref.~\cite{Koshida:2021eta} for a rigorous treatment of the
	first-row Virasoro subcategory at generic central charge.
	
	Consider the general four-point function
	\begin{equation}
		\begin{split}
			\braket{V_{(r_1,s_1),(\bar{r}_1,\bar{s}_1)}(z_1,\bar{z}_1)V_{(r_2,s_2),(\bar{r}_2,\bar{s}_2)}(z_2,\bar{z}_2)V_{(r_3,s_3),(\bar{r}_3,\bar{s}_3)}(z_3,\bar{z}_3)V_{(r_4,s_4),(\bar{r}_4,\bar{s}_4)}(z_4,\bar{z}_4)}\,,
		\end{split}
	\end{equation}
	whose crossing equation is
	\begin{equation}\label{eq:crossing_app}
		\begin{split}
			&\sum_{5}C_{125}C_{345}\mathcal{F}_{(r_1,s_1),(r_2,s_2),(r_3,s_3),(r_4,s_4)}^{(s),(r_5,s_5)}(z)\mathcal{F}_{(\bar{r}_1,\bar{s}_1),(\bar{r}_2,\bar{s}_2),(\bar{r}_3,\bar{s}_3),(\bar{r}_4,\bar{s}_4)}^{(s),(\bar{r}_5,\bar{s}_5)}(\bar{z}) \\
			&=\sum_{6}C_{236}C_{416}\mathcal{F}_{(r_1,s_1),(r_2,s_2),(r_3,s_3),(r_4,s_4)}^{(t),(r_6,s_6)}(z)\mathcal{F}_{(\bar{r}_1,\bar{s}_1),(\bar{r}_2,\bar{s}_2),(\bar{r}_3,\bar{s}_3),(\bar{r}_4,\bar{s}_4)}^{(t),(\bar{r}_6,\bar{s}_6)}(\bar{z})\,.
		\end{split}
	\end{equation}
	
	To solve Eq.~\eqref{eq:crossing_app}, recall the fusion
	kernel of the chiral Virasoro blocks
	\cite{Moore:1988qv,Felder:1989wv,Hou:1990gq,Furlan:1989ra,Koshida:2021eta}:
	\begin{equation}
		\begin{split}
			\mathcal{F}_{(r_1,s_1),(r_2,s_2),(r_3,s_3),(r_4,s_4)}^{(s),(r_5,s_5)}(z)=\sum_{6}\mathcal{K}_{(r_5,s_5),(r_6,s_6)}\,\mathcal{F}_{(r_1,s_1),(r_2,s_2),(r_3,s_3),(r_4,s_4)}^{(t),(r_6,s_6)}(z)\,.
		\end{split}
	\end{equation}
	Substitution into \eqref{eq:crossing_app} gives
	\begin{equation}\label{eq:crossing2}
		\begin{split}
			&\sum_{5}C_{125}C_{345}\mathcal{K}_{(r_5,s_5),(r_6,s_6)}\mathcal{K}_{(\bar{r}_5,\bar{s}_5),(\bar{r}_6,\bar{s}_6)} \\
			&=\begin{cases}
				C_{236}C_{416}\,, & \text{if}\ (r_6,s_6;\bar{r}_6,\bar{s}_6)\ \text{is in the spectrum}, \\
				0\,, & \text{otherwise}.
			\end{cases}
		\end{split}
	\end{equation}
	
	The fusion matrix is related to the $6j$ symbols of the quantum group by
	\begin{equation}
		\begin{aligned}
			\mathcal{K}_{5,6}
			\equiv\mathcal{K}_{(r_5,s_5),(r_6,s_6)}
			&=\frac{D_{(r_2,s_2)(r_3,s_3)(r_6,s_6)}
				D_{(r_4,s_4)(r_1,s_1)(r_6,s_6)}}
			{D_{(r_1,s_1)(r_2,s_2)(r_5,s_5)}
				D_{(r_3,s_3)(r_4,s_4)(r_5,s_5)}}\\
			&\quad\times
			\sixj{\ell^+_1}{\ell^+_2}{\ell^+_5}
			{\ell^+_3}{\ell^+_4}{\ell^+_6}_{q_+}
			\sixj{\ell^-_1}{\ell^-_2}{\ell^-_5}
			{\ell^-_3}{\ell^-_4}{\ell^-_6}_{q_-}\,.
		\end{aligned}
	\end{equation}
	Here $\ell^{\pm}$ are related to the Kac indices $(r,s)$ by
	\begin{equation}
		\begin{split}
			r=2\ell^++1\,,\quad s=2\ell^-+1\,.
		\end{split}
	\end{equation}
	For a Virasoro minimal model $\mathcal{M}(p,q)$, the quantum-group parameters
	$q_\pm$ are
	\begin{equation}
		\begin{split}
			q_+=e^{i\pi p/q}\,,\quad q_-=e^{i\pi q/p}\,.
		\end{split}
	\end{equation}
	We now define the remaining quantities.
	
	\paragraph{$6j$ symbol.} 
	
	To make Eq.~\eqref{def:sixj} unambiguous, we specify the logarithms used
	in every fractional power. The two quantum-group parameters are coupled
	through a single complex variable,
	\begin{equation}\label{eq:coupled_continuation}
		t:=\alpha_+^2,\qquad q_+=e^{i\pi t},\qquad
		q_-=e^{i\pi/t},\qquad \operatorname{Im}t>0.
	\end{equation}
	Thus $|q_+|<1$ and $|q_-|>1$. Physical minimal-model values are the
	boundary values $t\to p/q+i0$. In particular, there is no restriction on
	the sign of $\operatorname{Im}q_-$.

	For positive real $q$, every $q$-number $[n]_q$ with $n\geq1$ is positive.
	We fix the initial germ by its real logarithm on $0<q<1$ with
	$\arg q=0$, so that
	\begin{equation}\label{convention:qnumber}
		\left([n]_q\right)^\gamma>0,\qquad
		0<q<1,\quad \arg q=0,\quad
		n\in\mathbb{Z}_+,\quad\gamma\in\mathbb{R}.
	\end{equation}
	On the universal cover of the punctured open unit disk, write
	$q=e^{i\theta}$ with $\operatorname{Im}\theta>0$, and define
	\begin{equation}\label{eq:qnumber_log_interior}
		\begin{split}
			L_n^{\mathrm{in}}(\theta)
			&=-i(n-1)\theta+\operatorname{Log}(1-e^{2in\theta})
			-\operatorname{Log}(1-e^{2i\theta}),\\
			([n]_q)^\gamma&=\exp\bigl(\gamma L_n^{\mathrm{in}}(\theta)\bigr).
		\end{split}
	\end{equation}
	Here $\operatorname{Log}(1-z)$ is the analytic logarithm on $|z|<1$
	that vanishes at $z=0$. The variable $\theta$ is a lifted logarithm of
	$q$; it is not reduced modulo $2\pi$. Likewise,
	\begin{equation}
		\left([n]_q!\right)^\gamma
		=\exp\left(\gamma\sum_{k=1}^nL_k^{\mathrm{in}}(\theta)\right)
		=\frac{e^{-i\gamma n(n-1)\theta/2}}{(1-q^2)^{n\gamma}}
		\prod_{k=1}^n(1-q^{2k})^\gamma.
	\end{equation}
	These expressions are single-valued on the chosen universal cover.
	All products and ratios under fractional powers in Eq.~\eqref{def:sixj}
	are interpreted by adding and subtracting these logarithms. Boundary
	values are taken only after combining the factors in the admissible
	root-of-unity expression.

	For $q_+$, we use $L_n^{\mathrm{in}}(\pi t)$. To define the $q_-$
	factors, continue the same germ through the lift $\theta=\pi$ of $q=-1$.
	An explicit logarithm on the exterior cover is
	\begin{equation}\label{eq:qnumber_log_exterior}
		\begin{split}
			L_n^{\mathrm{out}}(\theta)
			&=i(n-1)(\theta-2\pi)+\operatorname{Log}(1-e^{-2in\theta})-\operatorname{Log}(1-e^{-2i\theta}),
			\qquad\operatorname{Im}\theta<0.
		\end{split}
	\end{equation}
	The matching follows from
	\begin{equation}
		\lim_{\epsilon\searrow\,0}L_n^{\mathrm{in}}(\pi+i\epsilon)
		=\lim_{\epsilon\searrow\,0}L_n^{\mathrm{out}}(\pi-i\epsilon)
		=\log n-i\pi(n-1).
	\end{equation}
	Since $[n]_{-1}=(-1)^{n-1}n\neq0$, these are restrictions of the same
	local analytic logarithm near this lift of $q=-1$. We use
	$L_n^{\mathrm{out}}(\pi/t)$ for $q_-$, with the same additive prescription
	for factorials and other fractional powers. This specifies both sets
	of $6j$ symbols along the coupled continuation
	\eqref{eq:coupled_continuation}. As $t$ approaches its physical value,
	$q_+$ and $q_-$ approach the unit circle from the interior and exterior,
	respectively; see Fig.~\ref{fig:q-number-continuation}.

	\begin{figure}[ht]
		\centering
		\begin{tikzpicture}[
			x=2.65cm, y=2.65cm,
			line cap=round, line join=round,
			>=latex,
			axis/.style={black!70, line width=0.45pt, ->},
			pathplus/.style={orange!85!black, line width=0.9pt, ->},
			pathminus/.style={green!45!black, line width=0.9pt, ->}
			]
			
			\fill[orange!18, even odd rule]
			(0,0) circle[radius=1]
			(0,0) circle[radius=0.96];
			\fill[green!14]
			(0:1.04) arc (0:180:1.04)
			-- (180:1) arc (180:0:1) -- cycle;
			
			\draw[black!80, line width=0.5pt]
			(0,0) circle[radius=1];
			
			\draw[axis] (0,0) -- (1.52,0)
			node[right=3pt, black] {$\mathrm{Re}\,q$};
			\draw[axis] (0,-1.15) -- (0,1.28)
			node[above=3pt, black] {$\mathrm{Im}\,q$};
			
			\draw[red!70!black, line width=0.45pt, line join=miter]
			(-1.60,0)
			\foreach \k in {0,...,19} {
				-- ({-1.58+0.08*\k}, 0.015)
				-- ({-1.54+0.08*\k},-0.015)
			}
			-- (0,0);
			
			\draw[line width=1.15pt] (0,0) -- (1,0);
			
			\draw[pathplus]
			(0.61,0)
			.. controls (0.55,0.38) and (0.13,0.64)
			.. (-0.15,0.49)
			.. controls (-0.40,0.355) and (-0.54,-0.09)
			.. (-0.36,-0.28)
			.. controls (-0.18,-0.47) and (0.27,-0.49)
			.. (0.43,-0.77);
			\node[orange!85!black] at (0.09,-0.65) {$q_+$};
			
			\draw[pathminus]
			(0.79,0)
			.. controls (0.68,0.51) and (0.19,0.80)
			.. (-0.18,0.67)
			.. controls (-0.61,0.52) and (-0.64,0)
			.. (-1,0)
			.. controls (-1.30,0) and (-1.39,0.55)
			.. (-1.17,0.88)
			.. controls (-0.95,1.21) and (-0.64,1.29)
			.. (-0.42,1.01);
			\node[green!45!black, above left=1pt]
			at (-1.17,0.96) {$q_-$};
			
			\fill (0,0) circle[radius=0.9pt];
			\fill (1,0) circle[radius=1.0pt];
			
			\node[below left=3pt] at (-1,0) {$-1$};
			\node[below right=3pt] at (1,0) {$1$};
			\node[anchor=north, inner sep=2pt] at (0.54,-0.07)
			{$[n]_q>0$};
			
		\end{tikzpicture}
		\caption{Analytic continuation of the $q$-numbers
			for $q_+$ (orange) and $q_-$ (green). For $n>1$, we have $\log([n]_q)>0$ when $q>0$ and $\arg(q)=0$. The red zigzag line is the branch cut.}
		\label{fig:q-number-continuation}
	\end{figure}

	\paragraph{OPE coefficients.} Let
	$D_{ijk}\equiv D_{(r_i,s_i),(r_j,s_j),(r_k,s_k)}$ denote the OPE coefficient
	of the chiral generalized minimal model. In our conventions, its explicit form
	is \cite{Gabai:2024qum}
	\begin{equation}\label{ope_chiral_general:final}
		\begin{split}
			&D_{(r_1,s_1)(r_2,s_2)(r_3,s_3)} 
			=D_{(r_1,1)(r_2,1)(r_3,1)}\,D_{(1,s_1)(1,s_2)(1,s_3)} \sqrt{\frac{Y(0,2\Bell_1)\,Y(0,2\Bell_2)\,Y(0,2\Bell_3)}
				{Y(1,2\Bell_1)\,Y(1,2\Bell_2)\,Y(1,2\Bell_3)}} \\
			&\times\frac{e^{i\pi(\ell_1^+\ell_1^-+\ell_2^+\ell_2^--\ell_3^+\ell_3^- -2\ell_1^+\ell_2^-)}\,(-1)^{(\ell_1^++\ell_3^+-\ell_2^+)(\ell_2^-+\ell_3^--\ell_1^-)}\,Y(1,2\Bell_1)\,Y(1,2\Bell_2)\,Y(1,2\Bell_3)}
			{Y(0,\Bell_1-\Bell_2+\Bell_3)\,
				Y(0,\Bell_2-\Bell_1+\Bell_3)\,Y(0,\Bell_1+\Bell_2-\Bell_3)\,
				Y(1,\Bell_1+\Bell_2+\Bell_3)}\,,
		\end{split}
	\end{equation}
	where
	$\Bell_i=(\ell_i^+,\ell_i^-)\equiv\left(\tfrac{r_i-1}{2},\tfrac{s_i-1}{2}\right)$,
	and the function $Y$ is defined by
	\begin{equation}\label{def:Y}
		\begin{split}
			Y(n,\Bell)&:=\prod_{k^+=1+n}^{\ell^++n}\,\prod_{k^-=1+n}^{\ell^-+n} \left(k^+\alpha_++k^-\alpha_-\right)\,, \\
			\alpha_+&=-\sqrt{\frac{p}{q}},\quad \alpha_-=\sqrt{\frac{q}{p}}. \\
		\end{split}
	\end{equation}
	
	$D_{(r_1,1)(r_2,1)(r_3,1)}$ and $D_{(1,s_1)(1,s_2)(1,s_3)}$ are the chiral
	minimal-model OPE coefficients in the $s=1$ and $r=1$ sectors, respectively:
	\begin{equation}\label{def:D_subsector}
		\begin{split}
			D_{(r_1,1)(r_2,1)(r_3,1)}&=\frac{1}{[r_3-1]_{q_+}!}\,\sqrt{\frac{[\frac{r_1+r_3-r_2-1}{2}]_{q_+}!\,[\frac{r_2+r_3-r_1-1}{2}]_{q_+}!\,[\frac{r_1+r_2+r_3-1}{2}]_{q_+}!}{\left([r_1]_{q_+}\,[r_2]_{q_+}\,[r_3]_{q_+}\right)^{1/2}\,[\frac{r_1+r_2-r_3-1}{2}]_{q_+}!}} \\
			&\times\prod_{k=1}^{\frac{r_1+r_2-r_3-1}{2}} 
			\frac{\Gamma\left(\left(\frac{r_1+r_2+r_3+1}{2}-k\right)\alpha_+^2 - 1\right)}{\Gamma\left((r_1-k)\alpha_+^2\right) 
				\Gamma\left((r_2-k)\alpha_+^2\right) 
				\Gamma(1-k\alpha_+^2)} \\
			&\times\left[
			\begin{aligned}
				&\left(\prod_{k=1}^{r_1-1}
				\frac{\Gamma\left(k\alpha_+^2\right)^2
					\Gamma(1-k\alpha_+^2)}{\Gamma\left((k+1)\alpha_+^2-1\right)}\right) \\
				&\quad\times\left(\prod_{k=1}^{r_2-1}
				\frac{\Gamma\left(k\alpha_+^2\right)^2
					\Gamma(1-k\alpha_+^2)}{\Gamma\left((k+1)\alpha_+^2-1\right)}\right) \\
				&\quad\times\left(\prod_{k=1}^{r_3-1}
				\frac{\Gamma\left((k+1)\alpha_+^2-1\right)}
				{\Gamma\left(k\alpha_+^2\right)^2\Gamma(1-k\alpha_+^2)}\right)
			\end{aligned}
			\right]^{1/2}\,, \\
			D_{(1,s_1),(1,s_2),(1,s_3)}
			&=D_{(r_1,1)(r_2,1)(r_3,1)}
			\Big{|}_{\substack{(r_1,r_2,r_3,\alpha_+,q_+)\rightarrow(s_1,s_2,s_3,\alpha_-,q_-)}} \\
		\end{split}
	\end{equation}
	The square roots in the first line of
	Eq.~\eqref{ope_chiral_general:final} require a correlated normalization.
	Set $\alpha_+=-\sqrt{t}$ and $\alpha_-=1/\sqrt{t}$, using the square root
	that is positive at $t=1$. Use the lifted quantum-number logarithms
	\eqref{eq:qnumber_log_interior} and \eqref{eq:qnumber_log_exterior},
	and continue the Gamma and $Y$ factors in the same variable $t$.
	For fixed admissible labels, cancel the removable singularities at
	$t=1$ in the complete products before taking this normalization limit.
	The squared subsector coefficients then have finite positive limits.
	Indeed, for $m=(r_1+r_2-r_3-1)/2$, the powers of $t-1$ cancel
	between the Gamma products, and their combined sign with the lifted
	quantum-number factors is $(-1)^{m(m+1)}=1$. The ratios $Y(0,2\Bell_i)/Y(1,2\Bell_i)$ also have finite positive limits.
	For the latter statement, the zero factors in the numerator and
	denominator at $t=1$ occur in equal numbers, and the remaining factors
	have matching signs. We choose the branches for which the first line
	of Eq.~\eqref{ope_chiral_general:final} has a positive limit at $t=1$.
	This is a normalization of the combined expression, not an instruction
	to approach $q_+$ and $q_-$ independently from inside the unit disk.

	Starting with this germ, we analytically continue to
	\begin{equation}
		\operatorname{Im}t>0
	\end{equation}
	and take $t\to p/q+i0$ for a physical minimal model. Gamma-function
	poles and zeros of the individual $Y$ factors occur only on the real
	$t$ axis, so the remaining square roots can be continued consistently
	in the upper half-plane. The second line of
	Eq.~\eqref{ope_chiral_general:final} has no additional fractional-power
	ambiguity once $\alpha_\pm$ and the displayed phases are fixed.
	Figure~\ref{fig:q-number-continuation} illustrates the corresponding
	interior/exterior prescription for $q_\pm$. Evaluating separate
	principal square roots at the physical endpoint is not equivalent to
	this continuation.

	For completeness, the reference coefficient $D_{ijk}$ is finite and
	nonzero on every chiral triplet allowed by the truncated fusion rules.
	For example, in the $r$ sector these rules bound every factorial index
	in Eq.~\eqref{def:D_subsector} by $q-1$. Coprimality of $p$ and $q$
	therefore gives $[k]_{q_+}\neq0$ for every factor with $1\leq k<q$.
	The possibly singular Gamma factors have arguments $kt$ or $kt-1$
	with $1\leq k<q$, or $1-kt$ with $1\leq k<q$; at $t=p/q$, none is a
	nonpositive integer. Gamma functions have no zeros. The same argument
	applies to the $s$ sector after interchanging $p$ and $q$.
	Finally, a zero of a $Y$ factor would require $p k^+=q k^-$, whereas
	its truncated product ranges obey $1\leq k^+<q$ and $1\leq k^-<p$.
	Thus no such zero occurs. Empty products are understood to equal one.
	This justifies dividing by the reference factors in the fusion matrix
	and in the definition of the reduced coefficients below; it does not
	assert that every physical bulk coefficient is nonzero.

	The preceding discussion applies to general Virasoro minimal models. We now
	specialize to the $E$ series, whose spectra are diagonal in the second Kac
	indices, $s=\bar{s}$. We therefore label a Virasoro primary by
	$(r,\bar r,s)$ and define the OPE-coefficient ratio
	\begin{equation}
		\begin{split}
			\alpha_{ijk}\equiv\alpha_{(r_i,\bar{r}_i,s_i),(r_j,\bar{r}_j,s_j),(r_k,\bar{r}_k,s_k)}:=\frac{C_{(r_i,\bar{r}_i,s_i),(r_j,\bar{r}_j,s_j),(r_k,\bar{r}_k,s_k)}}{D_{(r_i,s_i),(r_j,s_j),(r_k,s_k)}D_{(\bar{r}_i,s_i),(\bar{r}_j,s_j),(\bar{r}_k,s_k)}}\,.
		\end{split}
	\end{equation}
	Equation~\eqref{eq:crossing2} then becomes
	\begin{equation}\label{eq:crossing3}
		\begin{split}
			&\sum_{5}\alpha_{125}\alpha_{345}\sixj{\ell^+_1}{\ell^+_2}{\ell^+_5}{\ell^+_3}{\ell^+_4}{\ell^+_6}_{q_+}
			\sixj{\ell^-_1}{\ell^-_2}{\ell^-_5}
			{\ell^-_3}{\ell^-_4}{\ell^-_6}_{q_-}  \sixj{\bar\ell^+_1}{\bar\ell^+_2}{\bar\ell^+_5}{\bar\ell^+_3}{\bar\ell^+_4}{\bar\ell^+_6}_{q_+}
			\sixj{\bar\ell^-_1}{\bar\ell^-_2}{\bar\ell^-_5}
			{\bar\ell^-_3}{\bar\ell^-_4}{\bar\ell^-_6}_{q_-}\\
			&=\begin{cases}
				\alpha_{236}\alpha_{416}\,, & \text{if}\ (r_6,s_6;\bar{r}_6,\bar{s}_6)\ \text{is in the spectrum}, \\
				0\,, & \text{otherwise}.
			\end{cases}
		\end{split}
	\end{equation}
	Although physical external and \(s\)-channel fields obey
	\(\ell_i^-=\bar\ell_i^-\), the holomorphic and antiholomorphic fusion
	transformations introduce independent trial labels \(\ell_6^-\) and
	\(\bar\ell_6^-\).  The two \(q_-\) symbols therefore cannot in general be
	replaced by a square.  For the $s$-independent lift of Section \ref{sec:global_solution}, the reduced OPE coefficients factor out of the sum over the second Kac label, and orthogonality projects onto \(s_6=\bar s_6\).
	
	We begin by solving the $s=1$ sector, for which
	$\ell^-=\bar\ell^-=0$. The $q_-$ $6j$ symbols then equal unity, and
	Eq.~\eqref{eq:crossing3} reduces to
	\begin{equation}\label{eq:crossing4}
		\begin{split}
			&\sum_{\ell^+_5,\bar{\ell}^+_5}\alpha_{125}\alpha_{345}\sixj{\ell^{+}_1}{\ell^{+}_2}{ \ell^{+}_5}{\ell^{+}_3}{\ell^{+}_4}{ \ell^{+}_6}_{q_+}\sixj{\bar{\ell}^{+}_1}{ \bar{\ell}^{+}_2}{ \bar{\ell}^{+}_5}{ \bar{\ell}^{+}_3}{ \bar{\ell}^{+}_4}{ \bar{\ell}^{+}_6}_{q_+} \\
			&=\begin{cases}
				\alpha_{236}\alpha_{416}\,, & \text{if}\ (\ell^+_6,\bar{\ell}^+_6)\ \text{is in the spectrum}, \\
				0\,, & \text{otherwise}.
			\end{cases}
		\end{split}
	\end{equation}
	This is the reduced bootstrap equation used in the $s=1$ sector.
	
	\subsection{Reference OPE coefficient in the unitary case}
	
	The following lemma establishes the reality property needed in Section~\ref{subsec:unitary}.
	
	\begin{lemma}[Reality of the unitary reference OPE coefficients]\label{lem:unitary-reference-reality}
		Let $p=q\pm1$ and let $(r_i,s_i)$, $i=1,2,3$, be an admissible
		chiral triplet with odd $s_i$. Explicitly, $1\leq r_i<q$,
		$1\leq s_i<p$, the sums of the three labels in each sector are odd,
		and
		\[
		r_i+r_j-r_k\geq1,\quad r_1+r_2+r_3\leq2q-1,
		\qquad
		s_i+s_j-s_k\geq1,\quad s_1+s_2+s_3\leq2p-1
		\]
		for every permutation $(i,j,k)$. With the correlated continuation
		specified above, $D_{123}$ is finite and nonzero and
		\begin{equation}\label{eq:unitary-chiral-square-sign}
			\operatorname{sgn}(D_{123}^{\,2})
			=(-1)^{\sum_{i=1}^3(r_i-1)(s_i-1)/2}.
		\end{equation}
		Consequently, for a physical exceptional triplet whose left and right
		chiral triplets are both admissible,
		\begin{equation}\label{eq:unitary-reference-real}
			\bigl(C^{\mathrm{(ref)}}_{IJK}\bigr)^2>0,
			\qquad
			\bigl(C^{\mathrm{(ref)}}_{IJK}\bigr)^*
			=C^{\mathrm{(ref)}}_{IJK}.
		\end{equation}
		These assertions concern the nonzero reference factors, including when
		the corresponding exceptional reduced coefficient vanishes.
	\end{lemma}
	\begin{proof}
		Write $t=p/q$. First consider the $s=1$ factor in
		Eq.~\eqref{def:D_subsector}. Put
		$a=r_1-1$, $b=r_2-1$, $c=r_3-1$ and
		\[
		m=\frac{a+b-c}{2},\qquad n=\frac{a+c-b}{2},\qquad
		d=\frac{b+c-a}{2},\qquad H=m+n+d.
		\]
		These are nonnegative integers. Cancellation of the Gamma poles at
		$t=1$, with the lifted quantum-number square root, gives the explicitly
		positive limit
		\begin{equation}\label{eq:subsector-positive-limit}
			\lim_{t\to1}D_{(r_1,1)(r_2,1)(r_3,1)}^{\,2}
			=\frac{n!\,d!\,(H+1)!}
			{c!\,m!\,a!\,b!\sqrt{(a+1)(b+1)(c+1)}}
			\left[\prod_{k=1}^{m}
			\frac{k!\,(H-k)!}{(a-k)!\,(b-k)!}\right]^2>0.
		\end{equation}
		For example, the cancellation follows from
		$\Gamma(kt)^2\Gamma(1-kt)/\Gamma((k+1)t-1)
		\sim(-1)^k/[k(t-1)]$ and
		$\sqrt{[r_1][r_2][r_3]}\to(-1)^H\sqrt{r_1r_2r_3}$.
		Every remaining possible real zero or pole is at $t=v/k$ with
		$1\leq k<q$. If $v/k\neq1$, then
		$|v/k-1|\geq1/k>1/q$. Thus the real segment from $1$ to
		$1\pm1/q$, including its physical endpoint after removing the
		singularity at $1$, contains no zero or pole. Moreover,
		$[r_1][r_2][r_3]$ is positive on this segment: its limit at $1$ is
		positive by fusion parity, and no factor vanishes. Its lifted square
		root is therefore real throughout. The squared subsector coefficient
		stays positive. Applying the same argument to $t^{-1}=q/p$,
		whose unitary endpoint is $1\mp1/p$, proves positivity for the
		$r=1$ subsector.
		
		For the mixed factors, direct cancellation of the diagonal zero
		factors in Eq.~\eqref{def:Y} yields
		\begin{equation}
			\lim_{t\to1}\frac{Y(0,2\Bell_i)}{Y(1,2\Bell_i)}
			=\frac{1}{\min(r_i,s_i)}>0.
		\end{equation}
		Every zero of a nonempty $Y$ product is at $t=k^-/k^+$.
		The truncated fusion bounds give $1\leq k^+<q$ and
		$1\leq k^-<p$, so no such zero other than $1$ lies on the
		unitary segment. The displayed ratios therefore remain positive.
		The remaining quotient of $Y$ products in
		Eq.~\eqref{ope_chiral_general:final} is real and nonzero at the
		physical endpoint; its square is positive. Hence only the explicit
		phase can affect the sign of $D_{123}^{\,2}$. Since the $s_i$ are odd,
		$\ell_i^-$ are integers, and the truncated triangle conditions make
		$(\ell_1^++\ell_3^+-\ell_2^+)
		(\ell_2^-+\ell_3^--\ell_1^-)$ an integer. Squaring the phase in that
		equation gives precisely Eq.~\eqref{eq:unitary-chiral-square-sign}.
		
		In every exceptional spectral pairing $r_i\equiv\bar r_i\pmod2$,
		while the left and right second labels are the same. The two signs in
		Eq.~\eqref{eq:unitary-chiral-square-sign} therefore agree. Their
		product is strictly positive, proving
		Eq.~\eqref{eq:unitary-reference-real}. No enumeration of the individual
		branch signs is needed for this conclusion.
	\end{proof}
	
	\subsection{Barnes representation of the chiral coefficients}
	\label{app:barnes-representation}
	
	We now derive the Barnes representation
	\eqref{eq:alpha-intro-chiral-D} and its phase
	\eqref{eq:alpha-intro-chiral-phase} from the finite products above.
	The rewriting of minimal-model structure constants in terms of
	Barnes functions was discussed in
	Ref.~\cite[Appendix~A.1]{FredenhagenWellig:2007}.
	Here we also keep track of the phase in our continuation convention.
	
	\paragraph{Absolute value.}
	Put $t=\beta^2$ and define
	\begin{equation}
		\begin{split}
			\Upsilon_\beta(x)
			&:=\frac{1}
			{\Gamma_\beta(x)\Gamma_\beta(Q_\beta-x)},\\
			H(a,b)
			&:=\frac{\Upsilon_\beta((a+1)\beta-b\beta^{-1})}
			{\Upsilon_\beta(\beta)}.
		\end{split}
	\end{equation}
	The shift relations of $\Gamma_\beta$ imply
	\begin{equation}
		\begin{split}
			\frac{\Upsilon_\beta(x+\beta)}{\Upsilon_\beta(x)}
			&=\beta^{1-2\beta x}\gamma(\beta x),\\
			\frac{\Upsilon_\beta(x+\beta^{-1})}{\Upsilon_\beta(x)}
			&=\beta^{2x/\beta-1}\gamma(x/\beta),
			\qquad
			\gamma(x):=\frac{\Gamma(x)}{\Gamma(1-x)}.
		\end{split}
	\end{equation}
	In particular,
	\begin{equation}
		\begin{split}
			H(a,0)
			&=\beta^{a-ta(a+1)}
			\prod_{k=1}^{a}\gamma(kt),\\
			H(0,b)
			&=\beta^{-b+b(b+1)/t}
			\prod_{k=1}^{b}\gamma(k/t).
		\end{split}
	\end{equation}
	For $a,b\geq1$, the same shift relations give
	\begin{equation}
		\frac{H(a,b)H(a-1,b-1)}
		{H(a,b-1)H(a-1,b)}
		=-\frac{1}{(a\beta-b\beta^{-1})^2}.
	\end{equation}
	Iterating this identity yields
	\begin{equation}
		H(a,b)
		=H(a,0)H(0,b)
		\frac{(-1)^{ab}}{Y(0,(a,b))^2},
		\label{eq:barnes-H-factorization}
	\end{equation}
	where $Y$ is defined in \eqref{def:Y}.
	
	Let
	\[
	a_i=m_r-r_i,\qquad b_i=m_s-s_i,\qquad m_u=\frac{u_1+u_2+u_3-1}{2},
	\]
	where we use the same $m_u$ as in Section~\ref{subsec:alpha-intro-chiral}.
	The fusion rules imply that $a_i,b_i$ are nonnegative integers.
	Pairing opposite sign choices in
	\eqref{eq:alpha-intro-chiral-R}, and using
	$\Upsilon_\beta(x)=\Upsilon_\beta(Q_\beta-x)$, gives
	\begin{equation}
		\begin{split}
			\mathcal R_{123}
			&=\Upsilon_\beta(\beta)^4
			H(m_r,m_s)\prod_{i=1}^3H(a_i,b_i),\\
			\mathcal R_{0ii}
			&=\Upsilon_\beta(\beta)^4
			H(r_i,s_i)H(r_i-1,s_i-1),\\
			\mathcal R_{000}
			&=\Upsilon_\beta(\beta)^4H(1,1).
		\end{split}
	\end{equation}
	Substituting \eqref{eq:barnes-H-factorization} into these
	expressions reduces the normalized Barnes product to ordinary
	Gamma functions and $Y$ factors. Using
	\[
	\Gamma(x+1)=x\Gamma(x),\qquad
	\Gamma(x)\Gamma(1-x)=\frac{\pi}{\sin\pi x},\qquad
	\abs{[k]_{q_+}}
	=\left|\frac{\sin(\pi kt)}{\sin(\pi t)}\right|,
	\]
	the Gamma factors reproduce the absolute values of the two
	subsector coefficients in \eqref{def:D_subsector}. The remaining
	$Y$ factors reproduce the mixed factor in
	\eqref{ope_chiral_general:final}. The powers of $\beta$ cancel,
	using $\sum_i r_i=2m_r+1$ and $\sum_i s_i=2m_s+1$.
	Consequently,
	\begin{equation}
		\abs{D_{123}}^4
		=
		\left|
		\frac{\mathcal R_{123}^{\,2}\mathcal R_{000}}
		{\mathcal R_{011}\mathcal R_{022}\mathcal R_{033}}
		\right|.
		\label{eq:barnes-modulus-proof}
	\end{equation}
	All finite-product factors are nonzero and finite under the
	truncated fusion rules and coprimality assumptions stated above.
	The same identities therefore justify evaluating the Barnes
	expression directly at the rational point.
	
	\paragraph{Phase.}
	We use the correlated continuation $t\to p/q+i0$ and the
	normalization at $t=1$ specified above. In particular, the logarithm
	of each complete squared subsector coefficient is chosen to have
	a real limit at $t=1$; it is not a principal logarithm at the
	physical endpoint.
	
	First consider the $r$ subsector. The interior quantum-number
	logarithm \eqref{eq:qnumber_log_interior} obeys
	\begin{equation}
		\frac{1}{\pi}\operatorname{Im}
		L_n^{\mathrm{in}}(\pi t+i0)
		=-\lfloor nt\rfloor+\lfloor t\rfloor.
	\end{equation}
	For the ordinary Gamma factors, the logarithm that is real on
	the positive real axis gives
	\begin{equation}
		\begin{split}
			\frac{1}{\pi}\operatorname{Im}\log\Gamma(kt+i0)
			&=0,\\
			\frac{1}{\pi}\operatorname{Im}\log\Gamma(1-kt-i0)
			&=\lfloor kt\rfloor,\\
			\frac{1}{\pi}\operatorname{Im}\log\Gamma(kt-1+i0)
			&=-\mathbf 1_{\{kt<1\}}.
		\end{split}
	\end{equation}
	Substitution into \eqref{def:D_subsector} gives
	\begin{equation}
		\operatorname{Im}\log
		\left(D_{(r_1,1)(r_2,1)(r_3,1)}^2\right)
		=\pi Z_\beta(r),
	\end{equation}
	where $Z_\beta$ is defined in \eqref{def:D_phase_notation}. The floor sums arise from the quantum factorials and the
	$\Gamma(1-kt)$ factors. To see the origin of the last line of
	$Z_\beta$, put $B=\lfloor t^{-1}\rfloor$. The remaining
	indicator sums reduce to
	\[
	\sum_{i=1}^3\min(r_i,B)
	-2\min(m_r,B)-\min(1,B).
	\]
	Moreover, $Z_\beta(r)=0$ for $t$ sufficiently close to $1$ on
	either side. Thus the positive initial normalization fixes the
	subsector phase without an additional sign.
	
	For the $s$ subsector, $t^{-1}$ approaches the real axis from
	below. The constant multiples of $2\pi i$ from the exterior
	quantum-number logarithms are fixed by the same real initial
	value. We therefore obtain
	\begin{equation}
		\begin{split}
			\frac{D_{(r_1,1)(r_2,1)(r_3,1)}}
			{\abs{D_{(r_1,1)(r_2,1)(r_3,1)}}}
			&=e^{i\pi Z_\beta(r)/2},\\
			\frac{D_{(1,s_1)(1,s_2)(1,s_3)}}
			{\abs{D_{(1,s_1)(1,s_2)(1,s_3)}}}
			&=e^{-i\pi Z_{\beta^{-1}}(s)/2}.
		\end{split}
	\end{equation}
	
	Next, each factor in $Y$ has the form
	\[
	k^+\alpha_++k^-\alpha_-
	=\frac{k^- - k^+t}{\sqrt t}.
	\]
	When it is negative at the physical endpoint, it approaches
	the negative real axis from below. Counting these factors gives
	\begin{equation}
		\operatorname{Im}\log Y(n,(a,b))
		=-\pi N_\beta^{(n)}(a,b),
	\end{equation}
	with the function $N_\beta^{(n)}$ defined in \eqref{def:D_phase_notation}. The initial imaginary parts of
	$\log Y(0,(a,b))$ and $\log Y(1,(a,b))$ agree: simultaneous
	translation of both product indices preserves the numbers of
	negative factors and diagonal zeros at $t=1$. Their difference
	therefore has imaginary limit zero, consistently with the
	positive square root of $Y(0)/Y(1)$.
	
	Combining all the $Y$ factors in
	\eqref{ope_chiral_general:final}, their contribution to the phase
	divided by $\pi$ is
	\begin{equation}
		\begin{split}
			&-\frac12\sum_{i=1}^3
			\left[
			N_\beta^{(0)}(r_i-1,s_i-1)
			+N_\beta^{(1)}(r_i-1,s_i-1)
			\right]\\
			&\qquad
			+\sum_{i=1}^3N_\beta^{(0)}(a_i,b_i)
			+N_\beta^{(1)}(m_r-1,m_s-1).
		\end{split}
	\end{equation}
	Finally, the explicit phase in
	\eqref{ope_chiral_general:final} simplifies to
	\begin{equation}
		\begin{split}
			&\ell_1^+\ell_1^-+\ell_2^+\ell_2^-
			-\ell_3^+\ell_3^- -2\ell_1^+\ell_2^-\\
			&\quad
			+(\ell_1^++\ell_3^+-\ell_2^+)
			(\ell_2^-+\ell_3^--\ell_1^-)
			=\frac14\sum_{i=1}^3r_i(s_{i-1}-s_{i+1}).
		\end{split}
	\end{equation}
	Adding these contributions gives
	\[
	\frac{D_{123}}{\abs{D_{123}}}
	=e^{i\pi\Phi(r,s;\beta)},
	\]
	with $\Phi$ precisely as in
	\eqref{eq:alpha-intro-chiral-phase}.
	Together with \eqref{eq:barnes-modulus-proof}, this proves
	\eqref{eq:alpha-intro-chiral-D}.

	\section{\texorpdfstring{Conventions for $\alpha_{IJK}$ with one $\mathbb{Z}_2$ simple current}{Conventions for alpha IJK with one Z2 simple current}}\label{app:simple_current_conventions}
	This appendix expands the construction in Section \ref{sec:triplet_orbit} for the reference residues:
	\begin{itemize}
		\item $(p,q)=(1,12),(5,12)$ for $E_6$;
		\item $(p,q)=(1,18),(5,18),(7,18)$ for $E_7$;
		\item $(p,q)=(1,30),(7,30),(11,30),(13,30)$ for $E_8$.
	\end{itemize}
	The conventions for the other residues follow either by complex conjugation, which gives a consistent convention,
	$$\alpha_{IJK}\Big{|}_{(p,q)}=\left(\alpha_{IJK}\Big{|}_{(-p,q)}\right)^*\,,$$
	or by the phase relation \eqref{eq:periodic_phase}.
	
	We fix the sign conventions for reduced coefficients $\alpha_{IJK}$ involving one $\mathbb{Z}_2$ simple current. Specifically, we consider
	\begin{itemize}
		\item $\alpha_{(5,5)(\ell,\bar\ell)(5-\ell,5-\bar\ell)}$ for $E_6$ series;
		\item $\alpha_{(8,8)(\ell,\bar\ell)(8-\ell,8-\bar\ell)}$, $\alpha_{(8,0)(\ell,\bar\ell)(8-\ell,\bar\ell)}$ and $\alpha_{(0,8)(\ell,\bar\ell)(\ell,8-\bar\ell)}$ for $E_7$ series;
		\item $\alpha_{(14,14)(\ell,\bar\ell)(14-\ell,14-\bar\ell)}$, $\alpha_{(14,0)(\ell,\bar\ell)(14-\ell,\bar\ell)}$ and $\alpha_{(0,14)(\ell,\bar\ell)(\ell,14-\bar\ell)}$ for $E_8$ series.
	\end{itemize}
	
	Throughout this appendix we suppress the second Kac labels and use the odd-$s$ representatives. In a coefficient with one $\mathbb{Z}_2$ simple current, that insertion has $s=1$ and the other two fields have the same allowed $s$. A generic triplet can have any allowed odd labels $s_1,s_2,s_3$; the displayed relations use the independence of the reduced coefficients from these labels. The reference value $p=w=1$ is formal and does not define a physical minimal model.
	
	Ratios relating generic triplets below are shorthand for the corresponding cross-multiplied identities. This interpretation also applies when the generic coefficients vanish. The simple-current coefficients that occur as known denominators are nonzero.

	\subsection{\texorpdfstring{Convention for $E_6$ series}{Convention for E6 series}}
	The $E_6$ series has only a diagonal $\mathbb{Z}_2$ operator labeled by $(\ell,\bar{\ell},s)=(5,5,1)$.
	
	We first determine $\alpha_{IJK}$ with one external operator equal to the $\mathbb{Z}_2$ simple current:
	$$\alpha_{(5,5)(\ell,\bar\ell)(5-\ell,5-\bar\ell)}\,.$$
	Equation \eqref{alpha:constr_Z2diagonal2} fixes the squares of these coefficients and hence determines each coefficient up to a sign. To fix the sign, we use the operator-sign redundancy $V_I\rightarrow-V_I$. Let $(\ell,\bar{\ell},s)$ and $(5-\ell,5-\bar{\ell},s)$ be a pair of operators related by reflection. In the $E_6$ case, there are no operators with $(\ell,\bar{\ell})=(\tfrac{5}{2},\tfrac{5}{2})$, so each pair contains two different operators. Except for the pair containing the identity, whose coupling is already normalized, we can choose the sign of one field to fix the sign of $\alpha$. It is consistent to use the convention
	\begin{equation}\label{alpha:convention_E6}
		\begin{split}
			\alpha_{(5,5)(\ell,\bar\ell)(5-\ell,5-\bar\ell)}&=\sqrt{(-1)^{f_{12,k,w}(5,5-\ell,\ell)+f_{12,k,w}(5,5-\bar{\ell},\bar{\ell})}} \\
			&\times
			\begin{cases}
				i^{S_{\ell,\bar{\ell},1}+S_{5-\ell,5-\bar{\ell},1}} & \text{for } \ell<\tfrac{5}{2}\,,\\
				(-i)^{S_{\ell,\bar{\ell},1}+S_{5-\ell,5-\bar{\ell},1}} & \text{for } \ell>\tfrac{5}{2}\,.\\
			\end{cases}
		\end{split}
	\end{equation}
	Here we always take the principal branch for the square-root function:
	\begin{equation}
		\begin{split}
			\sqrt{1}=1\,,\quad\sqrt{-1}=i\,.
		\end{split}
	\end{equation}
	This convention is consistent with the permutation properties \eqref{eq:permutation_alpha}.
	
	The identity operator has a fixed sign. The remaining gauge transformations that preserve this convention are generated by
	\begin{itemize}
		\item Flip both signs in any reflection pair other than $\{V_{0,0,1},V_{5,5,1}\}$:
		$V_{\ell,\bar{\ell},s},V_{5-\ell,5-\bar{\ell},s}\rightarrow-V_{\ell,\bar{\ell},s},-V_{5-\ell,5-\bar{\ell},s}$.
		\item Choose one operator from every reflection pair other than $\{V_{0,0,1},V_{5,5,1}\}$ and flip all their signs together with the sign of $V_{5,5,1}$.
	\end{itemize}
	Both transformations leave $V_{0,0,1}$ unchanged and preserve the normalized identity couplings.
	
	The spin content of the $E_6$ series is given by
	\begin{equation}
		\begin{split}
			S_{\ell,\bar{\ell},s}
			&=
			\frac{(\ell-\bar{\ell})\bigl(p(\ell+\bar{\ell}+1)-12s\bigr)}{12}=
			\begin{cases}
				0 & \text{for }(\ell,\ell,s)\,, \\
				3s-p & \text{for }(0,3,s)\,, \\
				3s-2p & \text{for }(2,5,s)\,, \\
				2s-p & \text{for }(\frac{3}{2},\frac{7}{2},s)\,.
			\end{cases}
		\end{split}
	\end{equation}
	Thus we have
	\begin{equation}
		\begin{split}
			S_{\ell,\bar{\ell},1}+S_{5-\ell,5-\bar{\ell},1}=\begin{cases}
				p & (\ell,\bar{\ell})=(0,3)\,,(5,2)\,, \\
				-p & (\ell,\bar{\ell})=(3,0)\,,(2,5)\,, \\
				0 & (\ell,\bar{\ell})=(\frac{3}{2},\frac{7}{2})\,,(\frac{7}{2},\frac{3}{2})\,, \\
				0 & \ell=\bar{\ell}\,. \\			
			\end{cases}
		\end{split}
	\end{equation}
	Although \eqref{alpha:convention_E6} is a consistent convention, it differs from the convention induced by \eqref{eq:periodic_phase} by operator-sign redefinitions. We therefore use it only for the reference residues: 
	$$p=w\ (+\,0\cdot 12)\quad\text{with }w=1,5.$$
	In these cases, Eq.~\eqref{alpha:convention_E6} becomes
	\begin{equation}\label{alpha:convention_E6_final}
		\boxed{\begin{aligned}
				\alpha_{(5,5)(\ell,\bar\ell)(5-\ell,5-\bar\ell)}&=\sqrt{(-1)^{f_{12,0,w}(5,5-\ell,\ell)+f_{12,0,w}(5,5-\bar{\ell},\bar{\ell})}} \\
				&\times
				\begin{cases}
					i^{w} & \text{for } (\ell,\bar{\ell})=(0,3),(3,0)\,,\\
					i^{-w} & \text{for } (\ell,\bar{\ell})=(2,5),(5,2)\,,\\
					1 & \text{otherwise}\,. \\
				\end{cases}
		\end{aligned}}
	\end{equation}
	
	Next, let us look at the relation \eqref{alpha:constr_Z2diagonal1}. In the $E_6$ case, it reads
	\begin{equation}\label{alpha:constr_Z2diagonal1_E6}
		\begin{split}
			\frac{
				\alpha_{(\ell_1,\bar\ell_1)(\ell_2,\bar\ell_2)(\ell_3,\bar\ell_3)}
			}{
				\alpha_{(5-\ell_1,5-\bar\ell_1)(5-\ell_2,5-\bar\ell_2)(\ell_3,\bar\ell_3)}
			}
			&=
			\frac{
				\alpha_{(5,5)(\ell_2,\bar\ell_2)(5-\ell_2,5-\bar\ell_2)}
			}{
				\alpha_{(5,5)(\ell_1,\bar\ell_1)(5-\ell_1,5-\bar\ell_1)}
			}
			(-1)^{\ell_{231}+\bar\ell_{231}} \\
			&\quad\times(-1)^{f_{12,k,w}(\ell_2,\ell_3,\ell_1)+f_{12,k,w}(\bar{\ell}_2,\bar{\ell}_3,\bar{\ell}_1)}
		\end{split}
	\end{equation}
	The known phase on the right-hand side, together with the permutation relations, organizes the triplets $[IJK]$ into orbits. Within each orbit, the reduced coefficients are related by these phases.
	
	Below we demonstrate in detail how the orbit decomposition works for $E_6$. We will not repeat the same analysis for $E_7$ and $E_8$, since the final results for these two cases are similar but quite lengthy. 
	
	According to \eqref{def:model}, the $E_6$ spectrum is organized in three sectors
	\begin{equation}
		\begin{split}
			\id&=\left\{(0,0),(0,3),(3,0),(3,3)\right\}\,, \\
			\sigma&=\left\{(\tfrac{3}{2},\tfrac{3}{2}),(\tfrac{3}{2},\tfrac{7}{2}),(\tfrac{7}{2},\tfrac{3}{2}),(\tfrac{7}{2},\tfrac{7}{2})\right\}\,, \\
			\epsilon&=\left\{(2,2),(2,5),(5,2),(5,5)\right\}\,. \\
		\end{split}
	\end{equation}
	The three-sector decomposition suggests Ising-type fusion rules. The bootstrap solution obeys the corresponding sector selection rules
	\begin{equation}
		\begin{split}
			&\id\times\id=\id\,,\quad\id\times\sigma=\sigma\,,\quad\id\times\epsilon=\epsilon\,, \\
			&\sigma\times\sigma=\id+\epsilon\,,\quad\sigma\times\epsilon=\sigma\,,\quad\epsilon\times\epsilon=\id\,. \\
		\end{split}
	\end{equation}
	These are rules for the three sectors, supplemented by the second-Kac-label fusion conditions in the full theory. We do not impose them as extra assumptions in the bootstrap calculation; they are properties of the resulting solution. 
	
	Fusion with the $\mathbb{Z}_2$ simple current $(5,5)$, which belongs to the $\epsilon$ sector, acts on the sectors as
	\begin{equation}
		\begin{split}
			\id\rightarrow\epsilon\,,\quad\sigma\rightarrow\sigma\,,\quad \epsilon\rightarrow\id\,. \\
		\end{split}
	\end{equation}
	
	\noindent\textbf{Triplets of type $\id\id\id$. }
	The OPE coefficients for triplets of type $\epsilon\epsilon\id$ can be obtained from $\id\id\id$ by the $\mathbb{Z}_2$ transformation. Within $\id\id\id$ the triplet representatives are
	\begin{equation}
		\begin{split}
			[(0,3),(0,3),(0,3)]\,,\\
			[(3,0),(3,0),(3,0)]\,,\\
			[(0,3),(3,0),(3,3)]\,,\\
			[(0,3),(3,3),(3,3)]\,,\\
			[(3,0),(3,3),(3,3)]\,, \\  
			[(3,3),(3,3),(3,3)]\,. \\
		\end{split}
	\end{equation}
	Using $\mathbb{Z}_2$, we have the orbits
	\begin{align}
		\mathcal{C}_1^+&=\bigl\{[(0,3),(0,3),(0,3)]\,,\ [(5,2),(5,2),(0,3)]\bigr\}\,,\notag\\[2pt]
		\mathcal{C}_1^-&=\bigl\{[(3,0),(3,0),(3,0)]\,,\ [(2,5),(2,5),(3,0)]\bigr\}\,,\notag\\[2pt]
		\mathcal{C}_2&=\bigl\{[(0,3),(3,0),(3,3)]\,,\ [(2,5),(5,2),(3,3)]\,,\notag\\*
		&\qquad [(0,3),(2,5),(2,2)]\,,\ [(3,0),(5,2),(2,2)]\bigr\}\,,\notag\\[2pt]
		\mathcal{C}_3^+&=\bigl\{[(0,3),(3,3),(3,3)]\,,\ [(0,3),(2,2),(2,2)]\,,\notag\\*
		&\qquad [(5,2),(2,2),(3,3)]\bigr\}\,,\notag\\[2pt]
		\mathcal{C}_3^-&=\bigl\{[(3,0),(3,3),(3,3)]\,,\ [(3,0),(2,2),(2,2)]\,,\notag\\*
		&\qquad [(2,5),(2,2),(3,3)]\bigr\}\,,\notag\\[2pt]
		\mathcal{C}_4&=\bigl\{[(3,3),(3,3),(3,3)]\,,\ [(3,3),(2,2),(2,2)]\bigr\}\,.
	\end{align}
	
	Here the superscripts $\pm$ indicate the orbits which are related by a global parity transformation
	\begin{equation}
		\begin{split}
			(\ell,\bar{\ell})\rightarrow(\bar{\ell},\ell)\,.
		\end{split}
	\end{equation}
	
	\noindent\textbf{Triplets of type $\sigma\sigma\id$. } The triplets of type $\sigma\sigma\epsilon$ can be obtained by a $\mathbb{Z}_2$ transformation. Within $\sigma\sigma\id$ the representatives modulo $\mathbb{Z}_2$ transformations are
	\begin{equation}
		\begin{split}
			[(3/2,3/2),(3/2,3/2),(0,3)]\,,\\
			[(3/2,3/2),(3/2,3/2),(3,0)]\,,\\
			[(3/2,3/2),(3/2,3/2),(3,3)]\,,\\
			[(3/2,3/2),(3/2,7/2),(0,3)]\,,\\
			[(3/2,3/2),(7/2,3/2),(3,0)]\,,\\
			[(3/2,3/2),(3/2,7/2),(3,3)]\,,\\
			[(3/2,3/2),(7/2,3/2),(3,3)]\,,\\
			[(3/2,3/2),(7/2,7/2),(3,3)]\,, \\  
			[(3/2,7/2),(3/2,7/2),(0,3)]\,, \\
			[(7/2,3/2),(7/2,3/2),(3,0)]\,, \\
			[(3/2,7/2),(3/2,7/2),(3,3)]\,, \\
			[(3/2,7/2),(7/2,3/2),(3,3)]\,, \\
		\end{split}
	\end{equation}
	Using $\mathbb{Z}_2$, we have the orbits
	\begin{align}
		\mathcal{C}_5^+&=\bigl\{[(\tfrac{3}{2},\tfrac{3}{2}),(\tfrac{3}{2},\tfrac{3}{2}),(0,3)]\,,\ [(\tfrac{7}{2},\tfrac{7}{2}),(\tfrac{7}{2},\tfrac{7}{2}),(0,3)]\,,\notag\\*
		&\qquad [(\tfrac{3}{2},\tfrac{3}{2}),(\tfrac{7}{2},\tfrac{7}{2}),(5,2)]\bigr\}\,,\notag\\[2pt]
		\mathcal{C}_5^-&=\bigl\{[(\tfrac{3}{2},\tfrac{3}{2}),(\tfrac{3}{2},\tfrac{3}{2}),(3,0)]\,,\ [(\tfrac{7}{2},\tfrac{7}{2}),(\tfrac{7}{2},\tfrac{7}{2}),(3,0)]\,,\notag\\*
		&\qquad [(\tfrac{3}{2},\tfrac{3}{2}),(\tfrac{7}{2},\tfrac{7}{2}),(2,5)]\bigr\}\,,\notag\\[2pt]
		\mathcal{C}_6&=\bigl\{[(\tfrac{3}{2},\tfrac{3}{2}),(\tfrac{3}{2},\tfrac{3}{2}),(3,3)]\,,\ [(\tfrac{7}{2},\tfrac{7}{2}),(\tfrac{7}{2},\tfrac{7}{2}),(3,3)]\,,\notag\\*
		&\qquad [(\tfrac{3}{2},\tfrac{3}{2}),(\tfrac{7}{2},\tfrac{7}{2}),(2,2)]\bigr\}\,,\notag\\[2pt]
		\mathcal{C}_7^+&=\bigl\{[(\tfrac{3}{2},\tfrac{3}{2}),(\tfrac{3}{2},\tfrac{7}{2}),(0,3)]\,,\ [(\tfrac{7}{2},\tfrac{7}{2}),(\tfrac{7}{2},\tfrac{3}{2}),(0,3)]\,,\notag\\*
		&\qquad [(\tfrac{3}{2},\tfrac{3}{2}),(\tfrac{7}{2},\tfrac{3}{2}),(5,2)]\,,\ [(\tfrac{7}{2},\tfrac{7}{2}),(\tfrac{3}{2},\tfrac{7}{2}),(5,2)]\bigr\}\,,\notag\\[2pt]
		\mathcal{C}_7^-&=\bigl\{[(\tfrac{3}{2},\tfrac{3}{2}),(\tfrac{7}{2},\tfrac{3}{2}),(3,0)]\,,\ [(\tfrac{7}{2},\tfrac{7}{2}),(\tfrac{3}{2},\tfrac{7}{2}),(3,0)]\,,\notag\\*
		&\qquad [(\tfrac{3}{2},\tfrac{3}{2}),(\tfrac{3}{2},\tfrac{7}{2}),(2,5)]\,,\ [(\tfrac{7}{2},\tfrac{7}{2}),(\tfrac{7}{2},\tfrac{3}{2}),(2,5)]\bigr\}\,,\notag\\[2pt]
		\mathcal{C}_8^+&=\bigl\{[(\tfrac{3}{2},\tfrac{3}{2}),(\tfrac{3}{2},\tfrac{7}{2}),(3,3)]\,,\ [(\tfrac{7}{2},\tfrac{7}{2}),(\tfrac{7}{2},\tfrac{3}{2}),(3,3)]\,,\notag\\*
		&\qquad [(\tfrac{3}{2},\tfrac{3}{2}),(\tfrac{7}{2},\tfrac{3}{2}),(2,2)]\,,\ [(\tfrac{7}{2},\tfrac{7}{2}),(\tfrac{3}{2},\tfrac{7}{2}),(2,2)]\bigr\}\,,\notag\\[2pt]
		\mathcal{C}_8^-&=\bigl\{[(\tfrac{3}{2},\tfrac{3}{2}),(\tfrac{7}{2},\tfrac{3}{2}),(3,3)]\,,\ [(\tfrac{7}{2},\tfrac{7}{2}),(\tfrac{3}{2},\tfrac{7}{2}),(3,3)]\,,\notag\\*
		&\qquad [(\tfrac{3}{2},\tfrac{3}{2}),(\tfrac{3}{2},\tfrac{7}{2}),(2,2)]\,,\ [(\tfrac{7}{2},\tfrac{7}{2}),(\tfrac{7}{2},\tfrac{3}{2}),(2,2)]\bigr\}\,,\notag\\[2pt]
		\mathcal{C}_9&=\bigl\{[(\tfrac{3}{2},\tfrac{3}{2}),(\tfrac{7}{2},\tfrac{7}{2}),(3,3)]\,,\ [(\tfrac{3}{2},\tfrac{3}{2}),(\tfrac{3}{2},\tfrac{3}{2}),(2,2)]\,,\notag\\*
		&\qquad [(\tfrac{7}{2},\tfrac{7}{2}),(\tfrac{7}{2},\tfrac{7}{2}),(2,2)]\bigr\}\,,\notag\\[2pt]
		\mathcal{C}_{10}^+&=\bigl\{[(\tfrac{3}{2},\tfrac{7}{2}),(\tfrac{3}{2},\tfrac{7}{2}),(0,3)]\,,\ [(\tfrac{7}{2},\tfrac{3}{2}),(\tfrac{7}{2},\tfrac{3}{2}),(0,3)]\,,\notag\\*
		&\qquad [(\tfrac{3}{2},\tfrac{7}{2}),(\tfrac{7}{2},\tfrac{3}{2}),(5,2)]\bigr\}\,,\notag\\[2pt]
		\mathcal{C}_{10}^-&=\bigl\{[(\tfrac{7}{2},\tfrac{3}{2}),(\tfrac{7}{2},\tfrac{3}{2}),(3,0)]\,,\ [(\tfrac{3}{2},\tfrac{7}{2}),(\tfrac{3}{2},\tfrac{7}{2}),(3,0)]\,,\notag\\*
		&\qquad [(\tfrac{7}{2},\tfrac{3}{2}),(\tfrac{3}{2},\tfrac{7}{2}),(2,5)]\bigr\}\,,\notag\\[2pt]
		\mathcal{C}_{11}&=\bigl\{[(\tfrac{3}{2},\tfrac{7}{2}),(\tfrac{3}{2},\tfrac{7}{2}),(3,3)]\,,\ [(\tfrac{7}{2},\tfrac{3}{2}),(\tfrac{7}{2},\tfrac{3}{2}),(3,3)]\,,\notag\\*
		&\qquad [(\tfrac{3}{2},\tfrac{7}{2}),(\tfrac{7}{2},\tfrac{3}{2}),(2,2)]\bigr\}\,,\notag\\[2pt]
		\mathcal{C}_{12}&=\bigl\{[(\tfrac{3}{2},\tfrac{7}{2}),(\tfrac{7}{2},\tfrac{3}{2}),(3,3)]\,,\ [(\tfrac{3}{2},\tfrac{7}{2}),(\tfrac{3}{2},\tfrac{7}{2}),(2,2)]\,,\notag\\*
		&\qquad [(\tfrac{7}{2},\tfrac{3}{2}),(\tfrac{7}{2},\tfrac{3}{2}),(2,2)]\bigr\}\,.
	\end{align}
	
	We denote collectively by $\mathcal{C}_0$ the triplets containing the identity or the $\mathbb{Z}_2$ simple current: 
	\begin{equation}
		\begin{split}
			\mathcal{C}_0=\left\{[(\ell,\bar\ell),(\ell,\bar\ell),(0,0)],[(\ell,\bar\ell),(5-\ell,5-\bar\ell),(5,5)]\right\}
		\end{split}
	\end{equation}
	Here the identity and simple-current insertions have $s=1$, and the remaining two fields have the same allowed $s$. All reduced coefficients in $\mathcal{C}_0$ are fixed to $\pm1$ or $\pm i$.
	
	Thus the nontrivial triplets in the $E_6$ series form 18 orbits before imposing parity relations and the remaining crossing equations. The 18 representative reduced coefficients are variables to be determined, not freely adjustable parameters of the final solution. Parity is a property of the bootstrap solution; imposing the resulting parity relations leaves 12 nontrivial representatives.

	\subsection{\texorpdfstring{Convention for $E_7$ series}{Convention for E7 series}}
	The $E_7$ series has three $\mathbb{Z}_2$ operators labeled by $(\ell,\bar{\ell},s)=(0,8,1)$, $(8,0,1)$ and $(8,8,1)$. Together with the identity operator, they form a $\mathbb{Z}_2\times\mathbb{Z}_2$ fusion subalgebra. The operators in $E_7$ are organized into three types of orbits, presented in figures \ref{fig:first_type}, \ref{fig:second_type} and \ref{fig:third_type}.
	\begin{figure}[htbp]
		\centering
		\begin{tikzpicture}
			\node (00) at (0,0) {$(\ell,\ell)$};
			\node (08) at (2,2) {$(\ell,8-\ell)$};
			\node (80) at (2,-2) {$(8-\ell,\ell)$};
			\node (88) at (4,-0) {$(8-\ell,8-\ell)$};
			
			\draw (00) -- node[midway, above left] {$(0,8)$} (08);
			\draw (00) -- node[midway, below left] {$(8,0)$} (80);
			\draw (08) -- node[midway, above right] {$(8,0)$} (88);
			\draw (80) -- node[midway, below right] {$(0,8)$} (88);
			\draw (00) -- node[midway, above] {$(8,8)\quad$}(88);
			\draw (80) -- node[midway, below right] {$(8,8)$}(08);
		\end{tikzpicture}
		\caption{Orbits of the first type. $\ell=0,2,3$}
		\label{fig:first_type}
	\end{figure}
	
	\begin{figure}[htbp]
		\centering
		\begin{tikzpicture}
			\node (00) at (0,0) {$(4,1)$};
			\node (88) at (3,0) {$(4,7)$};
			
			\draw (00) -- node[midway, above] {$(8,8)$}node[midway, below] {$(0,8)$} (88);
			\draw (00) to[out=135,in=225,looseness=8] node[midway, left] {$(8,0)$} (00);
			\draw (88) to[out=45,in=-45,looseness=8] node[midway, right] {$(8,0)$} (88);
			
			\node (00) at (0,-2) {$(1,4)$};
			\node (88) at (3,-2) {$(7,4)$};
			
			\draw (00) -- node[midway, above] {$(8,8)$}node[midway, below] {$(8,0)$} (88);
			\draw (00) to[out=135,in=225,looseness=8] node[midway, left] {$(0,8)$} (00);
			\draw (88) to[out=45,in=-45,looseness=8] node[midway, right] {$(0,8)$} (88);
		\end{tikzpicture}
		\caption{Orbits of the second type.}
		\label{fig:second_type}
	\end{figure}
	
	\begin{figure}[htbp]
		\centering
		\begin{tikzpicture}
			\node (00) at (0,0) {$(4,4)$};
			\draw (00) to[out=45,in=-45,looseness=8] node[midway, right] {$(0,8),(8,0),(8,8)$} (00);
		\end{tikzpicture}
		\caption{Orbits of the third type.}
		\label{fig:third_type}
	\end{figure}
	
	We will use the spin content of the $E_7$ series, given by
	\begin{equation}
		\begin{split}
			S_{\ell,\bar{\ell},s}=\frac{(\ell-\bar{\ell})(p(\ell+\bar{\ell}+1)-18s)}{18}=\begin{cases}
				0 & \text{for }(\ell,\ell,s)\,, \\
				8s-4p & \text{for }(0,8,s)\,, \\
				4s-2p & \text{for }(2,6,s)\,, \\
				2s-p & \text{for }(3,5,s)\,, \\
				3s-p & \text{for }(1,4,s)\,, \\
				3s-2p & \text{for }(4,7,s)\,. \\
			\end{cases}
		\end{split}
	\end{equation}
	
	We first determine the reduced coefficients with one external $\mathbb{Z}_2$ simple current:
	$$\alpha_{(8,0)(\ell,\bar\ell)(8-\ell,\bar\ell)},\quad \alpha_{(0,8)(\ell,\bar\ell)(\ell,8-\bar\ell)},\quad\alpha_{(8,8)(\ell,\bar\ell)(8-\ell,8-\bar\ell)}\,.$$
	
	Consider the first type. There are three cases, generated by $(\ell,\bar{\ell})=(0,0),(2,2),(3,3)$. By \eqref{alpha:constr_Z2chiral2} we have
	\begin{equation}\label{alpha:constr_Z2chiral2_E7}
		\begin{split}
			\left[\alpha_{(8,0)(\ell,\ell)(8-\ell,\ell)}\right]^2&=\left[\alpha_{(0,8)(\ell,\ell)(\ell,8-\ell)}\right]^2 \\
			&=(-1)^{S_{8,0}+S_{\ell,\ell}+S_{8-\ell,\ell}+f_{18,k,w}(8,8-\ell,\ell)} \\
			&=\begin{cases}
				1 & \ell=0\,, \\
				(-1)^{f_{18,0,w}(8,6,2)} & \ell=2\,, \\
				(-1)^{1+f_{18,0,w}(8,5,3)} & \ell=3\,. \\
			\end{cases} \\
			\left[\alpha_{(8,8)(\ell,\ell)(8-\ell,8-\ell)}\right]^2&=1
		\end{split}
	\end{equation}
	Using the sign freedom of the non-$\mathbb{Z}_2$ operators, we set
	\begin{equation}\label{alpha:E7_type1_convention}
		\begin{split}
			\alpha_{(8,0)(\ell,\ell)(8-\ell,\ell)}&=\alpha_{(0,8)(\ell,\ell)(\ell,8-\ell)} =\begin{cases}
				1 & \ell=0\,, \\
				i^{f_{18,0,w}(8,6,2)} & \ell=2\,, \\
				i^{1+f_{18,0,w}(8,5,3)} & \ell=3\,. \\
			\end{cases} \\
			\alpha_{(8,8)(\ell,\ell)(8-\ell,8-\ell)}&=1\,.
		\end{split}
	\end{equation}
	This convention makes the possible parity symmetry manifest. For fixed signs of the $\mathbb{Z}_2$ operators, the remaining operator-sign freedom in each first-type orbit is a simultaneous flip of $(\ell,\ell)$, $(8-\ell,\ell)$, $(\ell,8-\ell)$ and $(8-\ell,8-\ell)$, with $s\neq1$ when $\ell=0$. The signs of the $\mathbb{Z}_2$ operators will be fixed below; any change in those signs must be accompanied by operator-sign changes that preserve \eqref{alpha:E7_type1_convention}. We now determine the remaining coefficients involving one simple current.
	
	For $\alpha_{(8,0)(8-\ell,8-\ell)(\ell,8-\ell)}$, we use \eqref{alpha:constr_Z2diagonal1}
	\begin{equation}
		\begin{split}
			\frac{
				\alpha_{(8,0)(8-\ell,8-\ell)(\ell,8-\ell)}
			}{
				\alpha_{(0,8)(\ell,\ell)(\ell,8-\ell)}
			}
			&=
			\frac{
				\alpha_{(8,8)(8-\ell,8-\ell)(\ell,\ell)}
			}{
				\alpha_{(8,8)(8,0)(0,8)}
			}
			(-1)^{f_{18,k,w}(8-\ell,\ell,8)+f_{18,k,w}(8-\ell,8-\ell,0)}
		\end{split}
	\end{equation}
	Later we will show that we can fix $\alpha_{(8,8)(8,0)(0,8)}=1$. By explicit calculation we have
	\begin{equation}\label{phi_q18_id1}
		\begin{split}
			(-1)^{f_{18,k,w}(8-\ell,\ell,8)}\equiv1\,,\quad \text{for }\ell=0,2,3\,.
		\end{split}
	\end{equation}
	The same analysis works for $\alpha_{(0,8)(8-\ell,8-\ell)(8-\ell,\ell)}$ and the result is the same. Therefore, together with the convention we get
	\begin{equation}
		\begin{split}
			\alpha_{(8,0)(8-\ell,8-\ell)(\ell,8-\ell)}&=\alpha_{(0,8)(8-\ell,8-\ell)(8-\ell,\ell)} \\
			&=\alpha_{(8,0)(\ell,\ell)(8-\ell,\ell)}(-1)^{f_{18,k,w}(8-\ell,8-\ell,0)} \\
			\text{for }&\ell=0,2,3\,. \\
		\end{split}
	\end{equation}
	
	The remaining coefficient is $\alpha_{(8,8)(\ell,8-\ell)(8-\ell,\ell)}$. We determine it using \eqref{alpha:constr_Z2antichiral1}:
	\begin{equation}
		\begin{split}
			\frac{
				\alpha_{(8,8)(\ell,8-\ell)(8-\ell,\ell)}
			}{
				\alpha_{(8,0)(\ell,\ell)(8-\ell,\ell)}
			}
			&=
			\frac{
				\alpha_{(0,8)(\ell,8-\ell)(\ell,\ell)}
			}{
				\alpha_{(0,8)(8,8)(8,0)}
			}
			(-1)^{f_{18,k,w}(8-\ell,\ell,8)}\,. \\
		\end{split}
	\end{equation}
	For $\ell=0$, it gives
	\begin{equation}
		\begin{split}
			\left(\alpha_{(8,8)(0,8)(8,0)}\right)^2=1\,.
		\end{split}
	\end{equation}
	Using the sign ambiguity of the diagonal $\mathbb{Z}_2$ operator, $(8,8,1)$, we set
	\begin{equation}
		\begin{split}
			\alpha_{(8,8)(0,8)(8,0)}=1\,.
		\end{split}
	\end{equation}
	This justifies the convention used above. For the other values of $\ell$, the same equation gives
	\begin{equation}
		\begin{split}
			\alpha_{(8,8)(\ell,8-\ell)(8-\ell,\ell)}&=\alpha_{(8,0)(\ell,\ell)(8-\ell,\ell)}\alpha_{(0,8)(\ell,8-\ell)(\ell,\ell)}(-1)^{f_{18,k,w}(8-\ell,\ell,8)}  \\
			&=\alpha_{(8,0)(\ell,\ell)(8-\ell,\ell)}\alpha_{(8,0)(8-\ell,\ell)(\ell,\ell)}(-1)^{f_{18,k,w}(8-\ell,\ell,8)}  \\
			&=(-1)^{f_{18,k,w}(8,8-\ell,\ell)}(-1)^{f_{18,k,w}(8-\ell,\ell,8)}  \\
			&=(-1)^{f_{18,k,w}(8,8-\ell,\ell)} \,. \\
		\end{split}
	\end{equation}
	Here, in the second line we used our parity-symmetric convention of gauge fixing; in the third line we used \eqref{alpha:constr_Z2chiral2}; in the last line we used \eqref{phi_q18_id1}.
	
	Table \ref{tab:E7_type1_conventions} summarizes these conventions.
	\begin{table}[htbp]
		\centering
		\renewcommand{\arraystretch}{1.35}
		\begin{tabular}{c c c c}
			\toprule
			Case & Triplet (ordered) & $\alpha$ & Permutation \\
			\midrule
			
			\multirow{6}{*}{$\ell=0$}
			& $(0,8)(0,0)(0,8)$
			& $1$
			& $+1$
			\\
			
			& $(8,0)(0,0)(8,0)$
			& $1$
			& $+1$
			\\
			
			& $(8,8)(0,0)(8,8)$
			& $1$
			& $+1$
			\\
			
			& $(8,0)(8,8)(0,8)$
			& $1$
			& $+1$
			\\
			
			& $(0,8)(8,8)(8,0)$
			& $1$
			& $+1$
			\\
			
			& $(8,8)(0,8)(8,0)$
			& $1$
			& $+1$
			\\
			
			\midrule
			
			\multirow{6}{*}{$\ell=2$}
			& $(0,8)(2,2)(2,6)$
			& $i^{f_{18,0,w}(8,6,2)}$
			& $+1$
			\\
			
			& $(8,0)(2,2)(6,2)$
			& $i^{f_{18,0,w}(8,6,2)}$
			& $+1$
			\\
			
			& $(8,8)(2,2)(6,6)$
			& $1$
			& $+1$
			\\
			
			& $(8,0)(6,6)(2,6)$
			& $i^{f_{18,0,w}(8,6,2)}
			(-1)^{f_{18,0,w}(6,6,0)}$
			& $+1$
			\\
			
			& $(0,8)(6,6)(6,2)$
			& $i^{f_{18,0,w}(8,6,2)}
			(-1)^{f_{18,0,w}(6,6,0)}$
			& $+1$
			\\
			
			& $(8,8)(2,6)(6,2)$
			& $(-1)^{f_{18,0,w}(8,6,2)}$
			& $+1$
			\\
			
			\midrule
			
			\multirow{6}{*}{$\ell=3$}
			& $(0,8)(3,3)(3,5)$
			& $i^{1+f_{18,0,w}(8,5,3)}$
			& $-1$
			\\
			
			& $(8,0)(3,3)(5,3)$
			& $i^{1+f_{18,0,w}(8,5,3)}$
			& $-1$
			\\
			
			& $(8,8)(3,3)(5,5)$
			& $1$
			& $+1$
			\\
			
			& $(8,0)(5,5)(3,5)$
			& $i^{1+f_{18,0,w}(8,5,3)}
			(-1)^{f_{18,0,w}(5,5,0)}$
			& $-1$
			\\
			
			& $(0,8)(5,5)(5,3)$
			& $i^{1+f_{18,0,w}(8,5,3)}
			(-1)^{f_{18,0,w}(5,5,0)}$
			& $-1$
			\\
			
			& $(8,8)(3,5)(5,3)$
			& $(-1)^{f_{18,0,w}(8,5,3)}$
			& $+1$
			\\
			
			\bottomrule
		\end{tabular}
		\caption{Reduced-coefficient conventions and permutation factors for the
			first-type orbits of the $E_7$ series. The $\alpha$ coefficients for
			other orderings are obtained by multiplying by the corresponding
			permutation factor: $+1$ for even permutations and the factor shown
			in the table for odd permutations.}
		\label{tab:E7_type1_conventions}
	\end{table}
	
	Then let us consider the second type. We give details on $\ell=4$ and $\bar\ell=1,7$, and the results for the other case can be obtained by swapping $\ell$ and $\bar\ell$.
	
	Consider $\alpha_{(8,0)(4,1)(4,1)}$. Using \eqref{alpha:constr_Z2chiral2} we get
	\begin{equation}
		\begin{split}
			\left[\alpha_{(8,0)(4,1)(4,1)}\right]^2=(-1)^{f_{18,0,w}(8,4,4)}=\begin{cases}
				1 & w=1,5\,, \\
				-1 & w=7\,. \\
			\end{cases}\,.
		\end{split}
	\end{equation}
	Using the gauge redundancy of the chiral $\mathbb{Z}_2$ operator, $(8,0)$, we set
	\begin{equation}
		\begin{split}
			\alpha_{(8,0)(4,1)(4,1)}=i^{f_{18,0,w}(8,4,4)}
		\end{split}
	\end{equation}
	
	The same analysis works when we swap the chiral and antichiral indices. We set
	\begin{equation}
		\begin{split}
			\alpha_{(0,8)(1,4)(1,4)}=i^{f_{18,0,w}(8,4,4)}\,.
		\end{split}
	\end{equation}
	
	The signs of all three $\mathbb{Z}_2$ operators are now fixed by $\alpha_{(8,0)(4,1)(4,1)}$, $\alpha_{(0,8)(1,4)(1,4)}$ and $\alpha_{(8,8)(0,8)(8,0)}$. No simultaneous sign flip involving these operators preserves the chosen conventions.
	
	For $\alpha_{(8,0)(4,7)(4,7)}$, we apply \eqref{alpha:constr_Z2antichiral1} to $\alpha_{(4,7)(4,7)(8,0)}$, which gives
	\begin{equation}
		\begin{split}
			\frac{
				\alpha_{(4,7)(4,7)(8,0)}
			}{
				\alpha_{(4,1)(4,1)(8,0)}
			}
			&=
			\frac{
				\alpha_{(0,8)(4,7)(4,1)}
			}{
				\alpha_{(0,8)(4,7)(4,1)}
			}
			(-1)^{f_{18,0,w}(7,0,7)}=1\,, \\
		\end{split}
	\end{equation}
	thus
	\begin{equation}
		\begin{split}
			\alpha_{(4,7)(4,7)(8,0)}=	\alpha_{(4,1)(4,1)(8,0)}=i^{f_{18,0,w}(8,4,4)}\,,
		\end{split}
	\end{equation}
	and similarly
	\begin{equation}
		\begin{split}
			\alpha_{(0,8)(7,4)(7,4)}=i^{f_{18,0,w}(8,4,4)}
		\end{split}
	\end{equation}
	
	Then let us look at $\alpha_{(0,8)(4,1)(4,7)}$. Using \eqref{alpha:constr_Z2antichiral2} we have
	\begin{equation}
		\begin{split}
			\left[\alpha_{(0,8)(4,1)(4,7)}\right]^2=(-1)^{1+f_{18,0,w}(8,7,1)}\,.
		\end{split}
	\end{equation}
	Using the gauge redundancy of either $(4,1)$ or $(4,7)$, we set
	\begin{equation}
		\begin{split}
			\alpha_{(0,8)(4,1)(4,7)}=i^{1+f_{18,0,w}(8,7,1)}\,.
		\end{split}
	\end{equation}
	
	The last element is $\alpha_{(8,8)(4,1)(4,7)}$. For this we use \eqref{alpha:constr_Z2chiral1} on $\alpha_{(0,8)(4,1)(4,7)}$, which gives
	\begin{equation}
		\begin{split}
			\frac{
				\alpha_{(0,8)(4,1)(4,7)}
			}{
				\alpha_{(8,8)(4,1)(4,7)}
			}
			&=
			\frac{
				\alpha_{(8,0)(4,1)(4,1)}
			}{
				\alpha_{(8,0)(0,8)(8,8)}
			}\,(-1)^{f_{18,k,w}(4,4,0)}\,.
		\end{split}
	\end{equation}
	So we get
	\begin{equation}
		\begin{split}
			\alpha_{(8,8)(4,1)(4,7)}=\frac{
				\alpha_{(0,8)(4,1)(4,7)}
			}{
				\alpha_{(8,0)(4,1)(4,1)}
			}(-1)^{f_{18,k,w}(4,4,0)}
			=i^{1+f_{18,0,w}(8,7,1)-f_{18,0,w}(8,4,4)+2f_{18,k,w}(4,4,0)}
		\end{split}
	\end{equation}
	By the same analysis, we have for $(1,4)$ and $(7,4)$:
	\begin{equation}
		\begin{split}
			\alpha_{(8,0)(1,4)(7,4)}&=i^{1+f_{18,0,w}(8,7,1)}\,, \\
			\alpha_{(8,8)(1,4)(7,4)}
			&=i^{1+f_{18,0,w}(8,7,1)-f_{18,0,w}(8,4,4)+2f_{18,k,w}(4,4,0)}\,. \\
		\end{split}
	\end{equation}
	
	Table \ref{tab:E7_type2_conventions} summarizes the conventions for the second-type orbits.
	\begin{table}[ht]
		\centering
		\renewcommand{\arraystretch}{1.35}
		\begin{tabular}{c c c c}
			\toprule
			Case & Triplet (ordered) & $\alpha$ & Permutation \\
			\midrule
			
			\multirow{4}{*}{$
				\substack{(4,1)\\  \\ (4,7)}
				$}
			& $(8,0)(4,1)(4,1)$
			& $i^{f_{18,0,w}(8,4,4)}$
			& $+1$
			\\
			
			& $(8,0)(4,7)(4,7)$
			& $i^{f_{18,0,w}(8,4,4)}$
			& $+1$
			\\
			
			& $(0,8)(4,1)(4,7)$
			& $i^{1+f_{18,0,w}(8,7,1)}$
			& $-1$
			\\
			
			& $(8,8)(4,1)(4,7)$
			& $i^{
				1+f_{18,0,w}(8,7,1)
				-f_{18,0,w}(8,4,4)+2f_{18,k,w}(4,4,0)
			}$
			& $-1$
			\\
			
			\midrule
			
			\multirow{4}{*}{$
				\substack{(1,4)\\  \\ (7,4)}
				$}
			& $(0,8)(1,4)(1,4)$
			& $i^{f_{18,0,w}(8,4,4)}$
			& $+1$
			\\
			
			& $(0,8)(7,4)(7,4)$
			& $i^{f_{18,0,w}(8,4,4)}$
			& $+1$
			\\
			
			& $(8,0)(1,4)(7,4)$
			& $i^{1+f_{18,0,w}(8,7,1)}$
			& $-1$
			\\
			
			& $(8,8)(1,4)(7,4)$
			& $i^{
				1+f_{18,0,w}(8,7,1)
				-f_{18,0,w}(8,4,4)+2f_{18,k,w}(4,4,0)
			}$
			& $-1$
			\\
			
			\bottomrule
		\end{tabular}
		\caption{Reduced-coefficient conventions and permutation factors for the
			second-type orbits of the $E_7$ series. The $\alpha$ coefficients
			for other orderings are obtained by multiplying by the corresponding
			permutation factor: $+1$ for even permutations and the factor shown
			in the table for odd permutations.}
		\label{tab:E7_type2_conventions}
	\end{table}
	
	For the third-type orbit, $(4,4)$, the relevant coefficients are
	$$\alpha_{(0,8)(4,4)(4,4)},\quad\alpha_{(8,0)(4,4)(4,4)},\quad\alpha_{(8,8)(4,4)(4,4)}\,.$$
	Each contains two copies of $(4,4)$, so a sign change of this field cannot alter the coefficient. The bootstrap equations determine these coefficients. To see this we consider the four-point function
	\begin{equation}
		\begin{split}
			\braket{(4,4)(4,1)(4,1)(4,4)}
		\end{split}
	\end{equation}
	whose bootstrap equation leads to
	\begin{equation}
		\begin{split}
			\alpha_{(8,0)(4,4)(4,4)}\alpha_{(8,0)(4,1)(4,1)}&=\begin{cases}
				1 & w=7\,, \\
				-1 & w=1,5\,, \\
			\end{cases} \\
			&=(-1)^{1+f_{18,0,w}(8,4,4)}\,. \\
		\end{split}
	\end{equation}
	Then we get
	\begin{equation}
		\begin{split}
			\alpha_{(8,0)(4,4)(4,4)}=-i^{f_{18,0,w}(8,4,4)}\,.
		\end{split}
	\end{equation}
	By a similar analysis we get
	\begin{equation}
		\begin{split}
			\alpha_{(0,8)(4,4)(4,4)}=-i^{f_{18,0,w}(8,4,4)}\,.
		\end{split}
	\end{equation}
	Applying \eqref{alpha:constr_Z2antichiral1} to $\alpha_{(8,0)(4,4)(4,4)}$, we get
	\begin{equation}
		\begin{split}
			\frac{
				\alpha_{(8,0)(4,4)(4,4)}
			}{
				\alpha_{(8,8)(4,4)(4,4)}
			}
			&=
			\frac{
				\alpha_{(0,8)(4,4)(4,4)}
			}{
				\alpha_{(0,8)(8,0)(8,8)}
			}
			(-1)^{f_{18,k,w}(4,4,0)}\,. \\
		\end{split}
	\end{equation}
	Therefore
	\begin{equation}
		\begin{split}
			\alpha_{(8,8)(4,4)(4,4)}=(-1)^{f_{18,0,w}(4,4,0)}\,.
		\end{split}
	\end{equation}
	One can check that
	\begin{equation}
		\begin{split}
			f_{18,0,w}(4,4,0)\equiv f_{18,0,w}(8,4,4)\mod{4}\,.
		\end{split}
	\end{equation}
	Table \ref{tab:E7_type3_conventions} summarizes the conventions for the third-type orbit.
	\begin{table}[htbp]
		\centering
		\renewcommand{\arraystretch}{1.35}
		\begin{tabular}{c c c c}
			\toprule
			Case & Triplet (ordered) & $\alpha$ & Permutation \\
			\midrule
			
			\multirow{3}{*}{$(4,4)$}
			& $(8,0)(4,4)(4,4)$
			& $-i^{f_{18,0,w}(4,4,0)}$
			& $+1$
			\\
			
			& $(0,8)(4,4)(4,4)$
			& $-i^{f_{18,0,w}(4,4,0)}$
			& $+1$
			\\
			
			& $(8,8)(4,4)(4,4)$
			& $(-1)^{f_{18,0,w}(4,4,0)}$
			& $+1$
			\\
			
			\bottomrule
		\end{tabular}
		\caption{Reduced-coefficient conventions and permutation factors for the
			third-type orbit of the $E_7$ series. The $\alpha$ coefficients for
			other orderings are obtained by multiplying by the corresponding
			permutation factor. In this orbit, all permutation factors are $+1$.}
		\label{tab:E7_type3_conventions}
	\end{table}
	
	Now we are ready to apply \eqref{alpha:constr_Z2diagonal1}, \eqref{alpha:constr_Z2chiral1} and \eqref{alpha:constr_Z2antichiral1} for the generic triplets 
	$$(\ell_1,\bar\ell_1)(\ell_2,\bar\ell_2)(\ell_3,\bar\ell_3)\,.$$
	The constraints are given by
	\begin{equation}\label{alpha:Z2_E7}
		\boxed{\begin{aligned}
				\frac{
					\alpha_{(\ell_1,\bar\ell_1)(\ell_2,\bar\ell_2)(\ell_3,\bar\ell_3)}
				}{
					\alpha_{(8-\ell_1,8-\bar\ell_1)(8-\ell_2,8-\bar\ell_2)(\ell_3,\bar\ell_3)}
				}
				&=
				\frac{
					\alpha_{(8,8)(\ell_2,\bar\ell_2)(8-\ell_2,8-\bar\ell_2)}
				}{
					\alpha_{(8,8)(\ell_1,\bar\ell_1)(8-\ell_1,8-\bar\ell_1)}
				}
				(-1)^{\ell_{231}+\bar\ell_{231}} \\
				&\quad\times(-1)^{f_{18,k,w}(\ell_2,\ell_3,\ell_1)+f_{18,k,w}(\bar{\ell}_2,\bar{\ell}_3,\bar{\ell}_1)} \\
				\frac{
					\alpha_{(\ell_1,\bar\ell_1)(\ell_2,\bar\ell_2)(\ell_3,\bar\ell_3)}
				}{
					\alpha_{(8-\ell_1,\bar\ell_1)(8-\ell_2,\bar\ell_2)(\ell_3,\bar\ell_3)}
				}
				&=
				\frac{
					\alpha_{(8,0)(\ell_2,\bar\ell_2)(8-\ell_2,\bar\ell_2)}
				}{
					\alpha_{(8,0)(\ell_1,\bar\ell_1)(8-\ell_1,\bar\ell_1)}
				}
				(-1)^{\ell_{231}+f_{18,k,w}(\ell_2,\ell_3,\ell_1)} \\
				\frac{
					\alpha_{(\ell_1,\bar\ell_1)(\ell_2,\bar\ell_2)(\ell_3,\bar\ell_3)}
				}{
					\alpha_{(\ell_1,8-\bar\ell_1)(\ell_2,8-\bar\ell_2)(\ell_3,\bar\ell_3)}
				}
				&=
				\frac{
					\alpha_{(0,8)(\ell_2,\bar\ell_2)(\ell_2,8-\bar\ell_2)}
				}{
					\alpha_{(0,8)(\ell_1,\bar\ell_1)(\ell_1,8-\bar\ell_1)}
				}
				(-1)^{\bar\ell_{231}+f_{18,k,w}(\bar\ell_2,\bar\ell_3,\bar\ell_1)} \\
		\end{aligned}}
	\end{equation}
	We would like to check the consistency condition. By combining the second and the third equations we have
	\begin{equation}
		\begin{split}
			&\frac{
				\alpha_{(\ell_1,\bar\ell_1)(\ell_2,\bar\ell_2)(\ell_3,\bar\ell_3)}
			}{
				\alpha_{(8-\ell_1,8-\bar\ell_1)(8-\ell_2,8-\bar\ell_2)(\ell_3,\bar\ell_3)}
			}\\
			&=	\frac{
				\alpha_{(8,0)(\ell_2,\bar\ell_2)(8-\ell_2,\bar\ell_2)}
			}{
				\alpha_{(8,0)(\ell_1,\bar\ell_1)(8-\ell_1,\bar\ell_1)}
			}
			(-1)^{\ell_{231}+f_{18,k,w}(\ell_2,\ell_3,\ell_1)} \\
			&\times\frac{
				\alpha_{(0,8)(8-\ell_2,\bar\ell_2)(8-\ell_2,8-\bar\ell_2)}
			}{
				\alpha_{(0,8)(8-\ell_1,\bar\ell_1)(8-\ell_1,8-\bar\ell_1)}
			}
			(-1)^{\bar\ell_{231}+f_{18,k,w}(\bar\ell_2,\bar\ell_3,\bar\ell_1)} \\
		\end{split}
	\end{equation}
	Consistency therefore requires
	\begin{equation}
		\begin{split}
			\frac{
				\alpha_{(8,0)(\ell_2,\bar\ell_2)(8-\ell_2,\bar\ell_2)}
			}{
				\alpha_{(8,0)(\ell_1,\bar\ell_1)(8-\ell_1,\bar\ell_1)}
			}\frac{
				\alpha_{(0,8)(8-\ell_2,\bar\ell_2)(8-\ell_2,8-\bar\ell_2)}
			}{
				\alpha_{(0,8)(8-\ell_1,\bar\ell_1)(8-\ell_1,8-\bar\ell_1)}
			}=\frac{
				\alpha_{(8,8)(\ell_2,\bar\ell_2)(8-\ell_2,8-\bar\ell_2)}
			}{
				\alpha_{(8,8)(\ell_1,\bar\ell_1)(8-\ell_1,8-\bar\ell_1)}
			}\,.
		\end{split}
	\end{equation}
	This condition holds provided that 
	$$
	\frac{\alpha_{(8,0)(\ell,\bar\ell)(8-\ell,\bar\ell)}\alpha_{(0,8)(8-\ell,\bar\ell)(8-\ell,8-\bar\ell)}}{\alpha_{(8,8)(\ell,\bar\ell)(8-\ell,8-\bar\ell)}}
	$$
	is independent of the operator label $(\ell,\bar\ell)$. To see this, we apply the relation
	\begin{equation}
		\begin{split}
			\frac{
				\alpha_{(\ell_1,\bar\ell_1)(\ell_2,\bar\ell_2)(\ell_3,\bar\ell_3)}
			}{
				\alpha_{(\ell_1,\ell_{\mathbb{Z}_2}-\bar\ell_1)(\ell_2,\bar\ell_2)(\ell_3,\ell_{\mathbb{Z}_2}-\bar\ell_3)}
			}
			&=
			\frac{
				\alpha_{(0,\ell_{\mathbb{Z}_2})(\ell_1,\bar\ell_1)(\ell_1,\ell_{\mathbb{Z}_2}-\bar\ell_1)}
			}{
				\alpha_{(0,\ell_{\mathbb{Z}_2})(\ell_3,\bar\ell_3)(\ell_3,\ell_{\mathbb{Z}_2}-\bar\ell_3)}
			}
			(-1)^{\bar\ell_{123}} \\
			&\quad\times(-1)^{f_{q,k,w}(\bar{\ell}_1,\bar{\ell}_2,\bar{\ell}_3)}
		\end{split}
	\end{equation}
	on $\alpha_{(8,0)(\ell,\bar\ell)(8-\ell,\bar\ell)}$, which gives
	\begin{equation}
		\begin{split}
			\frac{\alpha_{(8,0)(\ell,\bar\ell)(8-\ell,\bar\ell)}\alpha_{(0,8)(8-\ell,\bar\ell)(8-\ell,8-\bar\ell)}}{\alpha_{(8,8)(\ell,\bar\ell)(8-\ell,8-\bar\ell)}}&=
			\frac{
				\alpha_{(0,8)(8,0)(8,8)}\alpha_{(0,8)(8-\ell,\bar\ell)(8-\ell,8-\bar\ell)}
			}{
				\alpha_{(0,8)(8-\ell,\bar\ell)(8-\ell,8-\bar\ell)}
			}\\
			&\quad\times(-1)^{f_{18,k,w}(0,\bar{\ell},\bar{\ell})} \\
			&=1\,. \\
		\end{split}
	\end{equation}
	This verifies consistency.

	\subsection{\texorpdfstring{Convention for $E_8$ series}{Convention for E8 series}}
	Like the $E_7$ series, the $E_8$ series has three $\mathbb{Z}_2$ operators labeled by $(\ell,\bar\ell,s)=(0,14,1)$, $(14,0,1)$ and $(14,14,1)$. However, the $E_8$ series has only orbits of the first type.
	\begin{figure}[htbp]
		\centering
		\begin{tikzpicture}
			\node (00) at (0,0) {$(\ell,\bar{\ell})$};
			\node (08) at (3,2) {$(\ell,14-\bar{\ell})$};
			\node (80) at (3,-2) {$(14-\ell,\bar{\ell})$};
			\node (88) at (6,-0) {$(14-\ell,14-\bar{\ell})$};
			
			\draw (00) -- node[midway, above left] {$(0,14)$} (08);
			\draw (00) -- node[midway, below left] {$(14,0)$} (80);
			\draw (08) -- node[midway, above right] {$(14,0)$} (88);
			\draw (80) -- node[midway, below right] {$(0,14)$} (88);
			\draw (00) -- node[midway, above] {$(14,14)\quad$}(88);
			\draw (80) -- node[midway, below right] {$(14,14)$}(08);
		\end{tikzpicture}
		\caption{Orbits of the first type.}
		\label{fig:E8}
	\end{figure}
	
	The treatment parallels that of the $E_7$ series. We use the spin content of the $E_8$ series:
	\begin{equation}
		\begin{split}
			S_{\ell,\bar{\ell},s}=\frac{(\ell-\bar{\ell})(p(\ell+\bar{\ell}+1)-30s)}{30}=\begin{cases}
				0 & \text{for }(\ell,\ell,s)\,, \\
				5s-p\equiv0\mod{2} & \text{for }(0,5,s)\,, \\
				9s-3p \equiv0\mod{2}& \text{for }(0,9,s)\,, \\
				14s-7p\equiv1\mod{2} & \text{for }(0,14,s)\,, \\
				4s-2p\equiv0\mod{2} & \text{for }(5,9,s)\,, \\
				9s-6p \equiv1\mod{2}& \text{for }(5,14,s)\,, \\
				5s-4p \equiv1\mod{2}& \text{for }(9,14,s)\,, \\
				3s-p\equiv0\mod{2} & \text{for }(3,6,s)\,, \\
				5s-2p \equiv1\mod{2}& \text{for }(3,8,s)\,, \\
				8s-4p\equiv0\mod{2} & \text{for }(3,11,s)\,, \\
				2s-p \equiv1\mod{2}& \text{for }(6,8,s)\,, \\
				5s-3p\equiv0\mod{2} & \text{for }(6,11,s)\,, \\
				3s-2p \equiv1\mod{2}& \text{for }(8,11,s)\,. \\
			\end{cases}
		\end{split}
	\end{equation}
	We first consider $\alpha$'s with an external $\mathbb{Z}_2$ operator. To avoid overcounting, we choose the reference $(\ell,\bar{\ell})$ to be
	\begin{equation}
		\begin{split}
			(\ell,\bar\ell)=(0,0),(5,5),(0,5),(5,0),(3,3),(6,6),(3,6),(6,3)\,.
		\end{split}
	\end{equation}
	They generate all the operators under $\mathbb{Z}_2\times\mathbb{Z}_2$ actions.
	
	For the chiral $\mathbb{Z}_2$ operator, we use \eqref{alpha:constr_Z2chiral2}:
	\begin{equation}
		\begin{split}
			\left[\alpha_{(14,0)(\ell,\bar\ell)(14-\ell,\bar\ell)}\right]^2&=(-1)^{S_{14,0}+S_{\ell,\bar{\ell}}+S_{14-\ell,\bar{\ell}}+f_{30,k,w}(14,14-\ell,\ell)} \\
			&=\begin{cases}
				(-1)^{f_{30,0,w}(14,14-\ell,\ell)} & (\ell,\bar\ell)=(0,0),(0,5),(6,6),(6,3) \\
				(-1)^{1+f_{30,0,w}(14,14-\ell,\ell)} & (\ell,\bar\ell)=(5,5),(5,0),(3,3),(3,6) \\
			\end{cases}
		\end{split}
	\end{equation}
	Using the gauge redundancy of either $(\ell,\bar\ell)$ or $(14-\ell,\bar\ell)$, we set\footnote{For each reference residue $w=1,7,11,13$, we have $f_{30,0,w}(14,14,0)\equiv0\pmod4$. Thus $\alpha_{(14,0)(0,0)(14,0)}=1$, as required by the identity normalization. The same observation applies to the antichiral simple current.}
	\begin{equation}
		\begin{split}
			\alpha_{(14,0)(\ell,\bar\ell)(14-\ell,\bar\ell)}
			&=\begin{cases}
				i^{f_{30,0,w}(14,14-\ell,\ell)} & (\ell,\bar\ell)=(0,0),(0,5),(6,6),(6,3) \\
				i^{1+f_{30,0,w}(14,14-\ell,\ell)} & (\ell,\bar\ell)=(5,5),(5,0),(3,3),(3,6) \\
			\end{cases}
		\end{split}
	\end{equation}
	Similarly, for the antichiral $\mathbb{Z}_2$ operator, we use \eqref{alpha:constr_Z2antichiral2} and set
	\begin{equation}
		\begin{split}
			\alpha_{(0,14)(\ell,\bar\ell)(\ell,14-\bar\ell)}
			&=\begin{cases}
				i^{f_{30,0,w}(14,14-\bar\ell,\bar\ell)} & (\ell,\bar\ell)=(0,0),(5,0),(6,6),(3,6) \\
				i^{1+f_{30,0,w}(14,14-\bar\ell,\bar\ell)} & (\ell,\bar\ell)=(5,5),(0,5),(3,3),(6,3) \\
			\end{cases}\,. \\
		\end{split}
	\end{equation}
	
	For the diagonal $\mathbb{Z}_2$ operator, we use \eqref{alpha:constr_Z2diagonal2}
	\begin{equation}
		\begin{split}
			&\left[\alpha_{(14,14)(\ell,\bar\ell)(14-\ell,14-\bar\ell)}\right]^2\\
			&=(-1)^{S_{\ell,\bar{\ell}}+S_{14-\ell,14-\bar{\ell}}+f_{30,k,w}(14,14-\ell,\ell)+f_{30,k,w}(14,14-\bar{\ell},\bar{\ell})} \\
			&=\begin{cases}
				1 & (\ell,\bar\ell)=(0,0),(5,5),(3,3),(6,6) \\
				(-1)^{1+f_{30,k,w}(14,14-\ell,\ell)+f_{30,k,w}(14,14-\bar{\ell},\bar{\ell})} & (\ell,\bar\ell)=(0,5),(5,0),(3,6),(6,3) \\
			\end{cases}
		\end{split}
	\end{equation}
	Using the gauge redundancy of $(14-\ell,14-\bar{\ell})$, we set
	\begin{equation}
		\begin{split}
			&\alpha_{(14,14)(\ell,\bar\ell)(14-\ell,14-\bar\ell)} \\
			&=\begin{cases}
				1 & (\ell,\bar\ell)=(0,0),(5,5),(3,3),(6,6) \\
				i^{1+f_{30,k,w}(14,14-\ell,\ell)+f_{30,k,w}(14,14-\bar{\ell},\bar{\ell})} & (\ell,\bar\ell)=(0,5),(5,0),(3,6),(6,3) \\
			\end{cases}
		\end{split}
	\end{equation}
	
	The signs of the $\mathbb{Z}_2$ operators remain available. To constrain this freedom, we consider the identity orbit, with $s=1$, which contains
	$$(\ell,\bar\ell)=(0,0),(14,0),(0,14),(14,14)\,.$$
	The non-trivial triplet is $[(14,0),(0,14),(14,14)]$, for which we have
	\begin{equation}
		\begin{split}
			\left[\alpha_{(14,0)(0,14)(14,14)}\right]^2=(-1)^{f_{30,k,w}(14,0,14)}=1\,.
		\end{split}
	\end{equation}
	So we can set
	\begin{equation}
		\begin{split}
			\alpha_{(14,0)(0,14)(14,14)}=1\,.
		\end{split}
	\end{equation}
	Among the $\mathbb{Z}_2$ operators, this leaves only simultaneous flips of two signs, accompanied by the operator-sign changes needed to preserve the preceding conventions.
	
	For $\alpha_{(0,14)(14-\ell,\bar\ell)(14-\ell,14-\bar\ell)}$, we use \eqref{alpha:constr_Z2chiral1}:
	\begin{equation}
		\begin{split}
			\frac{
				\alpha_{(0,14)(14-\ell,\bar\ell)(14-\ell,14-\bar\ell)}
			}{
				\alpha_{(14,14)(\ell,\bar\ell)(14-\ell,14-\bar\ell)}
			}
			&=
			\frac{
				\alpha_{(14,0)(14-\ell,\bar\ell)(\ell,\bar\ell)}
			}{
				\alpha_{(14,0)(0,14)(14,14)}
			}
			(-1)^{f_{30,k,w}(14-\ell,14-\ell,0)} \\
		\end{split}
	\end{equation}
	which implies
	\begin{equation}
		\begin{split}
			\alpha_{(0,14)(14-\ell,\bar\ell)(14-\ell,14-\bar\ell)}&=\alpha_{(14,0)(14-\ell,\bar\ell)(\ell,\bar\ell)}\alpha_{(14,14)(\ell,\bar\ell)(14-\ell,14-\bar\ell)}(-1)^{f_{30,0,w}(14-\ell,14-\ell,0)}\,. \\
		\end{split}
	\end{equation}
	Similarly
	\begin{equation}
		\begin{split}
			\alpha_{(14,0)(\ell,14-\bar\ell)(14-\ell,14-\bar\ell)}&=\alpha_{(0,14)(\ell,14-\bar\ell)(\ell,\bar\ell)}\alpha_{(14,14)(\ell,\bar\ell)(14-\ell,14-\bar\ell)}(-1)^{f_{30,0,w}(14-\bar\ell,14-\bar\ell,0)}\,. \\
		\end{split}
	\end{equation}
	
	The last element is $\alpha_{(14,14)(14-\ell,\bar\ell)(\ell,14-\bar\ell)}$. We use \eqref{alpha:constr_Z2chiral1}:
	\begin{equation}
		\begin{split}
			\frac{
				\alpha_{(14,14)(14-\ell,\bar\ell)(\ell,14-\bar\ell)}
			}{
				\alpha_{(0,14)(\ell,\bar\ell)(\ell,14-\bar\ell)}
			}
			&=
			\frac{
				\alpha_{(14,0)(14-\ell,\bar\ell)(\ell,\bar\ell)}
			}{
				\alpha_{(14,0)(14,14)(0,14)}
			}
			(-1)^{f_{30,k,w}(14-\ell,\ell,14)} \\
		\end{split}
	\end{equation}
	which implies
	\begin{equation}
		\begin{split}
			\alpha_{(14,14)(14-\ell,\bar\ell)(\ell,14-\bar\ell)}=\alpha_{(0,14)(\ell,\bar\ell)(\ell,14-\bar\ell)}\alpha_{(14,0)(14-\ell,\bar\ell)(\ell,\bar\ell)}(-1)^{f_{30,0,w}(14-\ell,\ell,14)}
		\end{split}
	\end{equation}
	
	Tables~\ref{tab:E8_type1_conventions} and \ref{tab:E8_type1_conventions_2} summarize these conventions. For compactness, we use the notation
	\[
	\chi_w(\ell,\bar\ell):=
	1+f_{30,0,w}(14,14-\ell,\ell)
	+f_{30,0,w}(14,14-\bar\ell,\bar\ell).
	\]
	
	Now we are ready to apply \eqref{alpha:constr_Z2diagonal1}, \eqref{alpha:constr_Z2chiral1} and \eqref{alpha:constr_Z2antichiral1} for the generic triplets 
	$$(\ell_1,\bar\ell_1)(\ell_2,\bar\ell_2)(\ell_3,\bar\ell_3)\,.$$
	The constraints are given by
	\begin{equation}\label{alpha:Z2_E8}
		\boxed{\begin{aligned}
				\frac{
					\alpha_{(\ell_1,\bar\ell_1)(\ell_2,\bar\ell_2)(\ell_3,\bar\ell_3)}
				}{
					\alpha_{(14-\ell_1,14-\bar\ell_1)(14-\ell_2,14-\bar\ell_2)(\ell_3,\bar\ell_3)}
				}
				&=
				\frac{
					\alpha_{(14,14)(\ell_2,\bar\ell_2)(14-\ell_2,14-\bar\ell_2)}
				}{
					\alpha_{(14,14)(\ell_1,\bar\ell_1)(14-\ell_1,14-\bar\ell_1)}
				}
				(-1)^{\ell_{231}+\bar\ell_{231}} \\
				&\quad\times(-1)^{f_{30,k,w}(\ell_2,\ell_3,\ell_1)+f_{30,k,w}(\bar{\ell}_2,\bar{\ell}_3,\bar{\ell}_1)} \\
				\frac{
					\alpha_{(\ell_1,\bar\ell_1)(\ell_2,\bar\ell_2)(\ell_3,\bar\ell_3)}
				}{
					\alpha_{(14-\ell_1,\bar\ell_1)(14-\ell_2,\bar\ell_2)(\ell_3,\bar\ell_3)}
				}
				&=
				\frac{
					\alpha_{(14,0)(\ell_2,\bar\ell_2)(14-\ell_2,\bar\ell_2)}
				}{
					\alpha_{(14,0)(\ell_1,\bar\ell_1)(14-\ell_1,\bar\ell_1)}
				}
				(-1)^{\ell_{231}+f_{30,k,w}(\ell_2,\ell_3,\ell_1)} \\
				\frac{
					\alpha_{(\ell_1,\bar\ell_1)(\ell_2,\bar\ell_2)(\ell_3,\bar\ell_3)}
				}{
					\alpha_{(\ell_1,14-\bar\ell_1)(\ell_2,14-\bar\ell_2)(\ell_3,\bar\ell_3)}
				}
				&=
				\frac{
					\alpha_{(0,14)(\ell_2,\bar\ell_2)(\ell_2,14-\bar\ell_2)}
				}{
					\alpha_{(0,14)(\ell_1,\bar\ell_1)(\ell_1,14-\bar\ell_1)}
				}
				(-1)^{\bar\ell_{231}+f_{30,k,w}(\bar\ell_2,\bar\ell_3,\bar\ell_1)} \\
		\end{aligned}}
	\end{equation}
	
	\begin{table}[htbp]
		\centering
		\renewcommand{\arraystretch}{1.35}
		\resizebox{\textwidth}{!}{%
			\begin{tabular}{c c c c}
				\toprule
				Case & Triplet (ordered) & $\alpha$ & Permutation \\
				\midrule
				\multirow{6}{*}{$(0,0)$}
				& $(0,14)(0,0)(0,14)$ & $1$ & $+1$ \\
				& $(14,0)(0,0)(14,0)$ & $1$ & $+1$ \\
				& $(14,14)(0,0)(14,14)$ & $1$ & $+1$ \\
				& $(0,14)(14,0)(14,14)$ & $1$ & $+1$ \\
				& $(14,0)(0,14)(14,14)$ & $1$ & $+1$ \\
				& $(14,14)(14,0)(0,14)$ & $1$ & $+1$ \\
				\midrule
				
				\multirow{6}{*}{$(5,5)$}
				& $(0,14)(5,5)(5,9)$
				& $i^{1+f_{30,0,w}(14,9,5)}$ & $-1$ \\
				& $(14,0)(5,5)(9,5)$
				& $i^{1+f_{30,0,w}(14,9,5)}$ & $-1$ \\
				& $(14,14)(5,5)(9,9)$
				& $1$ & $+1$ \\
				& $(0,14)(9,5)(9,9)$
				& $-i^{1+f_{30,0,w}(14,9,5)}(-1)^{f_{30,0,w}(9,9,0)}$ & $-1$ \\
				& $(14,0)(5,9)(9,9)$
				& $-i^{1+f_{30,0,w}(14,9,5)}(-1)^{f_{30,0,w}(9,9,0)}$ & $-1$ \\
				& $(14,14)(9,5)(5,9)$
				& $-i^{2+2f_{30,0,w}(14,9,5)}(-1)^{f_{30,0,w}(9,5,14)}$ & $+1$ \\
				\midrule
				
				\multirow{6}{*}{$(0,5)$}
				& $(0,14)(0,5)(0,9)$
				& $i^{1+f_{30,0,w}(14,9,5)}$ & $-1$ \\
				& $(14,0)(0,5)(14,5)$
				& $i^{f_{30,0,w}(14,14,0)}$ & $+1$ \\
				& $(14,14)(0,5)(14,9)$
				& $i^{\chi_w(0,5)}$ & $-1$ \\
				& $(0,14)(14,5)(14,9)$
				& $i^{f_{30,0,w}(14,14,0)+\chi_w(0,5)}(-1)^{f_{30,0,w}(14,14,0)}$ & $-1$ \\
				& $(14,0)(0,9)(14,9)$
				& $-i^{1+f_{30,0,w}(14,9,5)+\chi_w(0,5)}(-1)^{f_{30,0,w}(9,9,0)}$ & $+1$ \\
				& $(14,14)(14,5)(0,9)$
				& $i^{1+f_{30,0,w}(14,14,0)+f_{30,0,w}(14,9,5)}(-1)^{f_{30,0,w}(14,0,14)}$ & $-1$ \\
				\midrule
				
				\multirow{6}{*}{$(5,0)$}
				& $(0,14)(5,0)(5,14)$
				& $i^{f_{30,0,w}(14,14,0)}$ & $+1$ \\
				& $(14,0)(5,0)(9,0)$
				& $i^{1+f_{30,0,w}(14,9,5)}$ & $-1$ \\
				& $(14,14)(5,0)(9,14)$
				& $i^{\chi_w(5,0)}$ & $-1$ \\
				& $(0,14)(9,0)(9,14)$
				& $-i^{1+f_{30,0,w}(14,9,5)+\chi_w(5,0)}(-1)^{f_{30,0,w}(9,9,0)}$ & $+1$ \\
				& $(14,0)(5,14)(9,14)$
				& $i^{f_{30,0,w}(14,14,0)+\chi_w(5,0)}(-1)^{f_{30,0,w}(14,14,0)}$ & $-1$ \\
				& $(14,14)(9,0)(5,14)$
				& $-i^{1+f_{30,0,w}(14,9,5)+f_{30,0,w}(14,14,0)}(-1)^{f_{30,0,w}(9,5,14)}$ & $-1$ \\
				\bottomrule
		\end{tabular}}
		\caption{Reduced-coefficient conventions and permutation factors for the first-type orbits of the $E_8$ series, part one. We use the notation $\chi_w(\ell,\bar\ell):=
			1+f_{30,0,w}(14,14-\ell,\ell)
			+f_{30,0,w}(14,14-\bar\ell,\bar\ell)$.}
		\label{tab:E8_type1_conventions}
	\end{table}
	\begin{table}[htbp]
		\centering
		\renewcommand{\arraystretch}{1.35}
		\resizebox{\textwidth}{!}{%
			\begin{tabular}{c c c c}
				\toprule
				Case & Triplet (ordered) & $\alpha$ & Permutation \\
				\midrule
				\multirow{6}{*}{$(3,3)$}
				& $(0,14)(3,3)(3,11)$
				& $i^{1+f_{30,0,w}(14,11,3)}$ & $-1$ \\
				& $(14,0)(3,3)(11,3)$
				& $i^{1+f_{30,0,w}(14,11,3)}$ & $-1$ \\
				& $(14,14)(3,3)(11,11)$
				& $1$ & $+1$ \\
				& $(0,14)(11,3)(11,11)$
				& $-i^{1+f_{30,0,w}(14,11,3)}(-1)^{f_{30,0,w}(11,11,0)}$ & $-1$ \\
				& $(14,0)(3,11)(11,11)$
				& $-i^{1+f_{30,0,w}(14,11,3)}(-1)^{f_{30,0,w}(11,11,0)}$ & $-1$ \\
				& $(14,14)(11,3)(3,11)$
				& $-i^{2+2f_{30,0,w}(14,11,3)}(-1)^{f_{30,0,w}(11,3,14)}$ & $+1$ \\
				\midrule
				
				\multirow{6}{*}{$(6,6)$}
				& $(0,14)(6,6)(6,8)$
				& $i^{f_{30,0,w}(14,8,6)}$ & $+1$ \\
				& $(14,0)(6,6)(8,6)$
				& $i^{f_{30,0,w}(14,8,6)}$ & $+1$ \\
				& $(14,14)(6,6)(8,8)$
				& $1$ & $+1$ \\
				& $(0,14)(8,6)(8,8)$
				& $i^{f_{30,0,w}(14,8,6)}(-1)^{f_{30,0,w}(8,8,0)}$ & $+1$ \\
				& $(14,0)(6,8)(8,8)$
				& $i^{f_{30,0,w}(14,8,6)}(-1)^{f_{30,0,w}(8,8,0)}$ & $+1$ \\
				& $(14,14)(8,6)(6,8)$
				& $i^{2f_{30,0,w}(14,8,6)}(-1)^{f_{30,0,w}(8,6,14)}$ & $+1$ \\
				\midrule
				
				\multirow{6}{*}{$(3,6)$}
				& $(0,14)(3,6)(3,8)$
				& $i^{f_{30,0,w}(14,8,6)}$ & $+1$ \\
				& $(14,0)(3,6)(11,6)$
				& $i^{1+f_{30,0,w}(14,11,3)}$ & $-1$ \\
				& $(14,14)(3,6)(11,8)$
				& $i^{\chi_w(3,6)}$ & $-1$ \\
				& $(0,14)(11,6)(11,8)$
				& $-i^{1+f_{30,0,w}(14,11,3)+\chi_w(3,6)}(-1)^{f_{30,0,w}(11,11,0)}$ & $+1$ \\
				& $(14,0)(3,8)(11,8)$
				& $i^{f_{30,0,w}(14,8,6)+\chi_w(3,6)}(-1)^{f_{30,0,w}(8,8,0)}$ & $-1$ \\
				& $(14,14)(11,6)(3,8)$
				& $-i^{1+f_{30,0,w}(14,11,3)+f_{30,0,w}(14,8,6)}(-1)^{f_{30,0,w}(11,3,14)}$ & $-1$ \\
				\midrule
				
				\multirow{6}{*}{$(6,3)$}
				& $(0,14)(6,3)(6,11)$
				& $i^{1+f_{30,0,w}(14,11,3)}$ & $-1$ \\
				& $(14,0)(6,3)(8,3)$
				& $i^{f_{30,0,w}(14,8,6)}$ & $+1$ \\
				& $(14,14)(6,3)(8,11)$
				& $i^{\chi_w(6,3)}$ & $-1$ \\
				& $(0,14)(8,3)(8,11)$
				& $i^{f_{30,0,w}(14,8,6)+\chi_w(6,3)}(-1)^{f_{30,0,w}(8,8,0)}$ & $-1$ \\
				& $(14,0)(6,11)(8,11)$
				& $-i^{1+f_{30,0,w}(14,11,3)+\chi_w(6,3)}(-1)^{f_{30,0,w}(11,11,0)}$ & $+1$ \\
				& $(14,14)(8,3)(6,11)$
				& $i^{1+f_{30,0,w}(14,8,6)+f_{30,0,w}(14,11,3)}(-1)^{f_{30,0,w}(8,6,14)}$ & $-1$ \\
				\bottomrule
		\end{tabular}}
		\caption{Reduced-coefficient conventions and permutation factors for the first-type orbits of the $E_8$ series, part two. The $\alpha$ coefficients for other orderings are obtained by multiplying by the corresponding permutation factor: $+1$ for even permutations and the factor shown in the table for odd permutations. We use the notation $\chi_w(\ell,\bar\ell):=
			1+f_{30,0,w}(14,14-\ell,\ell)
			+f_{30,0,w}(14,14-\bar\ell,\bar\ell)$.}
		\label{tab:E8_type1_conventions_2}
	\end{table}

\newpage
	\section{Comparison with previous results}\label{app:literature-comparison}

This appendix compares our OPE coefficients with the earlier calculations discussed in the introduction. The comparison requires matching the field labels, two-point normalizations, and the chiral factors divided out in defining the reduced coefficients. In particular, the squares of the OPE coefficients and their absolute squares need not coincide, since some coefficients are imaginary in our convention. For the $SU(2)$ WZW calculations, the affine level is $k=q-2$, and the representation labels are related to ours by $r=2\ell+1$ and $\bar r=2\bar\ell+1$. Their quantum-group parameter is $e^{i\pi/q}$, so these calculations can be compared with our formal $p=1$ data after the respective chiral factors are removed. The unitary minimal-model results at $p=q-1$ are related to this reference case by the phase and conjugation transformations of Section~\ref{subsec:finite_p_reduction}. We keep these convention changes explicit when comparing signs.

We begin with the WZW calculations. Douglas and Trivedi calculated the reduced OPE coefficients of the $E_7$ $SU(2)$ WZW model~\cite{Douglas:1988rv}. Their Table~1 gives the squares of these coefficients after removing diagonal Fateev--Zamolodchikov factors and fixing the normalization of the $SU(2)$ invariant tensors. We find exact agreement with all $26$ printed entries at our reference value $(p,q)=(1,18)$. The agreement includes negative squares, such as $\alpha_{(1,4)(1,4)(1,4)}^2=-\sqrt{2}$, as well as the vanishing coefficients. Using their reflection and left--right relations, together with identity normalization and the chiral fusion rules, the table determines all $17^3$ ordered squared coefficients. Every entry agrees with our result. This comparison determines the algebraic squares, but does not independently determine the relative signs of the unsquared coefficients.

Kato and Kitazawa studied the same $E_7$ WZW theory and expressed its squared coefficients in terms of diagonal WZW coefficients~\cite{Kato:1988ct}. Their normalization differs from that of Douglas and Trivedi: the diagonal coefficient used as a reference in a given row is not always the geometric mean of the left- and right-moving diagonal coefficients. After accounting for this difference, all $23$ ratios in their Table~1 agree exactly with our data, and the corresponding numerical values agree to the precision printed in the table. Thus the two $E_7$ calculations provide consistent checks of our magnitudes and zeros. We do not identify their positive squared quantities directly with $\alpha_{IJK}^2$ when the latter is negative in our convention.

Fuchs developed a method for extracting non-diagonal OPE coefficients from monodromy-invariant four-point functions~\cite{Fuchs:1989prl}. The explicit $E_6$ and $E_7$ WZW coefficients were given by Fuchs and Klemm~\cite{Fuchs:1989kz}. After removing the diagonal $A$-series factors as in their Eq.~(6.15), all six nonzero $E_6$ relative squares in Eq.~(6.16) and all twenty nonzero $E_7$ relative squares in Eq.~(6.23) agree exactly with $|\alpha_{IJK}|^2$ at $p=1$. Their one $E_6$ and three $E_7$ displayed zeros also agree.

For $E_8$, all $32$ relative squared magnitudes in Table~1 of Fuchs, Klemm, and Scheich~\cite{Fuchs:1989zp}, including its two zeros, agree exactly with our $p=1$ data in the same diagonal normalization. These comparisons concern magnitudes and zeros; we have not matched the unsquared coefficients in a common convention.

For unitary minimal models, Furlan, Ganchev, and Petkova calculated the $(A_{10},E_6)$ theory, corresponding to $(p,q)=(11,12)$ in our notation~\cite{Furlan:1989ra}. Their exceptional Kac label is the second one, whereas it is the first one here. After exchanging these labels and matching normalizations, the nine reduced-square relations in their Eq.~(5.4a) agree exactly with the squared magnitudes of our reduced coefficients. Their factorization formula, Eq.~(5.4c), and its symmetry relations reproduce the complete set of $12^3$ ordered magnitudes and zeros. The paper also studies larger fermionic and quasilocal algebras, whose additional fields are not part of the local spectrum considered here. Our comparison concerns the local algebra and its magnitudes. We have not completed the conversion of all their spinning-field signs; some of these signs are affected by the correction to the fusion-matrix convention discussed in Appendix~B of Ref.~\cite{Petkova:1994zs}.

Petkova and Zuber provide the most extensive comparison with the unitary exceptional minimal models~\cite{Petkova:1994zs}. They normalize the two-point function of a field $A$ to $(-1)^{S_A}$ and define relative coefficients $d_{AB}^{C}$ by removing positive diagonal $A$-series factors. Their scalar coefficients are the structure constants $M_{abc}$ of the Pasquier algebra. After putting the two sets of coefficients in the same quantum-group convention, we find agreement with all their scalar matrices. For $E_8$, the difference between our scalar signs and their nonnegative convention is removed by changing the signs of the scalar fields with $r=17,19,23,29$. Writing $A=(a,\bar a)$, $B=(b,\bar b)$, and $C=(c,\bar c)$ in first Kac labels, their factorization formula $|d_{AB}^{C}|^2=M_{abc}M_{\bar a\bar b\bar c}$ then reproduces every $E_6$ and $E_8$ magnitude and zero. For $E_7$, where this factorization does not hold in general, their Eqs.~(2.18)--(2.19) and the reflection relations likewise reproduce the complete squared-magnitude tensor.

The $E_6$ and $E_7$ comparisons with Petkova and Zuber also include all relative signs. We match their fusion-matrix convention using Eq.~(B.3) of their paper and convert the spin-dependent two-point metric to our normalization. For $E_6$, this conversion includes an additional phase for couplings containing two fields with half-integer chiral spins. Independently reconstructing their coefficients from the scalar matrices, identity normalization, and Eqs.~(2.5), (2.13), and (2.20)--(2.27), we find agreement for all $12^3$ and $17^3$ ordered entries, respectively, after one consistent field-sign redefinition in each theory. This verifies the complete signed $E_6$ and $E_7$ relative OPE coefficients in the common reference sector. For $E_8$, the authors report determining the remaining spinning-field signs, but do not print the full list. Our direct comparison with that paper therefore establishes the scalar signs and all magnitudes and zeros, rather than an entry-by-entry comparison of the complete signed $E_8$ tensor.

More recently, Nivesvivat and Ribault obtained explicit signed coefficients for the full $E_6$ family~\cite{Nivesvivat:2025odb}. Their reduced coefficients use a different reference factor and two-point normalization from ours. We therefore restore both reference factors and compare coefficients in the same two-point normalization. We checked all $18$ representatives with $s_I=s_J=s_K=1$ at $p=5,7,11,13,17,19,23,25$, with $q=12$. These physical models cover every coprime residue modulo $24$. All $144$ normalized-square comparisons agree numerically, with relative differences below $10^{-60}$.\footnote{We thank Rongvoram Nivesvivat for clarifying the parity conventions in email correspondence.} We also compared the signs of all $18$ representatives at the two unitary values $p=11,13$. In each case, one consistent choice of operator signs relates their coefficients to ours. The additional zeros agree for both reference cases $p=1,5$, and a separate spinning example with second Kac labels equal to three agrees after the same normalization procedure. These checks include nonunitary models, although the complete signed comparison was performed only at the two unitary values specified above. 

The earlier WZW and unitary minimal-model calculations thus agree with our data in their common reference sector. These comparisons do not independently test the other nonunitary $E_7$ and $E_8$ residues, or certify all crossing equations exactly. The comparison with Nivesvivat and Ribault supplies an additional check across all independent $E_6$ residues, with the numerical and signed scopes stated above.
	
	\newpage
	\section{An obstruction to a parity-invariant sign gauge}\label{app:parity_counter_example}
    A parity-odd scalar can force a relative sign between OPE coefficients
    that cannot be removed by operator-sign redefinitions. To illustrate this,
    consider three decoupled free bosons and $U(1)$ currents
    normalized as
    \begin{equation}
    	J_a(z)J_b(0)\sim\frac{\delta_{ab}}{z^2},
    	\qquad
    	T(z)=\frac12\sum_{a=1}^3:J_a^2:(z),
    \end{equation}
    with analogous antiholomorphic expressions. Compactness is not essential for this argument: the $U(1)$ currents here come from the shift symmetries and involve only derivative fields, whose local correlators are independent of the compactification radii. See
    \cite{DiFrancesco:1997nk} for the free-boson OPE formalism.
    Spatial parity exchanges the holomorphic and antiholomorphic currents
    and satisfies $P^2=1$.
    
    Define the normalized chiral fields
    \begin{equation}
    	V=J_1,\qquad
    	S=\frac{:J_1^2:-:J_2^2:}{2},\qquad
    	U=:J_1J_2:,\qquad
    	W=\frac{:(J_1+J_2)J_3:}{\sqrt2}.
    \end{equation}
    They are Virasoro primaries of weights
    $h_V=1$ and $h_S=h_U=h_W=2$.
    In particular, the quadratic fields are primary because their traceless
    coefficient matrices eliminate the fourth-order pole in their OPE with
    $T$. They have unit two-point coefficients, and $\langle SU\rangle=0$.
    The bulk fields
    \begin{equation}
    	A=\frac{S\bar U-U\bar S}{\sqrt2},\qquad
    	B=V\bar W,\qquad
    	\widetilde B=W\bar V
    \end{equation}
    are therefore normalized Virasoro primaries of weights
    $(2,2)$, $(1,2)$, and $(2,1)$, respectively.
    Here bars denote antiholomorphic counterparts.
    Their parity transformations, with the reflected coordinates suppressed,
    are
    \begin{equation}
    	P\cdot A=-A,\qquad
    	P\cdot B=\widetilde B,\qquad
    	P\cdot \widetilde B=B.
    \end{equation}
    Thus $A$ is a parity-odd scalar, and the underlying parity transformation of field
    labels is $P(A)=A$, $P(B)=\widetilde B$, and $P(\widetilde B)=B$.
    
    Writing $c_{XYZ}$ for the normalized chiral three-point coefficients,
    Wick contractions give
    \begin{equation}
    	c_{SVV}=c_{UWW}=1,\qquad
    	c_{UVV}=c_{SWW}=0.
    \end{equation}
    Left-right factorization then yields
    \begin{align}
    	C_{ABB}
    	&=\frac{c_{SVV}c_{UWW}-c_{UVV}c_{SWW}}{\sqrt2}
    	=\frac1{\sqrt2},\\
    	C_{A\widetilde B\widetilde B}
    	&=\frac{c_{SWW}c_{UVV}-c_{UWW}c_{SVV}}{\sqrt2}
    	=-\frac1{\sqrt2}.
    \end{align}
    The relative minus sign is required by parity symmetry.
    Under the residual operator-sign transformations
    $A\mapsto\varepsilon_A A$,
    $B\mapsto\varepsilon_B B$, and
    $\widetilde B\mapsto\varepsilon_{\widetilde B}\widetilde B$,
    with $\varepsilon_A,\varepsilon_B,\varepsilon_{\widetilde B}\in\{\pm1\}$,
    both coefficients acquire the same factor $\varepsilon_A$.
    Consequently,
    \begin{equation}
    	\frac{C'_{ABB}}{C'_{A\widetilde B\widetilde B}}=-1
    \end{equation}
    in every such gauge, so the equality
    $C'_{IJK}=C'_{P(I)\,P(J)\,P(K)}$ cannot hold for these fixed
    primary labels.\footnote{The free-boson theory has degenerate Virasoro
    	primaries, so this statement concerns the specified primary basis modulo
    	operator-sign redefinitions. Mixing degenerate primaries is an additional
    	freedom: for example, $S\bar U$ and $U\bar S$ are exchanged by parity,
    	whereas their antisymmetric combination is parity odd.}

    \section{Ocneanu fusion algebra for $E$ series}
    \label{app:ocneanu-E}
    
    This appendix gives explicit eigenvalues of all Virasoro-preserving
    simple topological defects in the $E$-series minimal models. We evaluate
    the Ocneanu fusion algebras constructed in
    \cite[Section~7 and Appendix~B]{PetkovaZuber:2001Ocneanu}, using the
    minimal-model factorization of \cite[Eq.~(5.4)]{Petkova:2000ip}.
    The formulas below express these data in our Kac-label and normalization
    conventions. They also determine the generalized Pasquier algebra
    entering \eqref{eq:ocn-ward-reduced}.
    
    \subsection{Labels, normalization, and the second Kac factor}
    
    Let $q=12,18,30$ for $G=E_6,E_7,E_8$, respectively, and let $p>1$
    be coprime to $q$. Since $p$ is odd, we use the representatives
    \begin{equation}
    	A=(r,\bar r,s),\qquad (r,\bar r)\in\mathcal I_G,
    	\qquad s=1,3,\ldots,p-2.
    	\label{eq:ocn-app-bulk-labels}
    \end{equation}
    Here $A$ denotes the bulk primary
    $V_{(r,s)(\bar r,s)}$. Table~\ref{tab:ocn-app-labels} specifies
    $\mathcal I_G$ and the exceptional defect labels $\beta$.
    The full simple-defect label is $\alpha=(\beta,t)$, with
    $t=1,3,\ldots,p-2$. Thus there are
    $|\mathcal I_G|(p-1)/2$ simple defects and the same number of bulk
    primaries. All these bulk representations occur with multiplicity one.
    
    \begin{table}[ht]
    	\centering
    	\small
    	\renewcommand{\arraystretch}{1.25}
    	\begin{tabular}{cclc}
    		\toprule
    		$G$ & $q$ & $\mathcal I_G$ & exceptional defect labels $\beta$\\
    		\midrule
    		$E_6$ & $12$ &
    		$\{1,7\}^2\sqcup\{4,8\}^2\sqcup\{5,11\}^2$
    		& $(a,b)\in\{1,\ldots,6\}\times\{1,2\}$\\
    		$E_7$ & $18$ &
    		$\{1,17\}^2\sqcup\{5,13\}^2\sqcup\{7,11\}^2$
    		& $\beta=1,\ldots,17$\\[-2pt]
    		&& $\quad\sqcup\{(3,9),(15,9),(9,3),(9,15),(9,9)\}$ &\\
    		$E_8$ & $30$ &
    		$\{1,11,19,29\}^2\sqcup\{7,13,17,23\}^2$
    		& $(a,b)\in\{1,\ldots,8\}\times\{7,6,5,8\}$\\
    		\bottomrule
    	\end{tabular}
    	\caption{Bulk first-label pairs and exceptional simple-defect labels.
    		Here $B^2$ means the Cartesian product $B\times B$.
    		The exceptional identity labels are $(1,1)$, $1$, and $(7,7)$,
    		respectively.}
    	\label{tab:ocn-app-labels}
    \end{table}
    
    In a multiplicity-free theory, a defect loop acts on a Virasoro primary field $A$ by
    \begin{equation}
    	\mathcal L_{(\beta,t)}[A]
    	=\lambda_{(\beta,t)}(A)A,
    	\qquad d_{(\beta,t)}=\lambda_{(\beta,t)}(\mathbf1).
    	\label{eq:ocn-app-action}
    \end{equation}
    The same eigenvalue acts on all Virasoro descendants of $A$. Define
    \begin{equation}
    	x_r=-2\cos\frac{\pi pr}{q},\qquad
    	S_t(s)=\frac{\sin(\pi qst/p)}{\sin(\pi qs/p)}.
    	\label{eq:ocn-app-variables}
    \end{equation}
    The complete action factorizes as
    \begin{equation}
    	\boxed{\lambda_{(\beta,t)}(r,\bar r,s)
    		=L_\beta^{G,p}(r,\bar r)\,S_t(s).}
    	\label{eq:ocn-app-factorization}
    \end{equation}
    Consequently, the fusion rule of the defects factorizes
    \begin{equation}
    	\begin{split}
    		\mathcal L_{(\beta_1,t_1)}\mathcal L_{(\beta_2,t_2)}=\sum_{\beta_3,t_3}N^G_{\beta_1\,\beta_2}{}^{\beta_3}N_{t_1\,t_2}{}^{t_3}\,\mathcal L_{(\beta_3,t_3)}.
    	\end{split}
    \end{equation}
    where $N_{t_1\,t_2}{}^{t_3}\equiv N^{A_{p-1}}_{t_1\,t_2}{}^{t_3}\in\{0,1\}$ is the structure constant of the diagonal minimal model.
    
    In particular, if $d_\beta=L_\beta^{G,p}(1,1)$, then
    \begin{equation}
    	d_{(\beta,t)}=d_\beta
    	\frac{\sin(\pi qt/p)}{\sin(\pi q/p)}.
    	\label{eq:ocn-app-dimensions}
    \end{equation}
    These are the defect eigenvalues themselves; their relation to the
    rescaled matrix $\Psi$ is given below.
    
    \subsection{The first Kac factor}
    \subsubsection{The $E_6$ eigenvalues}
    
    Number the $E_6$ graph by the chain $1$--$2$--$3$--$4$--$5$ and the extra
    edge $3$--$6$. The graph-fusion identity is vertex $1$. With
    $\epsilon_r=(-1)^{r+1}$, define
    \begin{equation}
    	\boldsymbol v^{(6)}(r)=
    	\bigl(1,\ x_r,\ x_r^2-1,\ \epsilon_rx_r,\
    	\epsilon_r,\ x_r-x_r^{-1}\bigr).
    	\label{eq:ocn-app-E6-vector}
    \end{equation}
    All twelve exceptional eigenvalues are
    \begin{equation}
    	\boxed{L_{(a,b)}^{E_6,p}(r,\bar r)
    		=v_a^{(6)}(r)\,x_{\bar r}^{\,b-1},\qquad
    		a=1,\ldots,6,\quad b=1,2.}
    	\label{eq:ocn-app-E6-action}
    \end{equation}
    Here we have used the defect basis of
    \cite{PetkovaZuber:2001Ocneanu}.

    \subsubsection{The $E_7$ eigenvalues}
    
    We use the 17 defect labels of
    \cite{PetkovaZuber:2001Ocneanu}. The identity is
    $\beta=1$, and the left and right fundamental defects are
    $\beta=2$ and $\beta=11$. Define
    \begin{equation}
    	\phi_n(r)=(-1)^{n-1}
    	\frac{\sin(n\pi pr/18)}{\sin(\pi pr/18)},
    	\qquad d_n=\phi_n(r),\quad e_n=\phi_n(\bar r).
    	\label{eq:ocn-app-E7-chi}
    \end{equation}
    Equivalently, $\phi_n(r)=U_{n-1}(x_r/2)$, where $U_n$ is a
    Chebyshev polynomial of the second kind, or
    \begin{equation}
    	\phi_0(r)=0,\qquad \phi_1(r)=1,\qquad
    	\phi_{n+1}(r)=x_r\phi_n(r)-\phi_{n-1}(r).
    \end{equation}
    All exceptional eigenvalues $L_\beta^{E_7,p}(r,\bar r)$ are specified by the compact table
    \begin{equation}
    	\boxed{\begin{array}{c|l}
    		\beta&L_\beta^{E_7,p}(r,\bar r)\\\hline
    		1-8&d_\beta\\
    		9&e_3\\
    		10&d_9-e_3\\
    		11-14&e_2d_{\beta-10}\\
    		15&e_2d_5-e_4\\
    		16&e_2d_6-d_2e_4\\
    		17&e_4
    	\end{array}}
    	\label{eq:ocn-app-E7-table}
    \end{equation}

    \subsubsection{The $E_8$ eigenvalues}
    
    The $E_8$ case has 32 defects in the first Kac factor. Set
    \begin{equation}
    	f_r=2\cos\frac{\pi pr}{5}\in\left\{\frac{1+\sqrt5}{2},\frac{1-\sqrt5}{2}\right\},
    	\qquad f_r^2=f_r+1,
    	\label{eq:ocn-app-E8-f}
    \end{equation}
    and define
    \begin{equation}
    	\begin{split}
    		\boldsymbol v^{(8)}(r)=\bigl(&f_r,\ f_rx_r,\ f_r(x_r^2-1),\
    		x_r(x_r^2-2),\ x_r^2-1,\ x_r,\ 1,\ f_r(x_r-x_r^{-1})\bigr).
    	\end{split}
    	\label{eq:ocn-app-E8-vector}
    \end{equation}
    The 32 eigenvalues are \cite{PetkovaZuber:2001Ocneanu}
    \begin{equation}
    	\boxed{L_{(a,b)}^{E_8,p}(r,\bar r)
    		=v_a^{(8)}(r)v_b^{(8)}(\bar r),\qquad
    		a=1,\ldots,8,\quad b\in\{7,6,5,8\}.}
    	\label{eq:ocn-app-E8-action}
    \end{equation}
    Thus an $8\times8$ table of $v_a^{(8)}(r)$, evaluated at
    $r=1,7,11,13,17,19,23,29$, determines the full exceptional table.
    For example, $(7,7)$ is the identity, while $(6,7)$ and $(7,6)$
    have eigenvalues $x_r$ and $x_{\bar r}$. The defect $(1,7)$ has
    eigenvalue $f_r$ and satisfies the Fibonacci fusion rule
    $\mathcal L_{(1,7)}^2=\mathcal L_{(7,7)}+\mathcal L_{(1,7)}$
    in the exceptional factor.

    \subsection{Generalized Pasquier algebra}
    
    We calculate the structure constants of the generalized Pasquier algebra from
    \eqref{eq:ocn-generalized-pasquier},
    \begin{equation}
    	\widetilde M_{AB}{}^C
    	=\sum_\alpha\Psi_A(\alpha)\Psi_B(\alpha)
    	(\Psi^{-1})_\alpha{}^C.
    	\label{eq:ocn-app-M-start}
    \end{equation}
    Write $A=(J,s)$, with $J=(r,\bar r)\in\mathcal I_G$, and
    $\alpha=(\beta,t)$, with $s,t\in\mathcal O_p:=\{1,3,\ldots,p-2\}$.
    To substitute $\Psi_A(\alpha)=\rho_A\lambda_\alpha(A)/d_\alpha$
    from \eqref{eq:ocn-ward-reduced}, separate \eqref{eq:ocn-rho} as
    \begin{equation}
    	\rho_{(J,s)}=R_J(p)[s]_{q_-},\qquad
    	R_J(p)=(-1)^{(r+\bar r-2)/2}
    	[r]_{q_+}^{1/2}[\bar r]_{q_+}^{1/2},
    	\label{eq:ocn-app-R}
    \end{equation}
    where $q_+=e^{i\pi p/q}$ and $q_-=e^{i\pi q/p}$.
    Using the eigenvalues \eqref{eq:ocn-app-factorization} and
    the dimensions \eqref{eq:ocn-app-dimensions}, we obtain
    \begin{equation}
    	\begin{aligned}
    		\Psi_{(J,s)}((\beta,t))
    		&=R_J(p)[s]_{q_-}
    		\frac{L_\beta^{G,p}(J)}{d_\beta}
    		\frac{\sin(\pi qst/p)}{\sin(\pi qs/p)}
    		\frac{\sin(\pi q/p)}{\sin(\pi qt/p)}\\
    		&=R_J(p)\frac{L_\beta^{G,p}(J)}{d_\beta}
    		\frac{\sin(\pi qst/p)}{\sin(\pi qt/p)}
    		=\Xi_J(\beta;p)\,\mathcal C_s(t),
    	\end{aligned}
    	\label{eq:ocn-app-Psi}
    \end{equation}
    where
    \begin{equation}
    	\Xi_J(\beta;p):=R_J(p)\frac{L_\beta^{G,p}(J)}{d_\beta},
    	\qquad
    	\mathcal C_s(t):=\frac{\sin(\pi qst/p)}{\sin(\pi qt/p)}.
    	\label{eq:ocn-app-factor-characters}
    \end{equation}
    Thus $\Psi=\Xi\otimes\mathcal C$ and
    $\Psi^{-1}=\Xi^{-1}\otimes\mathcal C^{-1}$.
    Substituting into \eqref{eq:ocn-app-M-start} gives
    \begin{equation}
    	\boxed{\widetilde M_{(J_1,s_1)(J_2,s_2)}{}^{(J_3,s_3)}
    		=M^{G,p}_{J_1J_2}{}^{J_3}\,N_{s_1s_2}{}^{s_3},}
    	\label{eq:ocn-app-full-M}
    \end{equation}
    with
    \begin{align}
    	M^{G,p}_{J_1J_2}{}^{J_3}
    	&=\sum_\beta\Xi_{J_1}(\beta;p)\Xi_{J_2}(\beta;p)
    	(\Xi^{-1})_\beta{}^{J_3},
    	\label{eq:ocn-app-first-factor-definition}\\
    	N_{s_1s_2}{}^{s_3}
    	&=\sum_{t\in\mathcal O_p}\mathcal C_{s_1}(t)\mathcal C_{s_2}(t)
    	(\mathcal C^{-1})_t{}^{s_3}.
    	\label{eq:ocn-app-second-factor-definition}
    \end{align}
    We evaluate these two factors in turn.
    
    \paragraph{First Kac factor.}
    The explicit exceptional eigenvalues obey
    \begin{equation}
    	\sum_\beta L_\beta^{G,p}(J)L_\beta^{G,p}(K)
    	=\frac{q\,\delta_{JK}}
    	{2\sin(\pi pr/q)\sin(\pi p\bar r/q)},
    	\qquad J=(r,\bar r).
    	\label{eq:ocn-app-character-orthogonality}
    \end{equation}
    Since $R_J(p)^2=[r]_{q_+}[\bar r]_{q_+}$, this implies
    \begin{equation}
    	\begin{gathered}
    		w_\beta(p):=\frac{2\sin^2(\pi p/q)}q\,d_\beta^2,\qquad
    		\sum_\beta w_\beta(p)\Xi_J(\beta;p)\Xi_K(\beta;p)=\delta_{JK},\\
    		(\Xi^{-1})_\beta{}^J=w_\beta(p)\Xi_J(\beta;p).
    	\end{gathered}
    	\label{eq:ocn-app-first-inverse}
    \end{equation}
    Consequently,
    \begin{equation}
    	\boxed{M^{G,p}_{J_1J_2}{}^{J_3}
    		=\frac{2\sin^2(\pi p/q)}q
    		\prod_{i=1}^3R_{J_i}(p)
    		\sum_\beta
    		\frac{\prod_{i=1}^3L_\beta^{G,p}(J_i)}{d_\beta}.}
    	\label{eq:ocn-app-M}
    \end{equation}
    
    The prescribed continuation fixes the square roots in $R_J$.
    With $\nu_p(n):=\lfloor pn/q\rfloor-\lfloor p/q\rfloor$, it gives
    \begin{equation}
    	[n]_{q_+}^{1/2}
    	=e^{-i\pi\nu_p(n)/2}
    	\sqrt{\left|\frac{\sin(\pi pn/q)}{\sin(\pi p/q)}\right|}.
    	\label{eq:ocn-app-root}
    \end{equation}
    For $J\in\mathcal I_G$, the integers $\nu_p(r)$ and
    $\nu_p(\bar r)$ have the same parity. Hence
    \begin{equation}
    	R_J(p)=
    	(-1)^{[r+\bar r-2-\nu_p(r)-\nu_p(\bar r)]/2}
    	\frac{\sqrt{|\sin(\pi pr/q)\sin(\pi p\bar r/q)|}}
    	{|\sin(\pi p/q)|}
    	\label{eq:ocn-app-R-magnitude-phase}
    \end{equation}
    is real. Substituting this expression and the explicit
    $L_\beta^{G,p}$ into \eqref{eq:ocn-app-M} yields
    \begin{equation}
    	\boxed{M^{G,p}_{J_1J_2}{}^{J_3}
    		=\mu_p^G(J_1)\mu_p^G(J_2)\mu_p^G(J_3)
    		T^G_{J_1J_2J_3},\qquad T^G_{J_1J_2J_3}\geq0.}
    	\label{eq:ocn-app-M-sign}
    \end{equation}
    Here $T^G$ is independent of $p$, and the signs $\mu_p^G(J)$
    include both the continuation phases and the signs of the finite
    sum in \eqref{eq:ocn-app-M}.
    
    For $G=E_6,E_8$, the absolute values are
    \begin{equation}
    	T^G_{J_1J_2J_3}
    	=\left|M^G_{r_1r_2r_3}
    	M^G_{\bar r_1\bar r_2\bar r_3}\right|,
    	\qquad J_i=(r_i,\bar r_i),
    	\label{eq:ocn-app-T-E68}
    \end{equation}
    where $M^G_{rst}$ is the graph Pasquier tensor
    \eqref{eq:alpha-intro-M}.
    For $E_7$, use the blocks $\id,B,C,D$ of
    Section~\ref{subsec:alpha-intro-e7}. Let $\mathfrak a$ denote
    the sorted triple of block labels, and let $n$ count the occurrences
    of $(3,9),(15,9),(9,3),(9,15)$ among $J_1,J_2,J_3$. Define
    \begin{equation}
    	\eta=\left(\sum_{i=1}^3\mathbf1_{r_i\ne\bar r_i}\right)\bmod2,
    	\qquad
    	P=\left|M^{E_7}_{r_1r_2r_3}
    	M^{E_7}_{\bar r_1\bar r_2\bar r_3}\right|\quad(n=0).
    \end{equation}
    For triples allowed by both chiral fusion rules, the result is
    \begin{equation}
    	T^{E_7}_{J_1J_2J_3}=
    	\begin{cases}
    		(1+2\eta)P,&n=0,\ \mathfrak a=BCC,\\
    		(1-\eta)P,&n=0,\ \mathfrak a=CCC,\\
    		P,&n=0,\ \mathfrak a\ne BCC,CCC,\\
    		0,&n=1,\ \mathfrak a=DDD,\\
    		1/\sqrt2,&n=1,\ \mathfrak a\ne DDD,\\
    		1/2,&n=2,\ \mathfrak a=CDD,\\
    		1,&n=2,\ \mathfrak a\ne CDD,\\
    		\sqrt2,&n=3.
    	\end{cases}
    	\label{eq:ocn-app-T-E7}
    \end{equation}
    It is zero otherwise. Explicitly, the two fusion conditions are
    $|r_1-r_2|+1\leq r_3\leq\min(r_1+r_2-1,35-r_1-r_2)$
    and the same inequalities for $\bar r_i$.
    
    To give the signs for arbitrary $p$, write
    \begin{equation}
    	p=kq\pm w,\qquad 0<w<q/2,\qquad
    	h_J=\left\lfloor\frac{r-1}{2}\right\rfloor
    	+\left\lfloor\frac{\bar r-1}{2}\right\rfloor.
    \end{equation}
    The possible $w$ are $\{1,5\}$, $\{1,5,7\}$, and
    $\{1,7,11,13\}$ for $E_6,E_7,E_8$, respectively, and
    \begin{equation}
    	\mu_{kq\pm w}^G(J)=(-1)^{kh_J}\mu_w^G(J).
    	\label{eq:ocn-app-mu-period}
    \end{equation}
    For $E_6$, both reference signs are $(-1)^{h_J}$, so
    \begin{equation}
    	\mu_p^{E_6}(J)=(-1)^{(k+1)h_J},\qquad
    	k=\left\lfloor\frac{p+6}{12}\right\rfloor.
    	\label{eq:ocn-app-mu-E6}
    \end{equation}
    For $E_7$, the following table lists every pair with
    $\mu_w^{E_7}(J)=-1$; all other pairs have sign $+1$:
    \begin{equation}
    	\begin{array}{c|l}
    		w&J\text{ with }\mu_w^{E_7}(J)=-1\\\hline
    		1&(3,9),\ (15,9),\ (9,3),\ (9,15)\\
    		5&(5,13),\ (13,5),\ (7,11),\ (11,7)\\
    		7&(1,17),\ (17,1),\ (7,11),\ (11,7)
    	\end{array}
    	\label{eq:ocn-app-mu-E7}
    \end{equation}
    For $E_8$, the signs factorize as
    $\mu_w^{E_8}(r,\bar r)=\zeta_w(r)\zeta_w(\bar r)$, where $\zeta_w(r)$ is given in the following table
    \begin{equation}
    	\begin{array}{c|rrrrrrrr}
    		w\backslash r&1&7&11&13&17&19&23&29\\\hline
    		1 &+&+&-&-&-&-&+&+\\
    		7 &+&+&+&+&-&+&-&+\\
    		11&+&+&+&-&+&-&-&-\\
    		13&+&+&+&-&-&-&+&-
    	\end{array}
    	\label{eq:ocn-app-mu-E8}
    \end{equation}
    Thus every nonzero first-factor structure constant has absolute
    value $T^G_{J_1J_2J_3}$ and phase $0$ or $\pi$, according to
    the product of the three signs in \eqref{eq:ocn-app-M-sign}.
    
    \paragraph{Second Kac factor.}
    Set $\theta=\pi q/p$. Since $q$ is even and coprime to the odd
    integer $p$, discrete sine orthogonality on $\mathcal O_p\equiv\{1,3,\ldots,p-2\}$ gives
    \begin{equation}
    	\sum_{t\in\mathcal O_p}\sin(st\theta)\sin(ut\theta)
    	=\frac p4\delta_{su},\qquad s,u\in\mathcal O_p.
    \end{equation}
    Equivalently,
    \begin{equation}
    	\begin{gathered}
    		w_t:=\frac4p\sin^2(t\theta),\qquad
    		\sum_{t\in\mathcal O_p}w_t\mathcal C_s(t)\mathcal C_u(t)
    		=\delta_{su},\\
    		(\mathcal C^{-1})_t{}^s=w_t\mathcal C_s(t).
    	\end{gathered}
    	\label{eq:ocn-second-orthogonality}
    \end{equation}
    Substitution into \eqref{eq:ocn-app-second-factor-definition} yields
    \begin{equation}
    	N_{s_1s_2}{}^{s_3}
    	=\frac4p\sum_{t\in\mathcal O_p}
    	\frac{\sin(s_1t\theta)\sin(s_2t\theta)\sin(s_3t\theta)}
    	{\sin(t\theta)}.
    	\label{eq:ocn-second-cubic-sum}
    \end{equation}
    To evaluate this sum, use
    $\mathcal C_s(t)=U_{s-1}(\cos(t\theta))$ and the Chebyshev
    product identity
    \begin{equation}
    	\mathcal C_{s_1}(t)\mathcal C_{s_2}(t)
    	=\sum_{\substack{n=|s_1-s_2|+1\\ \text{step }2}}^{s_1+s_2-1}
    	\mathcal C_n(t).
    \end{equation}
    At the present evaluation points,
    $\mathcal C_p(t)=0$ and
    $\mathcal C_{2p-n}(t)=-\mathcal C_n(t)$.
    Terms above $p$ therefore cancel their reflected partners,
    leaving the truncated product
    \begin{equation}
    	\mathcal C_{s_1}(t)\mathcal C_{s_2}(t)
    	=\sum_{s_3\in\mathcal O_p}
    	N_{s_1s_2}{}^{s_3}\mathcal C_{s_3}(t),
    	\label{eq:ocn-second-product}
    \end{equation}
    with
    \begin{equation}
    	\boxed{N_{s_1s_2}{}^{s_3}=
    		\begin{cases}
    			1,&|s_1-s_2|+1\leq s_3\leq
    			\min(s_1+s_2-1,\,2p-s_1-s_2-1),\\
    			0,&\text{otherwise}.
    	\end{cases}}
    	\label{eq:ocn-second-fusion-explicit}
    \end{equation}
    All three labels are odd, so the fusion parity condition is automatic.
    The second Kac factor is therefore exactly the $A$-type fusion
    coefficient, with no additional magnitude or phase. There is no
    square-root ambiguity: the factor $[s]_{q_-}$ in $\rho_A$ cancels
    algebraically in \eqref{eq:ocn-app-Psi}, leaving the signed sine
    ratio $\mathcal C_s(t)$.

	\bibliographystyle{unsrtnat}
	\bibliography{E_series}
	
\end{document}